\documentclass{elsart}

\usepackage{fancybox}
\usepackage[all]{xy}

\usepackage{amsfonts}
\usepackage[dvips]{graphicx}
\usepackage{epsfig}
\usepackage{amsmath,amssymb,latexsym}

\usepackage{float}
\usepackage{makeidx}
\usepackage{graphicx}
\usepackage{epsf}

\newtheorem{theorem}{Theorem}

\newtheorem{lemma}[theorem]{Lemma}

\newenvironment{proof}[1][Proof]{\textbf{#1.} }
     {\    \rule{0.5em}{0.5em}}

\newcommand{\AAA}{\circle*{4}}
\newcommand{\TT}{\mathcal{T}}
\newcommand{\UP}{\line(0,1){10}}
\newcommand{\BIF}{\line(-1,1){10}\line(1,1){10}}
\newcommand{\g}{\mathfrak{g}}
\newcommand{\cG}{\mathcal{G}}
\newcommand{\ad}{\mathrm{ad}}
\newcommand{\Ad}{\mathrm{Ad}}
\newcommand{\Om}{\Omega}
\newcommand{\tH}{\tilde{H}}
\newcommand{\SO}{\ensuremath{\mathrm{SO}}}
\newcommand{\so}{\ensuremath{\mathfrak{so}}}
\newcommand{\SP}{\ensuremath{\mathrm{Sp}}}
\newcommand{\SU}{\ensuremath{\mathrm{SU}}}
\newcommand{\cO}{\mathcal{O}}

\newcommand{\dd}{{\rm d}}

\begin{document}

\begin{frontmatter}



\title{The Magnus expansion and some of its applications}


\author[Poli]{S. Blanes},
\author[Castellon]{F. Casas},
\author[Valencia1]{J.A. Oteo} and
\author[Valencia1,Valencia2]{J. Ros}

\address[Poli]{Instituto de Matem\'atica Multidisciplinar,
Universidad Polit\'ecnica de Valencia, E-46022 Valencia, Spain}
\address[Castellon]{Departament de Matem\`atiques, Universitat Jaume I,
  E-12071 Castell\'on, Spain}
\address[Valencia1]{Departament de F\'{\i}sica Te\`orica,
  Universitat de Val\`encia, E-46100 Burjassot, Valencia, Spain}
\address[Valencia2]{IFIC, Centre Mixt
  Universitat de Val\`encia-CSIC, E-46100 Burjassot, Valencia, Spain}

\begin{abstract}

Approximate resolution of linear systems of differential
equations with varying coefficients is a recurrent problem
shared by a number of scientific and engineering areas, ranging from
Quantum Mechanics to Control Theory. When formulated in operator
or matrix form, the \textsl{Magnus expansion}
furnishes an elegant setting to built up approximate exponential
representations of the solution of the system.
It provides a power series expansion for the corresponding exponent and 
 is sometimes referred to as
\textsl{Time-Dependent
Exponential Perturbation Theory}. Every Magnus approximant corresponds
in Perturbation Theory to a partial re-summation of infinite terms with
the important additional property of preserving 
at any order certain symmetries
of the exact solution.

The goal of this review is threefold. First, to collect a number of
developments
scattered through half a century of scientific literature on Magnus
expansion.
They concern the methods for the generation of terms in the expansion,
estimates
of the radius of convergence of the series, generalizations and related
non-perturbative expansions. Second, to provide a bridge with its
implementation as generator of \textsl{especial purpose numerical
integration methods},
a field of intense activity during the last decade. Third, to
illustrate with
examples the kind of results one can expect from Magnus expansion
in comparison with those from both perturbative schemes and standard
numerical
integrators. We buttress this issue with a revision of the wide range
of physical applications found by Magnus expansion in the literature.

\end{abstract}


\end{frontmatter}

\tableofcontents{}




\section{Introduction}

\subsection{Motivation, overview and history}

The outstanding mathematician Wilhelm Magnus (1907-1990)
did important contributions to a wide variety of fields in
mathematics and mathematical physics \cite{abikoff94tml}. Among them
one can mention combinatorial group theory \cite{magnus76cgt} and
his collaboration in the Bateman project on higher transcendental
functions and integral transforms \cite{erdelyi53htf}. In this
report we review another of his long-lasting constructions: the
so-called Magnus expansion (hereafter referred to as ME). ME was
introduced as a tool to solve non-autonomous
linear differential equations for linear operators. It is
interesting to observe that in his seminal paper of 1954
\cite{magnus54ote}, although it is essentially mathematical in
nature, Magnus recognizes that his work was stimulated by results of
K.O. Friedrichs on the theory of linear operators in Quantum
Mechanics \cite{friedrichs53mao}. Furthermore, as the first
antecedent of his proposal he quotes a paper by R.P. Feynman in the
Physical Review \cite{feynman51aoc}. We stress these facts to show
that already in its very beginning ME was strongly related to
Physics and so has been ever since and there is no reason to doubt
that it will continue to be. This is the first motivation to offer
here a review as the present one.

Magnus proposal has the very attractive property of leading to
approximate solutions which exhibit at any order of approximation
some qualitative or physical characteristics which first principles
guarantee for the exact (but unknown) solution of the problem.
Important physical examples are the symplectic or unitary character
of the evolution operator when dealing with classical or quantum
mechanical problems, respectively. This is at variance with most
standard perturbation theories and is apparent when formulated in a
correct algebraic setting: Lie algebras and Lie groups. But this
great advantage has been some times darkened in the past by the
difficulties both in constructing explicitly higher order terms and
in assuring existence and convergence of the expansion.

In our opinion, recent years have witnessed great improvement in this
situation. Concerning general questions of existence and convergence new
results have appeared. From the point of view of applications some new
approaches in old fields have been published while completely new and
promising avenues have been open by the use of Magnus expansion in Numerical
Analysis. It seems reasonable to expect fruitful cross fertilization between
these new developments and the most conventional perturbative approach to ME
and from it further applications and new calculations.

This new scenario makes desirable for the Physics community in
different areas (and scientists and engineers in general) to have
access in as unified a way as possible to all the information
concerning ME which so far has been treated in very different
settings and has appeared scattered through very different
bibliographic sources.

As implied by the preceding paragraphs this report is mainly
addressed to a Physics audience, or closed neighbors, and
consequently we shall keep the treatment of its mathematical aspects
within reasonable limits and refer the reader to more detailed
literature where necessary. By the same token the applications
presented will be limited to examples from Physics or from the
closely related field of Physical Chemistry. We shall also emphasize
its instrumental character for numerically solving physical
problems.

In the present section as an introduction we present a brief
overview and sketchy history of more than 50 years of ME. To start
with, let us consider the initial value problem associated with
the linear ordinary differential equation
\begin{equation}  \label{eqdif}
Y^{\prime}(t)=A(t)Y(t),\qquad\qquad Y(t_0)=Y_{0},
\end{equation}
where as usual the prime denotes derivative with respect to the real
independent variable which we take as time $t$, although much of what will be said
applies also for a complex independent variable. In order of increasing
complexity we may consider the equation above in different contexts:

\begin{enumerate}
\item[(a)] $Y:\mathbb{R}\longrightarrow\mathbb{C}$, $A:\mathbb{R}\longrightarrow
\mathbb{C}$. This means that the unknown $Y$ and the given $A$ are
complex scalar valued functions of one real variable. In this case
there is no problem at all: the solution reduces to a quadrature
and an ordinary exponential
evaluation:
\begin{equation}   \label{eqescalar}
Y(t)= \exp \left( \int_{t_0}^{t}A(s)ds \right) \, Y_{0}.
\end{equation}
\item[(b)] $Y:\mathbb{R}\longrightarrow\mathbb{C}^{n}$, $A:\mathbb{R}%
\longrightarrow\mathbb{M}_{n}(\mathbb{C)}$, where
$\mathbb{M}_{n}(\mathbb{C)}$ is the set of $n\times n$ complex
matrices. Now $Y$ is a complex vector valued function and $A$ a
complex $n\times n$ matrix valued function. At variance with the
previous case, only in very special cases is the solution easy to
state: when for any pair of values of $t$, $t_{1}$ and $t_{2},$
one has $A(t_{1})A(t_{2})=A(t_{2})A(t_{1})$, which is certainly
the case if $A$ is constant. Then the solution reduces to a
quadrature (trivial or not) and a matrix exponential. With the
obvious changes in the meaning of the symbols, equation
(\ref{eqescalar}) still applies. In the general case, however, there
is not compact expression for the solution and (\ref{eqescalar})
is not any more the solution.

\item[(c)] $Y:\mathbb{R}\longrightarrow\mathbb{M}_{n}(\mathbb{C)}$,
$A:\mathbb{R}\longrightarrow\mathbb{M}_{n}(\mathbb{C)}$. Now both
$Y$ and $A$ are complex matrix valued functions. A particular
case, but still general enough to encompass the most interesting
physical and mathematical applications, corresponds to
$Y(t)\in\mathcal{G}$, $A(t)\in\mathfrak{g}$, where $\mathcal{G}$
and $\mathfrak{g}$ are respectively a matrix Lie group and its
corresponding Lie algebra. Why this is of interest is easy to
grasp: the key reason for the failure of (\ref{eqescalar}) is the
non-commutativity of matrices in general. So one can expect that
the (in general non-vanishing) commutators play an important role.
But when commutators enter the play one immediately thinks in Lie
structures. Furthermore, plainly speaking, the way from a Lie
algebra to its Lie group is covered by the exponential operation.
A fact that will be no surprise in this context. The same comments
of the previous case are valid here. In this report we shall
mostly deal with this matrix case.

\item[(d)] The most general situation one can think of corresponds to $Y(t)$
and $A(t)$ being operators in some space, e.g., Hilbert space in
Quantum Mechanics. Perhaps the most paradigmatic example of
(\ref{eqdif}) in this setting is the time-dependent Schr\"odinger
equation.
\end{enumerate}

Observe that case (b) above can be reduced to case (c). This is
easily seen if one introduces what in mathematical literature is
named the matrizant, a concept dating back at least to the
beginning of the 20th century in the work of Baker
\cite{baker02oti}. It is the $n\times n$ matrix $U(t,t_0)$ defined
through
\begin{equation}
Y(t)=U(t,t_0)Y_{0}.    \label{matrizant}%
\end{equation}
Without loss of generality we will take $t_0=0$ unless otherwise
explicitly stated, for the sake of simplicity. When no confusion
may arise, we write only one argument in $U$ and denote $U(t,0)
\equiv U(t)$, which then
satisfy the differential equation and initial condition%
\begin{equation}
U^{\prime}(t)=A(t)U(t),\qquad\qquad U(0)=I, \label{laecuacion}%
\end{equation}
where $I$ stands for the $n$-dimensional identity matrix. The
reader will have recognized $U(t)$ as what in physical terms is
known as time evolution operator.

We are now ready to state Magnus proposal: a solution to
(\ref{laecuacion})
which is a true matrix exponential%
\begin{equation}
U(t)=\exp \Omega(t),\qquad\qquad \Omega(0)=O,\label{magsol}%
\end{equation}
and a series expansion for the matrix in the exponent%
\begin{equation}
\Omega(t)=\sum_{k=1}^{\infty}\Omega_{k}(t),\label{ME}%
\end{equation}
which is what we call the \emph{Magnus expansion}. The mathematical elaborations explained in the
next section determine $\Omega_{k}(t)$. Here we just write down the three
first terms of that series:%
\begin{align}
\Omega_{1}(t)  & =\int_{0}^{t}A(t_{1})~\text{d}t_{1},\nonumber\\
\Omega_{2}(t)  & =\frac{1}{2}\int_{0}^{t}\text{d}t_{1}\int_{0}^{t_{1}}%
\text{d}t_{2}\ \left[  A(t_{1}),A(t_{2})\right] \label{ME123}\\
\Omega_{3}(t)  & =\frac{1}{6}\int_{0}^{t}\text{d}t_{1}\int_{0}^{t_{1}}%
\text{d}t_{2}\int_{0}^{t_{2}}\text{d}t_{3}\ (\left[  A(t_{1}),\left[
A(t_{2}),A(t_{3})\right]  \right]  +\left[  A(t_{3}),\left[  A(t_{2}%
),A(t_{1})\right]  \right]  )\nonumber
\end{align}
where $\left[  A,B\right]  \equiv AB-BA$ is the matrix commutator
of $A$ and $B$.

The interpretation of these equations seems clear: $\Omega_{1}(t)$
coincides exactly with the exponent in (\ref{eqescalar}). But this
equation cannot give the whole solution as has already been said.
So, if one insists in having an exponential solution the exponent
has to be corrected. The rest of the ME in (\ref{ME}) gives that
correction necessary to keep the exponential form of the solution.

The terms appearing in (\ref{ME123}) already suggest the most
appealing characteristic of ME. Remember that a matrix Lie algebra
is a linear space in which one has defined the commutator as the
second internal composition law. If, as we suppose, $A(t)$ belongs
to a Lie algebra $\mathfrak{g}$ for all $t$ so does any sum of
multiple integrals of nested commutators. Then, if all terms in ME
have a structure similar to that of the ones shown before, the whole
$\Omega(t)$ and any approximation to it obtained by truncation of ME
will also belong to the same Lie algebra. In the next section it
will be shown that this turns out to be the case. And \emph{a fortiori} its
exponential will be in the corresponding Lie group.

Why is this so important for physical applications? Just because
many of the properties of evolution operators derived from first
principles are linked to the fact that they belong to a certain
Lie group: e.g. unitary group in Quantum Mechanics, symplectic
group in Classical Mechanics. In that way use of (truncated) ME
leads to approximations which share with the exact solution of
equation (\ref{laecuacion}) important qualitative (very often,
geometric) properties. For instance, in Quantum Mechanics every approximant
preserves probability conservation.

From the present point of view we could say that the last
paragraphs summarize, in a nut shell, the main contents of the
famous paper of Magnus of 1954. With no exaggeration its
appearance can be considered a turning point in the treatment of
the initial value problem defined by (\ref{laecuacion}).

But important, as it certainly was, Magnus paper left some
problems, at least partially, open:
\begin{itemize}
 \item First, for what values of $t$ and for what operators $A$ does equation
(\ref{laecuacion})  admit a true exponential solution? This we call
the existence problem.
 \item Second, for what values of $t$ and for
what linear operators $A$ does the series in equation (\ref{ME})
converge? This we describe as the convergence problem. We want to
emphasize that, although related, these two are different
problems. To see why, think of the scalar equation $y^{\prime}=y^{2}$
with $y(0)=1$. Its solution $y(t)=(1-t)^{-1}$ exists for $t\neq1$,
but its form in power series $y(t)=\sum_{0}^{\infty}t^{n}$
converges only for $\left| t\right| <1$.
 \item Third, how to construct
higher order terms $\Omega_{k}(t)$, $k\geq3$, in the series? Moreover, is
there a closed-form expression for $\Omega_{k}(t)$?
 \item Fourth, how to calculate in an efficient way
$\exp\Omega^{[N]}$, where $\Omega^{\lbrack N]}\equiv\sum_{k=1}^{N}%
\Omega_{k}(t)$ is a truncation of the ME?
\end{itemize}

All these questions, and many others, will be dealt with in the rest
of the paper. But before entering that analysis we think interesting
to present  a view, however brief, from a historical perspective of
this half a century of developments on Magnus series. Needless to
say, we by no means try to present a detailed and exhaustive
chronological account of the many approaches followed by authors
from very different disciplines. To minimize duplication with later
sections we simply mention some representative samples in order the
reader can understand the evolution of the field.

Including some precedents, and with a, as undeniable as
unavoidable, dose of arbitrariness, we may distinguish four periods in
the history of our topic:

\begin{enumerate}
\item Before 1953. The problem which ME solves has a centennial history
dating back at least to the work of Peano, by the end of 19th
century, and Baker, at the beginning of the 20th (for references
to the original papers see e.g. \cite{ince56ode}). They combine
the theory of differential equations with an algebraic
formulation. Intimately related to these treatments from the very
beginning is the study of the so called Baker--Campbell--Hausdorff
or BCH formula for short
\cite{baker05aac,campbell98oal,hausdorff06dse} which gives $C$ in
terms of $A$, $B$ and their multiply nested commutators when
expressing $\exp(A)\exp(B)$ as $\exp(C)$. This topic has by itself
a long history up until today and will be discussed also in
section 2. As one of its early hallmarks we quote
\cite{dynkin47eot}. In Physics literature the interest in the
problem posed by equation (\ref{laecuacion}) highly revived with
the advent of Quantum Electrodynamics (QED). The works of F.J.
Dyson \cite{dyson49trt} and in particular R.P. Feynman
\cite{feynman51aoc} in the late forties and early fifties are
worth mentioning here.

\item 1953-1970. We have quoted as birth certificate of ME the paper \cite{magnus54ote} by
Magnus in 1954. This is not strictly true: there is a Research
Report \cite{magnus53asi} dated June 1953 which differs from the
published paper in the title and in a few minor details and which
should in fact be taken as a preliminary draft of it. In both
publications appears the result summarized above on ME with almost
identical words. The work of Pechukas and Light
\cite{pechukas66ote} gave for the first time a more specific
analysis of the problem of convergence than the rather vague
considerations in Magnus paper. Wei and Norman
\cite{wei63nog,wei63las}\ did the same for the existence problem.
Robinson, to the best of our knowledge, seems to have been the
first to apply ME to a physical problem \cite{robinson63mce}.
Special mention in this period deserves a paper by Wilcox
\cite{wilcox67eoa}, in which useful mathematical tools are given
and ME presented together with other algebraic treatments of
equation (\ref{laecuacion}), in particular Fer's infinite product
expansion \cite{fer58rdl}. Worth also of mention here is the first
application of ME as a numerical tool for integrating the
time-independent Schr\"odinger equation for potential scattering
by Chang and Light \cite{chang69eso}.

\item 1971-1990. During these years ME consolidated in different fronts. It
was successfully applied to a wide spectrum of fields in Physics
and Chemistry: from atomic \cite{baye73ats}\ and molecular
\cite{schek81aot}\ Physics to Nuclear Magnetic Resonance (NMR)
 \cite{ernst86pon,waugh82tob} to
Quantum Electrodynamics \cite{dahmen82ido} and elementary particle
Physics \cite{dolivo90mea}. A number of case studies also helped
to clarify its mathematical structure, see for example
\cite{klarsfeld89apo}. The construction of higher order terms was
approached from different angles. The intrinsic and growing
complexity of ME allows for different schemes. One which has shown
itself very useful in tackling other questions like the
convergence problem was the recurrent scheme by Klarsfeld and Oteo
\cite{klarsfeld89rgo}.

\item Since 1991. The last decade of the 20th century witnessed a renewed
interest in ME which still continues nowadays. It has followed different
lines. Concerning the basic problems of existence and convergence
of ME there has been definite progress
\cite{blanes98maf,casas07scf,moan99ote,moan01cot}. ME has also been adapted
for specific types of equations: Floquet theory when $A(t)$ is a
periodic function \cite{casas01fte}, stochastic differential
equations \cite{burrage99hso} or equations of the form
$Z^{\prime}=AZ-ZB$ \cite{iserles01ame}. Special mention should be
made to the new field open in this most recent period that uses
Magnus scheme to build novel algorithms \cite{iserles99ots} for
the numerical integration of differential equations within the
most wide field of geometric integration \cite{budd99gin}. After
optimization \cite{blanes00iho,blanes02hoo}, these integrators
have proved to be highly competitive.
\end{enumerate}

As a proof of the persistent impact the 1954 paper by Magnus has had
in scientific literature we present in Figures \ref{cites1} and \ref{cites2} the number
of citations per year and the cumulative number of citations, respectively, 
as December 2007 with data taken from ISI Web of
Science. The original paper appears about 750 times of which,
roughly, 50, 320 and 380 correspond respectively to each of the last
three periods we have considered. The enduring interest in that
seminal paper is clear from the figures.

\begin{figure}[th]
\begin{center}
\epsfxsize=4.6in {\epsfbox{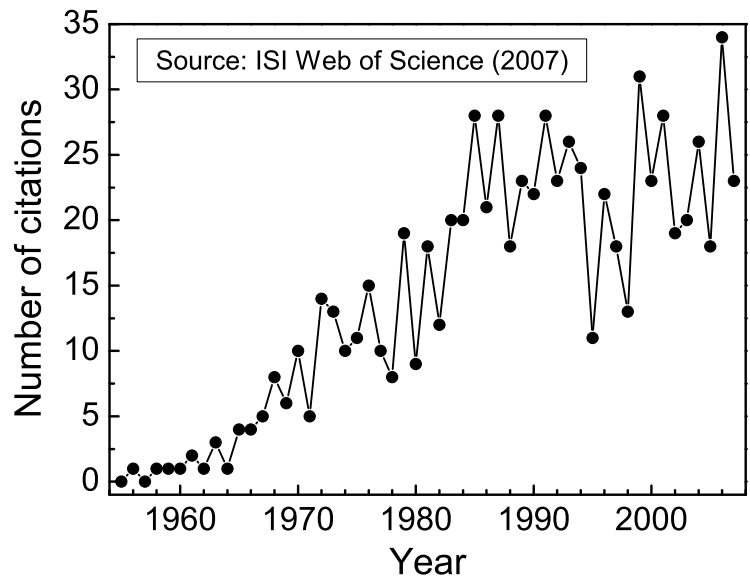}}
\end{center}
 \caption{Persistency of Magnus original paper: number of citations per year.}
 \label{cites1}

\end{figure}

\begin{figure}[th]
\begin{center}
\epsfxsize=4.6in {\epsfbox{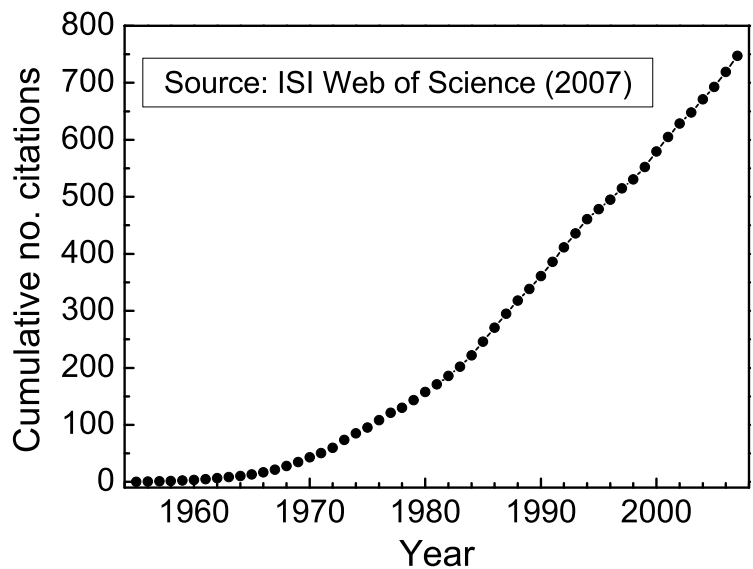}}
\end{center}
 \caption{Persistency of Magnus original paper: cumulative number of citations.}
 \label{cites2}

\end{figure}

The presentation of this report is organized as follows. In the remaining of 
this section we include some mathematical tools and notations that will
be used time and again in our treatment. In section 2 we introduce formally
the Magnus expansion, study its main features and analyze thoroughly 
the convergence issue. Next, in section 3 several generalizations of the
Magnus expansion are reviewed, with special emphasis in its application
to general nonlinear differential equations. In order to illustrate the main
properties of ME, in section 4 we consider simple examples for which the
computations required are relatively straightforward. Section 5 is devoted
to an aspect that has been most recently studied in this setting: the design of new
algorithms for the numerical integration of differential equations based on the
Magnus expansion. There, after a brief characterization of numerical integrators,
we present several methods that are particularly efficient, as shown by the
examples considered. Given the relevance of the new numerical schemes, 
we briefly review in section 6 some of its applications in different
contexts, ranging from boundary-value problems to stochastic differential equations. In
section 7, on the other hand, applications of the ME to significant physical problems
are considered. Finally, the paper ends with some concluding remarks.

\subsection{Mathematical preliminaries and notations} \label{notations}

Here we collect for the reader's convenience some mathematical
expressions, terminology and notations which appear most frequently
in the text. Needless to say that we have made no attempt to being
completely rigorous. We just try to facilitate the casual reading of
isolated sections.

As already mentioned, the natural mathematical habitat for
most of the objects we will deal with in this report is a Lie
group or its associated Lie algebra. Although most of the results
discussed in these pages are valid in a more general setting we will
essentially consider only matrix Lie groups and algebras.

By a Lie group $\mathcal{G}$ we understand a set which combines an
algebraic structure with a topological one. At the algebraic level
every two elements of $\mathcal{G}$ can be combined by an internal
composition law to produce a third element also in $\mathcal{G}$.
The law is required to be associative, to have an identity element
and every element must have an inverse. The ordinary  product and
inverse of invertible matrix play that role in the cases we are
more interested in. The topological exigence forces the
composition law and the association of an inverse to be
sufficiently smooth functions.

A Lie algebra $\mathfrak{g}$ is a vector space whose elements can be
combined by a second law, the Lie bracket, which we represent by
$[A,B]=C$, with $A,B,C$ elements of $\mathfrak{g}$, in such a way
that the law is bilinear, skew-symmetric and satisfies the well
known Jacobi identity,
\begin{equation}
    [A,[B,C]]+[B,[C,A]]+[C,[A,B]]=0.  \label{jacobi}
\end{equation}
When dealing with matrices we take as Lie bracket the familiar
commutator:
\begin{equation}\label{comutador}
    [A,B] = A B - B A, \qquad A\in \mathfrak{g}, \quad B\in
    \mathfrak{g},
\end{equation}
where $AB$ stands for the usual matrix product.
If we consider a finite-di\-men\-sio\-nal Lie algebra with dimension $d$
and denote by $A_i, i=1, \ldots, d$, the vectors of one of its
basis then the fundamental brackets one has to know are
\begin{equation}
    [A_i,A_j]=c_{ij}^{k}A_k,  \label{constestruct}
\end{equation}
where sum over repeated indexes is understood. The coefficients
$c_{ij}^{k}$ are the so-called structure constants of the algebra.

Associated with any $A\in \mathfrak{g}$ we can define a linear
operator $\ad_{A}: \mathfrak{g}\rightarrow \mathfrak{g}$ which
acts according to
\begin{equation}\label{definad}
  \ad_{A}B=[A,B], \qquad    \ad_{A}^j B = [A, \ad_{A}^{j-1} B], \qquad  
  \ad_{A}^0 B=B, \qquad j \in \mathbb{N},  B\in \mathfrak{g}.
\end{equation}
Also of interest is the exponential of this $\ad_A$ operator,
\begin{equation}\label{definAd}
    \Ad_A=\exp(\ad_A),
\end{equation}
whose action on $\mathfrak{g}$ is given by
\begin{equation}\label{actionAd}
  \Ad_A(B) = \exp(A) \, B \, \exp(-A) = \sum_{k=0}^{\infty} \frac{1}{k!} \ad_A^k B,
   \qquad B\in \mathfrak{g}.
\end{equation}

The type of matrices we will handle more frequently are orthogonal,
unitary and symplectic. Here are their characterization and the
notation we shall use for their group and algebra.

The special orthogonal group, $\SO(n)$, is the set of all $n\times
n$ real matrices with unit determinant satisfying $A^TA=AA^T=I$, where $A^T$ is the transpose
of $A$ and $I$ denotes the identity matrix. The corresponding
algebra $\so(n)$ consists of the skew-symmetric 
matrices.

A $n \times n$ complex matrix $A$ is called unitary if 
$A^{\dag}A=AA^{\dag}=I$, where $A^{\dag}$ is
the conjugate transpose or Hermitian adjoint of $A$.
The special unitary group, $\SU(n)$, is the set of all $n\times n$
unitary matrices with unit determinant. The
corresponding algebra $\mathfrak{su}(n)$ consist of the
skew-Hermitian traceless matrices. Special relevance in some
quantum mechanical problems we discuss will have the case $n=2$.
In this case a convenient basis for $\mathfrak{su}(2)$ is made up
by the Pauli matrices
\begin{equation}\label{Paulis}
    \sigma_1=\left(
               \begin{array}{ccr}
                 0 & & 1 \\
                 1 & & 0
               \end{array}
             \right), \quad
     \sigma_2=\left(
                  \begin{array}{ccr}
                    0 & & -i \\
                    i & & 0
                  \end{array}
                \right), \quad
     \sigma_3=\left(
                \begin{array}{ccr}
                  1 & & 0 \\
                  0 & & -1
                \end{array}
              \right).
\end{equation}
They satisfy the identity
\begin{equation}\label{producpaulis}
   \sigma_{j}\sigma_{k}=\delta_ {jk}+i\epsilon_{jkl}\sigma_l,
\end{equation}
and correspondingly
\begin{equation} \label{conmusigmas}
  [\sigma_{j},\sigma_{k}]=2i\epsilon_{jkl}\sigma_l,
\end{equation}
which directly give the structure constants for $\SU(2)$. The
following identities will prove useful for $\boldsymbol{a}$ and $\boldsymbol{b}$ in
$\mathbb{R}^{3}$:
\begin{equation}  \label{pauliescalar}
    (\boldsymbol{a}\cdot\boldsymbol{\sigma})(\boldsymbol{b}\cdot\boldsymbol{\sigma}) =  
    \boldsymbol{a}\cdot\boldsymbol{b} \  I+
    i(\boldsymbol{a}\times\boldsymbol{b})\cdot\boldsymbol{\sigma},  \qquad
     [\boldsymbol{a}\cdot\boldsymbol{\sigma},\boldsymbol{b}\cdot\boldsymbol{\sigma}]  =   
       2 i (\boldsymbol{a}\times\boldsymbol{b}) \cdot \boldsymbol{\sigma},
\end{equation}
where we have denoted $\boldsymbol{\sigma} = (\sigma_1, \sigma_2, \sigma_3)$. 
Any $U \in \SU(2)$ can be written as
\begin{equation}\label{exppaulis}
  U =\exp(i\boldsymbol{a}\cdot\boldsymbol{\sigma})=
  \cos ( a ) \,I + i \frac{\sin  (a) }{a} \boldsymbol{a}\cdot\boldsymbol{\sigma},
\end{equation}
where $a = \|\boldsymbol{a}\| = \sqrt{a_1^2 + a_2^2 + a_3^2}$.
A more elaborate expression which we shall make use of in later
sections is (with $a=1$)
\begin{equation}\label{cambiopict}
    \exp(i\boldsymbol{a}\cdot\boldsymbol{\sigma} t) \, (\boldsymbol{b}\cdot\boldsymbol{\sigma}) \, 
    \exp(-i\boldsymbol{a}\cdot\boldsymbol{\sigma} t)=%
    \boldsymbol{b}\cdot\boldsymbol{\sigma}+\sin 2t \, (\boldsymbol{b}\times\boldsymbol{a})\cdot\boldsymbol{\sigma}%
    +\sin^2 t \, (\boldsymbol{a}\times(\boldsymbol{b}\times\boldsymbol{a}))\cdot\boldsymbol{\sigma}
\end{equation}
In Hamiltonian problems the symplectic group $\SP(n)$ plays a
fundamental role. It is the group of $2n\times 2n$ real matrices
satisfying
\begin{equation}\label{simplectic}
    A^{T}JA=J,   \qquad \mbox{ with } \quad \qquad J=\left(
                      \begin{array}{rcc}
                        O_n & & I_n \\
                        -I_n & & O_n \\
                      \end{array}
                    \right)
\end{equation}
and $I_n$ denotes the $n$-dimensional identity matrix.
Its corresponding Lie algebra $\mathfrak{sp}(n)$ consists of
matrices verifying $B^{T}J + JB=O_{2n}$. In fact, these can be
considered particular instances of the so-called $J$-orthogonal
group, defined as \cite{postnikov94lga}
\begin{equation}   \label{j-ortho}
   \mathrm{O}_J(n) = \{ A \in \mathrm{GL}(n) \, : \, A^T J A = J \},
\end{equation}
where $\mathrm{GL}(n)$ is the group of all $n \times n$
nonsingular real matrices and $J$ is some constant matrix in
$\mathrm{GL}(n)$. Thus, one recovers the orthogonal group when
$J=I$, the symplectic group $\mathrm{Sp}(n)$ when $J$ is the basic
symplectic matrix given in (\ref{simplectic}), and the Lorentz
group $\mathrm{SO}(3,1)$ when $J = \mathrm{diag}(1,-1,-1,-1)$.
The corresponding Lie algebra is the set
\begin{equation}  \label{j-algebra}
\mathrm{o}_J(n) = \{ B \in \mathfrak{gl}_n (\mathbb{R}) \, : \,
B^T J + J B = O \},
\end{equation}
where $\mathfrak{gl}_n (\mathbb{R})$ is the Lie algebra of all $n
\times n$ real matrices. If $B \in \mathrm{o}_J(n)$, then its
Cayley transform
\begin{equation}  \label{cayley1}
A = (I - \alpha B)^{-1} (I + \alpha B)
\end{equation}
is $J$-orthogonal.

Another important matrix Lie group not included in the previous
characterization is the special linear group $\mathrm{SL}(n)$,
formed by all $n \times n$ real matrices with unit determinant. The
corresponding Lie algebra $\mathfrak{sl}(n)$ comprises all
traceless matrices. For real $2\times 2$ matrices in $\mathfrak{sl}(2)$ one has
\begin{equation}\label{expMat2}
 \exp  \left(  \begin{array}{cc}
   a &  \ b \\ c & -a
            \end{array}  \right) =
             \left(  \begin{array}{cc}
  \cosh(\eta) + \frac{a}{\eta} \sinh(\eta) &  \  \frac{b}{\eta}
               \sinh(\eta)    \\
 \frac{c}{\eta} \sinh(\eta) & \cosh(\eta) - \frac{a}{\eta} \sinh(\eta)
            \end{array}  \right)
\end{equation}
with $\eta=\sqrt{a^2+bc}$.

When dealing with convergence problems it is necessary to use some
type of norm for a matrix. By such we mean a non-negative real
number $\|A\|$ associated with each matrix $A \in \mathbb{C}^{n \times n}$ and satisfying
\begin{itemize} 
  \item[a)] $\|A\| \ge  0$ for all $A$ and $\|A\| = 0$  iff  $A=O_n$.  
  \item[b)] $\| \alpha A\| = |\alpha| \, \| A\|$, for all scalars $\alpha$.
  \item[c)] $\| A+B\|  \le   \| A \| + \| B\|$.
\end{itemize}
Quite often one adds the sub-multiplicative property
\begin{equation}\label{subaditiva}
    \| A B \| \le  \| A \| \, \| B \|,
\end{equation}
but not all matrix norms satisfy this condition \cite{golub96mc}.

There exist different families of matrix norms. Among the more
popular ones we have the $p$-norm $\| A\|_p$ and the Frobenius norm
$\|A\|_F$. For a matrix $A$ with elements $a_{ij}$, $i,j=1 \ldots
n$, they are defined as
\begin{eqnarray}
  \| A\|_p & = & \max_{\|\textbf{x}\|_p =1} \|A\textbf{x}\|_p     \label{spectral} \\
  \| A\|_F & = & \sqrt {\sum_{i=1}^{n} \sum_{j=1}^{n}
  |a_{ij}|^2}=\sqrt{\mathrm{tr}(A^{\dag} A) },   \label{Frobenius}
\end{eqnarray}
respectively, where $\|\textbf{x}\|_p =(\sum_{j=1}^{n}
  |x_j|^p)^{\frac{1}{p}}$ and $\mathrm{tr}($A$)$ is the trace of the
  matrix $A$. Although both verify (\ref{subaditiva}), the
$p$-norms have the important property that for every matrix $A$
and $\mathbf{x} \in \mathbb{R}^n$ one has $\| A \mathbf{x}\|_p \le
\|A\|_p \, \|\mathbf{x}\|_p$. The most used $p$-norms correspond to
$p=1$, $p=2$  and $p = \infty$. 

Of paramount importance in numerical linear algebra is the case
$p=2$. The resulting $2$-norm  of a vector is nothing but the Euclidean norm,
whereas in the matrix case it is also called the spectral norm of
$A$ and can be characterized as the square root of the largest
eigenvalue of $A^{\dag} A$. A frequently used inequality relating
Frobenius and spectral norms is the following:
\begin{equation}  \label{ineqf1}
   \|A\|_2 \le \|A\|_F  \le \sqrt{n} \, \|A\|_2.
\end{equation}
In fact, this last inequality can be made more stringent \cite{tyrtyshnikov97abi}:
\begin{equation}   \label{ineqf2}
   \|A\|_F \le \sqrt{\mathrm{rank}(A)} \, \|A\|_2.
\end{equation}   

Considering in a matrix Lie algebra $\mathfrak{g}$ a norm satisfying property
(\ref{subaditiva}), it is clear that $\|[A,B]\| \le 2 \|A\| \|B\|$, and the 
$\ad$ operator defined by (\ref{definad}) is bounded, since
\[
    \|\ad_A\| \le 2 \|A\| 
\]
for any matrix $A$.

A matrix norm is said to be unitarily invariant if $\|U A V\| =
\|A\|$ whenever $U$, $V$ are unitary matrices. Frobenius and $p$-norms are 
both unitarily invariant \cite{horn85man}.

In some of the most basic formulas for the Magnus expansion there
will appear the so-called Bernoulli numbers $B_n$, which are
defined through the generating function \cite{abramowitz65hom}
\[
   \frac{t \e^{z t}}{\e^t - 1} = \sum_{n=0}^{\infty} B_n(z) \,
   \frac{t^n}{n!}, \qquad |t| < 2 \pi
\]
as $B_n = B_n(0)$. Equivalently,
\[
    \frac{x}{\e^x - 1} = \sum_{n=0}^{\infty} \, \frac{B_n}{n!} \, x^n,
\]
whereas the formula
\[
      \frac{\e^x - 1}{x} = \sum_{n=0}^{\infty} \,  \frac{1}{(n+1)!} \, x^n
\]   
will be also useful in the sequel.   
The first few nonzero Bernoulli numbers are $B_0 = 1$, $B_1 =
-\frac{1}{2}$, $B_2 = \frac{1}{6}$, $B_4 = -\frac{1}{30}$. In
general one has $B_{2m+1}=0$ for $m\geq 1$.



\section{The Magnus expansion (ME)}\label{section2}

Magnus proposal with respect to
the linear evolution equation
\begin{equation}   \label{eq:evolution}
   Y^{\prime}(t)=A(t) Y(t)
\end{equation}
with initial condition $Y(0)=I$, was to express the solution as the exponential of a certain
function, 
\begin{equation} \label{eq:Omega}
  Y(t)=\exp{\Omega (t)}.
\end{equation}
This is in contrast to the representation
\[
    Y(t) = \mathcal{T} \left( \exp \int_0^t A(s) ds \right)
\]
in terms of the \emph{time-ordening operator} $\mathcal{T}$ introduced by Dyson \cite{dyson49trt}.
    
It turns out that $\Omega(t)$ in (\ref{eq:Omega}) 
can be obtained explicitly in a number of ways. The
crucial point is to derive a differential equation for the
operator $\Omega$ that replaces (\ref{eq:evolution}). Here we
reproduce the result first established by Magnus as Theorem III
in \cite{magnus54ote}:
\begin{theorem}  \label{thMag}
(Magnus 1954). Let $A(t)$ be a known function of $t$
(in general, in an associative ring), and
let $Y(t)$ be an unknown function satisfying (\ref{eq:evolution})
with $Y(0)=I$. Then, if certain unspecified conditions of
convergence are satisfied, $Y(t)$ can be written in the form
\[
   Y(t) = \exp{ \Omega(t)},
\]
where
\begin{equation}   \label{OmegaDE}
  \frac{d \Omega}{dt} = \sum_{n=0}^\infty \frac{B_n}{n!} \,
  {\ad}^n_\Omega A,
\end{equation}
and $B_n$ are the Bernoulli numbers. Integration of
(\ref{OmegaDE}) by iteration leads to an infinite series for
$\Omega$ the first terms of which are
\[
  \Omega(t) = \int_0^t A(t_1) dt_1 - \frac{1}{2} \int_0^t \left[
  \int_0^{t_1} A(t_2) dt_2, A(t_1) \right] dt_1 + \cdots
\]
\end{theorem}

\subsection{A proof of Magnus Theorem}

The proof of this theorem is largely based on the derivative of
the matrix exponential map, which we discuss next. Given a scalar function
$\omega(t) \in \mathbb{R}$, the derivative of the exponential is
given by $d \exp(\omega(t))/dt = \omega'(t) \exp(\omega(t))$. One
could think of a similar formula for a matrix $\Omega(t)$.
However, this is not the case, since in general $[\Omega, \Omega']
\ne 0$. Instead one has the following result.
\begin{lemma}  \label{lemma1}
The derivative of a matrix exponential can be written
alternatively as
\begin{eqnarray}  \label{lem1a}
  \mathit{(a)} \quad   \frac{d}{dt} \exp(\Omega(t)) & = &
         d \exp_{\Omega(t)}(\Omega'(t)) \, \exp(\Omega(t)),  \\
  \mathit{(b)} \quad   \frac{d}{dt} \exp(\Omega(t)) & = &
        \exp(\Omega(t)) \,  d \exp_{-\Omega(t)}(\Omega'(t)),    \label{lem1c} \\         
  \mathit{(c)} \quad   \frac{d}{dt} \exp(\Omega(t)) & = & \int_0^1 \e^{x \Omega(t)}
   \, \Omega'(t) \, \e^{(1-x) \Omega(t)} dx, \label{lem1b}
\end{eqnarray}
where $d \exp_{\Omega}(C)$ is defined by its (everywhere convergent)
power series
\begin{equation}  \label{fdexp1}
  d \exp_{\Omega}(C)  =    \sum_{k=0}^{\infty} \frac{1}{(k+1)!} \,
   \mathrm{ad}_{\Omega}^k(C)
  \equiv
  \frac{\exp(\mathrm{ad}_{\Omega})-I}{\mathrm{ad}_{\Omega}}(C).
\end{equation}
\end{lemma}

\begin{proof}
 Let $\Omega(t)$ be a matrix-valued differentiable function and
 set
 \[
   Y(\sigma,t) \equiv \frac{\partial }{\partial t} \left(
   \exp(\sigma \Omega(t)) \right) \exp(-\sigma \Omega(t))
\]
for $\sigma, t \in \mathbb{R}$. Differentiating with respect to
$\sigma$,
\begin{eqnarray*}
  \frac{\partial Y}{\partial \sigma} & = &
    \frac{\partial }{\partial t} \left( \exp(\sigma \Omega) \Omega
    \right) \exp(-\sigma \Omega) + \frac{\partial }{\partial t} \left( \exp(\sigma \Omega) \right)
    (- \Omega) \exp(-\sigma \Omega) \\
    & = & \left( \exp(\sigma \Omega) \Omega' + \frac{\partial }{\partial t}
       \left( \exp(\sigma \Omega) \right) \Omega \right)
       \exp(-\sigma \Omega)  \\
    & & - \frac{\partial }{\partial t} \left( \exp(\sigma \Omega)
    \right) \Omega \exp(-\sigma \Omega) = \exp(\sigma \Omega) \Omega' \exp(-\sigma \Omega) \\
    & = & \exp(\mathrm{ad}_{\sigma \Omega})(\Omega') =  \sum_{k=0}^{\infty} \frac{\sigma^k}{k!}
     \mathrm{ad}_{\Omega}^k (\Omega'),
\end{eqnarray*}
where the first equality in the last line follows readily from
(\ref{definAd}) and (\ref{actionAd}). On the other hand
\begin{equation}   \label{auxl1}
  \frac{d }{dt}(\exp \Omega) \exp(-\Omega) = Y(1,t) = \int_0^1
  \frac{\partial }{\partial \sigma} Y(\sigma, t) d\sigma
\end{equation}
since $Y(0,t) = 0$, and
\[
  \int_0^1 \frac{\partial }{\partial \sigma} Y(\sigma, t) d\sigma
  = \int_0^1 \sum_{k=0}^{\infty} \frac{\sigma^k}{k!}
     \mathrm{ad}_{\Omega}^k (\Omega') d\sigma = \sum_{k=0}^{\infty} \frac{1}{(k+1)!}
     \mathrm{ad}_{\Omega}^k (\Omega'),
\]
from which formula (\ref{lem1a}) follows. The convergence of the
power series (\ref{fdexp1}) is a consequence of the boundedness of
the ad operator: $\|\mathrm{ad}_{\Omega}\| \le 2 \|\Omega\|$.

Multiplying both sides of (\ref{lem1a}) by $\exp(-\Omega)$, we have
\[
  \e^{-\Omega} \frac{d \e^{\Omega}}{dt}  =  \e^{-\Omega} d \exp_{\Omega}(\Omega') \e^{\Omega} 
     = \e^{\ad_{-\Omega}} d \exp_{\Omega}(\Omega') =
      \frac{\e^{\ad_{-\Omega}} - I}{\ad_{-\Omega}} \Omega' = d \exp_{-\Omega}(\Omega')
\]
from which (\ref{lem1c}) follows readily.
Finally, equation (\ref{lem1b}) is obtained by taking
\[
  \int_0^1 \frac{\partial }{\partial \sigma} Y(\sigma, t) d\sigma
  = \int_0^1 \exp(\sigma \Omega) \Omega' \exp(-\sigma \Omega) d\sigma
\]
in (\ref{auxl1}).
\end{proof}

According to Rossmann \cite{rossmann02lgr} and Sternberg \cite{sternberg04lal},
formula (\ref{lem1a}) was first proved by F. Schur in
1890 \cite{schur90nbd}  and was
taken up later from a different point of view by Poincar\'e (1899),
whereas the integral formulation
(\ref{lem1b}) has been derived a number of times in the physics
literature \cite{wilcox67eoa}.

As a consequence of the Inverse Function Theorem, the exponential
map has a local inverse in the vicinity of a point $\Omega$ at
which $d \exp_{\Omega} =
(\exp(\mathrm{ad}_{\Omega})-I)/\mathrm{ad}_{\Omega}$ is
invertible. The following lemma establishes when this takes place.

\begin{lemma}  \label{lemma2}
(Baker 1905). If the eigenvalues of the linear operator
$\mathrm{ad}_{\Omega}$ are different from $2 m \pi i$ with $m \in
\{ \pm 1, \pm 2, \ldots \}$, then $d \exp_{\Omega}$ is invertible.
Furthermore,
\begin{equation}  \label{fdexpinv}
  d \exp_{\Omega}^{-1}(C) =  \frac{\ad_{\Omega}}{\e^{\ad_{\Omega}} - I} C = 
  \sum_{k=0}^{\infty} \frac{B_k}{k!}
  \mathrm{ad}_{\Omega}^k (C)
\end{equation}
and the convergence of the $d \exp_{\Omega}^{-1}$ expansion is
certainly assured if  $\| \Omega \| < \pi$.
\end{lemma}
\begin{proof}
The eigenvalues of $d \exp_{\Omega}$ are of the form
\[
 \mu = \sum_{k \ge 0} \frac{\nu^k}{(k+1)!} = \frac{\e^{\nu}
- 1}{\nu},
\]
where $\nu$ is an eigenvalue of $\mathrm{ad}_{\Omega}$. By
assumption, the values of $\mu$ are non-zero, so that $d
\exp_{\Omega}$ is invertible. By definition of the Bernoulli
numbers, the composition of (\ref{fdexpinv}) with (\ref{fdexp1})
gives the identity. Convergence for $\| \Omega \| < \pi$ follows
from $\|\mathrm{ad}_{\Omega}\| \le 2 \|\Omega\|$ and from the fact
that the radius of convergence of the series expansion for
$x/(\e^x - 1)$ is $2 \pi$.
\end{proof}

It remains to determine the eigenvalues of the operator
$\mathrm{ad}_{\Omega}$. In fact, it is not difficult to show that
if $\Omega$ has $n$ eigenvalues $\{ \lambda_j, \; j=1,2,\ldots,n
\}$, then $\mathrm{ad}_{\Omega}$ has $n^2$ eigenvalues $\{
\lambda_j - \lambda_k, \; j,k=1,2,\ldots,n \}$.

As a consequence of the previous discussion, Theorem \ref{thMag} can be
rephrased more precisely in the following terms.
\begin{theorem}  \label{thMrefor}
 The solution of the differential equation $Y' =
A(t) Y$ with initial condition $Y(0) = Y_0$
can be written as $Y(t) = \exp(\Omega(t)) Y_0$ with
$\Omega(t)$ defined by
\begin{equation}  \label{fmag1}
  \Omega' = d \exp_{\Omega}^{-1}(A(t)), \qquad \Omega(0) = O,
\end{equation}
where
\[
   d \exp_{\Omega}^{-1}(A) = \sum_{k=0}^{\infty} \frac{B_k}{k!}
  \mathrm{ad}_{\Omega}^k (A).
\]
\end{theorem}
\begin{proof}
 Comparing the derivative of $Y(t) = \exp(\Omega(t)) Y_0$,
\[
  \frac{dY}{dt} = \frac{d }{dt} \left( \exp(\Omega(t)) \right) Y_0
  = d \exp_{\Omega}(\Omega') \, \exp(\Omega(t)) Y_0
\]
with $Y' = A(t) Y$, we obtain $A(t) = d \exp_{\Omega}(\Omega')$.
Applying the inverse operator $d \exp_{\Omega}^{-1}$ to this
relation yields the differential equation (\ref{fmag1}) for
$\Omega(t)$.
\end{proof}

Taking into account the numerical values of the first few
Bernoulli numbers, the differential equation (\ref{fmag1})
therefore becomes
\[
  \Omega' = A(t) - \frac{1}{2} [\Omega, A(t)] + \frac{1}{12}
  [\Omega,[\Omega,A(t)]] + \cdots,
\]
which is nonlinear in $\Omega$. By defining
\[
  \Omega^{[0]} = O, \qquad \Omega^{[1]} = \int_0^t A(t_1) dt_1,
\]
and applying Picard fixed point iteration, one gets
\[
  \Omega^{[n]} = \int_0^t \left( A(t_1) dt_1 - \frac{1}{2}
  [\Omega^{[n-1]}, A] + \frac{1}{12}
  [\Omega^{[n-1]},[\Omega^{[n-1]},A]] + \cdots \right)
     dt_1
\]
and $\lim_{n \rightarrow \infty} \Omega^{[n]}(t) = \Omega(t)$ in
a suitably small neighbourhood of the origin.

\subsection{Formulae for the first terms in Magnus expansion}
\label{sec2.2}
 Suppose now that $A$ is of first order in some
parameter $\varepsilon $ and try a solution in the form of a series
\begin{equation}   \label{M-series}
 \Omega(t) =\sum_{n=1}^\infty \Omega_n(t),
\end{equation}
where $\Omega_n$ is supposed to be of order $\varepsilon^n$. 
Equivalently, we replace $A \longmapsto \varepsilon A$ in (\ref{eq:evolution})
and determine the successive terms of
\begin{equation}  \label{M-series-eps}
    \Omega(t) =\sum_{n=1}^\infty  \varepsilon^n \Omega_n(t).
\end{equation}
This can be done explicitly, at least for the first terms,
by substituting the series (\ref{M-series-eps})
in (\ref{fmag1}) and equating powers of
$\varepsilon $. Obviously, the Magnus series (\ref{M-series}) is recovered by taking $\varepsilon = 1$.
Thus, using the
notation $A(t_i)\equiv A_i$, the first four orders read
\begin{enumerate}
\item $\Omega_1^{\prime}= A $, so that
\begin{equation}\label{O1}
  \Omega_1(t)=\int_{0}^t {\rm d}t_1 A _1
\end{equation}
\item $\Omega_2^{\prime}=-\frac{1}{2}[\Omega_1,A ]$. Thus
\begin{equation}\label{O2}
  \Omega_2(t)=\frac{1}{2}\int_{0}^t {\rm d}t_1 \int_{0}^{t_1} {\rm d}t_2
  [A _1,A _2]
\end{equation}
\item $\Omega_3^{\prime}
=-\frac{1}{2}[\Omega_2,A ]+\frac{1}{12}[\Omega_1,[\Omega_1,A ]]$.
After some work and using the formula
\begin{equation}\label{intxy}
  \int_0^\alpha {\rm d}x \int_0^x f(x,y){\rm d}y=\int_0^\alpha {\rm
  d}y \int_y^\alpha f(x,y){\rm d}x
\end{equation}
we obtain
\begin{equation}\label{O3}
  \Omega_3(t)=\frac{1}{6}\int_{0}^t {\rm d}t_1 \int_{0}^{t_1}{\rm d}t_2
  \int_{0}^{t_2} {\rm d}t_3  \{
  [A _1,[A _2,A _3]]+[[A _1,A _2],A _3]\}
\end{equation}
\item $\Omega_4^{\prime}=-\frac{1}{2}[\Omega_3,A ]+\frac{1}{12}[\Omega_2,[\Omega_1,A ]]
+\frac{1}{12}[\Omega_1,[\Omega_2,A ]]$, which yields
\begin{eqnarray}\label{O4}
  \Omega_4(t)&=&\frac{1}{12}\int_{0}^t {\rm d}t_1 \int_{0}^{t_1}{\rm d}t_2
  \int_{0}^{t_2} {\rm d}t_3  \int_{0}^{t_3} {\rm d}t_4 \{
  [[[A _1,A _2],A _3]A _4] \\
  &+&[A _1,[[A _2,A _3],A _4]]+[A _1,[A _2,[A _3,A _4]]]+[A _2,[A _3,[A _4,A _1]]]
  \} \nonumber
\end{eqnarray}

\end{enumerate}
The apparent symmetry in the formulae above is deceptive. High
orders require repeated use of (\ref{intxy}) and become unwieldy.
Prato and Lamberti
\cite{prato97ano} give explicitly the fifth order using an
algorithmic point of view. One can also find in the literature quite
involved explicit expressions for an arbitrary order \cite
{bialynicki69eso,mielnik70cat,saenz02cat,strichartz87tcb,suarez01laa}.
 In the next subsection we describe
a recursive procedure to generate the terms in the expansion.

\subsection{Magnus expansion generator}\label{MEG}

The above procedure can provide indeed a recursive procedure to generate
all the terms in the Magnus series (\ref{M-series}). Thus, by substituting
$\Omega(t)=\sum_{n=1}^{\infty }\Omega _{n}$ into equation
(\ref{fmag1}) and equating terms of the same order one gets in
general
\begin{eqnarray}
\Omega_{1}^{\prime} &=& A   \nonumber  \label{omegapunto} \\
\Omega_{n}^{\prime}
&=&\sum_{j=1}^{n-1}\frac{B_{j}}{j!}S_{n}^{(j)},\quad n\geq 2,
\label{omegapunt}
\end{eqnarray}
where
\begin{equation}   \label{Sk}
  S_n^{(k)}=\sum \,
[\Om _{i_1},[ \ldots [\Om _{i_k},A ] \ldots ]]  \qquad (i_1+
\cdots +i_k=n-1).
\end{equation}
Notice that in the last equation the order in $A$ has been
explicitly reckoned, whereas $k$ represents the number of $\Om
$'s. The newly defined operators $S_n^{(k)}$ can again be
calculated recursively. The recurrence relations are now given by
\begin{eqnarray}
S_{n}^{(j)} &=&\sum_{m=1}^{n-j}\left[ \Omega
_{m},S_{n-m}^{(j-1)}\right]
,\qquad\qquad  2\leq j\leq n-1  \label{eses} \\
S_{n}^{(1)} &=&\left[ \Omega _{n-1},A \right] ,\qquad
S_{n}^{(n-1)}= \ad _{\Omega _{1}}^{n-1} (A)  .  \nonumber
\end{eqnarray}
After integration we reach the final result in the form
\begin{eqnarray}
\Omega _{1} &=&\int_{0}^{t}A (\tau )d\tau  \nonumber \\
\Omega _{n}
&=&\sum_{j=1}^{n-1}\frac{B_{j}}{j!}\int_{0}^{t}S_{n}^{(j)}(\tau
)d\tau ,\qquad\qquad n\geq 2.  \label{omegn}
\end{eqnarray}
Alternatively, the expression of $S_n^{(k)}$ given by (\ref{Sk}) 
can be inserted into (\ref{omegn}), thus arriving at
\begin{equation}   \label{recur2}
  \Omega_n(t) =  \sum_{j=1}^{n-1} \frac{B_j}{j!} \,
    \sum_{
            k_1 + \cdots + k_j = n-1 \atop
            k_1 \ge 1, \ldots, k_j \ge 1}
            \, \int_0^t \,
       \ad_{\Omega_{k_1}(s)} \,  \ad_{\Omega_{k_2}(s)} \cdots
          \, \ad_{\Omega_{k_j}(s)} A(s) \, ds    \qquad n \ge 2.
\end{equation}
Notice that each term $\Omega_n(t)$ in the Magnus series is
a multiple integral of combinations of $n-1$ nested commutators containing
$n$ operators $A(t)$. If, in particular, $A(t)$ belongs to some Lie algebra $\mathfrak{g}$,
then it is clear that $\Omega(t)$ (and in fact any truncation of the Magnus series) also stays
in $\mathfrak{g}$ and therefore $\exp(\Omega) \in \mathcal{G}$, where $\mathcal{G}$
denotes the Lie group whose corresponding Lie algebra (the tangent space at the
identity of $\mathcal{G}$) is $\mathfrak{g}$.

\subsection{Magnus expansion and time-dependent perturbation
theory}
\label{sec2.4}

It is not difficult to establish a connection between Magnus
series and Dyson perturbative series \cite{dyson49trt}. The later gives
the solution of (\ref{eq:evolution}) as
\begin{equation}  \label{Dys1}
  Y(t) = I + \sum_{n=1}^{\infty} P_n(t),
\end{equation}
where $P_n$ are time-ordered products
\[
   P_n(t) = \int_0^t dt_1 \ldots \int_0^{t_{n-1}} dt_n \, A_1 A_2
   \ldots A_n,
\]
where $A_i \equiv A(t_i)$. Then
\[
\sum_{j=1}^{\infty }\Omega _{j}(t)=\log \left( I+\sum_{j=1}^{\infty
}P_{j}(t)\right).
\]
As stated by Salzman \cite{salzman87ncf},
\begin{equation}
\Omega
_{n}=P_{n}-\sum_{j=2}^{n}\frac{(-1)^{n}}{j}R_{n}^{(j)},\qquad
n\geq 2, \label{MagDay}
\end{equation}
where
\[
  R_n^{(k)}=\sum P_{i_1}P_{i_2}\ldots P_{i_k}  \qquad
(i_1+\cdots+i_k=n)
\]
obeys to the quadratic recursion formula
\begin{eqnarray}
R_{n}^{(j)} &=&\sum_{m=1}^{n-j+1}R_{m}^{(1)}R_{n-m}^{(j-1)},
\label{erres}
\\
R_{n}^{(1)} &=&P_{n},\qquad R_{n}^{(n)}=P_{1}^{n}.  \nonumber
\end{eqnarray}
Equation (\ref{erres}) represents the Magnus expansion generator
in Salzman's approach.
It may be useful to write down the first few equations provided by
this formalism:
\begin{eqnarray}\label{Salzeqs}
  \Omega_1&=& P_1 \nonumber \\
  \Omega_2&=& P_2-\frac{1}{2}P_1^2 \\
  \Omega_3&=& P_3-\frac{1}{2}(P_1P_2+P_2P_1)+\frac{1}{3}P_1^3.   \nonumber
\end{eqnarray}
A similar set of equations was developed by Burum
\cite{burum81meg}, thus providing
\begin{eqnarray}\label{Burum1}
  P_1&=&\Om  _1,\nonumber \\
P_2&=&\Om  _2+\frac{1}{2!}\Om  _1^2,  \\
P_3&=&\Om  _3+\frac{1}{2!}(\Om  _1 \Om  _2 + \Om  _2 \Om  _1) +
\frac{1}{3!} \Om_1^3  \nonumber
\end{eqnarray}
and so on. The general term reads
\begin{equation}\label{MagDay1}
\Omega_{n}=P_{n}-\sum_{j=2}^{n}\frac{1}{j}Q_{n}^{(j)}, \qquad
n\geq 2,
\end{equation}
where
\begin{equation}   \label{BurumQ}
  Q_n^{(k)} = \sum \Om _{i_1}\ldots\Om _{i_k},
\qquad\quad (i_1+ \cdots +i_k=n).
\end{equation}
As before, subscripts indicate the order with respect to the
parameter $\varepsilon$, while superscripts represent the number
of factors in each product. Thus, the summation in (\ref{BurumQ})
extends over all possible products of $k$ (in general
non-commuting) operators $\Om _i$  such that the overall order of
each term is equal to $n$. By regrouping terms, one has
\begin{eqnarray}\label{BurumQQ}
  Q_n^{(k)} &=& \Om  _1 \sum_{i_2+\cdots+i_k=n-1}
\Om _{i_2}\cdots\Om _{i_k}+ \Om _2 \sum_{i_2+\cdots+i_k=n-2}
\Om _{i_2}\cdots\Om _{i_k}\\   \nonumber
&+& \cdots +\Om
_{n-k+1}\sum_{i_2+\cdots+i_k=k-1} \Om _{i_2}\cdots\Om _{i_k},
\end{eqnarray}
where $Q_{n}^{(j)}$ may also be obtained recursively from
\begin{eqnarray}\label{qs}
Q_{n}^{(j)} &=&\sum_{m=1}^{n-j+1}Q_{m}^{(1)}Q_{n-m}^{(j-1)},\\
Q_{n}^{(1)} &=&\Omega_{n},\qquad Q_{n}^{(n)}=\Omega_{1}^{n}.
\nonumber
\end{eqnarray}
By working out this recurrence one gets the same expressions as
(\ref{Salzeqs}) for the first terms. Further aspects of the
relationship between Magnus, Dyson series and time-ordered
products can be found in \cite{lam98dot} and \cite{oteo00ftp}.

\subsection{Graph theoretical analysis of Magnus
expansion}   \label{graph}

The previous recursions allow us in principle to express any
$\Omega_k$ in the Magnus series in terms of $\Omega_1, \ldots,
\Omega_{k-1}$. In fact, this procedure has some advantages from a
computational point of view. On the other hand, as we have mentioned before,
when the
recursions are solved explicitly,  $\Omega_k$ can
be expanded as a linear combination of terms that are composed
from integrals and commutators acting iteratively on $A$. The
actual expression, however, becomes increasingly complex with $k$,
as it should be evident from the first terms
(\ref{O1})-(\ref{O4}). An alternative form of the Magnus
expansion, amenable also for recursive derivation by using
graphical tools, can be obtained by associating each term in the
expansion with a \emph{binary rooted tree}, an approach worked out
by Iserles and N{\o}rsett \cite{iserles99ots}. For completeness,
in the sequel we show the equivalence of the recurrence
(\ref{eses})-(\ref{omegn}) with this graph theoretical
approach.

In essence, the idea of Iserles and N{\o}rsett is to associate
each term in $\Omega_k$ with a rooted tree, according to the
following prescription.

Let $\TT_0$ be the set consisting of the single rooted tree
with one vertex, then $  \TT_0=\{ \begin{picture}(10,5)
              \put (5,3) \AAA
             \end{picture} \}$, establish the relationship
             between this tree and $A$ through the map
\[
    \begin{picture}(10,5)
      \put (5,3) \AAA
    \end{picture} \leadsto A(t)
\]
and define recursively
\begin{displaymath}
    \TT_m=\left\{\rule[-12pt]{0pt}{24pt}\right. \;\;
      \begin{picture}(20,30)(0,10)
    \put (10,0) \BIF
    \put (0,10) \UP
    \put (-3,23) {$\tau_1$}
    \put (17,13) {$\tau_2$}
      \end{picture}\quad :\; \tau_1\in\TT_{k_1},\; \tau_2\in\TT_{k_2},
      \; k_1+k_2=m-1\left.\rule[-12pt]{0pt}{24pt}\right\}.
\end{displaymath}
Next, given two expansion terms $H_{\tau_1}$ and $H_{\tau_2}$,
which have been associated previously with $\tau_1 \in \TT_{k_1}$
and $\tau_2 \in \TT_{k_2}$, respectively ($k_1 + k_2 = m-1$), we
associate
\begin{displaymath}
    H_\tau(t)=\left[\int_0^t H_{\tau_1}(\xi) d\xi,
    H_{\tau_2}(t)\right]  \qquad\mbox{with}\qquad \tau=
  \begin{picture}(20,30)
     \put (10,0) \BIF
     \put (0,10) \UP
     \put (-3,23) {$\tau_1$}
     \put (17,13) {$\tau_2$}
  \end{picture}.
\end{displaymath}
Thus, each $H_\tau$ for $\tau \in \TT_{m}$ involves exactly $m$
integrals and $m$ commutators.

 These composition rules establish a one-to-one relationship
between a rooted tree $\displaystyle \tau\in\TT \equiv \cup_{m\geq0}\TT_m$, and
a matrix function $H_\tau(t)$ involving $A$, multivariate
integrals and commutators.

From here it is easy to deduce that every $\tau \in \TT_m$, $m
\geq 1$, can be written in a unique way as
\begin{displaymath}
  \tau=
  \begin{picture}(50,65)
    \put (10,0) \BIF
    \put (0,10) \UP
    \put (-3,23) {$\tau_1$}
    \put (20,10) \BIF
    \put (10,20) \UP
    \put (7,33) {$\tau_2$}
    \put (30,20) {\line(-1,1){10}}
    \put (20,30) \UP
    \put (17,43) {$\tau_3$}
    \put (40,30) \BIF
    \put (50,40) \AAA
    \put (30,40) \UP
    \put (27,53) {$\tau_s$}
    \multiput (31,21.5)(2,2){4} {{\tiny .}}
  \end{picture}
\end{displaymath}
or $\tau \equiv a(\tau_1, \tau_2, \ldots, \tau_s)$. Then the
Magnus expansion can be expressed in the form
\cite{iserles00lgm,iserles99ots}
\begin{equation}   \label{m.2}
  \Omega(t)=\sum_{m=0}^\infty \sum_{\tau\in\TT_m} \alpha(\tau)
  \int_0^t H_\tau(\xi) d\xi,
\end{equation}
with the scalar $\alpha(\begin{picture}(10,5)  \put (5,3) \AAA
\end{picture})=1$ and, in general,
\[
  \alpha(\tau)=\frac{{B}_s}{s!} \prod_{l=1}^s
    \alpha(\tau_l).
\]
Let us illustrate this procedure by writing down explicitly the
first terms in the expansion in a tree formalism. In $\TT_1$ we
only have $k_1 = k_2 = 0$, so that a single tree is possible,
\[
  \tau_1 =  \begin{picture}(10,5)
      \put (5,3) \AAA
    \end{picture},   \qquad
  \tau_2 =  \begin{picture}(10,5)
      \put (5,3) \AAA
    \end{picture},   \qquad  \Rightarrow \qquad
  \tau = \begin{picture}(20,20)
    \put (10,0) \BIF
    \put (0,10) \UP
    \put (0,20) \AAA
    \put (20,10) \AAA
  \end{picture},
\]
with $\alpha(\tau) = -1/2$. In $\TT_2$ there are two
possibilities, namely $k_1=0$, $k_2=1$ and $k_1=1$, $k_2=0$, and
thus one gets
\[
   \begin{array}{lclclc}
   \tau_1 = \begin{picture}(10,5)
      \put (5,3) \AAA
    \end{picture}, & \quad & \tau_2 = \begin{picture}(20,20)
    \put (10,0) \BIF
    \put (0,10) \UP
    \put (0,20) \AAA
    \put (20,10) \AAA
  \end{picture} &  \qquad \Rightarrow \qquad &
  \tau = \begin{picture}(30,35)
      \put (10,0) \BIF \put (0,10) \UP \put (20,10) \BIF \put (10,20) \UP
      \put (0,20) \AAA
      \put (10,30) \AAA
      \put (30,20) \AAA
    \end{picture}, &  \quad \alpha(\tau) = \frac{1}{12} \\
    \tau_1 = \begin{picture}(20,20)
    \put (10,0) \BIF
    \put (0,10) \UP
    \put (0,20) \AAA
    \put (20,10) \AAA
  \end{picture}, & \quad & \tau_2 = \begin{picture}(10,5)
      \put (5,3) \AAA
    \end{picture} & \qquad \Rightarrow \qquad &
  \tau = \begin{picture}(30,45)
      \put (20,0) \BIF \put (10,10) \UP \put (10,20) \BIF \put (0,30) \UP
      \put (0,40) \AAA
      \put (20,30) \AAA
      \put (30,10) \AAA
    \end{picture} &  \quad \alpha(\tau) = \frac{1}{4}
 \end{array}
\]
and the process can be repeated for any $\TT_m$. The
correspondence between trees and expansion terms should be clear
from the previous graphs. For instance, the last tree is nothing
but the integral of $A$, commuted with $A$, integrated and
commuted with $A$. In that way, by truncating the expansion
(\ref{m.2}) at $m=2$ we have
\begin{equation}  \label{m.3}
  \Omega(t) = \;\; \begin{picture}(10,10)
    \put (5,0) \UP
    \put (5,10) \AAA
  \end{picture} \;\; - \;\; \frac{1}{2} \;\; \begin{picture}(20,35)
    \put (10,0) \UP
    \put (10,10) \BIF
    \put (0,20) \UP
    \put (0,30) \AAA
    \put (20,20) \AAA
  \end{picture} \quad + \quad \frac{1}{4} \;\; \begin{picture}(30,55)
    \put (20,0) \UP
    \put (20,10) \BIF
    \put (10,20) \UP
    \put (10,30) \BIF
    \put (0,40) \UP
    \put (0,50) \AAA
    \put (20,40) \AAA
    \put (30,20) \AAA
  \end{picture}\quad+\quad \frac{1}{12} \;\; \begin{picture}(30,45)
    \put (10,0) \UP
    \put (10,10) \BIF
    \put (0,20) \UP
    \put (20,20) \BIF
    \put (10,30) \UP
    \put (0,30) \AAA
    \put (10,40) \AAA
    \put (30,30) \AAA
  \end{picture} \quad + \cdots,
\end{equation}
i.e., the explicit expressions collected in subsection
\ref{sec2.2}.

Finally, the relationship between the tree formalism and the recurrence
(\ref{eses})-(\ref{omegn}) can be established as follows. From
(\ref{m.2}) we can write
\[
  \sum_{m=1}^\infty \sum_{\tau\in\TT_m} \alpha(\tau) H_\tau(t)
  = \sum_{s=1}^m \frac{{B}_s}{s!}
    \underset{k_1 + \cdots + k_s = m-s}{\underset{k_1, \ldots, k_s}{\sum}}
       \sum_{\tau_i \in \TT_{k_i}} \, \alpha(\tau_1) \cdots
     \alpha(\tau_s) H_{a(\tau_1, \ldots, \tau_s)}.
\]
Thus, by comparing (\ref{omegn}) and (\ref{m.2}) we have
\[
  \Omega_m(t)=\sum_{\tau \in \TT_{m-1}} \alpha(\tau)
  \int_0^t H_\tau(\xi) d\xi = \sum_{j=1}^{m-1} \frac{B_j}{j!}
  \int_0^t S_m^{(j)}(\xi) d\xi
\]
so that
\[
  S_m^{(j)} =
    \underset{k_1 + \cdots + k_j = m-1-j}{\underset{k_1, \ldots, k_j}{\sum}}
       \sum_{\tau_i \in \TT_{k_i}} \, \alpha(\tau_1) \cdots
     \alpha(\tau_j) H_{a(\tau_1, \ldots, \tau_j)}.
\]
In other words, each term $S_n^{(j)}$ in the recurrence
(\ref{eses}) carries on a complete set of binary trees. Although
both procedures are equivalent, the use of (\ref{eses}) and
(\ref{omegn}) can be particularly well suited when high orders of
the expansion are considered, for two reasons: (i) the enormous
number of trees involved for large values of $m$ and (ii) in
(\ref{m.2}) many terms are redundant, and a careful graph
theoretical analysis is needed to deduce which terms have to be
discarded \cite{iserles99ots}.

Recently, an ME-type formalism has been developed in the more
abstract setting of dendriform algebras. This generalized expansion incorporates the usual
one as a limit, but is formulated more in line with (non-commutative) Butcher
series. In this context, the use of planar rooted trees to represent 
the expansion and the so-called pre-Lie product allows one to reduce the number
of terms at each order in comparison with expression (\ref{m.2}) \cite{ebrahimi-fard08ama}.

\subsection{Time-symmetry of the expansion}
    \label{time-symmetry}

The map $\varphi^t : Y(t_0) \longrightarrow Y(t)$ corresponding to
the linear differential equation (\ref{eq:evolution}) with $Y(t_0) = Y_0$ is time
symmetric, $\varphi^{-t} \circ \varphi^t = \mathrm{Id}$, since
integrating (\ref{eq:evolution}) from $t=t_0$ to $t=t_f$ for every
$t_f \geq t_0$ and back to $t_0$ leads us to the original initial
value $Y(t_0)=Y_0$. Observe that, according with (\ref{matrizant}),
 the map $\varphi^t$ can be expressed
in terms of the fundamental matrix (or evolution operator) $U(t,t_0)$ as
$\varphi^{t_f}(Y_0) = U(t_f,t_0)
Y_0$. Then time-symmetry establishes that
\[
    U(t_0,t_f) = U^{-1}(t_f,t_0)
\]
or, in terms of the Magnus expansion,
\[
   \Omega(t_f,t_0) = - \Omega(t_0,t_f).
\]
To take advantage of this feature, let us write the solution of
(\ref{eq:evolution}) at the final time $t_f=t_0 + s$ as
\begin{equation}    \label{t-s3}
   Y(t_{1/2} + \frac{s}{2}) = \exp \left( \Omega(t_{1/2} + \frac{s}{2},
      t_{1/2} - \frac{s}{2}) \right) \, Y(t_{1/2} - \frac{s}{2}),
\end{equation}
where $t_{1/2} = (t_0 + t_f)/2$. Then
\begin{equation}    \label{t-s4}
   Y(t_{1/2} - \frac{s}{2}) = \exp \left(- \Omega(t_{1/2} + \frac{s}{2},
      t_{1/2} - \frac{s}{2}) \right) \, Y(t_{1/2} + \frac{s}{2}).
\end{equation}
On the other hand, the solution at $t_0$ can be written as
\begin{equation}    \label{t-s5}
   Y(t_{1/2} - \frac{s}{2}) = \exp \left( \Omega(t_{1/2} - \frac{s}{2},
      t_{1/2} + \frac{s}{2}) \right) \, Y(t_{1/2} + \frac{s}{2}),
\end{equation}
so that, by comparing (\ref{t-s4}) and (\ref{t-s5}),
\begin{equation}    \label{t-s6}
    \Omega(t_{1/2} - \frac{s}{2},t_{1/2} + \frac{s}{2}) = -
    \Omega(t_{1/2} + \frac{s}{2},t_{1/2} - \frac{s}{2})
\end{equation}
and thus $\Omega$ does not contain even powers of $s$. If $A(t)$
is an analytic function and a Taylor series centered around
$t_{1/2}$ is considered, then each term in $\Omega_k$ is an odd
function of $s$ and, in particular, $\Omega_{2i+1}(s) =
\mathcal{O}(s^{2i+3})$. This fact has been noticed in
\cite{iserles01rah,munthe-kaas99cia} and will be fully
exploited in section \ref{Mintegrators} when analyzing the Magnus expansion as a
numerical device for integrating differential equations.

\subsection{Convergence of the Magnus expansion}

As we pointed out in the introduction, from a mathematical point of view, there are at least two 
different  issues of paramount importance at the very basis of the Magnus expansion:
\begin{enumerate}
 \item (\emph{Existence}) For what values of $t$ and for what operators $A$ does equation
 (\ref{eq:evolution}) admit an exponential solution in the form $Y(t) = \exp(\Omega(t))$
 for a certain $\Omega(t)$? 
 \item (\emph{Convergence}) Given a certain operator $A(t)$, 
 for what values of $t$  does the Magnus
 series (\ref{M-series}) converge? In other words, when $\Omega(t)$ in (\ref{eq:Omega})
 can be obtained as the sum of the series (\ref{M-series})?
\end{enumerate}

Of course, given the relevance of the expansion, both problems have been  
extensively treated in the literature since Magnus proposed this formalism in
1954. We next review some of the most relevant 
contributions 
available regarding both aspects, with special emphasis on the convergence of the
Magnus series.

\subsubsection{On the existence of $\Omega(t)$}

In most cases one is interested in the case where $A$ belongs to a Lie algebra
$\mathfrak{g}$ under the commutator product. In this general
setting the Magnus theorem can be formulated as four statements
concerning the solution of $Y^\prime = A(t) Y$, each one more
stringent than the preceding \cite{wei63nog}. Specifically,
\begin{itemize}
 \item[(A)] The differential equation $Y^\prime = A(t) Y$
 has a solution of the form $Y(t) = \exp \Omega(t)$.
 \item[(B)] The exponent $\Omega(t)$ lies in the Lie algebra
  $\mathfrak{g}$.
 \item[(C)] The exponent $\Omega(t)$ is a continuous
 differentiable function of $A(t)$ and $t$, satisfying the
 nonlinear differential equation $\Omega^\prime = d
 \exp_{\Omega}^{-1}(A(t))$.
 \item[(D)] The operator $\Omega(t)$ can be computed by the Magnus series
(\ref{M-series}).
\end{itemize}

Let us analyze in detail now the conditions under which statements
(A)-(D) hold.

\vspace*{0.2cm}

\noindent (A) If $A(t)$ and $Y(t)$ are $n \times n$ matrices,
from well-known general theorems on differential equations it is clear that the initial
value problem defined by (\ref{eq:evolution}) and $Y(0)=I$ always
has a uniquely determined solution $Y(t)$ which is continuous and
has a continuous first derivative in any interval in which $A(t)$
is continuous \cite{coddington55tod}. Furthermore, the determinant
of $Y$ is always different from zero, since
\[
   \det Y(t) = \exp \left( \int_0^t \, \mathrm{tr}\, A(s) ds
   \right).
\]
On the other hand, it is well known that any matrix $Y$ can be
written in the form $\exp \Omega$ if and only if $\det Y \ne 0$
\cite[p. 239]{gantmacher59tto}, so that it is
always possible to write $Y(t) = \exp \Omega(t)$.

In the general context of Lie groups and Lie algebras, it is
indeed the regularity of the exponential map from the Lie algebra
$\mathfrak{g}$ to the Lie group $\mathcal{G}$ that determines the
global existence of an $\Omega(t) \in \mathfrak{g}$ \cite{dixmier57led,saito57scg}: the
exponential map of a complex Lie algebra is globally one-to-one if
and only if the algebra is nilpotent, i.e. there exists a finite
$n$ such that $\ad_{x_{1}}\ad_{x_{2}}\cdots \ad_{x_{n-1}}x_{n}=0
$, where $x_{j}$ are arbitrary elements from the Lie algebra. In
general, however, the injectivity of the exponential map is only
assured for $\xi \in \mathfrak{g}$ such that $\|\xi\| <
\rho_{\mathcal{G}}$ for a real number $\rho_{\mathcal{G}}
> 0$ and some norm in $\mathfrak{g}$ \cite{moan99ote,moan02obe}. 

\vspace*{0.2cm}

\noindent (B) Although in principle $\rho_{\mathcal{G}}$
constitutes a sharp upper bound for the mere existence of the
operator $\Omega \in \mathfrak{g}$, its practical value in the
case of differential equations is less clear. As we have noticed,
any nonsingular matrix has a logarithm, but this logarithm might
be in $\mathfrak{gl}(n,\mathbb{C})$ even when the matrix is real.
The logarithm of $Y(t)$ may be complex even for real $A(t)$
 \cite{wei63nog}. In such a situation, the
solution of (\ref{eq:evolution}) cannot be written as the
exponential of a matrix belonging to the Lie algebra over the
field of real numbers. One might argue that this is indeed
possible over the field of complex numbers, but (i) the element
$\Omega$ cannot be computed by the Magnus series (\ref{M-series}), since it
contains only real rational coefficients, and (ii) examples exist
where the logarithm of a complex matrix does not lie in the
corresponding Lie subalgebra \cite{wei63nog}.

It is therefore interesting to determine for which range of $t$ a
real $A(t)$ in (\ref{eq:evolution}) leads to a real logarithm.
This issue has been tackled by Moan in \cite{moan02obe} in the
context of a complete normed (Banach) algebra, proving that if
 \begin{equation}   \label{mo.1}
     \int_0^t \|A(s)\|_2 \, ds < \pi
 \end{equation}
then the solution of (\ref{eq:evolution}) can be written indeed as
 $Y(t) = \exp \Omega(t)$, where $\Omega(t)$ is in the Banach
 algebra.
 
\vspace*{0.2cm} 

\noindent (C) In his original paper \cite{magnus54ote},
Magnus was well aware that if
the function $\Omega(t)$ is assumed to be differentiable, it may
not exist everywhere. In fact, he related the differentiability
issue to the problem of solving $d \exp_{\Omega}(\Omega') = A(t)$
with respect to $\Omega^\prime$ and provided an implicit condition
for an arbitrary $A$. More specifically, he proved the following
result for the case of $n \times n$ matrices (Theorem V in \cite{magnus54ote}).
\begin{theorem}  \label{con-mag}
  The equation $A(t) = d \exp_{\Omega}(\Omega')$ can be solved by 
  $\Omega' = d \exp_{\Omega}^{-1} A(t)$ for an arbitrary $n \times n$ matrix
  $A$ if and only
  if none of the differences between any two of the eigenvalues of 
  $\Omega$ equals $2\pi i m$, where
 $m= \pm 1, \pm 2,\ldots$, ($m \ne 0$).
\end{theorem}  
This result can be considered, in fact, as a reformulation of Lemma \ref{lemma2},
but, unfortunately, 
has not very much practical application unless 
the eigenvalues of $\Omega$ can easily be determined from those of $A(t)$. One
would like instead to have conditions based directly on $A$.

\subsubsection{Convergence of the Magnus series}
\label{convergence}

For dealing with the validity of statement (D) one has to analyze
the convergence of the series $\sum_{k=1}^\infty \Omega_k$. Magnus
also considered the question of when the series terminates at some
finite index $m$, thus giving a globally valid $\Omega = \Omega_1
+ \cdots +\Omega_m$. This will happen, for instance, if
\[
   \left[ A(t), \int_0^t A(s) ds \right] = 0
\]
identically for all values of $t$, since then $\Omega_k = 0$ for
$k >1$. A sufficient (but not necessary) condition for the
vanishing of all terms $\Omega_k$ with $k > n$ is that
\[
    [A(s_1), [A(s_2), [A(s_3), \cdots ,[A(s_n), A(s_{n+1})] \cdots
    ]]] = 0
\]
for any choice of $s_1, \ldots, s_{n+1}$. In fact, the termination
of the series cannot be established solely by consideration of the
commutativity of $A(t)$ with itself, and Magnus considered an
example illustrating this point.

In general, however, the Magnus series does not converge unless
$A$ is small in a suitable sense. Several bounds to the actual
radius of convergence in terms of $A$ have been obtained in the
literature. Most of these results can be stated as follows. 
If $\Omega_m(t)$ denotes the homogeneous element with $m-1$
commutators in the Magnus series as given by (\ref{recur2}),
then $\Omega(t) = \sum_{m=1}^{\infty} \Omega_m(t)$ is absolutely
convergent for $0 \le t < T$, with
\begin{equation}   \label{cm1.1}
    T = \max \left\{ t \ge 0 \, : \, \int_0^t \|A(s)\|_2 \, ds < r_c
      \right\}.
\end{equation}
Thus, both Pechukas and Light \cite{pechukas66ote} and
Karasev and Mosolova \cite{karasev77ipa} obtained $r_c=\log 2 =
0.693147\ldots$, whereas Chacon and Fomenko \cite{Chacon91rff} got
$r_c=0.57745\ldots$. In 1998, Blanes \textit{et al.}
\cite{blanes98maf} and Moan \cite{moan98eao} obtained independently
the improved bound
\begin{equation}   \label{cm1.2}
    r_c = \frac{1}{2}  \int_0^{2 \pi} \frac{1}{2 + \frac{x}{2} (1 - \cot \frac{x}{2})}
     dx  \equiv \xi = 1.08686870\ldots
\end{equation}
by analyzing the recurrence (\ref{eses})-(\ref{omegn}) and (\ref{recur2}),
respectively. Furthermore, Moan also obtained a bound on the 
individual terms $\Omega_m$ of the Magnus series \cite{moan02obe} which is useful,
in particular, for estimating
errors when the series is truncated. Specifically, he showed that
 \[
   \|\Omega_m(t)\| \le \frac{f_m}{2}
      \left( 2 \int_0^t \|A(s)\|_2 \, ds \right)^m \le
      \pi \left( \frac{1}{\xi} \int_0^t \|A(s)\|_2 \, 
      ds \right)^m,
 \]
 where $f_m$ are the coefficients of 
 \[
   G^{-1}(x) = \sum_{m \ge 1} f_m x^m = x + \frac{1}{4} x^2 +
   \frac{5}{72} x^3 + \frac{11}{576} x^4 + \frac{479}{86400} x^5 +
   \cdots,
 \]
 the inverse function of
\[
     G(s) = \int_0^s \frac{1}{2 + \frac{x}{2}(1- \cot \frac{x}{2})} \, dx.
\]

On the other hand, by analyzing 
some selected examples, Moan \cite{moan02obe}
concluded that, in order to get convergence \emph{for all real matrices}
$A(t)$, necessarily $r_c \le \pi$ in (\ref{cm1.1}), and more recently Moan and Niesen
\cite{moan06cot} have been able to prove that indeed $r_c=\pi$ if only
real matrices are involved.

In any case, it is important to remark that statement (D) is locally valid,
but cannot be used to compute $\Omega$ in the large. However, as
we have seen, the other statements need not depend on the validity
of (D). In particular, if (B) and (C) are globally valid, one can
still investigate many of the properties of $\Omega$ even though
one cannot compute it with the aid of (D).

\subsubsection{An improved radius of convergence}
\label{subsec273}

The previous results on the convergence of the Magnus series have
been established for $n \times n$ real matrices: if $A(t)$ is a real $n \times n$ matrix, then
(\ref{mo.1}) gives a condition for $Y(t)$ to have a real logarithm. In fact,
under the same condition, the Magnus series (\ref{M-series}) converges precisely
to this logarithm, i.e., its sum $\Omega(t)$ satisfies $\exp(\Omega(t)) = Y(t)$
\cite{moan06cot}.

One should have in mind, however, that the original expansion was conceived 
by requiring only that $A(t)$ be a linear operator depending on a real variable $t$
in an associative ring (Theorem \ref{thMag}). The idea was to define,
in terms of $A$, an operator $\Omega(t)$ such that 
the solution of the initial value problem
$Y^\prime = A(t) Y$,  $Y(0)=I$,
for a second operator $Y$ is given as $Y = \exp \Omega$.
The proposed expression for $\Omega$ is an infinite series satisfying the condition that
 ``its partial sums become Hermitian after multiplication by $i$ if $i A$ is a Hermitian operator"
\cite{magnus54ote}.  As this quotation illustrates, Magnus expansion was first derived 
 in the context of quantum mechanics, and so one typically assumes that it is
also valid when $A(t)$ is a linear operator in a Hilbert space. Therefore, it might be
desirable to have conditions for the convergence of the Magnus series in this more general
setting. In \cite{casas07scf}, by applying standard techniques of
complex analysis and some elementary properties of the unit sphere, the bound
$r_c = \pi$ has been shown to be also valid for \emph{any} bounded normal 
operator $A(t)$ in a Hilbert space of arbitrary dimension. 
Next we review the main issues involved and refer the reader to  \cite{casas07scf}
for a more detailed treatment.

Let us assume that $A(t)$ is a bounded operator in a Hilbert space $\mathcal{H}$,
with $2 \le \mathrm{dim} \ \mathcal{H} \le \infty$. Then we introduce a new parameter
$\varepsilon \in \mathbb{C}$ and denote by $Y(t;\varepsilon)$  the 
solution of the initial value problem
\begin{equation}   \label{cmeq.3.1}
   \frac{dY}{dt} = \varepsilon A(t) Y,  \qquad Y(0)=I,
\end{equation}
where now $I$ denotes the identity operator in $\mathcal{H}$.  It is known that  
$Y(t;\varepsilon)$ is an analytic function of $\varepsilon$ for a fixed value of $t$. 
Let us introduce the set $B_{\gamma} \subset \mathbb{C}$ 
characterized by the real parameter $\gamma$,
\[
  B_{\gamma} = \{ \varepsilon \in \mathbb{C} \, : \, |\varepsilon| 
         \int_0^t \|A(s)\| ds <  \gamma \}.
\]            
Here $\|.\|$ stands for the norm defined by the inner product on $\mathcal{H}$,
i.e., the $2$-norm introduced in subsection \ref{notations}.

If $t$ is fixed, the operator function 
$\varphi(\varepsilon) = \log Y(t;\varepsilon)$ is well defined
in $B_{\gamma}$ when $\gamma$ is small enough, say $\gamma < \log 2$, as an
analytic function of $\varepsilon$. 
As a matter of fact, this is a direct consequence of the results collected in section \ref{convergence}:
if, in particular, $|\varepsilon| \int_0^t \|A(s)\| ds < \log 2$, the 
Magnus series corresponding to (\ref{cmeq.3.1}) converges and its sum
$\Omega(t; \varepsilon)$ satisfies $\exp(\Omega(t; \varepsilon)) = Y(t; \varepsilon)$.
In other words, the power series $\Omega(t; \varepsilon)$ coincides with
$\varphi(\varepsilon)$ when $|\varepsilon| \int_0^t \|A(s)\| ds < \log 2$, and so
the Magnus series is the power series expansion of $\varphi(\varepsilon)$ around
$\varepsilon = 0$.
\begin{theorem}   \label{main-theorem}
       The function $\varphi(\varepsilon) = \log Y(t;\varepsilon)$ is an analytic 
       function of $\varepsilon$ in the set $B_{\pi}$, with
\[
  B_{\pi} = \{ \varepsilon \in \mathbb{C} \, : \, |\varepsilon| 
         \int_0^t \|A(s)\| ds <  \pi \}.
\]           
If $\mathcal{H}$ is infinite dimensional, the statement holds true if $Y$ is a normal operator.            
\end{theorem}
In other words, $\gamma = \pi$. The proof of this theorem is based 
on some elementary properties of the unit
sphere $S^1$ in a Hilbert space. Let us define the angle between any two vectors
$x \ne 0$, $y \ne 0$ in $\mathcal{H}$, $\mathrm{Ang}\{x,y\} = \alpha$, $0 \le \alpha \le \pi$, from
\[
      \cos \alpha = \frac{\mathrm{Re} \langle x, y \rangle}{\|x\| \, \|y\|},
\]    
where $\langle \cdot, \cdot \rangle$ is the inner product on $\mathcal{H}$.
This angle is a metric in $S^1$, i.e., the triangle inequality holds in $S^1$. 

The first property we need
is given by the next lemma \cite{moan02obe}.
\begin{lemma}  \label{lem-moan}
 For any $x \in \mathcal{H}$, $x \ne 0$, 
$\mathrm{Ang}\{Y(t;\varepsilon) x, x\} \le |\varepsilon| \int_0^t \|A(s)\| ds$.
\end{lemma}  
Observe that if $Y$ is a normal operator in $\mathcal{H}$, i.e., $Y Y^{\dag} = Y^{\dag} Y$
(in particular, if $Y$ is unitary), then
$\|Y^{\dag} x\| = \|Yx\|$ for all $x \in \mathcal{H}$ and therefore 
$\mathrm{Ang}\{Y^{\dag} x, x\} = \mathrm{Ang}\{Yx, x\}$.

The second required property provides useful information on the location of the 
eigenvalues of a given bounded operator in $\mathcal{H}$ \cite{mityagin90unp}.
  \begin{lemma}  \label{lem-mityagin}
  Let $T$ be a (bounded) operator on $\mathcal{H}$. 
    If $\mathrm{Ang}\{T x,x\} \le \gamma$ and $\mathrm{Ang}\{T^{\dag} x,x\} \le \gamma$
    for any $x \ne 0$, $x \in \mathcal{H}$, where $T^{\dag}$ denotes the adjoint operator of $T$,
    then the spectrum of $T$, $\sigma(T)$, is contained in the set
    \[
       \Delta_{\gamma} = \{ z = |z| \e^{i \omega} \in \mathbb{C} \, : \, |\omega| \le \gamma\}.
    \]   
  \end{lemma}
  \begin{proof}
(of Theorem \ref{main-theorem}). Let us introduce the 
operator $T \equiv Y(t;\epsilon)$, with $\varepsilon \in B_{\gamma}$, $\gamma < \pi$. Then
by Lemma \ref{lem-moan}, $\mathrm{Ang}\{T x,x\} \le \gamma$ for all $x \ne 0$, and thus, by
Lemma \ref{lem-mityagin}, 
\begin{equation}   \label{eq.pr.1}
    \sigma(T) \subset  \Delta_{\gamma}.
\end{equation}    
If $\mathrm{dim} \ \mathcal{H} = \infty$ and we assume that $Y(t;\epsilon)$ is a normal operator,
then (\ref{eq.pr.1}) also holds.

From equation (\ref{cmeq.3.1}) in integral form,
\[
  Y(t; \varepsilon) = I + \varepsilon \int_0^t A(s) Y ds,
\]
one gets $\|Y\| \le 1 + |\varepsilon| \int_0^t \|A(s)\| \, \|Y\| ds$, and application of Gronwall's
lemma \cite{gronwall19not} leads to
\[
    \|Y(t;\varepsilon)\| \le \exp \left( |\varepsilon| \int_0^t \|A(s)\| ds \right).
\]
An analogous reasoning for the inverse operator also proves that      
\[
    \|Y^{-1}(t;\varepsilon)\| \le \exp \left( |\varepsilon| \int_0^t \|A(s)\| ds \right).
\]
In consequence,
\[
   \|T\| \le \e^{\gamma}  \qquad \mbox{ and } \qquad \|T^{-1}\| \le \e^{\gamma}.
\]

If $\lambda \ne 0 \in \sigma(T)$, then $|\lambda| \le \|T\|$ \cite{hunter01aan}
and therefore $|\lambda| \le \e^{\gamma}$. In addition, $\frac{1}{\lambda} \in
\sigma(T^{-1})$, so that $|\lambda| \ge \e^{-\gamma}$. Equivalently,
\begin{equation}  \label{eq.pr.2}
  \sigma(T) \subset \{ z \in \mathbb{C} : \e^{-\gamma} \le |z| \le \e^{\gamma} \} \equiv G_{\gamma}.   
\end{equation}
Putting together (\ref{eq.pr.1}) and (\ref{eq.pr.2}), one has
\[
   \sigma(T) \subset G_{\gamma} \cap \Delta_{\gamma} \equiv \Lambda_{\gamma}.
\]
Now choose any value $\gamma'$ such that  $\gamma < \gamma' < \pi$ 
(e.g., $\gamma' = (\gamma + \pi)/2$) and consider the closed curve 
$\Gamma = \partial \Lambda_{\gamma'}$.  Notice that the curve $\Gamma$
encloses $\sigma(T)$ in its interior, so that it is possible to define \cite{dunford58lop} the function
$\varphi(\varepsilon) = \log Y(t; \varepsilon)$ by the equation 
 \begin{equation}   \label{eq.pr.3}
      \varphi(\epsilon) = \frac{1}{2 \pi i} \int_{\Gamma} \log z \, (z I - Y(t;\epsilon))^{-1}  \, dz,
  \end{equation}
where the integration along $\Gamma$ is performed in the counterclockwise direction.
As is well known, (\ref{eq.pr.3}) defines an analytic function of $\varepsilon$ in
$B_{\gamma'}$ \cite{dunford58lop} and the result of the theorem follows.
\end{proof}

\begin{theorem}  \label{conv-mag}
  Let us consider the differential equation $Y^\prime = A(t) Y$ defined in 
  a Hilbert space $\mathcal{H}$ with $Y(0)=I$, and let $A(t)$ be a bounded operator
  in $\mathcal{H}$. Then, the Magnus series 
  $\Omega(t) = \sum_{k=1}^{\infty}  \Omega_k(t)$, with
$\Omega_k$ given by (\ref{recur2}) converges  in the interval $ t \in [0,T)$  such that
\[
   \int_0^T \|A(s)\| ds < \pi
\]
and the sum $\Omega(t)$ satisfies $\exp \Omega(t) = Y(t)$. The statement also holds when
$\mathcal{H}$ is infinite-dimensional if $Y$ is a normal operator (in particular, if $Y$ is
unitary).
\end{theorem}    

\begin{proof}
Theorem \ref{main-theorem} shows that $\log Y(t;\varepsilon) \equiv \varphi(\varepsilon)$ is a 
 well defined and analytic function of $\varepsilon$ for 
  \[
       |\varepsilon| \int_0^t \|A(s)\| ds  < \pi.
  \]     
It has also been shown that the Magnus series 
$\Omega(t;\varepsilon) = \sum_{k=1}^{\infty} \varepsilon^k \Omega_k(t)$, with
$\Omega_k$ given by (\ref{recur2}), is absolutely convergent when
$|\varepsilon| \int_0^t \|A(s)\| ds  < \xi = 1.0868...$ and its sum satisfies
$\exp \Omega(t;\varepsilon) = Y(t;\varepsilon)$. Hence, the Magnus series is
the power series of the analytic function $\varphi(\varepsilon)$ in the disk 
$|\varepsilon| < \xi / \int_0^t \|A(s)\| ds$. But $\varphi(\varepsilon)$ is analytic
in $B_{\pi} \supset B_{\xi}$ and the power series has to be unique. In consequence,
the power series of $\varphi(\varepsilon)$ in $B_{\pi}$ has to be the same as the power
series of $\varphi(\varepsilon)$ in $B_{\xi}$, which is precisely the Magnus series. Finally,
by taking $\varepsilon = 1$ we get the desired result.
\end{proof}

\subsubsection{Further discussion}

Theorem \ref{conv-mag} provides sufficient conditions for the convergence 
of the Magnus series based on an estimate by the norm of the operator $A$. In
particular, it guarantees that the operator $\Omega(t)$  in $Y(t) = \exp \Omega(t)$
can safely be obtained with the convergent series $\sum_{k \ge 1} \Omega_k(t)$ for
$0 \le t < T$ when the terms $\Omega_k(t)$ are computed with (\ref{recur2}).  
A natural question at this
stage is what is the optimal convergence domain. In other words,
is the bound estimate $r_c = \pi$ given by Theorem \ref{conv-mag}
sharp or is there still room for improvement? In order to clarify this issue,
we next analyze two simple examples involving $2 \times 2$ matrices.

\noindent \textbf{Example 1}. Moan and Niesen \cite{moan06cot} consider the 
coefficient matrix
 \begin{equation}   \label{ej1.1}
   A(t) = \left(  \begin{array}{rr}
        2  &  \ t \\
        0  &  -1
           \end{array}   \right).
 \end{equation}
If we introduce, as before, the complex parameter $\varepsilon$ in the problem, the
corresponding exact solution $Y(t;\varepsilon)$ of (\ref{cmeq.3.1}) is given by
 \begin{equation}   \label{ej1.2}
   Y(t;\varepsilon) = \left(  \begin{array}{lc}
        \e^{2 \varepsilon t}   &  \ \ \frac{1}{9 \varepsilon} \e^{2 \varepsilon t} - \left(
           \frac{1}{9 \varepsilon} + \frac{1}{3} t \right) \e^{-\varepsilon t}   \\
        0  &   \e^{-\varepsilon t} 
           \end{array}   \right)
 \end{equation}
and therefore
   \[
       \log Y(t;\varepsilon) = \left(  \begin{array}{rc}
        2t  &  \ g(t;\varepsilon) \\
        0  & -t
           \end{array}   \right), \quad \mbox{ with } \ \ g(t;\varepsilon) = \frac{t (1 - \e^{3 \varepsilon t} + 
           3 \varepsilon t)}{3(1- \e^{3 \varepsilon t})}.
    \]           
The Magnus series can be obtained by computing the Taylor expansion of  $\log Y(t;\varepsilon)$ 
around $\varepsilon = 0$. Notice that the function $g$ has a singularity when
$\varepsilon t = \frac{2\pi}{3} i$, and thus, by taking  $\varepsilon = 1$, the Magnus series only converges up to $t=  \frac{2}{3} \pi$. On the other hand, condition  
$\int_0^T \|A(s)\| ds  < \pi$ leads to $T \approx 1.43205 < \frac{2}{3} \pi$. In consequence, the
actual convergence domain of the Magnus series is larger than the estimate provided
by Theorem \ref{conv-mag}.
\hfill{$\Box$}

\vspace*{0.3cm}

\noindent \textbf{Example 2}. Let us introduce the matrices
\begin{equation}   \label{ej2.1}
    X_{1}=  \left(
  \begin{array}{cr}
  1  & \  0 \\
  0  &  -1
  \end{array} \right) = \sigma_3, \qquad
  X_{2}=  \left(
  \begin{array}{cc}
  0  & \  1 \\
  0  & 0
  \end{array} \right)
\end{equation}
and define 
\[
  A(t)= \left\{
  \begin{array}{ll}
  \beta \, X_{2} \quad & 0\leq t \leq 1 \\
  \alpha \, X_{1} \quad & t > 1
  \end{array} \right.
\]  
with $\alpha, \beta$ complex constants. Then, the solution of
equation (\ref{eq:evolution}) at $t=2$ is 
\[
   Y(2) = \e^{\alpha X_{1}} \e^{\beta X_{2}} = \left(
  \begin{array}{lc}
  \e^{\alpha} \quad & \beta \e^{\alpha} \\
  0 & \e^{-\alpha}
  \end{array} \right),
\]
so that
\begin{equation}   \label{ex1.4}
 \log Y(2) = \log(\e^{\alpha X_{1}} \e^{\beta X_{2}})  = \alpha X_{1} +
    \frac{2\alpha \beta}{1- \e^{-2\alpha}} \, X_{2},
\end{equation}
an analytic function if $|\alpha| < \pi$ with first singularities
at $\alpha = \pm i \pi$. Therefore, the Magnus series
 cannot converge at $t=2$ if $|\alpha| \ge \pi$, independently
of $\beta \ne 0$, even when it is possible in this case to get a closed-form expression
for the general term. Specifically, a straightforward computation with 
the recurrence (\ref{eses})-(\ref{omegn}) shows that
\begin{equation}  \label{ex2.2}
  \sum_{n=1}^{\infty} \Omega_{n}(2) =
     \alpha X_{1} + \beta X_{2} + \sum_{n=2}^{\infty}
    (-1)^{n-1}\frac{2^{n-1}B_{n-1}}{(n-1)!} \ \alpha^{n-1}
     \beta \, X_{2}.
\end{equation}

If we take the spectral norm, then $\|X_1\| = \|X_2\| = 1$ and
\[
   \int_0^{t=2} \|A(s)\| ds = |\alpha| + |\beta|,
\]
so that the convergence domain provided by
Theorem \ref{conv-mag} is $|\alpha| + |\beta| < \pi$ for this example. Notice that in the
limit $|\beta| \rightarrow 0$ this domain is optimal.
\hfill{$\Box$}

\vspace*{0.3cm}

From the analysis of Examples 1 and 2 we can conclude the following. First,
the convergence domain of the Magnus series provided by Theorem \ref{conv-mag}
is the best result one can get for a generic bounded operator $A(t)$ in a Hilbert
space, in the sense that one may consider specific matrices $A(t)$, as in Example 2, where
the series diverges for any time $t$ such that  $\int_0^t \|A(s)\| ds > \pi$. Second, there
are also situations (as in Example 1) where the bound estimate $r=\pi$ 
is still rather conservative: \emph{the  Magnus series converges indeed for a 
larger time interval than that
given by Theorem \ref{conv-mag}}. This is particularly evident if one takes 
$A(t)$ as a diagonal matrix: then, 
the exact solution $Y(t;\varepsilon)$ of (\ref{cmeq.3.1}) is a diagonal matrix whose elements
are non-vanishing entire functions of $\varepsilon$, and obviously $\log Y(t;\varepsilon)$
is also an entire function of $\varepsilon$. In such circumstances, the convergence domain
 $|\varepsilon| \int_0^t \|A(s)\| ds  < \pi$ for the Magnus series does not make much sense.  
Thus a natural question arises: is it possible to obtain a more precise criterion of
convergence? In trying to answer this question, in \cite{casas07scf} an alternative
characterization of the convergence has been developed which is valid for $n \times n$
complex matrices.  More precisely, a connection has been established between the convergence
of the Magnus series and the existence of multiple eigenvalues of the 
fundamental matrix $Y(t; \varepsilon)$ for a fixed $t$, which we denote by
$Y_t(\varepsilon)$. By using the theory of analytic matrix functions,  and in particular,
of the logarithm of an 
analytic matrix function (such as is done e.g. in \cite{yakubovich75lde}), the following
result has been proved in \cite{casas07scf}: 
if the analytic matrix function $Y_t(\varepsilon)$
has an eigenvalue $\rho_0(\varepsilon_0)$ of multiplicity $l > 1$ 
for a certain $\varepsilon_0$ such that: (a) there is a curve in the $\varepsilon$-plane
joining $\varepsilon = 0$ with $\varepsilon = \varepsilon_0$, and
(b) the number of equal terms in 
$\log \rho_1(\varepsilon_0)$, $\log \rho_2(\varepsilon_0)$, $\ldots, 
\log \rho_{l}(\varepsilon_0)$ such that  $\rho_k(\varepsilon_0) = \rho_0$, $k=1,\ldots, l$ is less
than the maximum dimension of the elementary Jordan block corresponding to $\rho_0$, then
the radius of convergence of the series $\Omega_t(\varepsilon) = 
\sum_{k \ge 1} \varepsilon^k \Omega_{t,k}$ verifying
$\exp \Omega_t(\varepsilon) = Y_t(\varepsilon)$ is precisely $r = |\varepsilon_0|$. An analysis
 along the same line has been carried out in \cite{veshtort03nsi}.

When this criterion is applied to Example 1, it gives as the radius of convergence of the
Magnus series corresponding to equation (\ref{cmeq.3.1}) for a fixed $t$,
 \begin{equation}  \label{Y-Seq1}
    \Omega_t(\varepsilon) = \sum_{k=1}^{\infty} \varepsilon^k  \ \Omega_{t,k}, 
 \end{equation}
the value
\begin{equation}   \label{cex1}
        r = |\varepsilon| = \frac{2 \pi}{3t}.
\end{equation}
To get the actual convergence domain of the usual Magnus expansion 
we have to take $\varepsilon = 1$, and so, from (\ref{cex1}), we get $2\pi/(3t) = 1$,
or equivalently $t = 2 \pi/3$, i.e., the result achieved from the analysis of the exact solution.

With respect to Example 2, one gets \cite{casas07scf}
\[
    |\varepsilon| =  \frac{\pi}{|\alpha| (t-1)}.
\]
If we now fix $\varepsilon = 1$, the actual $t$-domain of convergence of the Magnus series
is
\[
    t = 1 + \frac{\pi}{|\alpha|}.
\]        
Observe that, when $t=2$, we get $|\alpha| = \pi$ and thus the previous
result  is recovered: the Magnus
series converges only for $|\alpha| < \pi$. 

It should also be mentioned that
the case
of a diagonal matrix $A(t)$ is compatible with this alternative
characterization \cite{casas07scf}.

\subsection{Magnus expansion and the BCH formula}

The Magnus expansion can also be used to get explicitly
the terms of the series $Z$ in
\[
    Z = \log( \e^{X_{1}} \, \e^{X_{2}} ),
\]
$X_{1}$ and $X_{2}$ being two non commuting indeterminate variables. As it
is well known \cite{postnikov94lga},
\begin{equation}  \label{eq.5.0}
    Z = X_{1} + X_{2} + \sum_{n=2}^{\infty} G_n(X_{1},X_{2}),
\end{equation}
where $G_n(X_{1},X_{2})$ is a homogeneous Lie polynomial in  $X_{1}$ and $X_{2}$ of
grade $n$; in other words, $G_n$ can be expressed in terms of $X_{1}$
and $X_{2}$ by addition, multiplication by rational numbers and nested
commutators. This result is often known as the
Baker--Campbell--Hausdorff (BCH) theorem and proves to be very
useful in various fields of mathematics (theory of linear
differential equations \cite{magnus54ote}, Lie group theory
\cite{gorbatsevich97fol}, numerical analysis \cite{hairer06gni})
and theoretical physics (perturbation theory, transformation
theory, Quantum Mechanics and Statistical Mechanics
\cite{kumar65oet,weiss62tbf,wilcox67eoa}). In particular, in the
theory of Lie groups, with this theorem one can explicitly write
the operation of multiplication in a Lie group in canonical
coordinates in terms of the Lie bracket operation in its 
algebra and also prove the existence of a local Lie group with a
given Lie algebra \cite{gorbatsevich97fol}.

If $X_{1}$ and $X_{2}$ are matrices and one considers the piecewise
constant matrix-valued function
\begin{equation}   \label{eq.2.2.3b}
   A(t) = \left\{  \begin{array}{ccl}
           X_{2} &  \quad &  0 \le t \le 1  \\
                   X_{1} & \quad  &  1 < t \le 2
                    \end{array}  \right.
\end{equation}
then the exact solution of (\ref{eq:evolution}) at $t=2$ is $Y(2)
= \e^{X_{1}} \, \e^{X_{2}}$. By computing $Y(2) = \e^{\Omega(2)}$ with
recursion (\ref{recur2}) one gets for the first terms
\begin{eqnarray}   \label{eq.2.2.4}
   \Omega_1(2) & = & X_{1} + X_{2}  \nonumber  \\
   \Omega_2(2) & = & \frac{1}{2} [X_{1},X_{2}]   \\
   \Omega_3(2) & = & \frac{1}{12} [X_{1},[X_{1},X_{2}]] - \frac{1}{12}
                 [X_{2},[X_{1},X_{2}]]  \nonumber  \\
   \Omega_4(2) & = & \frac{1}{24} [X_{1},[X_{2},[X_{2},X_{1}]]]  \nonumber
\end{eqnarray}
In general, each $G_n(X_{1},X_{2})$ is a linear combination of the
commutators of the form $[V_1,[V_2, \ldots,[V_{n-1},V_n] \ldots]]$
with $V_i \in \{X_{1},X_{2}\}$ for $1 \le i \le n$, the coefficients being
universal rational constants. This is perhaps one of the reasons why
the Magnus expansion is often referred to in the literature as the
continuous analogue of the BCH formula. As a matter of fact,
Magnus proposed a different method for obtaining the first terms in
the series (\ref{M-series}) based on (\ref{eq.5.0}) \cite{magnus54ote}.

Now we can apply Theorem \ref{conv-mag} and obtain the following sharp bound.
\begin{theorem}
The Baker--Campbell--Hausdorff series in the form (\ref{eq.5.0})
converges absolutely when $\|X_{1}\| + \|X_{2}\| < \pi$.
\end{theorem}
This result can be generalized, of course, to any number of non
commuting operators $X_1, X_2, \ldots, X_q$. Specifically, the series
\[
    Z = \log( \e^{X_{1}} \, \e^{X_{2}} \, \cdots \, \e^{X_{q}} ),
\]
converges absolutely if $\|X_{1}\| + \|X_{2}\| + \cdots + \|X_q\| < \pi$.

\subsection{Preliminary linear transformations}\label{PLT}

To improve the accuracy and the bounds on the convergence domain
of the Magnus series for a given problem, it is quite common to
consider first a linear transformation on the system in such a way that
the resulting differential equation can be more easily handled in a certain sense to be specified
for each problem. To illustrate the procedure, let us consider a simple example.

\noindent \textbf{Example}. Suppose we have the $2 \times 2$ matrix
\begin{equation}  \label{ex1.7}
   A(t) = \alpha(t)  X_{1} + \beta(t) X_{2},
\end{equation}
where $X_1$ and $X_2$ are given by (\ref{ej2.1}) and 
$\alpha$ and $\beta$ are complex functions of time,
$\alpha, \beta : \mathbb{R} \longrightarrow \mathbb{C}$. Then the
exact solution of  $Y^{\prime} = A(t)  Y$, $Y(0)=I$ is 
\begin{equation}\label{ex1.8}
 Y(t) =  \left(
  \begin{array}{cc}
 \e^{\int_0^t\alpha(s)ds}  \quad &
 \displaystyle \int_0^t ds_{1}\e^{\int_{s_1}^t\alpha(s_2)ds_2} \ \beta(s_1) \
 \e^{-\int_{0}^{s_1}\alpha(s_2)ds_2} \\
  0  \quad & \e^{-\int_0^t\alpha(s)ds}
  \end{array} \right).
\end{equation}
Let us factorize the solution as $Y(t)= \tilde{Y}_0(t) \tilde{Y}_1(t)$, with $\tilde{Y}_0(t)$ the
solution of the initial value problem defined by 
\begin{equation}\label{ej1.11}
 \tilde{Y}_0^{\prime} = A_0(t) \, \tilde{Y}_0  \qquad  A_0(t) = \alpha(t) X_1 = \left(
  \begin{array}{cc}
  \alpha(t) & \ 0 \\
  0 & -\alpha(t)
  \end{array} \right)
\end{equation}
and $\tilde{Y}_0(0)=I$. Then,  the
equation satisfied by $\tilde{Y}_1$ is
\begin{equation}\label{ej1.12}
  \tilde{Y}_1^{\prime} = A_1(t) \tilde{Y}_1, \qquad \mbox{ with } \quad A_1=\tilde{Y}_0^{-1}A \, \tilde{Y}_0
  - A_0 = \left(
  \begin{array}{cc}
 \quad 0 \quad  &  \beta(t) \, \e^{2\int_0^t\alpha(s)ds}  \\
  0   & 0
  \end{array} \right),
\end{equation}
so that the first term of the Magnus expansion applied to (\ref{ej1.12})
already provides the exact solution (\ref{ex1.8}). 
\hfill{$\Box$}

This, of course, is not
the typical behaviour, but in any case if the transformation $\tilde{Y}_0$ in the factorization
$Y(t)= \tilde{Y}_0(t) \tilde{Y}_1(t)$ is chosen
appropriately,  the first few terms in the Magnus series applied to the equation satisfied by
of $\tilde{Y}_1$ give usually very accurate approximations.

Since this kind of preliminary transformation is frequently used in Quantum
Mechanics, we specialize the treatment to this particular setting
here and consider equation (\ref{laecuacion}) instead. In other words, we
write (\ref{eq:evolution}) in the more conventional form
of the time
dependent Sch\"odinger equation
\begin{equation}   \label{eq:evolU}
   \frac{dU(t)}{dt} = \tilde{H}(t) U(t),
\end{equation}
where $\tilde{H} \equiv H/(i \hbar)$, $\hbar$ is the reduced Planck constant, $H$ is the Hamiltonian
and $U$ corresponds to the evolution operator.

As in the example, suppose that $\tilde{H}$ can be split into two pieces,
$\tilde{H} = \tilde{H}_{0} + \varepsilon \tilde{H}_{1}$, with
$\tilde{H}_{0}$ a solvable Hamiltonian and $\varepsilon \ll 1$ a
small perturbation parameter. In such a situation one tries to
integrate out the $\tH _0$ piece so as to circumscribe the
approximation to the $\tH _1$ piece. In the case of 
equation (\ref{eq:evolU}) this is carried out by means of
a linear time-dependent transformation. In Quantum Mechanics this
preliminary linear transformation corresponds to a new evolution
picture, such as the interaction or the adiabatic picture.  

 Among other possibilities we may factorize the
time-evolution operator as
\begin{equation}\label{Ufac}
  U(t)=G(t)U_G(t) G^\dag (0),
\end{equation}
where $G(t)$ is a linear transformation whose purpose is to be
defined yet. In the new $G$-Picture, the corresponding time-evolution operator
$U_G$ obeys the equation
\begin{equation}\label{UGeq}
   U_G^{\prime}(t)=\tH _G(t)U_G(t),\quad \tH _G(t)= G^\dag (t)\tH (t)G(t)-G^\dag
  (t) G^{\prime}(t).
\end{equation}
The choice of $G$ depends on the nature of the problem at hand.
There is no generic formal recipe to find out the most appropriate
$G$. In the spirit of canonical transformations of Classical
Mechanics, one should built up the very $U_G$  perturbatively.
However, the aim here is different because $G$ is defined from the
beginning. Two rather common choices are:
\begin{itemize}
\item \textbf{Interaction Picture}. It is well suited when $\tH _0(t)$ is
diagonal in some basis, or else, it is constant. In that case
\begin{equation}\label{GInt}
  G(t)=\exp\left( \int_0^t \tH _0(\tau) d \tau \right)
\end{equation}
so that
\begin{equation}\label{HGInt}
  \tH _G(t)=\varepsilon \exp\left( -\int_0^t \tH _0(\tau) d \tau\right)
  \tH _1(t) \exp\left( \int_0^t \tH _0(\tau) d
  \tau\right) .
\end{equation}

\item \textbf{Adiabatic Picture}. A time scale of the
system  much smaller than that of the interaction defines an
adiabatic regime. For instance, suppose that the Hamiltonian 
operator $H(t)$ depends smoothly on time through the variable
$\tau = t/T$, where $T$ determines the time scale and
$T \rightarrow \infty$. Then the quantum mechanical evolution of the system
is described by $dU/dt = \tilde{H}(\varepsilon t) U$, with
$\varepsilon \equiv 1/T \ll 1$, or equivalently
\begin{equation}   \label{adiab1}
      \frac{dU(\tau)}{d\tau} =  \frac{1}{\varepsilon} \tilde{H}(\tau) U(\tau),
\end{equation}
with $\tau \equiv \varepsilon t$. In this case
the appropriate transformation is a $G$ that
renders $\tH (t)$ instantaneously diagonal, i.e.,
\begin{equation}\label{GAdiab}
  G^\dag (t)\tH (t)G(t)=E(t)={\rm diag}[E_1(t),E_2(t),\ldots ].
\end{equation}
The term $G^\dag  G^{\prime}$  of the new Hamiltonian in (\ref{UGeq})
is, under adiabatic conditions, very small. Its main diagonal
generates the so-called Berry, or geometric, phase \cite{berry84qpf}.
\end{itemize}

Both types of $G$ do not exclude mutually, but they may be used in
succession. As a matter of fact, corrections to the adiabatic
approximation must be followed by the former one. In turn, an
adiabatic transformation may be iterated, as proposed by Garrido
\cite{garrido64gai} and Berry \cite{berry87qpc}. 

In section 4 we shall use extensively
these preliminary linear transformations on several
standard problems of Quantum Mechanics to illustrate the practical features of the
Magnus expansion.

\subsection{Exponential product representations}

In contrast to Magnus expansion, much less attention has been paid
to solutions of (\ref{eq:evolution}) in the form of a product of
exponential operators. Both approaches are by no means equivalent,
since in general the operators $\Om_n$ do not commute with each
other. For instance, for a quantum system as in equation
(\ref{eq:evolU}), the ansatz $U=\prod \exp (\Phi_n)$ (where
$\Phi_n$ are skew-Hermitian operators to be determined) is an
alternative to the Magnus expansion, also preserving the unitarity
of the time-evolution operator. One such procedure was devised by
Fer  in 1958 in a paper devoted to the study of systems of linear 
differential equations \cite{fer58rdl}. Although the original
result obtained by Fer was cited and explicitly stated by Bellman
\cite[p. 204]{bellman60itm}, sometimes it has been misquoted as a
reference for the Magnus expansion \cite{baye73ats}. On the other
hand, Wilcox associated Fer's name with an interesting alternative
infinite product expansion which is indeed a continuous analogue
of the Zassenhaus formula \cite{wilcox67eoa} (something also
attributed to the Fer factorization \cite[p. 372]{magnus76cgt}).
This however also led to some confusion since his approach is in
the spirit of perturbation theory, whereas Fer's original one was
essentially nonperturbative. The situation was clarified in
\cite{klarsfeld89eip}, where also some applications to Quantum
Mechanics were carried out for the first time.

In this section we discuss briefly the main features of the Fer
and Wilcox expansions, and how the latter can be derived from the
successive terms of the Magnus series. This will make clear the
different character of the two expansions. We also include some
details on the factorization of the solution proposed by Wei and
Norman \cite{wei63las,wei64ogr}. Finally we provide another
interpretation of the Magnus expansion as the continuous analogue
of the BCH formula in linear control theory.

\subsubsection{Fer method}\label{Fer-section}

 An intuitive way to introduce Fer formalism is the following
 \cite{iserles00lgm}.
Given the matrix linear system $Y^{\prime} = A(t) Y$, $Y(0)=I$, we know
that
\begin{equation}   \label{Fer1a}
    Y(t)=\exp(F_1(t))
\end{equation}
is the exact solution if $A$
commutes with its time integral $F_1(t)=\int_0^t  A(s) ds$, and
$Y(t)$ evolves in
the Lie group $\mathcal{G}$ if $A$ lies in its corresponding Lie algebra
$\mathfrak{g}$. If the goal is to respect the Lie-group structure in
the general case, we need to `correct' (\ref{Fer1a}) without loosing this
important feature.

Two possible remedies arise in a quite natural way. The first is just to seek a
correction $\Delta(t)$ evolving in the Lie algebra $\mathfrak{g}$ so that
\[
    Y(t) = \exp \left( F_{1}(t) + \Delta(t) \right).
\]
This is nothing but the Magnus expansion. Alternatively, one may correct
with $Y_{1}(t)$ in the Lie group $\mathcal{G}$,
\begin{equation}   \label{Fer2}
    Y(t) = \exp (F_{1}(t)) Y_{1}(t).
\end{equation}
This is precisely the approach pursued by Fer, i.e. representing the
solution of (\ref{eq:evolution}) in the factorized form (\ref{Fer2}), where
(hopefully)  $Y_1$ will be closer to the identity matrix
than $Y$ at least for small $t$.

The question now is to find the differential equation satisfied by
$Y_1$. Substituting (\ref{Fer2}) into equation (\ref{eq:evolution}) we have
\begin{equation}\label{Fer3}
  \frac{d}{d t} Y =\left( \frac{d}{d t}
\e^{F_1} \right) Y_1 \; + \; \e^{F_1} \frac{d}{d t}
Y_1 = A \e^{F_1} Y_1,
\end{equation}
so that, taking into account the expression for the derivative of
the exponential map (Lemma \ref{lemma1}), one arrives easily at
\begin{equation}\label{Fer5}
  Y^{\prime}_1 = A_{1}(t) Y_1  \qquad   Y_{1}(0) = I,
\end{equation}
where
\begin{equation}\label{Fer6}
  A_{1}(t) = \e^{-F_1} A  \e^{F_1} - \int_0^{1}{\rm d} x \;\e^{-x
F_1 } A  \e^{x F_1}.
\end{equation}
The above procedure can be repeated to yield a sequence of
iterated matrices $A_{k}$. After $n$ steps we have the following recursive
scheme, known as the Fer expansion:
\begin{eqnarray}
Y &=&\e^{F_{1}}\e^{F_{2}}\cdots \e^{F_{n}}Y_{n}  \label{Fer7} \\
Y_{n}^{\prime} &=&A_{n}(t)Y_{n}\quad \ \ \ \ \ Y_{n}(0)=I,\qquad j=1,2,\ldots
\nonumber
\end{eqnarray}
with $F_{n}(t)$ and $A_{n}(t)$  given by
\begin{eqnarray}
F_{n+1}(t) &=&\int_{0}^{t}A_{n}(s)ds\ \ \ \ \ \ \
A_{0}(t)=A(t),\quad n=0,1,2...  \nonumber \\
A_{n+1} &=&\e^{-F_{n+1}}A_{n}\e^{F_{n+1}}-\int_{0}^{1}dx\
\e^{-xF_{n+1}}A_{n}\e^{xF_{n+1}}  \nonumber \\
&=&\int_{0}^{1}dx\int_{0}^{x}du\ \e^{-(1-u)F_{n+1}}\left[
A_{n},F_{n+1}\right] \e^{(1-u)F_{n+1}}  \label{Fer8} \\
&=&\sum_{j=1}^{\infty }\frac{(-)^{j}\ j}{(j+1)!}
\ad_{F_{n+1}}^{j}(A_{n}) ,\quad n=0,1,2... \nonumber
\end{eqnarray}
When after $n$ steps we impose $Y_{n}=I$ we are left with an approximation 
 to the exact solution $Y(t)$.

Inspection of the expression of $A_{n+1}$ in (\ref{Fer8}) reveals
an interesting feature of the Fer expansion. If we assume that a
perturbation parameter $\varepsilon$ is introduced in $A$, i.e. if
we substitute $A$ by $\varepsilon A$ in the formalism,
since $F_{n+1}$ is of the same order in $\varepsilon$ as
$A_n$, then an elementary recursion shows that the matrix $A_{n}$
starts with a term of order $\varepsilon^{2^n}$ (correspondingly
the operator $F_n$ contains terms of order $\varepsilon^{2^{n-1}}$
and higher). This should greatly enhance the rate of convergence
of the product in equation (\ref{Fer7}) to the exact solution.

It is possible to derive a bound on the
convergence domain in time of the expansion \cite{blanes98maf}. The idea
is just to look for conditions on $\ A(t)\ $ which insure $\
F_{n}\rightarrow 0$ as $n\rightarrow \infty $. As in the case of
the Magnus expansion, we take $A(t)$ to be a bounded matrix with
$\|A(t)\| \le k(t)\equiv k_{0}(t)$. Fer's algorithm, equations
(\ref{Fer7}) and
(\ref{Fer8}), provides then a recursive relation among corresponding bounds $%
\ k_{n}(t)\ $ for $\| A_{n}(t)\| $. If we denote $\ K_{n}(t)\equiv
\int_{0}^{t}k_{n}(s)ds$, we can write this
relation in the generic form $k_{n+1}=f(k_{n},K_{n})$, which after
integration gives
\begin{equation}
K_{n+1}=M(K_{n}).  \label{iteracio}
\end{equation}

The question now is: When is $\ K_{n}\ \rightarrow 0$ as
$n\rightarrow \infty $? This is certainly so if $0$ is a stable
fixed point for the iteration of the mapping $\ M\ $ and $\ K_{0}\
$ is within its basin of attraction. To see when this is the case
we have to solve the equation $\ \xi =M(\xi )\ $ to find where the
next fixed point lies. Let us do it explicitly. By taking norms
in the recursive scheme (\ref{Fer8}) we have
\[
 \|A_{n+1}\| \ \leq \int_{0}^{1}dx\int_{0}^{x}du\
\e^{2(1-u)K_{n}}\left\| \left[ A_{n},F_{n+1}\right] \right\| ,
\]
which can be written as $\| A_{n+1}\| \ \leq k_{n+1},$%
with
\[
k_{n+1}=\frac{1-\e^{2K_{n}}(1-2K_{n})}{2K_{n}}\frac{dK_{n}}{dt}
\]
and consequently $K_{n+1}$ is given by eq. (\ref{iteracio}) with
\begin{equation}
M(K_{n})=\int_{0}^{K_{n}}\frac{1-\e^{2x}(1-2x)}{2x}\ dx.
\label{rec}
\end{equation}
That is the mapping we have to iterate. It is clear that $\xi =0\
$is a stable fixed point of $M$. The next, unstable, fixed point
is $\xi =0.8604065 $. So we can conclude that we have a convergent
Fer expansion at least for values of time $\ t\ $ such that
\begin{equation}
\int_{0}^{t} \|A(s)\| ds\leq
K_{0}(t)<0.8604065.  \label{con1}
\end{equation}
Notice that the bound for the convergence domain provided by this result  is
smaller than the corresponding to the Magnus expansion (Theorem
\ref{conv-mag}).

\subsubsection{Wilcox method} \label{Wilcox-section}

    A more tractable form of infinite product expansion has been devised by
Wilcox \cite{wilcox67eoa} in analogy with the Magnus approach. The
idea, as usual, is to treat $\varepsilon$ in
\begin{equation}   \label{wil.1}
    Y^{\prime} = \varepsilon A(t) Y, \qquad Y(0)=I
\end{equation}
as an expansion parameter and to determine the successive factors
in the product
\begin{equation}\label{W1}
  Y(t) = \e^{W_1} \, \e^{W_2} \, \e^{W_3} \cdots
\end{equation}
by assuming that $W_n$ is exactly of order $\varepsilon^n$. Hence,
it is clear from the very beginning that the methods of Fer and
Wilcox give rise indeed to completely different infinite product
representations of the solution $Y(t)$.

    The explicit expressions of $W_1$, $W_2$ and $W_3$ are given in
\cite{wilcox67eoa}. It is noteworthy that the operators $W_n$ can
be expressed in terms of Magnus operators $\Omega_k$, for which
compact formulae and recursive procedures are available. To this
end we simply use the Baker--Campbell--Hausdorff formula to
extract formally from the identity
\begin{equation}\label{W3}
  \e^{W_1} \, \e^{W_2} \, \e^{W_3} \cdots =
   \e^{\Omega_1+\Omega_2+\Omega_3+\cdots} \; ,
\end{equation}
terms of the same order in $\varepsilon$. After a straightforward
calculation one finds for the first terms
\begin{eqnarray}\label{W4}
  W _1 &=& \Omega_1, \qquad W _2 = \Omega_2, \qquad W_3 = \Omega_3 -
      \frac{1}{2}[\Omega_1,\Omega_2],\\
W_4 &=& \Omega_4 - \frac{1}{2}[\Omega_1,\Omega_3]
+\frac{1}{6}[\Omega_1,[\Omega_1,\Omega_2]],\quad {\rm etc.}
\end{eqnarray}

The main interest of the Wilcox formalism stems from the fact that
it provides explicit expressions for the successive approximations
to a solution represented as an infinite product of exponential
operators. This offers a useful alternative to the Fer expansion
whenever the computation of $F_n$ from equation (\ref{Fer8}) is
too cumbersome. We note in passing that to first order the three
expansions yield the same result: $F_1=W_1=\Omega_1$.

\subsubsection{Wei--Norman factorization}
\label{w-nfacto}

Suppose now that $A$ and $Y$ in equation (\ref{eq:evolution})  
are linear operators and that $A(t)$ can be
expressed in the form
\begin{equation}   \label{wn1}
   A(t) = \sum_{i=1}^m u_i(t) X_i, \qquad m \; \mbox{ finite,}
\end{equation}
where the $u_i(t)$ are scalar functions of time, and $X_1, X_2,
\ldots, X_m$ are time-independent operators. Furthermore, suppose
that the Lie algebra $\mathfrak{g}$ generated by the $X_i$ is of
finite dimension $l$ (this is obviously true if $A$ and $Y$ are
finite dimensional matrix operators). Under these conditions, if
$X_1, X_2, \ldots, X_l$ is a basis for $\mathfrak{g}$, the Magnus
expansion allows to express the solution locally in the form $Y(t)
= \exp( \sum_{i=1}^l f_i(t) X_i)$. Wei and Norman, on the other
hand, show that there exists a neighborhood of $t=0$ in which the
solution can be written as a product \cite{wei63las,wei64ogr}
\begin{equation}  \label{wn2}
  Y(t) = \exp(g_1(t) X_1) \, \exp(g_2(t) X_2) \cdots \exp(g_l(t)
  X_l),
\end{equation}
where the $g_i(t)$ are scalar functions of time. Moreover, the
$g_i(t)$ satisfy a set of nonlinear differential equations which
depend only on the Lie algebra $\mathfrak{g}$ and the $u_i(t)$'s.
These authors also study the conditions under which the solution
converges globally, that is, for all $t$. In particular, this
happens for all solvable Lie algebras in a suitable basis and for
any real $2 \times 2$ system of equations \cite{wei64ogr}.

In the terminology of Lie algebras and Lie groups, the representation
$Y(t) = \exp(\sum_{i=1}^l f_i(t) X_i)$ corresponds to the \emph{canonical
coordinates of the first kind}, whereas equation (\ref{wn2}) defines a
system of \emph{canonical coordinates of the second kind}
\cite{belinfante72aso,gorbatsevich97fol,postnikov94lga}.

This class of factorization has been used in combination with
the Fer expansion to obtain closed-form solutions of the Cauchy
problem defined by certain classes of parabolic linear partial
differential equations \cite{casas96sol}. When the algorithm is
applied, the solution is written as a finite product of
exponentials depending on certain ordering functions for which
convergent approximations are constructed in explicit form.

Notice that the representation (\ref{wn2}) is clearly useful when the spectral
properties of the individual operators $X_i$ are readily
available. Since the $X_i$ are constant and often have simple
physical interpretation, the evaluation of the eigenvalues and
eigenvectors can be done once for all times, and this may
facilitate the computation of the exponentials. This situation arises, in
particular, in control theory \cite{belinfante72aso}. The functions $u_i(t)$
are known as the controls, and the operator $Y(t)$ acts on the states of the
system, describing how the states are transformed along time.

\subsection{The continuous BCH formula}
\label{lcbch}

When applied to the
equation $Y^\prime = A(t) Y$ with the matrix $A(t)$ given by
(\ref{wn1}), the Magnus expansion adopts a particularly simple form. Furthermore, by making
use of the structure constants of the Lie algebra, it is
relatively easy to get explicit expressions for
the canonical coordinates of the first kind $f_i(t)$. 
Let us
illustrate the procedure by considering the particular case
\[
    A(t) = u_1(t) X_1 + u_2(t) X_2.
\]
Denoting by $\alpha_i(t) = \int_0^t u_i(s) ds$, and, for a given
function $\mu$,
\[
  \left(\int_i \mu \right)(t)  \equiv  \int_0^t u_i(s) \mu(s) ds,
\]
a straightforward calculation shows that the first terms of $\Omega$ in
the Magnus expansion
can be written as
\begin{eqnarray}  \label{wn3}
  \Omega(t) & = & \beta_{1}(t) X_1 + \beta_{2}(t) X_2 +
  \beta_{12}(t) [X_1,X_2] + \beta_{112}(t) [X_1,[X_1,X_2]] \nonumber \\
  & &  + \beta_{212}(t) [X_2,[X_1,X_2]] + \cdots
\end{eqnarray}
where
\begin{eqnarray}  \label{wn4}
  \beta_i &=& \alpha_i, \qquad\quad  i=1,2,  \nonumber\\
  \beta_{12} &=&  \frac12  (\int_1 \alpha_2 - \int_2 \alpha_1), \\
   \beta_{112} &=&
\frac{1}{12}  (\int_2 \alpha_1^2 - \int_1 \alpha_1 \alpha_2 )
- \frac{1}{4} (\int_1 \int_2 \alpha_1 - \int_1 \int_1 \alpha_2), \nonumber \\
\beta_{212} &=&
\frac{1}{12}  (\int_2 \alpha_1 \alpha_2 - \int_1 \alpha_2^2 )
+ \frac{1}{4} (\int_2 \int_1 \alpha_2 - \int_2 \int_2 \alpha_1). \nonumber
\end{eqnarray}
Taking into account the structure constants of the particular
finite dimensional Lie algebra under consideration, from (\ref{wn3})
one easily gets the functions $f_i(t)$. In the general case, (\ref{wn3})
allows us to express $\Omega$ as a linear combination of
elements of a basis of the free Lie algebra generated by $X_1$ and
$X_2$. In this case, the recurrence (\ref{eses})-(\ref{omegn})
defining the Magnus expansion can be carried out only with the nested
integrals
\begin{equation}   \label{wn5}
   \alpha_{i_1 \cdots i_s}(t) \equiv  \left(\int_{i_s} \cdots \int_{i_1} 1
    \right)(t) =
   \int_0^t \int_0^{t_{s}} \cdots \int_0^{t_3} \int_0^{t_2} u_{i_s}(t_s)
   \cdots u_{i_1}(t_1) dt_1 \cdots dt_s
\end{equation}
involving the functions $u_1(t)$ and $u_2(t)$. Thus,
for instance, the coefficients in (\ref{wn4}) can be written (after
successive integration by parts) as
\begin{eqnarray*}
  \beta_i &=& \alpha_i, \qquad \quad i=1,2, \\
  \beta_{12} &=&  \frac12 (\int_1 \alpha_{2} - \int_2
  \alpha_1) \ = \  \frac12 (\alpha_{21} - \alpha_{12}), \\
   \beta_{112} &=&
\frac{1}{6}  (\int_2 \int_1 \alpha_1 + \int_1 \int_1 \alpha_2 - 2
  \int_1 \int_2 \alpha_1) \ = \
\frac16 (\alpha_{112}+\alpha_{211}-2 \alpha_{121}), \\
\beta_{212} &=&
\frac{1}{6}  (2 \int_2 \int_1 \alpha_2 - \int_2 \int_2 \alpha_1 -
  \int_1 \int_2 \alpha_2) \ = \
\frac16 (2 \alpha_{212}-\alpha_{122}- \alpha_{221}).
\end{eqnarray*}
The series (\ref{wn3}) expressed in terms of the integrals
(\ref{wn5}) is usually referred to  as
the continuous Baker--Campbell--Hausdorff formula 
\cite{kawski02tco,murua06tha}  for the linear case. We will generalize this formalism
to the nonlinear case in the next section.



\section{Generalizations of the Magnus expansion}
\label{section3}

In view of the attractive properties of the Magnus expansion as a
tool to construct approximate solutions of non-autonomous systems
of linear ordinary differential equations, it is hardly surprising
that several attempts have been made along the years either to
extend the procedure to a more general setting or to manipulate
the series to achieve further improvements. In this section we
review some of these generalizations, with special emphasis on the
treatment of nonlinear differential equations.

 First we reconsider an iterative method originally
devised by Voslamber \cite{voslamber72oea} for computing
approximations $\Omega^{(n)}(t)$
 in $Y(t)=\exp(\Omega(t))$ for the linear equation
$Y^\prime = A(t) Y$. The resulting approximation may be
interpreted as a re-summation of terms in the Magnus series and
possesses interesting features not shared by the corresponding
truncation of the conventional Magnus expansion. Then we adapt the
Magnus expansion to the physically relevant case of a periodic
matrix $A(t)$ with period $T$ which incorporates in a natural way
the structure of the solution ensured by the Floquet theorem. Next
we go one step further and generalize the Magnus expansion to the
so-called nonlinear Lie equation $Y^\prime = A(t,Y) Y$. Finally,
we show how the procedure can be applied to \emph{any} nonlinear
explicitly time-dependent differential equation. Although the
treatment is largely formal, in section~\ref{section5} we will see
that it is of paramount importance for designing new and highly
efficient numerical integration schemes for this class of
differential equations. We particularize the treatment to the
important case of Hamiltonian systems and also establish an interesting connection
with the Chen--Fliess series for nonlinear differential equations.

\subsection{Voslamber iterative method}  \label{Voslamber}

Let us consider equation (\ref{eq:evolution}) when there is a perturbation
parameter $\varepsilon$ in the (in general, complex) matrix $A$, i.e.,
equation (\ref{wil.1}). Theorem \ref{conv-mag}
guarantees that, for sufficiently small values of $t$,
$Y(t;\varepsilon) = \exp \Omega(t;\varepsilon)$, where
\begin{equation}   \label{vos.2}
    \Omega(t;\varepsilon) = \sum_{n=1}^{\infty}
          \varepsilon^n \, \Omega_{n}(t).
\end{equation}
The advantages of this representation and the approximations
obtained when the series is truncated have been sufficiently
recognized in the treatment done in previous sections.  There is, however, a property
of the exact solution not shared by any truncation of the series
(\ref{vos.2}) which could be relevant in certain physical
applications: $(1/\varepsilon) \sum_{n=1}^{m} \varepsilon^n
\Omega_{n}(t)$ with $m>1$ is unbounded for $\varepsilon
\rightarrow \infty$ even when $\Omega(t,\varepsilon)/\varepsilon$
is bounded uniformly with respect to $\varepsilon$ under rather
general assumptions on the matrix $A(t)$ \cite{voslamber72oea}. Notice that
this is  the case, in particular, for the adiabatic problem (\ref{adiab1}).

When Schur's unitary triangularization
theorem \cite{horn85man} is applied to the exact solution
$Y(t;\varepsilon)$ one has
\begin{equation}   \label{vos.3}
   T_{\varepsilon} = U^{\dag} \,Y \,U,
\end{equation}
where $T_{\varepsilon}$ is an upper triangular matrix and $U$ is unitary. In
other words, $Y$ is unitarily equivalent to an upper triangular
matrix $T_{\varepsilon}$. Differentiating (\ref{vos.3}) and using (\ref{wil.1})
one arrives at
\[
   T_{\varepsilon}' = \varepsilon U^{\dag} A U T_{\varepsilon} +
     \left[ T_{\varepsilon}, U^{\dag} U' \right].
\]
Since the second term on the right hand side is not upper
triangular, it follows at once that
\[
   T_{\varepsilon}(t) = \exp \left( \varepsilon \int_0^t (U^{\dag} A U
   )_{\vartriangle} ds \right),
\]
where the subscript $\vartriangle$ denotes the upper triangular
part (including terms on the main diagonal) of the corresponding
matrix. Taking into account (\ref{vos.3}) one gets
\begin{equation}   \label{vos.4}
   \Omega(t,\varepsilon) = \varepsilon U \left( \int_0^t (U^{\dag} A U
   )_{\vartriangle} ds \right) U^{\dag}.
\end{equation}
Considering now the Frobenius norm (which is unitarily invariant,
section \ref{notations}) of both sides of this equation, one has
\begin{eqnarray}  \label{suine2f}
  \| \Omega\|_F & = &  |\varepsilon|
   \left\| \int_0^t (U^{\dag} A U
   )_{\vartriangle} ds \right\|_F \le |\varepsilon| \int_0^t
   \| (U^{\dag} A U )_{\vartriangle} \|_F \, ds  \nonumber \\
   & & \le |\varepsilon| \int_0^t
   \| U^{\dag} A U \|_F \, ds = |\varepsilon| \int_0^t \|  A  \|_F \, ds,
\end{eqnarray}
If the the spectral norm is considered instead, from inequalities 
(\ref{ineqf1}), (\ref{ineqf2}) and (\ref{suine2f}), one concludes that
\[
   \| \Omega\|_2 \le \sqrt{\mathrm{rank} (A)} \,\,
        |\varepsilon| \int_0^t \|  A  \|_2 \, ds.
\]
In any case, what is important to stress here is that for the exact solution
$\Omega(t;\varepsilon)/\varepsilon$ is bounded uniformly with
respect to the $\varepsilon$ parameter. Voslamber proceeds by
deriving an algorithm for generating successive approximations of
$Y(t;\varepsilon) = \exp(\Omega(t;\varepsilon))$ which, contrarily to the
direct series expansion (\ref{vos.2}), preserve this property. His
point of departure is to get a series expansion for the so-called
\textit{dressed derivative of $\Omega$} \cite{oteo05iat}
\begin{equation}   \label{vos.5}
   \Gamma \equiv \e^{\Omega/2} \, \Omega^{\prime} \,
   \e^{-\Omega/2}.
\end{equation}
This is accomplished by inserting (\ref{fmag1}) in (\ref{vos.5}).
Specifically, one has
\begin{eqnarray*}
  \Gamma  & = &  \e^{ \mathrm{ad}_{\Omega/2}} \,
  \Omega^{\prime} = \e^{ \mathrm{ad}_{\Omega/2}} \, d
  \exp_{\Omega}^{-1} (\varepsilon A) = \e^{ \mathrm{ad}_{\Omega/2}} \,
   \frac{\mathrm{ad}_{\Omega}}{\e^{\mathrm{ad}_{\Omega}}-1} (\varepsilon A) \\
    & = &
   \frac{\mathrm{ad}_{\Omega/2}}{\sinh \Omega/2}  (\varepsilon A) =
   \sum_{n=0}^{\infty} \frac{B_n(1/2)}{n!} \mathrm{ad}_{\Omega}^n
   (\varepsilon A)
\end{eqnarray*}
and finally \cite{voslamber72oea,oteo05iat}
\begin{equation}   \label{vos.6}
   \Gamma = \sum_{n=0}^{\infty} \frac{2^{1-n}-1}{n!} B_n
     \, \mathrm{ad}_{\Omega}^n
   (\varepsilon A),
\end{equation}
where, as usual, $B_n$ denote Bernoulli numbers. In order to
express $\Gamma$ as a power series of $\varepsilon$ one has to
insert the Magnus series (\ref{vos.2}) into eq. (\ref{vos.6}).
Then we get
\begin{equation}   \label{vos.7}
  \Gamma(t;\varepsilon) = \sum_{n=1}^{\infty} \varepsilon^n \, 
  \Gamma_n(t),
\end{equation}
where the terms $\Gamma_n$ can be expressed as a function of
$\Omega_k$ with $k \le n-2$ through the recursive procedure
\cite{oteo05iat}
\begin{eqnarray}   \label{vos.8}
  \Gamma_1 & = & A, \qquad\quad \Gamma_2 = 0,  \\
  \Gamma_n & = & \sum_{j=2}^{n-1} c_j \sum_{
            k_1 + \cdots + k_j = n-1 \atop
            k_1 \ge 1, \ldots, k_j \ge 1}
            \,
       \ad_{\Omega_{k_1}} \,  \ad_{\Omega_{k_2}} \cdots
          \, \ad_{\Omega_{k_j}} A,  \qquad n \ge 3.  \nonumber
\end{eqnarray}
Here
\[
  c_j \equiv \frac{2^{1-j}-1}{j!} B_j,
\]
with $c_{2j+1} = 0$, $c_2=-1/24$, $c_4=7/5760$, etc. In
particular,
\begin{eqnarray*}
  \Gamma_3 & = & -\frac{1}{24} [\Omega_1,[\Omega_1,A]] \\
  \Gamma_4 & = & -\frac{1}{24} ([\Omega_1,[\Omega_2,A]] +
  [\Omega_2,[\Omega_1,A]]).
\end{eqnarray*}
Now, from the definition of $\Gamma$, eq. (\ref{vos.5}), we write
\[
  \Omega^{\prime} = \e^{-\Omega/2} \, \Gamma \, \e^{\Omega/2},
\]
which, after integration over $t$, can be used for constructing
successive approximations to $\Omega$ once the terms $\Gamma_n$
are known in terms of $\Omega_k$, $k \le n-2$. Thus, the $n$th
approximant $\Omega^{(n)}$ is defined by
\begin{equation}  \label{vos.9}
  \Omega^{(n)}(t) = \int_0^t \e^{-\frac{1}{2}
  \Omega^{(n-1)}(s)} \Gamma^{(n)}(s) \, \e^{\frac{1}{2}
  \Omega^{(n-1)}(s)} ds, \qquad n=1,2,\ldots
\end{equation}
where the $\varepsilon$ dependence has been omitted by simplicity
and $\Gamma^{(n)} = \sum_{k=1}^{n} \varepsilon^k \Gamma_k$,
$\Omega^{(0)}=O$. The first two approximants read explicitly
\begin{eqnarray}  \label{vos.10}
  \Omega^{(1)}(t,\varepsilon) & = & \varepsilon \, \Omega_1(t) =
  \varepsilon \int_0^t A(s) ds    \nonumber  \\
  \Omega^{(2)}(t,\varepsilon) & = & \varepsilon \int_0^t \e^{-\frac{1}{2}
  \Omega^{(1)}(s,\varepsilon)} A(s) \, \e^{\frac{1}{2}
  \Omega^{(1)}(s,\varepsilon)} ds.
\end{eqnarray}
In this approach the solution is approximated by $Y(t) \simeq
\exp(\Omega^{(n)})$. Observe that $\Omega^{(n)}$ contains
contributions from an infinity of orders in $\varepsilon$, whereas
the $n$th term in the Magnus series (\ref{vos.2}) is proportional
to $\varepsilon^n$. Furthermore, $\Omega^{(n)}$ contains
$\sum_{k=1}^n \varepsilon^k \Omega_k$ \emph{and} also higher
powers $\varepsilon^m$ ($m>n$). In particular, one gets easily
\[
  \Omega^{(2)}(t;\varepsilon) = \varepsilon \, \Omega_1(t) +
  \varepsilon^2 \, \Omega_2(t) + \sum_{k=3}^{\infty} \,
  \frac{(-1)^{k-1}}{2^{k-1} (k-1)!} \varepsilon^k \int_0^t
  \ad_{\Omega_1(s)}^{k-1} A(s) ds.
\]
From the structure of the expression (\ref{vos.9}) it is also
possible to find the asymptotic behaviour of
$\Omega^{(n)}/\varepsilon$ ($n \ge 3$) for $\varepsilon
\rightarrow \infty$ and prove that it remains bounded
\cite{voslamber72oea}, just as the exact solution does. This property
of the Voslamber iterative algorithm may lead to better
approximations of $Y(t)$ when the parameter $\varepsilon$ is not
very small, since in that case $\Omega^{(n)}/\varepsilon$ is
expected to remain close to $\Omega/\varepsilon$, as shown in
\cite{oteo05iat}.

\subsection{Floquet--Magnus expansion}  \label{Floquet}

We now turn our attention to a specific case of equation
(\ref{eq:evolution}) with important physical and mathematical
applications, namely when the (complex) $n \times n$ matrix-valued
function $A(t)$ is periodic with period $T$. Then further
information is available on the structure of the exact solution as
is given by the celebrated Floquet theorem, which ensures the
factorization of the solution in a periodic part and a purely
exponential factor. More specifically,
\begin{equation}  \label{fflo2}
   Y(t) = P(t) \, \exp(t F),
\end{equation}
where $F$ and $P$ are $n \times n$ matrices, $P(t)=P(t+T)$ for all
$t$ and $F$ is constant. Thus, albeit a solution of
(\ref{eq:evolution}) is not, in general, periodic the departure
from periodicity is determined by (\ref{fflo2}). This result, when
applied in quantum mechanics, is referred to as Bloch wave theory
\cite{bloch28udq,galindo90qme}. It is widely used in problems of
solid state physics where space-periodic potentials are quite
common. In Nuclear Magnetic Resonance this structure is exploited
as far as either time-dependent periodic magnetic fields or sample
spinning are involved \cite{ernst86pon}. Asymptotic stability of the
solution $Y(t)$ is dictated by the nature of the eigenvalues of
$F$, the so-called characteristic exponents of the original
periodic system \cite{hale80ode}.

An alternative manner of interpreting equation (\ref{fflo2}) is to
consider the piece $P(t)$, provided it is invertible, to perform a
transformation of the solution in such a way that the coefficient
matrix corresponding to the new representation has all its matrix
entries given by constants. Thus the piece $\exp (tF)$ in
(\ref{fflo2}) may be considered as an exact solution of the system
(\ref{eq:evolution}) previously moved to a representation where
the coefficient matrix is the constant matrix $F$
\cite{yakubovich75lde}. The
$t$-dependent change of representation is carried out by $P(t)$.
Connecting with section \ref{PLT}, $P(t)$ is the appropriate
preliminary linear transformation
for periodic systems. Of course, Floquet theorem by itself gives
no practical information about this procedure. It just states that
such a representation does exist. In fact, a serious difficulty in
the study of differential equations with periodic coefficients is
that no general method to compute either the matrix $P(t)$ or the
eigenvalues of $F$ is known.

Mainly, two ways of exploiting the above structure of $Y(t)$ are
found in the literature \cite{levante95fqm}. 
The first one consists in performing a
Fourier expansion of the formal solution leading to an infinite
system of linear differential equations with constant
coefficients. Thus, the $t$-dependent finite system is replaced
with a constant one at the price of handling infinite dimension.
Resolution of the truncated system furnishes an approximate
solution. The second approach is of perturbative nature. It deals
directly with the form (\ref{fflo2}) by expanding
\begin{equation}  \label{fflo3}
      P(t)=\sum_{n=1}^{\infty } P_{n}(t),  \qquad
      F=\sum_{n=1}^{\infty}F_{n}.
\end{equation}
Every term $F_{n}$ in (\ref{fflo3}) is fixed so as to ensure
$P_{n}(t)=P_{n}(t+T)$, which in turn guarantees the Floquet
structure (\ref{fflo2}) at any order of approximation.

Although the Magnus expansion, such as it has been formulated in
this work, does not provide explicitly  the structure of the
solution ensured by Floquet theorem, it can be adapted without
special difficulty to cope also with this situation. 
The starting point is to introduce the Floquet form (\ref{fflo2})
into the differential equation $Y^{\prime}=A(t)Y$. In that way the
evolution equation for $P$ is obtained:
\begin{equation}   \label{ffloPeq}
   P^{\prime}(t)= A(t) P(t) - P(t) F,  \qquad  P(0)=I.
\end{equation}
The constant matrix $F$ is also unknown and we will determine it
so as to ensure $P(t+T)=P(t).$ Now we replace the usual
perturbative scheme in equation (\ref{fflo3}) with the exponential
ansatz
\begin{equation}   \label{fflo4}
    P(t)=\exp (\Lambda (t)),  \qquad   \Lambda (0)=O.
\end{equation}
Obviously, $\Lambda (t+T)=\Lambda (t)$ so as to preserve
periodicity. Now equation (\ref{ffloPeq}) conveys
\begin{equation}   \label{ffloDeq}
\frac{\text{d}}{\text{d}t}\exp (\Lambda )=A  \exp (\Lambda )-\exp
(\Lambda )F,
\end{equation}
from which, as with the conventional Magnus expansion, it follows
readily that
\begin{equation}  \label{ffloFMeq}
   \Lambda^{\prime}= \sum_{k=0}^{\infty} \frac{B_{k}}{k!}
    \text{ad}_{\Lambda }^{k} \, (A+(-1)^{k+1}F).
\end{equation}
This equation is now, in the Floquet context, the analogue of
Magnus equation (\ref{fmag1}). Notice that if we put $F=O$ then
(\ref{fmag1}) is recovered. The next move is to consider the
series expansions for $\Lambda $ and $F$
\begin{equation}   \label{ffloseries}
    \Lambda(t) = \sum_{k=1}^{\infty } \Lambda _{k}(t), \qquad
    F=\sum_{k=1}^{\infty } F_{k},
\end{equation}
with $\Lambda _{k}(0)=O,$ for all $k$. Equating terms of the same
order in (\ref{ffloFMeq}) one gets the successive contributions to
the series (\ref{ffloseries}). Therefore, the explicit ansatz we
are propounding reads
\begin{equation}   \label{F_ansatz}
Y(t)=\exp \left( \sum_{k=1}^{\infty }\Lambda _{k}(t)\right) \,
\exp \left( t\sum_{k=1}^{\infty }F_{k}\right) .
\end{equation}
This can be properly referred as the \emph{Floquet--Magnus expansion}.

Substituting the expansions of equation (\ref{ffloseries}) into
(\ref{ffloFMeq}) and equating terms of the same order one can
write
\begin{equation}
\Lambda_{n}^{\prime}=\sum\limits_{j=0}^{n-1}\dfrac{B_{j}}{j!}\left(
W_{n}^{(j)}(t)+(-1)^{j+1}T_{n}^{(j)}(t)\right) \qquad \ \ \ \ \ \ \
\ \ (n\geq 1).
\end{equation}
The terms $W_{n}^{(j)}(t)$ may be obtained by a similar recurrence
to that given in equation (\ref{eses})
\begin{equation}  \label{ffloW}
\begin{tabular}{l}
$W_{n}^{(j)}=\sum\limits_{m=1}^{n-j}\left[ \Lambda _{m},W_{n-m}^{(j-1)}%
\right] \qquad \ \ \ \ \ \ \ \ \ \ (1\leq j\leq n-1),$ \\
\\
$W_{1}^{(0)}=A,\qquad \ \ \ \ \ W_{n}^{(0)}=O\qquad \ \ \ \ \ \ \ \ \ (n>1),$%
\end{tabular}
\end{equation}
whereas the terms $T_{n}^{(j)}(t)$ obey to the recurrence relation
\begin{equation}  \label{ffloT}
\begin{array}{l}
T_{n}^{(j)}=\sum\limits_{m=1}^{n-j}\left[ \Lambda
_{m},T_{n-m}^{(j-1)}\right]
\qquad \ \ \ \ \ \ \ \ \ \ \ \ (1\leq j\leq n-1), \\
\\
T_{n}^{(0)}=F_{n}\qquad \ \ \ \ \ \ \ \ \ (n>0).
\end{array}
\end{equation}
Every $F_{n}$ is fixed by the condition $\Lambda _{n}(t+T)=
\Lambda_{n}(t)$. An outstanding feature is that $F_{n}$ can be
determined independently of $\Lambda _{n}(t)$ as the solution
$Y(t)=P(t)\exp (tF)$ shrinks to $Y(T)=\exp (TF)$. Consequently,
the conventional Magnus expansion $Y(t)=\exp (\Omega (t))$
computed at $t=T$ must furnish
\begin{equation}  \label{fflostrobo}
  F_{n}=\frac{\Omega _{n}(T)}{T},\qquad \text{for all } \; n.
\end{equation}
The first contributions to the Floquet-Magnus expansion read
explicitly
\begin{eqnarray}   \label{firste}
\Lambda _{1}(t) & = & \int_{0}^{t} A(x)~\text{d}x-tF_{1},
\nonumber  \\
F_{1} & = & \frac{1}{T} \int_{0}^{T}A(x)~\text{d}x,  \nonumber \\
\Lambda _{2}(t) & = & \frac{1}{2} \int_{0}^{t}\left[ A(x)+F_{1},\Lambda _{1}(x)%
\right] ~\text{d}x-tF_{2}, \\  \nonumber
F_{2} & = & \frac{1}{2T} \int_{0}^{T}\left[ A(x)+F_{1},\Lambda _{1}(x)\right] ~%
\text{d}x. \nonumber
\end{eqnarray}
Moreover, from the recurrence relations (\ref{ffloW}) and
(\ref{ffloT}) it is possible to obtain a sufficient condition such
that convergence of the series $\sum \Lambda _{n}$ is guaranteed
in the whole interval $t \in [0,T]$ \cite{casas01fte}. In fact,
one can show that absolute convergence of the Floquet--Magnus
series is ensured at least if
\begin{equation}
   \int_{0}^{T} \| A(t)\| ~\text{d}t  < \xi _{F} \equiv
0.20925.
\end{equation}
Notice that convergence of the series $\sum F_{n}$ is already
guaranteed by (\ref{fflostrobo}) and the discussion concerning the
conventional Magnus expansion in subsections \ref{convergence} and
\ref{subsec273}. The
bound $\xi _{F}$ in the periodic Floquet case turns out to be
smaller than the corresponding bound $r_c = \pi$ in the conventional
Magnus expansion. At first sight this could be understood as an
impoverishment of the result. However it has to be recalled that,
due precisely to Floquet theorem, once the condition is fulfilled
in one period convergence is assured for any value of time. On the
contrary, in the general Magnus case the bound gives always a
running condition.

\subsection{Magnus expansions for nonlinear matrix equations}
\label{NL-Magnus}

It is possible to extend the procedure leading to the
Magnus expansion for the linear equation (\ref{eq:evolution}) and
obtain approximate solutions for the nonlinear matrix equation
\begin{equation}  \label{nlm1}
   Y^\prime  =   A(t, Y) Y,  \quad\qquad Y(0) = Y_0 \in \cG,
\end{equation}
where $\cG$ is a matrix Lie group, $A: \mathbb{R}_{+} \times \cG
\longrightarrow \g$ and $\g$ denotes the corresponding Lie algebra.
 Equation (\ref{nlm1}) appears in relevant physical fields such as rigid
body mechanics, in the calculation of Lyapunov exponents ($\cG
\equiv \SO(n)$) and other problems arising in 
Hamiltonian dynamics ($\cG \equiv \SP(n)$). In fact, it can be shown
that every differential equation evolving on a matrix Lie group
$\cG$ can be written in the form (\ref{nlm1}) \cite{iserles00lgm}. Moreover, the
analysis of generic differential equations defined in homogeneous
spaces can be reduced to the Lie-group equation (\ref{nlm1})
\cite{munthe-kaas97nio}.

In \cite{casas06eme} a general procedure for devising Magnus
expansions for the nonlinear equation (\ref{nlm1}) is introduced.
It is based on applying Picard's iteration on the associated
differential equation in the Lie algebra and retaining in each
iteration the terms necessary to increase the order while
maintaining the explicit character of the expansion. The resulting
methods are thus explicit by design and are expressed in terms of
integrals.

As usual, the starting point in the formalism is to represent the
solution of (\ref{nlm1}) in the form
\begin{equation}   \label{nlm2}
   Y(t) = \exp(\Omega(t,Y_0)) Y_0.
\end{equation}
Then one obtains the differential equation satisfied by $\Omega$:
\begin{equation}    \label{nlm3}
   \Omega^\prime = d \exp_{\Omega}^{-1} \big(A(t, \e^{\Omega} Y_0)
    \big), \qquad \Omega(0) = O,
\end{equation}
where $d \exp_{\Omega}^{-1}$ is given by (\ref{fdexpinv}). Now, as
in the linear case, one can apply Picard's iteration to equation
(\ref{nlm3}), giving instead
\begin{eqnarray*}
  \Omega^{[m+1]}(t) & = & \int_0^t d
  \exp_{\Omega^{[m]}(s)}^{-1} A(s,\e^{\Omega^{[m]}(s)} Y_0)
  ds  \\
  & = & \int_0^t \sum_{k=0}^{\infty} \frac{B_k}{k!} \ad_{\Omega^{[m]}(s)}^k
  A(s,\e^{\Omega^{[m]}(s)} Y_0) ds, \qquad m \ge 0.
\end{eqnarray*}
The next step in getting explicit approximations is to truncate
appropriately the $d \exp^{-1}$ operator in the above expansion.
Roughly speaking, when the whole series for $d \exp^{-1}$ is
considered, the power series expansion of the iterate function
$\Omega^{[k]}(t)$, $k \ge 1$, only reproduces the expansion of the
solution $\Omega(t)$ up to certain order, say $m$. In
consequence, the (infinite) power series of $\Omega^{[k]}(t)$ and
$\Omega^{[k+1]}(t)$ differ in terms $\cO(t^{m+1})$. The
idea is then to discard in $\Omega^{[k]}(t)$ all terms of order
greater than $m$. This of course requires careful analysis
of each term in the expansion. For instance, $\Omega^{[0]} = O$
implies that $(\Omega^{[1]})^\prime = A(t,Y_0)$ and therefore
\[
   \Omega^{[1]}(t) = \int_0^t A(s,Y_0) ds = \Omega(t,Y_0) + \cO(t^2).
\]
Since
\[
    A(s,\e^{\Omega^{[1]}(s)} Y_0) = A(0,Y_0) + \cO(s)
\]
it follows at once that
\[
   -\frac{1}{2} \int_0^t [ \Omega^{[1]}(s),A(s,\e^{\Omega^{[1]}(s)} Y_0) ]
     \, ds  = \cO(t^3).
\]
When this second term in $\Omega^{[2]}(t)$ is included and
$\Omega^{[3]}$ is computed, it turns out that $\Omega^{[3]}$
reproduces correctly the expression of $\Omega^{[2]}$ up to
$\cO(t^2)$. Therefore we truncate $d\exp^{-1}$ at the $k=0$ term
and take
\[
   \Omega^{[2]}(t) = \int_0^t A(s,\e^{\Omega^{[1]}(s)} Y_0) ds.
\]
With greater generality, we let
\begin{eqnarray}    \label{nlm4}
  \Omega^{[1]}(t) & = & \int_0^t A(s,Y_0) ds   \\
  \Omega^{[m]}(t) & = & \sum_{k=0}^{m-2} \frac{B_k}{k!} \int_0^t
  \ad_{\Omega^{[m-1]}(s)}^k A(s,\e^{\Omega^{[m-1]}(s)} Y_0)) ds, \qquad
  m \ge 2  \nonumber
\end{eqnarray}
and take the approximation $\Omega(t) \approx \Omega^{[m]}(t)$.
This results in an explicit approximate solution that involves a
linear combination of multiple integrals of nested commutators, so
that $\Omega^{[m]}(t) \in \g$ for all $m \ge 1$. It can also be
proved that $\Omega^{[m]}(t)$, once inserted in (\ref{nlm2}),
provides an explicit approximation $Y^{[m]}(t)$ for the solution
of (\ref{nlm1}) that is correct up to terms $\mathcal{O}(t^{m+1})$
\cite{casas06eme}. In addition, $\Omega^{[m]}(t)$ reproduces
exactly the sum of the first $m$ terms in the $\Omega$ series of
the usual Magnus expansion for the linear equation $Y'= A(t) Y$.
It makes sense, then, to regard the scheme (\ref{nlm4}) as an
explicit Magnus expansion for the nonlinear equation (\ref{nlm1}).

This procedure can be easily adapted to construct an exponential
representation of the solution for the differential system
\begin{equation}   \label{nlm5}
   Y' = [A(t,Y), Y],  \qquad Y(0) = Y_0 \in \mbox{Sym}(n).
\end{equation}
Here $\mbox{Sym}(n)$ stands for the set of $n \times n$ symmetric
real matrices and the (sufficiently smooth) function $A$ maps $
\mathbb{R}_{+} \times \mbox{Sym}(n)$ into $\so(n)$, the Lie
algebra of $n \times n$ real skew-symmetric matrices. It is well
known that the solution itself remains in $\mbox{Sym}(n)$ for all
$t \ge 0$. Furthermore, the eigenvalues of $Y(t)$ are independent
of time, i.e., $Y(t)$ has the same eigenvalues as $Y_0$. This
remarkable qualitative feature of the system (\ref{nlm5}) is the
reason why it is called an \textit{isospectral flow}. Such flows
have several interesting applications in physics and applied
mathematics, from molecular dynamics to micromagnetics to linear
algebra \cite{calvo97nso}.

Since $Y(t)$ and $Y(0)$ share the same spectrum, there exists a
matrix function $Q(t) \in \SO(n)$ such that
$Y(t) Q(t) = Q(t) Y(0)$ or, equivalently,
\begin{equation}   \label{nlm6}
    Y(t) = Q(t) Y_0 Q^T(t).
\end{equation}
Then, by inserting (\ref{nlm6}) into (\ref{nlm5}), it is clear
that the time evolution of $Q(t)$ is described by
\begin{equation}   \label{nlm7}
    Q' = A(t,QY_0Q^T) \, Q, \qquad Q(0)=I,
\end{equation}
i.e., an equation of type (\ref{nlm1}). Yet there is another
possibility: if we seek the orthogonal matrix solution of
(\ref{nlm7}) as $Q(t) = \exp(\Omega(t))$ with $\Omega$
skew-symmetric,
\begin{equation}  \label{nlm8}
    Y(t) = \e^{\Omega(t)} Y_0 \,\e^{-\Omega(t)},  \qquad t \ge 0,
    \qquad \Omega(t) \in \so(n),
\end{equation}
then the corresponding equation for $\Omega$ reads
\begin{equation}  \label{nlm9}
   \Omega' = d\exp_{\Omega}^{-1} \big( A(t,\e^{\Omega} Y_0
   \e^{-\Omega}) \big),  \qquad \Omega(0) = O.
\end{equation}
In a similar way as for equation (\ref{nlm3}), we apply Picard's
iteration to (\ref{nlm9}) and truncate the $d\exp^{-1}$ series at
$k=m-2$. Now we can also truncate consistently the operator
\[
    \mathrm{Ad}_{\Omega} Y_0 \equiv \e^{\Omega} Y_0
   \e^{-\Omega} = \e^{\ad_{\Omega}} Y_0
\]
and the outcome still lies in $\so(n)$. By doing so, we replace
the computation of one matrix exponential by several commutators.

In the end, the scheme reads
\begin{eqnarray}    \label{nlm10}
  \Omega^{[1]}(t) & = & \int_0^t A(s,Y_0) ds  \nonumber  \\
  \Theta_{m-1}(t) & = & \sum_{l=0}^{m-1} \frac{1}{l!}
  \ad_{\Omega^{[m-1]}(t)}^l Y_0   \\
  \Omega^{[m]}(t) & = & \sum_{k=0}^{m-2} \frac{B_k}{k!} \int_0^t
  \ad_{\Omega^{[m-1]}(s)}^k A(s,\Theta_{m-1}(s)) ds, \qquad
  m \ge 2  \nonumber
\end{eqnarray}
and, as before, one has $\Omega(t) = \Omega^{[m]}(t) +
\cO(t^{m+1})$. Thus
\begin{eqnarray*}
  \Theta_1(t) & = & Y_0 + [\Omega^{[1]}(t),Y_0] \\
  \Omega^{[2]}(t) & = & \int_0^t A(s,\Theta_1(s)) ds  \\
  \Theta_2(t) & = & Y_0 + [\Omega^{[2]}(t), Y_0] + \frac{1}{2}
     [\Omega^{[2]}(t),[\Omega^{[2]}(t), Y_0]] \\
  \Omega^{[3]}(t) & = & \int_0^t A(s,\Theta_2(s)) ds - \frac{1}{2}
   \int_0^t [\Omega^{[2]}(s),A(s,\Theta_2(s))] ds
\end{eqnarray*}
and so on. Observe that this procedure preserves the
isospectrality of the flow since the approximation
$\Omega^{[m]}(t)$ lies in $\so(n)$ for all $m \ge 1$ and $t \ge
0$. It is also equally possible to develop a formalism based on
rooted trees in this case, in a similar way as for the standard
Magnus expansion.

\noindent \textbf{Example}. The double bracket equation
\begin{equation}  \label{fdbe1}
  Y^\prime = [[Y,N], Y], \quad\qquad Y(0)=Y_0  \in \mbox{Sym}(n)
\end{equation}
was introduced by Brockett \cite{brockett91dst} and Chu \& Driessel
\cite{chu90tpg} to solve certain standard problems in applied
mathematics, although similar equations also appear in the
formulation of physical theories such as micromagnetics
\cite{moore94nga}. Here $N$ is a constant $n \times n$ symmetric
matrix. It clearly constitutes an example of an
isospectral flow with $A(t,Y) \equiv [Y,N]$.  When the
procedure (\ref{nlm10}) is applied to (\ref{fdbe1}), one
reproduces exactly the expansion
obtained in \cite{iserles02otd} with the convergence domain
established in \cite{casas04nim}.
\hfill{$\Box$}

\subsection{Treatment of general nonlinear equations}
\label{GNLM}

 As a matter of fact, the Magnus expansion can be formally generalized
to any nonlinear explicitly time-dependent differential equation. Given
the importance of the expansion, it has been indeed (re)derived a number
of times along the years in different settings. 
We have to mention in this respect the work of Agrachev and
Gamkrelidze \cite{agrachev81caa,agrachev79ter,gamkrelidze79ero},
and Strichartz \cite{strichartz87tcb}. In the context of
Hamiltonian dynamical systems, the expansion was first proposed in
\cite{oteo91tme} and subsequently applied in a more general
context in \cite{blanes01smf} with the aim of designing new numerical
integration algorithms.

By introducing nonstationary vector fields and flows, it turns out that
one gets a linear differential equation in terms of operators which can
be analyzed in exactly the same way as in section \ref{section2}.
Thus it is in principle possible to
build approximate solutions of the differential equation which
preserve some geometric properties of the exact solution. The
corresponding Magnus series expansion allows us to write the
formal solution and then different approximations  can be
obtained by truncating the series. Obviously, this formal
expansion presents two difficulties in order to render a useful
algorithm in practice: (i) it is not evident what the domain of
convergence is, and (ii) some device has to be designed to compute
the exponential map once the series is truncated.

Next we briefly summarize the main ideas involved in
the procedure. To begin with, let us consider the autonomous equation
\begin{equation} \label{nlp1}
 {\bf x}' = {\bf f(x}), \qquad
  {\bf x}(0)={\bf x}_0\in \mathbb{R}^n.
\end{equation}
If $\varphi^t$ denotes the exact flow of (\ref{nlp1}), i.e. ${\bf
x}(t) =\varphi^t({\bf x}_0)$,  then for each infinitely
differentiable map $g: \mathbb{R}^n \longrightarrow \mathbb{R}$,
$g(\varphi^t({\bf y}))$ admits the representation
\begin{equation}\label{nlp2}
 g(\varphi^t({\bf y})) = \Phi^t [g]({\bf y})
\end{equation}
where the operator $\Phi^t$ acts on differentiable functions
\cite{olver93aol}. To
be more specific, let us introduce the Lie derivative (or Lie
operator) associated with ${\bf f}$,
\begin{equation} \label{nlp4}
  L_{{\bf f}} = \sum_{i=1}^n f_i \frac{\partial}{\partial x_i}.
\end{equation}
It acts on differentiable functions $F: \mathbb{R}^n
\longrightarrow \mathbb{R}^m$ (see \cite[Chap. 8]{arnold89mmo}
for more details) as
\[
  L_{{\bf f}} F( \mathbf{y}) = F^\prime(\mathbf{y})
  \mathbf{f}(\mathbf{y}),
\]
where $F^\prime(\mathbf{y})$ denotes the Jacobian matrix of $F$.
It follows from the chain rule that, for the solution
$\varphi^t(\mathbf{x}_0)$ of (\ref{nlp1}),
\begin{equation}   \label{nlp2a1}
  \frac{d }{dt} F(\varphi^t(\mathbf{x}_0)) = (L_{{\bf f}}
  F)(\varphi^t(\mathbf{x}_0)),
\end{equation}
and applying the operator iteratively one gets
\[
  \frac{d^k }{dt^k} F(\varphi^t(\mathbf{x}_0)) = (L_{{\bf f}}^k
  F)(\varphi^t(\mathbf{x}_0)), \qquad k \ge 1.
\]
Therefore, the Taylor series of $F(\varphi^t(\mathbf{x}_0))$ at
$t=0$ is given by
\begin{equation}  \label{nlp2a}
  F(\varphi^t(\mathbf{x}_0)) = \sum_{k \ge 0} \frac{t^k}{k!} (L_{{\bf
  f}}^k F) (\mathbf{x}_0) = \exp (t L_{{\bf f}})[F](\mathbf{x}_0).
\end{equation}
Now, putting $F(\mathbf{y}) = \mathrm{Id}(\mathbf{y}) = \mathbf{y}$, the
identity map, this is just the Taylor series of the solution
itself
\[
  \varphi^t(\mathbf{x}_0) = \sum_{k \ge 0} \frac{t^k}{k!} (L_{{\bf
  f}}^k \mathrm{Id}) (\mathbf{x}_0) = \exp 
   (t L_{{\bf f}})[\mathrm{Id}](\mathbf{x}_0).
\]
If we substitute $F$ by $g$ in (\ref{nlp2a}) and compare with
(\ref{nlp2}), then it is clear that $\Phi^t[g](\mathbf{y}) = \exp
(t L_{{\bf f}})[g](\mathbf{y})$. The object $\exp (t L_{{\bf f}})$
is called the Lie transform associated with $\mathbf{f}$.

At this point, let us suppose that $\mathbf{f}(\mathbf{x})$ can be split
as $\mathbf{f}(\mathbf{x}) = \mathbf{f}_1(\mathbf{x}) +
\mathbf{f}_2(\mathbf{x})$, in such a way that the systems
\[
    \mathbf{x}^\prime = \mathbf{f}_1(\mathbf{x}), \qquad
    \mathbf{x}^\prime = \mathbf{f}_2(\mathbf{x})
\]
have flows $\varphi_1^t$ and $\varphi_2^t$, respectively, so that
\[
 g(\varphi_{i}^{t}({\bf y}))  =
   \exp(t L_{{\bf f}_i}) [g] ({\bf y}) \qquad  \qquad i=1,2.
\]
Then, for their composition one has
\begin{equation}  \label{nlp7}
 g(\varphi_{2}^{t}\circ\varphi_{1}^{s}({\bf y}))  =
   \exp(s L_{{\bf f}_1}) \exp(t L_{{\bf f}_2}) [g] ({\bf y}).
\end{equation}
This is precisely formula (\ref{nlp2a}) with
$\mathbf{f}=\mathbf{f}_1$, $t$ replaced with $s$ and with
$F(\mathbf{y}) = \exp(t L_{{\bf f}_2}) [g] ({\bf y})$. Notice that
the indices 1 and 2 as well as $s$ and $t$ to the left and right of
eq. (\ref{nlp7}) are permuted. In other words, \emph{the Lie transforms appear in
the reverse order to their corresponding maps} \cite{hairer06gni}.

The Lie derivative $L_{\mathbf{f}}$ satisfies some remarkable
properties. Given two functions $\psi_1, \ \psi_2$, it can be
easily verified that
\begin{eqnarray*}
 & &  L_{{\bf f}}(\alpha_1 \psi_1 + \alpha_2 \psi_2) =
   \alpha_1 L_{{\bf f}} \psi_1 + \alpha_2 L_{{\bf f}} \psi_2,
   \qquad \qquad   \alpha_1,\alpha_2 \in \mathbb{R}   \\
 & & L_{{\bf f}}(\psi_1\, \psi_2)= (L_{{\bf f}} \psi_1) \, \psi_2 + \psi_1\,L_{{\bf f}} \psi_2
\end{eqnarray*}
and by induction we can prove the Leibniz rule
\[
  L_{{\bf f}}^k(\psi_1\, \psi_2)= \sum_{i=0}^k
\left( \begin{array}{c}
        k \\
        i
    \end{array} \right)
    \left(L_{{\bf f}}^i \psi_1 \right) \,
    \left(L_{{\bf f}}^{k-i} \psi_2 \right),
\]
with $ L_{{\bf f}}^i \psi=L_{{\bf f}}\, \left(L_{{\bf f}}^{i-1}
\psi \right) $ and $ L_{{\bf f}}^0 \psi=\psi$, justifying the name
of Lie derivative.
 In addition, given two vector
fields ${\bf f}$ and ${\bf g}$, then
\begin{eqnarray*}
 & & \alpha_1 L_{{\bf f}} +  \alpha_2 L_{{\bf g}} =
 L_{\alpha_1 {\bf f} + \alpha_2{\bf g}} ,
 \nonumber  \\
 & & [L_{{\bf f}},L_{{\bf g}}]=
  L_{{\bf f}}\,L_{{\bf g}} - L_{{\bf g}}\,L_{{\bf f}}
  = L_{{\bf h}} ,
\end{eqnarray*}
where $ {\bf h} $ is another vector field corresponding to the Lie
bracket of the vector fields $\mathbf{f}$ and $\mathbf{g}$, denoted by
$ {\bf h}=({\bf f},{\bf g})$, whose
components are
\begin{equation}  \label{commut-nl}
 h_i = ({\bf f},{\bf g})_i = L_{{\bf f}}g_i -  L_{{\bf g}}f_i =
   \sum_{j=1}^n \left(
   f_j \frac{\partial g_i}{\partial x_j} -
   g_j \frac{\partial f_i}{\partial x_j}  \right).
\end{equation}
Moreover, from (\ref{nlp2}) and (\ref{nlp2a1}) (replacing $F$
with $g$) we can write
\[
 \frac{d }{dt} \Phi^t[g](\mathbf{x}_0) = \frac{d }{dt}
 g(\varphi^t(\mathbf{x}_0)) = (L_{\mathbf{f}}
 g)(\varphi^t(\mathbf{x}_0))= \Phi^t
 L_{\mathbf{f}}[g](\mathbf{x}_0).
\]
Particularizing to the function $g(\mathbf{x}) = \mathrm{Id}_j(\mathbf{x})
= x_j$, we get
\[
 \frac{d}{dt}\Phi^t[\mathrm{Id}_j]({\bf y}) = \Phi^t
 L_{{\bf f(y})}[\mathrm{Id}_j]({\bf y}),
 \qquad \quad j=1,\ldots,n, \quad {\bf y}={\bf x}_0
\]
or, in short,
\begin{equation} \label{nlp5}
 \frac{d}{dt}\Phi^t = \Phi^t L_{{\bf f(y})},
 \qquad \quad {\bf y}={\bf x}_0,
\end{equation}
i.e., a linear differential equation for the operator $\Phi^t$.
Notice that, as expected, equation (\ref{nlp5}) admits as formal
solution
\begin{equation}  \label{nlp6}
 \Phi^t  = \exp(tL_{{\bf f(y)}}),
   \qquad  \qquad   {\bf y}={\bf x}_0.
\end{equation}
  We can follow the same steps for the non-autonomous equation
\begin{equation} \label{non-lin}
 {\bf x}' = {\bf f}(t,{\bf x}) ,
\end{equation}
where now the operational equation to be solved is
\begin{equation} \label{nlp8}
 \frac{d}{dt}\Phi^t = \Phi^t L_{\mathbf{f}(t,\mathbf{y})},
 \qquad \quad  {\bf y}={\bf x}_0.
\end{equation}
To simplify notation, from now on we consider ${\bf x}_0$ as a set
of coordinates such that $\mathbf{f}(t,\mathbf{x}_0)$ is a
differentiable function of ${\bf x}_0$. Since
$L_{{\bf f}}$ is a linear operator we can then use directly the Magnus
series expansion to obtain the formal solution  of  (\ref{nlp8}) as
$ \Phi^t= \exp(L_{\mathbf{w}(t,\mathbf{x}_0)})$, with ${\bf w}=\sum_i{\bf w}_i$.
The first two terms are now
\begin{eqnarray} \label{nlp9}
 {\bf w}_1(t,{\bf x}_0) & = & \int_{0}^{t} {\bf f}(s,{\bf x}_0) ds
 \nonumber \\
 {\bf w}_2(t,{\bf x}_0) & = & - \frac{1}{2} \int_{0}^{t} ds_1
   \int_{0}^{s_1} ds_2  ({\bf f}(s_1,{\bf x}_0), {\bf f}(s_2,{\bf x}_0)).
\end{eqnarray}
Observe that the sign of ${\bf w}_2$ is changed when compared with
$\Omega_{2}$ in (\ref{O2}) and the integrals affect only the
explicit time-dependent part of the vector field. In general, due to the
structure of equations (\ref{eq:evolution}) and (\ref{nlp8}), the expression
for $\mathbf{w}_n(t,{\bf x}_0)$ can be obtained from the corresponding
$\Omega_n(t)$ in the linear case by applying the following rules:
\begin{enumerate}
  \item replace $A(t_i)$ by ${\bf f}(t_i,{\bf x}_0)$;
  \item replace the commutator $[\cdot,\cdot]$ by the Lie bracket (\ref{commut-nl});
  \item change the sign in $\mathbf{w}_n(t,{\bf x}_0)$ for even $n$.
\end{enumerate}

  Once ${\bf w}^{[n]}=\sum_{i=1}^n{\bf w}_i(t,{\bf x}_0)$ is
computed, there still remains to evaluate the action of the Lie
transform $\exp(L_{\mathbf{w}(t,\mathbf{x}_0)})$ on the initial conditions
${\bf x}_0$. At time $t=T$, this can be seen as the 1-flow
solution of the autonomous differential equation
\begin{equation}\label{}
 {\bf y}' = {\bf w}^{[n]}(T,{\bf y}) , \qquad \quad
  {\bf y}(0) = {\bf x}_0,
\end{equation}
since ${\bf y}(1)=\exp(L_{\mathbf{w}(T,\mathbf{x}_0)}) \mathbf{x}_0 = \mathbf{x}(T)$.

Although this is arguably the most direct way to construct a Magnus
expansion for \emph{arbitrary} time dependent nonlinear differential equations, it is
by no means the only one. In particular, Agrachev and
Gamkrelidze \cite{agrachev79ter,gamkrelidze79ero} obtain a similar expansion
by transforming (\ref{nlp8}) into the integral equation
\begin{equation}   \label{nlp10}
     \Phi^t = \mathrm{Id} + \int_0^t \Phi^s \vec{X}_s ds
\end{equation}
which is subsequently solved by successive approximations. Here, for clarity,
we have denoted $\vec{X}_s \equiv L_{\mathbf{f}(s,\mathbf{x}_0)}$.  Then one gets
the formal series
\begin{eqnarray}  \label{nlp11}
  \Phi^t & = & \mathrm{Id} + \int_0^t dt_1 \vec{X}_{t_1} + \int_0^t dt_1 \int_0^{t_1}
     dt_2 \vec{X}_{t_2} \vec{X}_{t_1}  + \cdots  \\
    & = &   \mathrm{Id} + \sum_{m=1}^{\infty} \int_0^t dt_1 \int_0^{t_1}
     dt_2 \cdots \int_0^{t_{m-1}} dt_m \; \vec{X}_{t_m} \cdots \vec{X}_{t_1}. \nonumber
\end{eqnarray}
An object with this shape is called a \emph{formal chronological series} \cite{agrachev79ter},
and the set of all formal chronological series can be endowed with a real associative
algebra structure. It is then possible to show that there exists an absolutely continuous
formal chronological series
\begin{equation}   \label{nlp12}
    V_t(\vec{X}_t) = \sum_{m=1}^{\infty} \int_0^t dt_1 \int_0^{t_1}
     dt_2 \cdots \int_0^{t_{m-1}} dt_m \; G_m(\vec{X}_{t_1}, \ldots, \vec{X}_{t_m})
\end{equation}
such that
\[
      \Phi^t = \exp(V_t(\vec{X}_t)).
\]
Here $G_m(\vec{X}_{t_1}, \ldots, \vec{X}_{t_m})$ are Lie polynomials homogeneous of
the first grade in each variable, which can be algorithmically constructed.
 In particular,
\begin{eqnarray*}
   G_1(\vec{X}_{t_1}) & = & \vec{X}_{t_1}  \\
   G_2(\vec{X}_{t_1}, \vec{X}_{t_2} ) & = & \frac{1}{2} [\vec{X}_{t_2}, \vec{X}_{t_1} ]  \\
   G_3(\vec{X}_{t_1}, \vec{X}_{t_2}, \vec{X}_{t_3} ) & = & \frac{1}{6} (
       [ \vec{X}_{t_3},  [\vec{X}_{t_2}, \vec{X}_{t_1} ] ] +  [ [\vec{X}_{t_3},  \vec{X}_{t_2}], \vec{X}_{t_1} ] )
\end{eqnarray*}
The series (\ref{nlp12}) in general diverges, even if the Lie operator $\vec{X}_t$ is analytic
\cite{agrachev81caa}.
Nevertheless, in certain cases convergence holds. For instance, if $\vec{X}_t$
belongs to a Banach Lie algebra $\mathcal{B}$ for all $t \in \mathbb{R}$, where one has a norm satisfying
$\| [X,Y]\| \le \|X\| \|Y\|$ for all $X,Y \in \mathcal{B}$ and $\int_0^t \| \vec{X}_s \| ds
\equiv \int_0^t \| L_{\mathbf{f}(s,\mathbf{x}_0)} \| ds
\le 0.44$, then $V_t(\vec{X}_t) $ converges absolutely in $\mathcal{B}$  \cite{agrachev79ter}. As a matter
of fact, an argument analogous to that used in 
\cite{blanes98maf} and \cite{moan98eao}
may allow us to
improve this bound and get convergence for
\[
    \int_0^t \| \vec{X}_s \| ds \le \frac{1}{2} \int_0^{2 \pi} \frac{1}{2 + \frac{x}{2} (1 - \cot \frac{x}{2})}
     dx  = 1.08686870\ldots
\]

\subsubsection{Treatment of Hamiltonian systems}

We have seen how the algebraic setting we have developed for linear systems of
differential equations may be extended formally to nonlinear systems. We will review next how it
can be adapted to the important class of Hamiltonian systems. In this context, the
role of a Lie bracket of vector fields (\ref{commut-nl}) is played by the classical Poisson
bracket \cite{thirring92cds}. 

The Lie algebraic presentation of Hamiltonian systems in
Classical Mechanics has
been approached in different ways and the Magnus expansion invoked in this context by diverse
authors \cite{marcus70hot,spirig79aao,tani68ctg}. More explicit use of the Magnus
expansion is
done in \cite{oteo91tme} where the evolution operator for a classical system
is constructed and its differential equation analyzed.

To particularize to this situation the preceding general treatment, 
let us consider a system with $l$ degrees of freedom and phase
space variables $ {\bf x=
(q,p)}=(q_1,\ldots,q_l,p_1,\ldots,p_l)$, 
where $(q_{i}, p_{i}),$ $i=1,\ldots ,l$ are the usual pairs of canonical
conjugate coordinate and momentum, respectively. By defining the Poisson bracket of two scalar
functions $F(\mathbf{q},\mathbf{p})$ and $G(\mathbf{q},\mathbf{p})$ of phase space variables in the
conventional way \cite{thirring92cds}
\begin{equation*}
 \{ F,G \} \equiv \sum_{i=1}^{l}\left( \frac{\partial F}{\partial
q_{i}}\frac{\partial G}{\partial p_{i}}-\ \frac{\partial F}{\partial p_{i}}%
\frac{\partial G}{\partial q_{i}}\right),
\end{equation*}%
we have
\begin{equation*}
\{ F, G \}=\sum_{i,j=1}^{2l}\frac{\partial F}{\partial x_{i}}%
J_{ij}\frac{\partial G}{\partial x_{j}},
\end{equation*}%
and in particular
\begin{equation*}
   \{ x_{i}, x_{j} \}=J_{ij}.
\end{equation*}%
Here $J$ is the basic symplectic matrix appearing in equation (\ref{simplectic}) (with $n=l$).
With these definitions the set of (sufficiently smooth) functions on
phase space acquires the structure of a Lie algebra and we can
associate with any such function $F(\mathbf{x})$ a Lie operator
\begin{equation}    \label{lie-op}
L_{F}=\sum_{i,j=1}^{2l}\frac{\partial F}{\partial x_{i}}J_{ij}\frac{%
\partial }{\partial x_{j}}
\end{equation}%
which acts on the same set of functions as $L_F G = \{F, G\}$. It is then a simple
exercise to show that the set of all Lie operators is also a Lie
algebra under the usual commutator $\left[ L_{F},L_{G}\right]
=L_{F}L_{G}-L_{G}L_{F}$ and furthermore
\begin{equation*}
\left[ L_{F},L_{G}\right] =L_{\{F,G\}}.
\end{equation*}

 Given the Hamiltonian function $H({\bf q,p},t):
\mathbb{R}^{2l}\times \mathbb{R} \rightarrow \mathbb{R}$, where $
{\bf q, \, p} \in \mathbb{R}^l$, the equations of motion are
\begin{equation} \label{eq.HamNL}
 \mathbf{q}' =  \boldsymbol{\nabla}_{\mathbf{p}}  H,
     \qquad \qquad
 \mathbf{p}' =  -\boldsymbol{\nabla}_{\mathbf{q}}  H,
\end{equation}
or, equivalently, in terms of $\mathbf{x}$,
\[
     \mathbf{x}' =  J \, \boldsymbol{\nabla}_{\mathbf{x}}  H.
\]
It is then elementary to show that the Lie operator $L_{-H}$
 is nothing but the Lie derivative $L_{\bf f}$ 
 (\ref{nlp4}) associated with the function 
 \[
     \mathbf{f} = J \, \boldsymbol{\nabla}_{\mathbf{x}}  H.
\]
Therefore, the operational equation (\ref{nlp8}) becomes
\begin{equation} \label{hami2}
 \frac{d}{dt}\Phi^t_H = \Phi^t_H L_{-H(\mathbf{y},t)},
 \qquad \quad  {\bf y}={\bf x}_0
\end{equation}
and the previous treatment also holds in this setting. As a result, the
Magnus expansion reads
\begin{eqnarray}
 \Phi_H^t= \exp({L_W}),
\end{eqnarray}
 where $W=\sum_{i=1}^{\infty}W_i$ and the first two terms are
\begin{eqnarray}
W_1({\bf x}_0) & = &  - \int_{0}^{t} H({\bf x}_0,s) ds  \\
W_2({\bf x}_0) & = & - \frac{1}{2} \int_{0}^{t} ds_1
   \int_{0}^{s_1} ds_2  \{H({\bf x}_0,s_1), H({\bf x}_0,s_2)\}.
  \nonumber
\end{eqnarray}

\subsection{Magnus expansion and the Chen--Fliess series}
\label{MECF}

Suppose that $\mathbf{f}$ in equation (\ref{non-lin}) has the form
$\mathbf{f}(t,\mathbf{x}) = \sum_{i=1}^m u_i(t)
\mathbf{f}_i(\mathbf{x})$, i.e., we are dealing with the nonlinear
differential equation
\begin{equation}   \label{cf1}
   \mathbf{x}^\prime(t) = \sum_{i=1}^m u_i(t) \;
\mathbf{f}_i(\mathbf{x}(t)), \qquad \mathbf{x}(0) = \mathbf{p},
\end{equation}
where $u_i(t)$ are integrable functions of time. Systems of the
form (\ref{cf1}) appear for instance in nonlinear control theory. In that
context the functions $u_i$ are the controls and $\mathbf{f}_i$
are related to the nonvarying geometry of the system. Observe that this problem
constitutes the natural (nonlinear) generalization of the case studied in
section \ref{lcbch}.

One of the most basic procedures for obtaining $\mathbf{x}(T)$ for a given $T$
 is by applying simple Picard iteration. For an analytic \emph{output
function} $g: \mathbb{R}^n \longrightarrow \mathbb{R}$, from
(\ref{nlp2a1}) it is clear that
\begin{equation}   \label{cf2}
  \frac{d }{dt} g(\mathbf{x}(t)) = (L_{(\sum u_i
  \mathbf{f}_i)}g)(\mathbf{x}(t)) = \sum_{i=1}^m u_i(t)
  (E_ig)(\mathbf{x}(t)), \qquad g(\mathbf{x}(0)) = g(\mathbf{p}),
\end{equation}
where, for simplicity, we have denoted by $E_i$ the Lie derivative
$L_{\mathbf{f}_i}$. This can be particularized to the case $g = x_i$, the
$i$th component function.

By rewriting (\ref{cf2}) as an equivalent
integral equation and iterating we get
\begin{eqnarray}   \label{cf3}
  g(\mathbf{x}(t)) & = & g(\mathbf{p}) + \int_0^t \sum_{i_1 = 1}^m
  u_{i_1}(t_1) (E_{i_1} g)(\mathbf{x}(t_1)) dt_1 \nonumber  \\
   & = &  g(\mathbf{p}) + \int_0^t \sum_{i_1 = 1}^m u_{i_1}(t_1)
   \bigg( (E_{i_1} g)(\mathbf{p}) +  \nonumber \\
 & & \hspace*{1cm} \left. \int_0^{t_1} \sum_{i_2=1}^m u_{i_2}(t_2)
   (E_{i_2} E_{i_1} g)(\mathbf{x}(t_2)) dt_2 \right) dt_1  \\
   & = &  g(\mathbf{p}) + \int_0^t \sum_{i_1 = 1}^m u_{i_1}(t_1)
   \left( (E_{i_1} g)(\mathbf{p}) + \int_0^{t_1} \sum_{i_2=1}^m u_{i_2}(t_2)
    \Big( (E_{i_2} E_{i_1} g)(\mathbf{p})  \right. \nonumber  \\
     & & \hspace*{1cm} + \left.  \int_0^{t_2} \sum_{i_3=1}^m
    u_{i_3}(t_3) (E_{i_3} E_{i_2} E_{i_1} g)(\mathbf{x}(t_3)) dt_3
    \Big) dt_2 \right) dt_1 \nonumber
\end{eqnarray}
and so on. Notice that in this expression the time
dependence of the solution is separated from the nonvarying
geometry of the system, which is contained in the vector fields
$E_i$ and need to be computed only once at the beginning of the
calculation. Next we reverse the names of the integration
variables and indices used (e.g., rename $i_1$ to become $i_3$ and
vice versa), so that
\begin{eqnarray}   \label{cf4}
  g(\mathbf{x}(t)) & = & g(\mathbf{p}) + \sum_{i_1 = 1}^m 
    \left( \int_0^t u_{i_1}(t_1)dt_1 \right)
    (E_{i_1} g)(\mathbf{p})  \nonumber \\
   & & +  \sum_{i_2=1}^m \sum_{i_1=1}^m \left( \int_0^{t} \int_0^{t_2}
       u_{i_2}(t_2) u_{i_1}(t_1) dt_1 dt_2 \right)
   (E_{i_1} E_{i_2} g)(\mathbf{p})  \\
   & & + \sum_{i_3 = 1}^m \sum_{i_2 = 1}^m \sum_{i_1 = 1}^m  \left(
      \int_0^t \int_0^{t_3} \int_0^{t_2} u_{i_3}(t_3) u_{i_2}(t_2)
      u_{i_1}(t_1) dt_1 dt_2 dt_3 \right) \nonumber \\
   & & \hspace*{1cm} (E_{i_1} E_{i_2} E_{i_3} g)(\mathbf{p}) + \cdots \nonumber
\end{eqnarray}
Observe that the indices in the Lie derivatives and in the
integrals are in the opposite order. This procedure can be further
iterated, thus yielding the formal infinite series
\begin{eqnarray}   \label{cf5}
  g(\mathbf{x}(t)) & = & g(\mathbf{x}(0)) + \sum_{s \ge 1} \sum_{i_1 \cdots i_s}
   \int_0^t
  \int_0^{t_{s}} \cdots \int_0^{t_3} \int_0^{t_2} u_{i_s}(t_s)
  \cdots u_{i_1}(t_1) dt_1 \cdots dt_s \nonumber \\
   & & \hspace*{1cm} E_{i_1} \cdots E_{i_s}
  g(\mathbf{x}(0)),
\end{eqnarray}
where each $i_j \in L = \{1,\ldots,m \}$. An expression of the form
(\ref{cf5}) is referred to as the Chen--Fliess series, and it can be
proved that, under certain circumstances, it actually converges
uniformly to the solution of (\ref{cf2}) \cite{kawski02tco}. This series
originates in K.T. Chen's work \cite{chen57iop} on geometric invariants
and iterated integrals of paths in $\mathbb{R}^n$. Later, Fliess 
\cite{fliess81fcn} applied the theory to the analysis of control systems.

One of the great advantages of the Chen--Fliess series is that it
can be manipulated with purely algebraic and combinatorial tools,
instead of working directly with nested integrals. To emphasize
this aspect, observe that each term in the series can be
identified by a sequence of indices or \emph{word} $w = i_1 i_2
\cdots i_s$ in the \emph{alphabet} $L$ through the following
two maps:
\begin{eqnarray*}
 \mathcal{M}_1 : \,  w = i_1 i_2 \cdots i_s & \longmapsto & \big( g \mapsto (E_w
  g)(\mathbf{p}) = (E_{i_1}  E_{i_2} \cdots E_{i_s} g)(\mathbf{p})
  \big),  \\
 \mathcal{M}_2 : \,   w = i_1 i_2 \cdots i_s & \longmapsto & \left( u \mapsto \int_0^t
  u_{i_s}(t_s) \int_0^{t_{s}} \cdots \int_0^{t_2} u_{i_1}(t_1) dt_1
  \cdots dt_{s-1} dt_s \right)
\end{eqnarray*}
In fact, the nested integral appearing in the map $\mathcal{M}_2$ can be
expressed in a simple way, as we did for the linear case in (\ref{wn5}) 
\begin{equation}  \label{cf5b}
  \alpha_{i_1\cdots i_s} = \int_0^t
  u_{i_s}(t_s) \int_0^{t_{s}} \cdots \int_0^{t_2} u_{i_1}(t_1) dt_1
  \cdots dt_{s-1} dt_s
\end{equation}
With this notation, the series of linear differential operators
appearing at the right-hand side of (\ref{cf5}) can be written in
the compact form \cite{murua06tha}
\begin{equation}  \label{cf6}
  \sum_{w \in L^*} \alpha_w E_w,
\end{equation}
where $L^*$ denotes the set of words on the 
alphabet $L = \{1, 2, \ldots, m \}$, the function
$\alpha_w$ is given by (\ref{cf5b}) for each word $w \in L^*$ and
\[
   E_w = E_{i_1} \cdots E_{i_s}, \qquad \mbox{ if } \quad w = i_1 \cdots
   i_s \in L^*.
\]
It was proved by Chen that the series (\ref{cf6}) is an
exponential Lie series \cite{chen57iop}, i.e., it can be rewritten as the
exponential of a series of vector fields obtained as nested
commutators of $E_1,\ldots, E_m$. Such an expression is referred to
in nonlinear control as the formal analogue of a continuous
Baker--Campbell--Hausdorff formula and also as the logarithm of the Chen--Fliess
series \cite{kawski00ctl}. 

Notice the similarities of this procedure with the more general treatment carried out in 
subsection \ref{GNLM} for the nonlinear differential equation (\ref{nlp1}). Thus,
expression (\ref{nlp11}) constitutes the generalization of (\ref{cf5}) to an arbitrary
function $\mathbf{f}$ in (\ref{nlp1}). Conversely, the logarithm of the Chen--Fliess series
can be viewed as the corresponding nonlinear Magnus series for
the particular nonlinear system (\ref{cf1}).

From these considerations, it is clear that, in principle,
one can obtain an explicit formula for the terms of the logarithm of the Chen--Fliess series
in a basis of the Lie algebra generated by $E_1, \ldots,
E_m$, but this problem has been only recently solved for any number
of operators $m$ and arbitrary order, using labelled rooted trees
\cite{murua06tha}. Thus, for instance, when $m=2$, it holds that
\begin{eqnarray}  \label{amcf1}
  \sum_{w \in I^*} \alpha_w E_w & = & \exp( \beta_1 E_1 + \beta_2 E_2
  + \beta_{12} [E_1,E_2] + \beta_{112} [E_1,[E_1,E_2]]  \nonumber \\
  & & + \beta_{212} [E_2,[E_1,E_2]] + \cdots),
\end{eqnarray}
where, not surprisingly, the expressions of the $\beta$ coefficients are given by
(\ref{wn4}) with the corresponding change of sign in $\beta_{12}$
due to the nonlinear character of equation (\ref{cf1}). 

Another relevant consequence of the connection between Magnus series and the
Chen--Fliess series is the following: the Lie series 
defining the logarithm of the Chen--Fliess series can be obtained explicitly from
the recurrence (\ref{eses})-(\ref{omegn}), valid in principle
 for the linear case. Of course, the successive terms of the Chen--Fliess
series itself can be generated by expanding the exponential.



\section{Illustrative examples}
\label{section4}

After having reviewed in the preceding two sections the main theoretical aspects of
the Magnus expansion and other exponential methods,
in this section we gather some examples of their application.
All of them are standard problems of Quantum Mechanics
where the exact solution for the evolution operator $U(t)$ is well
known. Due to their simplicity, higher order computations are
possible within a reasonable amount of effort. The comparison
between approximate and exact analytic results may help the reader
to grasp the advantages as well as the technical difficulties of the
methods we have analyzed.

The examples considered here are treated in
\cite{casas01fte,klarsfeld89eip,oteo05iat,pechukas66ote}, although
some results are unpublished material, in particular those
involving highest order computations. In subsection \ref{Ex:ME} we
present results concerning the most straightforward way of dealing
with ME, namely computations in the Interaction Picture. In
subsection \ref{Ex:AME} an application of ME in the adiabatic basis
is developed. Subsection \ref{Ex:FW} is devoted to illustrate the
exponential infinite-product expansions of Fer and Wilcox. An
example on the application of the iterative version of ME by Voslamber
is given in subsection \ref{Ex:IME}. Eventually, subsection \ref{Ex:FME}
contains an application of the Floquet--Magnus formalism.

\subsection{ME in the Interaction Picture}\label{Ex:ME}

We illustrate the application of ME in the Interaction Picture
(see subsection \ref{PLT}) by means of two simple time-dependent physical
systems frequently encountered in the literature, for which exact
solutions are available:  the time-dependent forced harmonic
oscillator, and a particle of spin $\frac{1}{2}$ in a constant
magnetic field. In the first case we fix $\hbar =1$ for convenience.

As we will see, ME in the Interaction Picture is appropriate
whenever the characteristic time scale of the perturbation is
shorter than the proper time scale of the system.

To illustrate and evaluate the quality of the various approximations
for the time-evolution operator, we compute the transition
probabilities among non-perturbed eigenstates induced by the small
perturbation.

\subsubsection{Linearly forced harmonic oscillator}\label{LHO}

The Hamiltonian function describing a linearly driven
harmonic oscillator reads ($\hbar =1$)
\begin{equation}\label{eq:OA}
H =H _0+V(t), \quad \mbox{ with }
\qquad H _0=\frac{1}{2}\omega _{0}(p^2+q^2),\qquad V(t)=\sqrt 2 f(t)q
\end{equation}
and $f(t)$ is real. Here $q$ and $p$ stand for the position and
momentum operators satisfying $[q,p]=i$ and $\omega
_{0}$ gives the energy level spacing in absence of the perturbation
$V(t)$.
We introduce the usual operators $a_\pm
\equiv \frac{1}{\sqrt 2 }(q\mp i p)$, so that $[a_-,a_+]=1$. With
this notation we have
\begin{equation}\label{H0V}
  H _0= \omega_0 \left(a_+ a_- +  \frac{1}{2} \right),
     \qquad V=f(t)(a_+ + a_-).
\end{equation}
The eigenstates of $H _0$ are denoted by $| n\rangle$, so that $H
_0 | n\rangle = \omega_0(n+\frac{1}{2})| n\rangle $, where $n$ stands for
the quantum number. With this notation $n=0$ corresponds to the
ground state.

For simplicity in the computations we choose $\omega_0 = 1$.
In accordance with the prescriptions in section \ref{PLT},  the
Hamiltonian in the Interaction Picture is given by (\ref{HGInt}) and
reads
\begin{equation}\label{LHOIP}
  H _I(t)=\e^{iH _0t} \, V(t) \, \e^{-iH _0t}=f(t)(\e^{it}a_+ + \e^{-it}a_-) .
\end{equation}
Accordingly, the evolution operator is factorized as
\begin{equation}\label{eq:HO_HI}
U(t,0)=\exp(-iH _0 t) \, U_I(t,0),
\end{equation}
where the new evolution operator $U_I$ is obtained from
$U_I' = \tH _I(t) U_I \equiv -i H_I(t) U_I$.

The infinite Magnus series terminates in the present example. It
happens because the second order Magnus approximant, which involves
the computation of
\begin{eqnarray}\label{2HO}
  [\tH_I (t_1),\tH_I (t_2)]&=&f(t_1)f(t_2) \left( \e^{i(t_1-t_2)} [a_+,a_-] +
     \e^{-i(t_1-t_2)}[a_-,a_+]  \right) \nonumber \\
  &=& 2if(t_1)f(t_2)\sin(t_2-t_1)
\end{eqnarray}
reduces to a scalar function. Thus Magnus series in the Interaction
Picture furnishes the exact evolution operator irrespective of
$f(t)$:
\begin{eqnarray}\label{UHO}
  U_I(t,0)&=&\exp\left( \int_0^t {\rm d}t_1 \tH_I (t_1)
  -\frac{1}{2}\int_0^t {\rm d}t_1 \int_0^{t_1} {\rm d}t_2
  [\tH_I (t_1),\tH_I (t_2)]\right) \nonumber \\
  &=& \exp \big( -i(\alpha a_++\alpha^*a_-)-i\beta \big)  \\
  &=& \exp( -i\alpha a_+) \, \exp(-i\alpha^*a_-) \, \exp (
  -i\beta-|\alpha|^2/2),\nonumber
\end{eqnarray}
where we have defined
\begin{eqnarray}\label{ab}
  \alpha &\equiv & \int_0^t {\rm d}t_1 f(t_1) \e^{it_1},\label{eq:al} \\
  \beta &\equiv & \int_0^t {\rm d}t_1 \int_0^{t_1} {\rm d}t_2
  f(t_1)f(t_2) \sin(t_2-t_1).\label{eq:be}
\end{eqnarray}
Equations (\ref{UHO}) and (\ref{eq:HO_HI}) yield the exact
time-evolution operator for the linearly forced harmonic oscillator
Hamiltonian (\ref{eq:OA}) \cite{klarsfeld93dho}.

To compute transition probabilities between free harmonic
oscillator states of quantum numbers $n$ and $m$,
\begin{equation}\label{PHO}
  P_{n\to m}=| \langle m|U_I |n \rangle |^2,
\end{equation}
the last form in (\ref{UHO}) results most convenient. Specifically,
assuming that the oscillator was initially in its ground state $| 0
\rangle $, we get in particular the familiar Poisson distribution
for the transition probabilities
\begin{equation}\label{PoissonHO}
P_{0\to n}= \frac{1}{n!} \, |\alpha|^{2n} \, \exp(-|\alpha|^2).
\end{equation}

\subsubsection{Two-level quantum systems}\label{2L}

The generic Hamiltonian for a two-level quantum  system can be
written down in the form
\begin{equation}\label{H2L}
  H (t)= \left( \begin{array}{ccc}
  E_1(t) & \  & \  C(t) \\ C^*(t) & & E_2(t) \end{array} \right)
\end{equation}
where $E_1(t),E_2(t)$ are real functions and $C(t)$ is, in general,
a complex function of $t$. We define the solvable piece of the
Hamiltonian as the diagonal matrix
\begin{equation}\label{H2L0}
  H _0(t) = \left( \begin{array}{cc}
  E_1(t) & \  \  0 \\ 0 & \  E_2(t) \end{array} \right)
\end{equation}
and all the time-dependent interaction described by the function
$C(t)$ is considered as a perturbation. In the Interaction Picture the
new Hamiltonian reads (see (\ref{HGInt}))
\begin{equation}\label{HI2L}
  H _I(t)= \left( \begin{array}{ccc}
  0 & \  \  &  C(t) \displaystyle \exp \left( i\int_0^t {\rm d}t' \omega(t') \right) \\
  C^*(t)  \displaystyle \exp \left( -i\int_0^t {\rm d}t' \omega(t') \right)  & & 0 \end{array} \right)
\end{equation}
with $\omega=(E_1-E_2)/\hbar$. Suppose now that $H_0$ is time-independent.
Then $U(t)=\exp( \tH _0t ) \, U_I(t)$. Without loss of generality, the
$H_0$ may be rendered traceless, so that $E_1=-E_2\equiv E$. Thus
$\pm E$ denote the eigenenergies associated to the eigenvectors $|+
\rangle \equiv (1, 0)^T$, $|- \rangle \equiv (0, 1)^T$ of $H_0$,
the unperturbed system. In terms of
Pauli matrices the Hamiltonian in this case may be expressed as
\begin{equation}\label{H2Lsig}
  H(t) =\frac{1}{2}\hbar \omega \sigma_3 +f(t) \sigma_1 +g(t)\sigma_2,
\end{equation}
where $f=\mathrm{Re} (C)$ and $g=-\mathrm{Im} (C)$.

Since $H _0$ is diagonal, the transition probability between
eigenstates $|+ \rangle$, $|- \rangle$ of $H_0$ is simply
\begin{equation}\label{Ppm}
  P(t) = |\langle +|U_I(t)|-\rangle |^2.
\end{equation}
As the evaluation of (\ref{Ppm}) requires the computation and manipulation of
exponential matrices involving Pauli matrices, formulas (\ref{exppaulis}) and
(\ref{cambiopict}) in section \ref{notations} come in hand here.

Next, we study two particular cases of interaction for which the
exact solution of the time evolution operator admits an analytic
expression.

{\bf 1- Rectangular step}. Suppose that in (\ref{H2Lsig}) $g=0$,
namely,
\begin{equation}\label{eq:Hround}
  H(t)=\frac{1}{2}\hbar \omega \sigma_3 +f(t) \sigma_1
\end{equation}
with $f=0$ for $t<0$  and $f=V_0$ for $t \ge 0$. Alternatively, if we
restrict ourselves to compute an observable such as the transition
probability, this example is equivalent to a rectangular mound
(or rectangular barrier) of width $T=t$ above.
The exact solution for this problem reads
\begin{equation}\label{eq:analytic}
  U(t,0)=\exp\left(- i\left( \frac{\omega}{2}\sigma_3 +
  \frac{V_0}{\hbar} \sigma_1 \right) t \right),
\end{equation}
which yields the exact transition probability
\begin{equation}\label{Pex_rect}
  P_{ex}=\frac{4\gamma^2}{4\gamma^2+\xi^2}\sin^2\sqrt{\gamma^2+\xi^2/4}
\end{equation}
between eigenstates
$|+ \rangle$, $|- \rangle$ of $H_0$. Here we have denoted
$\gamma \equiv V_0 t/\hbar $ and $\xi \equiv \omega t$.

The Interaction Picture is defined here by the explicit integration
of the diagonal piece in the Hamiltonian, so that
  $U=\exp(-i\xi \sigma_3 /2) U_I$,
where $U_I$ stands for the time evolution operator in the
Interaction Picture and obeys
\begin{equation}\label{Eq:Dirac_eq}
  U_I'=\tH _I(t) \, U_I , \qquad U_I (0)=I
\end{equation}
with
\begin{equation} \label{eq:H2}
H_I(t) =f(t) (\sigma_1 \cos \xi - \sigma_2 \sin \xi) .
\end{equation}

A computation with the usual time-dependent perturbation theory gives for the
first orders (formula (\ref{Dys1}))
\begin{eqnarray}\label{PT1}
  P_{pt}^{(1)}&=&P_{pt}^{(2)}=\frac{4\gamma^2}{\xi^2}\sin^2(\xi/2)
  \\
  P_{pt}^{(3)}&=&P_{pt}^{(4)}=\frac{\gamma^2}{\xi^2}\left[ 2\sin \frac{\xi}{2}-
  \frac{\gamma^2}{3\xi^2}\left(  9\sin \frac{\xi}{2}+\sin \frac{3\xi}{2}-
  6\xi \cos \frac{\xi}{2} +4\sin^3 \frac{\xi}{2}\right) \right] ^2.
  \nonumber
\end{eqnarray}
Notice that $P_{pt}^{(i)}>1$ may happen in the equations above
because the unitary character of the operator $U(t)$ is not preserved by the usual time-dependent
perturbation formalism.

In this example it is not difficult to compute the  first four terms in the
Magnus series corresponding to $U_I(t) = \exp \Omega(t)$ in
(\ref{Eq:Dirac_eq}). To facilitate the notation, we define $s = \sin \xi$
and $c = \cos \xi$. The Magnus approximants in the Interaction
Picture may be written down in terms of Pauli matrices and read
explicitly
\begin{eqnarray}\label{M4}
  \Omega_1&=&-i\frac{\gamma}{\xi}[\sigma_1 s+\sigma_2 (1-c)] \nonumber \\
  \Omega_2&=&-i\left( \frac{\gamma}{\xi}\right) ^2\sigma_3 (s-\xi) \nonumber  \\
  \Omega_3&=&-i\left( \frac{\gamma}{\xi} \right) ^3 \frac{1}{3}\{ \sigma_1 [3\xi
  (1+c)-(5+c)s]
  +\sigma_2 [(3\xi - s)s-4(1-c)] \} \nonumber  \\
  \Omega_4&=&-i\left( \frac{\gamma}{\xi}\right) ^4 \frac{1}{3}\sigma_3 [(4 c +5)\xi
  -(c+8)s].
\end{eqnarray}
The first two formulae for the approximate transition probabilities
are, respectively
\begin{eqnarray}\label{P4}
  P_{M}^{(1)} &=& \sin^2\left( \frac{2\gamma}{\xi}\sin(\xi /2)\right)
  \nonumber  \\
  P_{M}^{(2)} &=& \frac{4\gamma^2}{\xi^2} \frac{\sin^2\lambda}{\lambda^2} \sin^2(\xi/2), \quad
  \lambda= [4\sin^2(\xi /2)+
  \frac{\gamma^2}{\xi^2}(\sin \xi -\xi )^2]^{1/2} .
\end{eqnarray}
We omit explicit expressions for $P_{M}^{(3)}$ and $P_{M}^{(4)}$
since they are quite involved. However, we include their outputs in
Figures {\ref{plot:Rect1}, {\ref{plot:Rect2} and \ref{plot:Rect3},
where we plot the first to fourth order approximate transition
probabilities with ME in the Interaction Picture  and compare them
to the exact case and also with perturbation theory outputs. In
Figure {\ref{plot:Rect1} and {\ref{plot:Rect2} we set $\gamma=1.5$
and $\gamma=2$ respectively, whereas in Figure {\ref{plot:Rect3} we
fix $\xi=1$.

\begin{figure}[tb]
\begin{center}
\includegraphics[width=14cm]{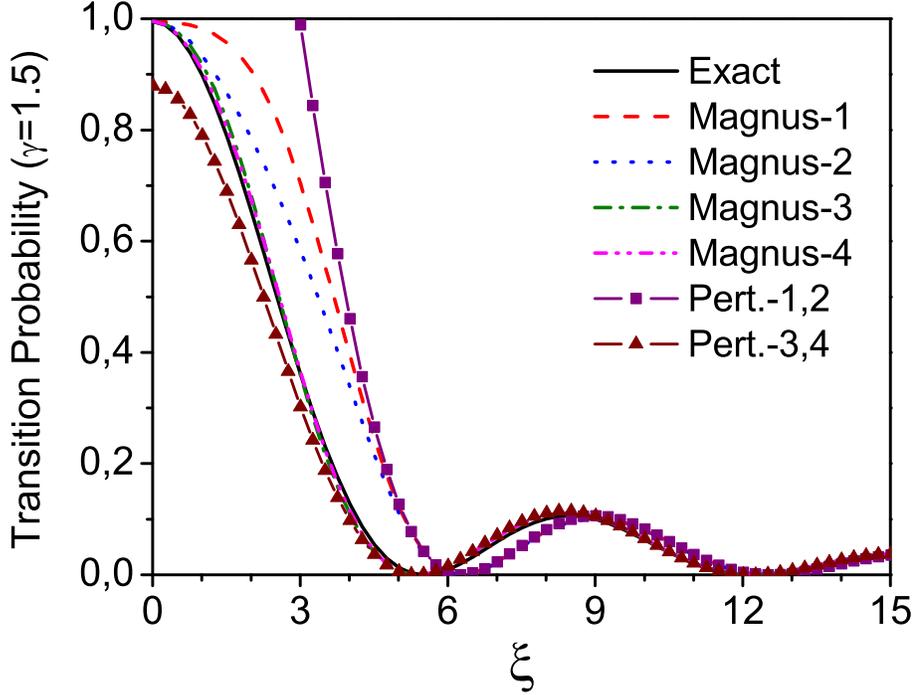}
 \caption{Rectangular step:
 Transition probabilities as a function of $\xi$, with $\gamma=1.5$.
 The solid line corresponds to the exact result (\ref{Pex_rect}). Broken lines
 stand for approximations obtained via ME and lines with symbols correspond
 to perturbation theory, according to the legend. Computations up to fourth order,
 in the Interaction Picture. } \label{plot:Rect1}
 \end{center}
\end{figure}

\begin{figure}[htb]
\begin{center}
\includegraphics[width=14cm]{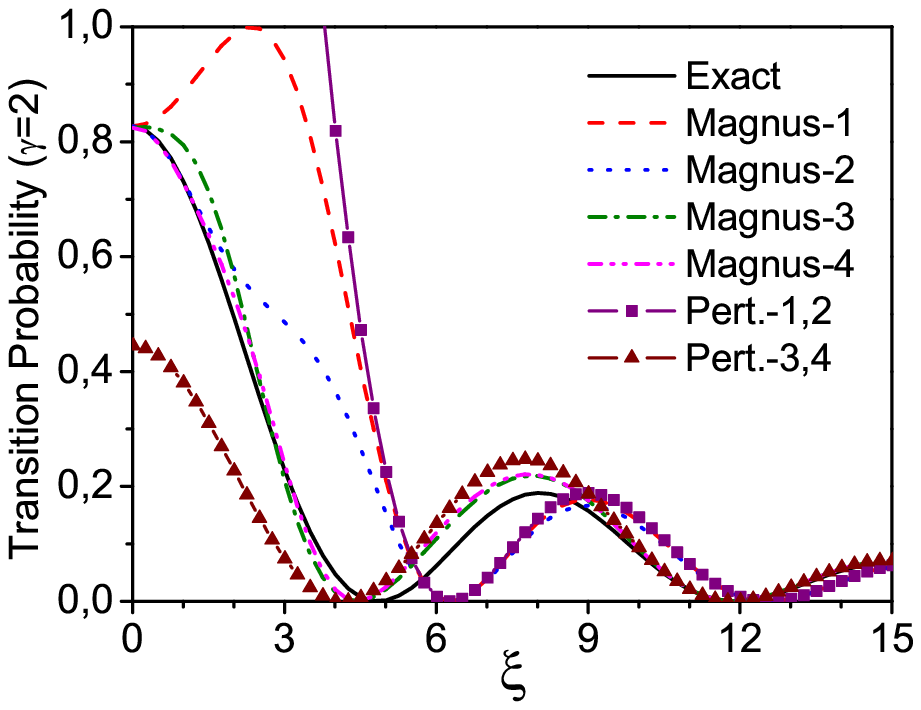}
 \caption{Rectangular step:
 Transition probabilities as a function of $\xi$, with $\gamma=2$.
 Lines are coded as in Figure \ref{plot:Rect1}.
 Computations up to fourth order,
 in the Interaction Picture.} \label{plot:Rect2}
 \end{center}
\end{figure}

We observe that for the Magnus expansion in the Interaction Picture, the smaller
the value of the parameter $\xi$
the better works the approximate solution. As a matter of fact, in the
sudden limit, $\xi \ll 1$, ME furnishes the exact result; unlike
perturbation theory. As far as the intensity of the perturbation
$\gamma$ increases, the quality of the approximations spoils. This
effect is much more dramatic for the standard perturbation theory.

\begin{figure}[htb]
\begin{center}
\includegraphics[width=14cm]{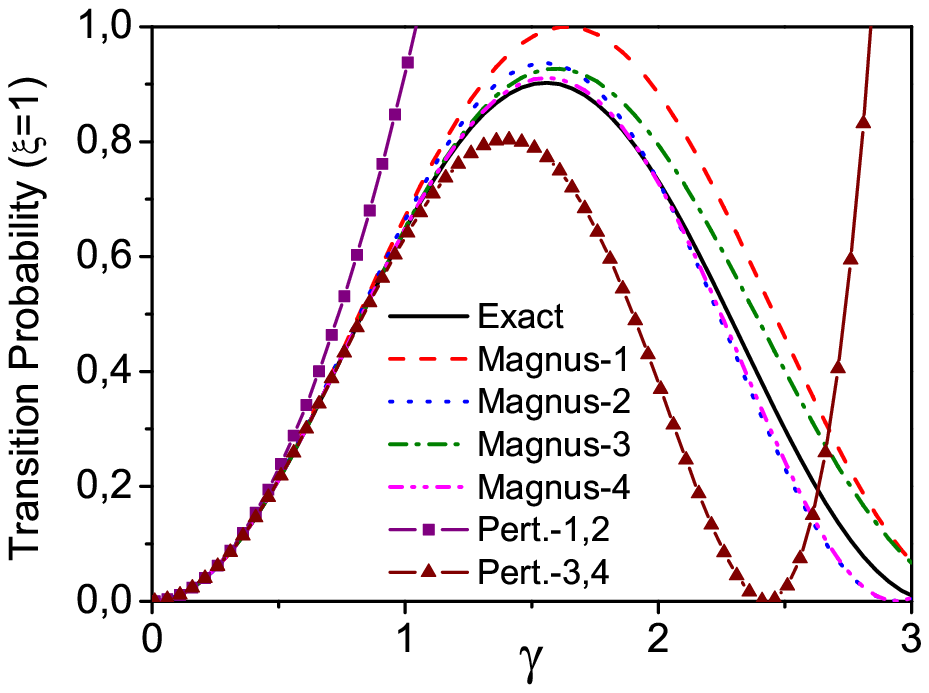}
 \caption{Rectangular step:
 Transition probabilities as a function of $\gamma$, with $\xi=1$.
 Lines are coded as in Figure \ref{plot:Rect1}.
 Computations up to fourth order,
 in the Interaction Picture.} \label{plot:Rect3}
 \end{center}
\end{figure}

On the other hand, it is clear from (\ref{eq:H2}) that
\[
   \int_{-\infty}^t  \|H_I(t_1)\|_2 \, dt_1 = \int_0^t |f(t_1)| \, dt_1 = V_0 \, t,
\]
whence   $\int_{-\infty}^t  \|\tH_I(t_1)\|_2 \, dt_1 = \gamma$, and thus
Theorem \ref{conv-mag} guarantees that
the Magnus expansion in the Interaction Picture is convergent if $\gamma < \pi$. Notice
that this is always the case for the parameters considered in Figures
\ref{plot:Rect1}-\ref{plot:Rect3}. The estimate $\gamma < \pi$ for the convergence
domain in the Interaction Picture should be compared with the corresponding one
in the Schr\"odinger picture:
\[
   \int_{-\infty}^t  \|\tH(t_1)\|_2 \, dt_1 = \sqrt{\gamma^2 + \frac{\xi^2}{4}} < \pi.
\]
Notice then that, as pointed out in subsection \ref{PLT},
a change of picture allows us to improve the convergence of the
Magnus expansion.

\vspace*{0.3cm}

{\bf 2- Hyperbolic secant step: Rosen--Zener model}. In the Rosen--Zener
Hamiltonian \cite{rosen32dsg} the interaction $C(t)$  in (\ref{H2L})
is given by the real function $V(t)= V_0 \, {\rm sech}( t/T)$, where $T$
determines the time-scale. We will use the notation $\gamma=\pi
V_0T/\hbar$ and $\xi =\omega T=2ET/\hbar$.

The corresponding Hamiltonian in terms of Pauli matrices is
\begin{equation} \label{eq:RZ_H}
H(t)=E\sigma_3+V(t)\sigma_1 \equiv \boldsymbol{a}(t) \cdot \boldsymbol{\sigma},
     \qquad V(t)=V_0/\cosh(t/T),
\end{equation}
with $\boldsymbol{a} \equiv (V(t),0,E)$. In the Interaction
Picture one has
\begin{equation}\label{eq:HIs}
H_I(s)=V(s)(\sigma_1 \cos(\xi s) - \sigma_2 \sin(\xi s))
\end{equation}
in terms of the dimensionless time-variable $s=t/T$.
Notice that $\xi$ measures the ratio between the
\emph{interaction} time $T$ and the \emph{internal} time of the
system $\hbar /2E$. From (\ref{eq:HIs}), and after straightforward
calculation, the first and second ME operators are readily obtained.

The exact result for the transition probability (provided the time
interval extends from $-\infty$ to $+\infty$), as well as
perturbation theory and Magnus expansion up to second order read
\cite{pechukas66ote}
\begin{eqnarray}\label{P_Sech}
  P_{ex}&=& \frac{\sin^2\gamma}{\cosh^2(\pi\xi/2)} \nonumber \\
  P_{pt}^{(1)}&=&P_{pt}^{(2)}=\frac{\gamma^2}{\cosh^2(\pi\xi /2)} \nonumber \\
  P_{M}^{(1)}&=&\sin^2[\gamma/\cosh(\pi\xi /2)] \nonumber \\
  P_{M}^{(2)}&=&\frac{\sin^2\lambda}{\lambda^2}\frac{\gamma^2}{\cosh^2(\pi\xi
  /2)} \\
  \lambda&=&\gamma\left[ \frac{1}{\cosh^2(\pi\xi/2)}+\frac{\gamma^2g^2(\xi)}{\pi^4}\right] ^{1/2},
  \quad g(\xi)=8\xi \sum_{k=0}^\infty \frac{2k+1}{[(2k+1)^2+\xi^2]^2}.
  \nonumber
\end{eqnarray}

In Figures {\ref{plot:Sech2} and \ref{plot:Sech3} we plot some
results from the formulae in (\ref{P_Sech}). In  Figure
{\ref{plot:Sech2} we take $\gamma=1.5$ and in Figure
{\ref{plot:Sech3} we set  $\xi=0.3$.

\begin{figure}[htb]
\begin{center}
\includegraphics[width=14cm]{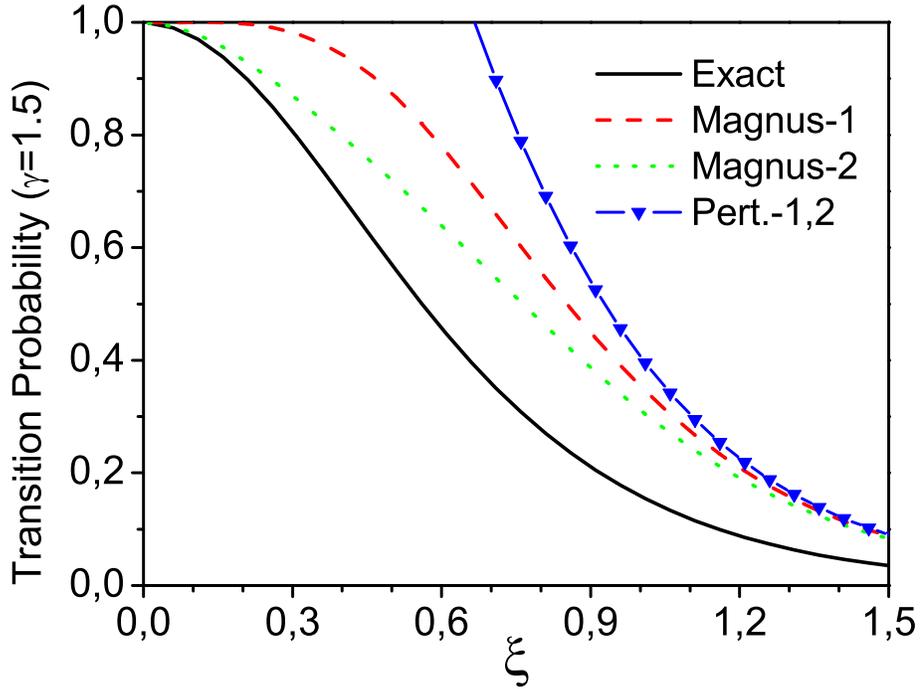}
\caption{Rosen--Zener model: Transition probabilities (\ref{P_Sech})
as a function of $\xi$, with $\gamma=1.5$. The solid line stand for
the exact result. Broken lines stand for approximations obtained via
ME and triangles correspond to perturbation theory, according to the
legend. Computations up to second order, in the Interaction Picture.
} \label{plot:Sech2}
 \end{center}
\end{figure}

\begin{figure}[tb]
\begin{center}
\includegraphics[width=14cm]{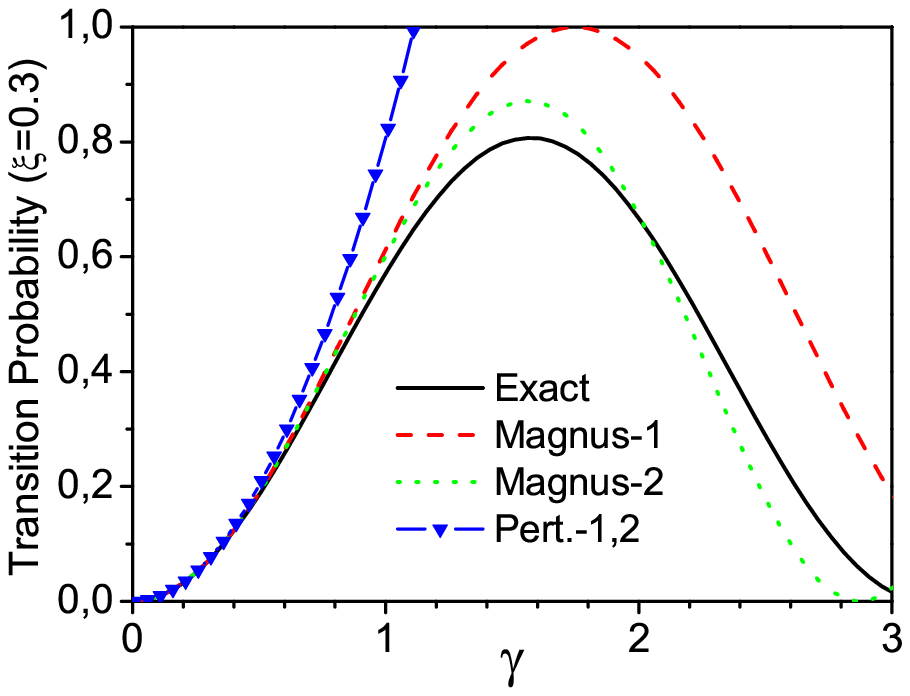}
 \caption{Rosen--Zener model:
 Transition probabilities as a function of $\gamma$, with $\xi=0.3$.
 Lines are coded as in Figure \ref{plot:Sech2}.
 Computations up to second order, in the Interaction Picture.}
\label{plot:Sech3}
 \end{center}
\end{figure}

Similarly to the case of the rectangular step, we observe in Figure
\ref{plot:Sech2} that \emph{the Magnus expansion works better in the sudden regime} defined
by $\xi \ll 1$, namely, when the internal time of the system
$\hbar/2E$ is much larger than the time scale $T$ of the
perturbation. Also, Figure \ref{plot:Sech3} illustrates how the
approximations spoil as far as the intensity $\gamma$ increases.
Notice in both figures the unitarity violation of the approximation built with the usual
perturbation theory.

In this case a simple calculation shows that
\[
   \int_{-\infty}^{\infty}   \| \tH_I(t)\|_2 \, dt  = (1/\hbar) \, \int_{\infty}^{\infty} |V(t)| \, dt =
      V_0 \pi T /\hbar = \gamma,
\]
and thus the Magnus series converges at least for $\gamma < \pi$.

\subsection{ME in the Adiabatic Picture}\label{Ex:AME}

Here we will illustrate the effect of using the Adiabatic Picture
introduced in Section \ref{PLT}. The use of this type of preliminary
transformation is convenient whenever the time scale of the
interaction is much larger than the proper time of the unperturbed
system.

Since an adiabatic regime conveys a smooth profile for the
perturbation, namely, existence of derivatives, the case of the
rectangular step cannot be properly used for the sake of
illustration.

\subsubsection{Linearly forced harmonic oscillator}\label{LHO_ad}
For the linearly driven harmonic oscillator the procedure yields the
exact solution, as in the preceding subsection \ref{2L}, albeit the
method is a bit more involved technically \cite{klarsfeld93dho}.

\subsubsection{Rosen--Zener model}\label{RZ_ad}
Following \cite{klarsfeld92cmi} we will deal with the Rosen--Zener
model (see Section \ref{2L} and \cite{pechukas66ote}) since it
allows a clear illustration of the adiabatic regime.

The preliminary linear transformation $G(s)$ defined in
(\ref{GAdiab}) for the Hamiltonian (\ref{eq:RZ_H}) is given by
$G(s)=\boldsymbol{\hat{b}}\cdot \boldsymbol{\sigma}$, where the unit
vector $\boldsymbol{\hat{b}} $ points in the direction
$\boldsymbol{b}=\boldsymbol{\hat{a}}+\boldsymbol{\hat{k}}$
($\boldsymbol{\hat{k}}=$unit vector along the $z$-axis). Remind that
$s=t/T$ is the dimensionless time-variable. The evolution operator
gets then factorized as
\begin{equation}
U_G(s)=G^\dag (s)U(s)G(s_0),
\end{equation}
which according to (\ref{UGeq}), satisfies the equation
\begin{equation}
\frac{dU_G}{ds} =\tilde H_G(s)U_G.
\end{equation}
Here $\tilde H_G\equiv -i H_G/\hbar$ is given by
\begin{equation}\label{HGs}
\tilde H_G(s)=\frac{T}{i\hbar} a \sigma_3- i \frac{\theta'}{ 2} \sigma_2 ,
\end{equation}
with  $a^2=E^2_0+V^2_0/\cosh^2 s$ and $\cot \theta =(E_0/V_0)\cosh
s$.

Next, in analogy to (\ref{eq:HIs}), we introduce the Adiabatic
Interaction Picture which allows us to integrate the diagonal piece
of $\tH_G(s)$. The time-evolution operator gets eventually factorized
as
\begin{equation}
U_G(s)=\exp \left( (-iT/\hbar)\int^{\infty}_0 \textrm{d}s'\,
a(s')\sigma_3 \right) U_G^{(I)}(s) \exp \left(
(-iT/\hbar)\int_{-\infty}^0 \textrm{d}s'\,  a(s')\sigma_3 \right) ,
\end{equation}
where $U_G^{(I)}(s)$ obeys the equation
\begin{equation}
\frac{dU_G^{(I)}}{ds}=\tilde H_G^{(I)}(s) \, U_G^{(I)},
\end{equation}
with
\begin{equation}\label{HGIs}
\tilde H_G^{(I)}(s)= -i(\theta' /2)[\sigma_1 \sin A(s)+\sigma_2
\cos A(s)],
\end{equation}
and
\begin{equation}\label{AHGIs}
A(s)=\frac{2T}{\hbar} \int_0^s \textrm{d}s'\, a(s') =
\frac{\xi}{2}\ln \frac{1+\rho}{1-\rho}+\frac{2\gamma}{\pi}\arctan
\frac{2\gamma}{\pi\xi}\rho .
\end{equation}
We have introduced the definition
\begin{equation}
\rho=\{ 1-[1+(\pi \xi/2\gamma)]\sin^2 \theta\} ^{1/2},
\end{equation}
in terms of the dimensionless strength parameter $\gamma=\pi
V_0T/\hbar$ and $\theta$. Using the ME to first order in the
adiabatic basis (which coincides with the fixed one at $s=\pm
\infty$) one finds the spin-flip approximate (first order)
transition probability
\begin{equation}\label{P_ad_M1}
^{ad}P_{M}^{(1)}=\sin^2 \left[ \int_0^{\theta_0} \textrm{d} \theta
\sin A(s(\theta)) \right].
\end{equation}
In Figure \ref{plot:RZ_Ad} we compare the numerical results given by
the new approximation (\ref{P_ad_M1}) with the exact formula
$P_{ex}$ in (\ref{P_Sech}). For the sake of illustration we plot
also the results in the usual Interaction Picture to first order in
ME (see $P_{M}^{(1)}$ in (\ref{P_Sech})). It is noteworthy the gain
achieved when using the adiabatic ME in the intermediate regime,
namely, moderate values of $\xi$, although only the first
order is considered. Here the adiabatic regime corresponds to large
values of $\xi = 2 E T/\hbar$ ($\varepsilon = 1/T \ll 1$).

It should be also noticed that for the Hamiltonian
(\ref{HGIs}) one has
\[
   \int_{s_0}^{s}   \| \tH_G^{(I)}(s_1)\|_2 \, ds_1  = \frac{1}{2} |\theta(s) - \theta(s_0)|
    < \frac{1}{2} 2 \pi = \pi
\]
and thus the convergence condition given by Theorem \ref{conv-mag} is
always satisfied. In other words, for this example
\emph{the Magnus expansion is always convergent in the Adiabatic Interaction
Picture}.

More involved illustrative examples of ME in the Adiabatic Picture
may be found in \cite{klarsfeld92mai,klarsfeld93dho}.

\begin{figure}[htb]
\begin{center}
\includegraphics[width=14cm]{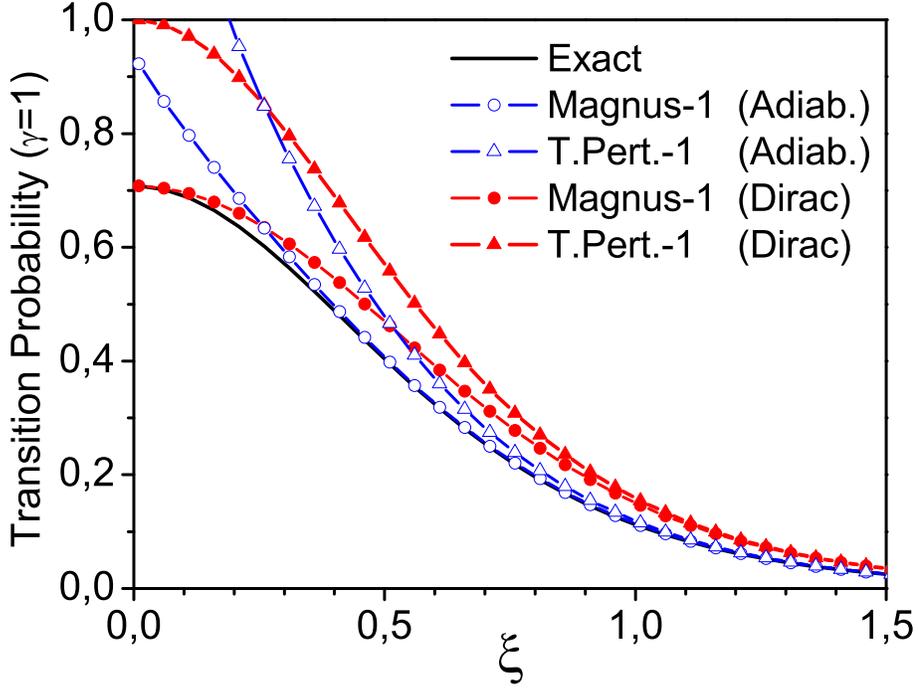}
 \caption{Rosen--Zener model: Transition probability as a function of $\xi$,
 with $\gamma =1$. The solid line stands for the exact result in (\ref{P_Sech}).
 The remaining lines stand for first order computations with ME (circles) and
 perturbation theory (triangles). Open symbols are for the Adiabatic Picture
 and solid symbols for the Interaction Picture.} \label{plot:RZ_Ad}
 \end{center}
\end{figure}

\subsection{Fer and Wilcox infinite-product expansions}\label{Ex:FW}

Next we illustrate the use of Fer and Wilcox infinite-product expansions
by using the same time-dependent systems  as before.

\subsubsection{Linearly forced harmonic oscillator}

Since the commutator $[a_-, a_+]$ is a $c$-number, Fer iterated
Hamiltonians $H^{(n)}$ with $n>1$ eventually vanish so that $F_n=0$
for $n>2$. The Wilcox operators $W_n$ with $n>2$ in eq.(\ref{W1})
vanish for the same reason. Thus, in this particular case, the
second-order approximation in either method leads to the exact
solution of the Schr\"odinger equation just as ME did. To sum up,
the final result reads
\begin{equation}\label{FWex3}
  U_I= \e^{\Om _1+\Om _2}= \e^{W_1} \e^{W_2}= \e^{F_1} \e^{F_2}=
\e^{-i\beta} \e^{-i(\alpha a_+  + \alpha^* a_-)}\,,
\end{equation}
where $\alpha(t)$ and $\beta(t)$ are given in (\ref{eq:al}) and
(\ref{eq:be}), respectively.

\subsubsection{Two-level quantum system: Rectangular step}

For the Hamiltonian (\ref{eq:Hround}), the first-order Fer and Wilcox operators
in the Interaction Picture verify
\begin{equation}\label{FWex7}
  F_1 = W_1 =  \Omega_1 = \int_0^t {\rm d}t_1 \tH _I(t_1),
\end{equation}
where $\tH_I(t)$ is given by (\ref{eq:H2}). The explicit expression is collected
in (\ref{M4}) (first equation). Analogously, the second equation there
also corresponds to the second-order Wilcox operator $W_2 = \Omega_2$.

To proceed further with Fer's method we must calculate the modified
Hamiltonian $\tH ^{(1)}$ in (\ref{Fer2}). After straightforward algebra
one eventually obtains
\begin{equation}\label{FWex12}
  \tH ^{(1)}=\frac{1}{2\theta}\left( \frac{\sin^2\theta}{\theta} - \sin2\theta +
\frac{1}{\theta}\left( \frac{\sin2\theta}{2\theta} -
\cos2\theta\right) F_1\right) \, [F_1,\tH _I]\; ,
\end{equation}
where $\theta=(2\gamma /\xi)\sin(\xi /2)$ (notice that $\tH ^{(1)}$
and therefore $F_2$  depend  on $\sigma_1$ and $\sigma_2$, while
$W_2$ is proportional to $\sigma_3$). Since it does not seem
possible to derive an analytical expression for $F_2$, the
corresponding matrix elements have been computed by replacing
the integral by a conveniently chosen quadrature.

The transition probability $P(t)$ from an initial state with spin up
to a state with spin down (or viceversa) is given by (\ref{Ppm}).
This expression has been computed on assuming: $U_I\simeq
\e^{F_1}=\e^{W_1}$, $U_I\simeq \e^{F_1} \e^{F_2}$, $U_I\simeq
\e^{W_1} \e^{W_2}$, $U_I\simeq \e^{F_1} \e^{F_2} \e^{F_3}$ and $U_I\simeq
\e^{W_1} \e^{W_2} \e^{W_3}$, and the results have been compared with the
exact analytical solution (\ref{Pex_rect}).

In Figures \ref{plot:FW1} and \ref{plot:FW2} we show the transition probability $P$ as a
function of $\xi$  for two different values of $\gamma$ ($\gamma = 1.5$ and
$\gamma = 2$, respectively), while in
Figure \ref{plot:FW3} we have plotted $P$ versus $\gamma$ for $\xi$
fixed. Notice that the second order in the Wilcox expansion does not
contribute to the transition probability (this is similar to what
happens in perturbation theory). On the other hand, Fer's
second-order approximation is already in remarkable agreement with
the exact result whereas the third order cannot even be
distinguished from the exact result in Figure \ref{plot:FW2} at the cost of a much larger
computational effort.
Wilcox approximants preserve unitarity but do not give acceptable
approximations.
\begin{figure}[htb]
\begin{center}
\includegraphics[width=14cm]{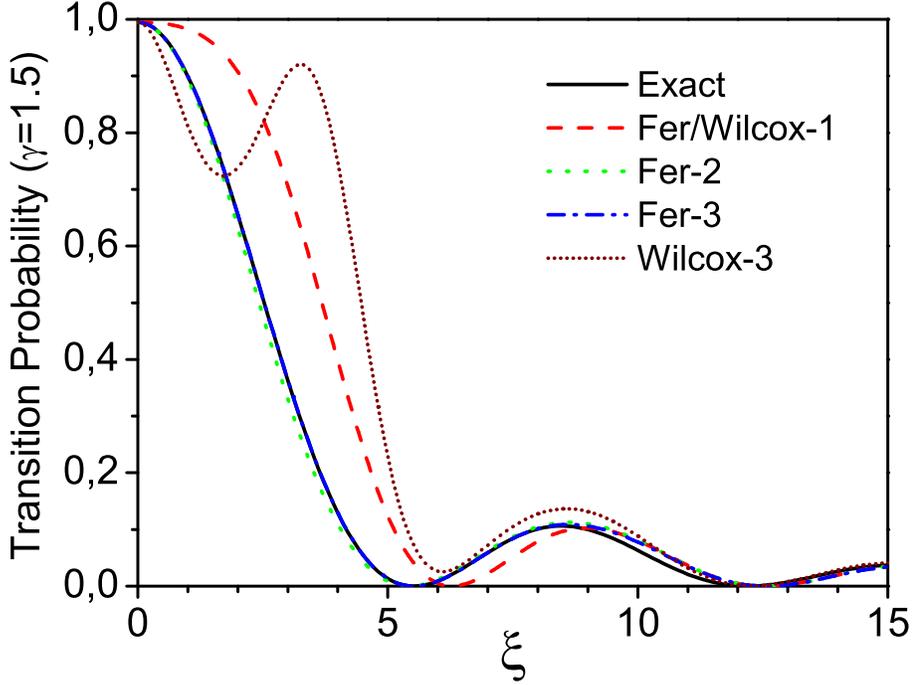}
 \caption{Rectangular step: Transition probability as a function of $\xi$, with $\gamma=1.5$.
 The solid line corresponds to the exact result (\ref{Pex_rect}). Broken lines stand for Fer and
 Wilcox approximations up to third order.
 Computations up to third order, in the Interaction Picture.} \label{plot:FW1}
 \end{center}
\end{figure}

\begin{figure}[tb]
\begin{center}
\includegraphics[width=14cm]{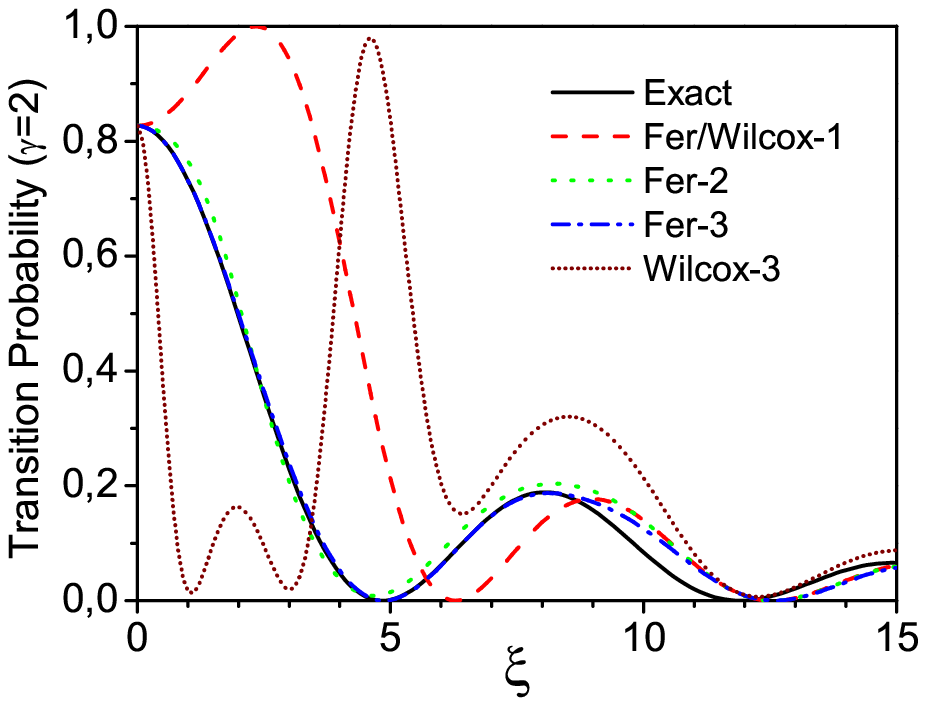}
 \caption{Rectangular step: Transition probability as a function of $\xi$, with $\gamma=2$.
 Lines are coded as in Figure \ref{plot:FW1}.
 Computations up to third order, in the Interaction Picture.} \label{plot:FW2}
 \end{center}
\end{figure}

\begin{figure}[tb]
 \begin{center}
\includegraphics[width=14cm]{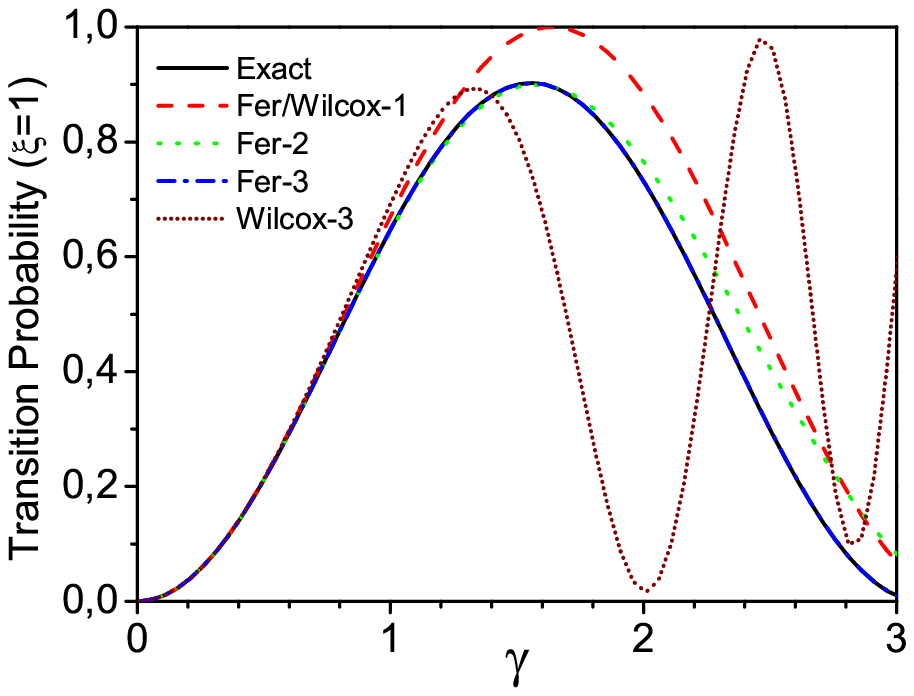}
 \caption{Rectangular step: Transition probability as a function of $\gamma$, with $\xi=1$.
  Lines are coded as in Figure \ref{plot:FW1}.
 Computations up to third order, in the Interaction Picture.} \label{plot:FW3}
 \end{center}
\end{figure}

\subsection{Voslamber iterative method}\label{Ex:IME}

Just to keep the same structure as in preceding subsections, we
mention that the Voslamber iterative method of subsection \ref{Voslamber}
also yields the exact solution for the
linearly driven harmonic oscillator after computing the second
iteration.

Next, for the two-level system with a rectangular step
described by the Hamiltonian (\ref{eq:Hround}) we compute the
second iterate $\Om^{(2)}$ and compare
with second order ME approximation for $U_I$.  As a  test, we
shall obtain again the transition probability $P(t)$
 given by (\ref{Pex_rect}). The
expression (\ref{Ppm}) will be calculated here on assuming: $U_I\simeq
\exp{\Om^{(1)}}   =\exp \Om_1$, $U_I\simeq \exp({\Om_1+\Om_2})$ and
$U_I\simeq \exp{\Om^{(2)} }$.

The  second order Magnus approximation to the transition probability
is given by (\ref{P4}), whereas the second iterate is obtained from
(\ref{vos.10}),
\begin{eqnarray}\label{eq:dressed2+}
  \Om^{(2)}(t)&=& -i\frac{\gamma}{\omega} \int_0^{\omega t} \{
  [\sin^2(\Delta)+\cos^2(\Delta)\cos{\xi}]\,\sigma_1 \nonumber \\
  &&+\cos^2(\Delta)\sin{\xi}\,\sigma_2
  - \sin(2\Delta)\sin(\xi /2)\,\sigma_3 \}\dd \xi,
\end{eqnarray}
where  $\Delta\equiv \frac{\gamma}{\omega}|\sin(\xi/2)|$, $\gamma
=V_0 t/\hbar$, $\xi =\omega t$. Since it does not seem possible to
derive an analytical expression for $\Om^{(2)} $, the corresponding
matrix elements have been computed by approximating the integral in
(\ref{eq:dressed2+}) with a sufficiently accurate quadrature.

In Figure \ref{plot:G}, the various approximated transition
probabilities  as well as the exact result (\ref{Pex_rect}) have
been plotted as a function of $\xi$ for a fixed value  $\gamma =1.5$.
We observe that the approximation from the second iterate keeps the
trend of the exact solution in a better way that the second order
Magnus approximation does.

In Figure \ref{plot:xi} we have plotted the corresponding transition
probabilities versus $\gamma$ for fixed value $\xi =1$. Although
locally the second order Magnus approximation may be more accurate,
it seems that the trend of the exact solution is mimicked for a
longer interval of $\gamma$.

\begin{figure}[htb]
\begin{center}
\includegraphics[width=14cm]{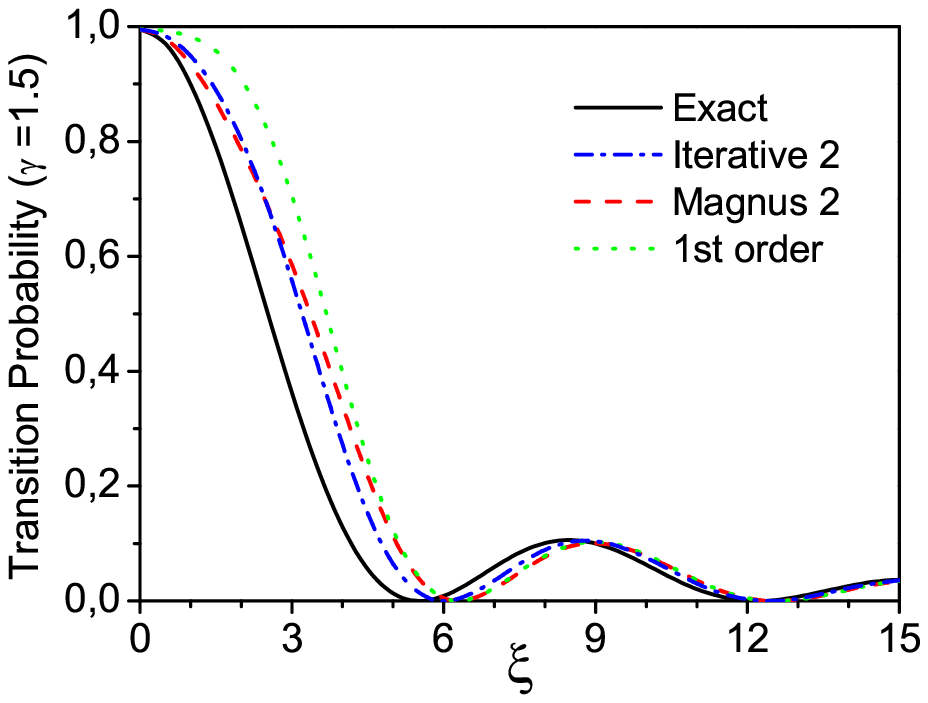}
 \caption{Rectangular step:
 Transition probability in the two level system as a function of $\xi$, for
 $\gamma =1.3$: Exact result (\ref{P_Sech}) (solid line), second iterate of
 Voslamber method
 (dashed line), second order  ME (dotted line) and
 first Voslamber iterate (or order in ME) (dash-dotted line).
 Computations are done in the Interaction Picture.} \label{plot:G}
 \end{center}
\end{figure}

\begin{figure}[tb]
\begin{center}
\includegraphics[width=14cm]{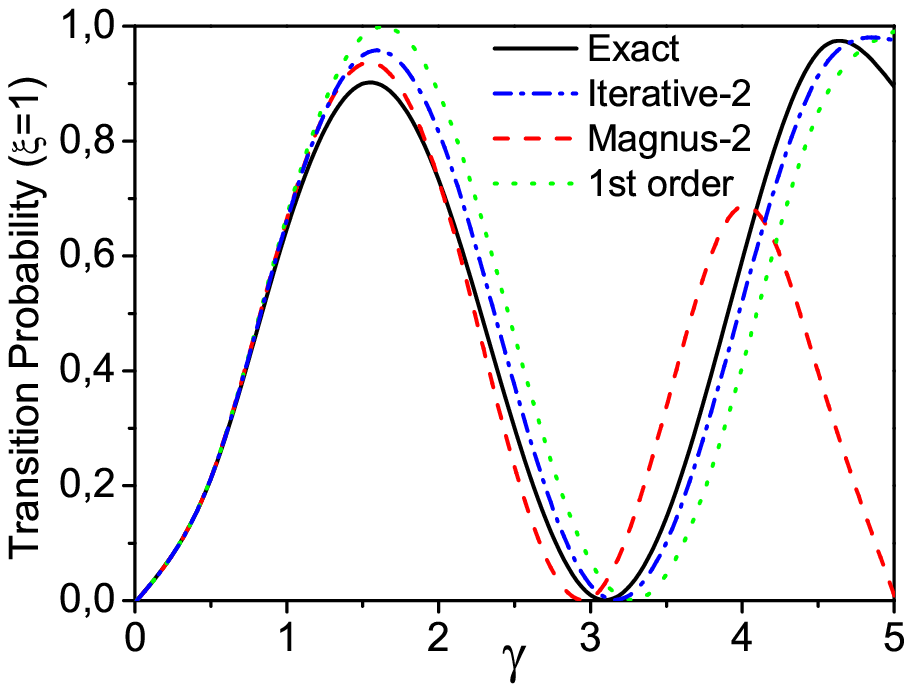}
 \caption{Rectangular step:
 Transition probability in the two level system as a function of $\gamma $, for
 $\xi =1$. Lines are coded as in Figure \ref{plot:G}.
 Computations are done in the Interaction Picture.} \label{plot:xi}
 \end{center}
\end{figure}

As it has already been pointed out above, the Magnus expansion works
the better the more sudden the perturbation. Thus, the re-summation
involved in the iterative method improves a bit that issue. Further
results on the present example may be found in \cite{oteo05iat}.

\subsection{Linear periodic systems: Floquet--Magnus formalism}
\label{Ex:FME}

Next we deal with a periodically driven harmonic oscillator, where
the infinite expansions obtained from Floquet--Magnus recurrences
(see subsection \ref{Floquet}), are utterly summed. Comparison with the exact
solution illustrates the feasibility of the method.

The particular system we consider is described by the Hamiltonian (\ref{eq:OA})
with $f(t)=\frac{\beta}{\sqrt{2}}
\cos \omega t$ and $\omega_0 < \omega$. Once the recurrences of the Floquet--Magnus expansion (see
section \ref{Floquet}) are explicitly computed for several orders, their
general term may be guessed by inspection. For the Floquet operator
we get
\begin{equation}
F=-i\left[ \frac{\omega _{0}}{2}\left( p^{2}+q^{2}\right) -\beta \frac{%
\omega _{0}}{\omega }\sum_{k=0}^{\infty }\left( \frac{\omega _{0}}{\omega }%
\right) ^{k}q+\beta ^{2}\frac{\omega _{0}}{4\omega
^{2}}\sum_{k=0}^{\infty }\left( 2k+1\right) \left( \frac{\omega
_{0}}{\omega }\right) ^{k}\right] ,
\end{equation}
and the associated transformation results from
\begin{eqnarray}  \label{eq:FLambda}
\Lambda (t) & = & i \beta ^{2}\frac{\omega _{0}}{\omega ^{3}}\left[ \sin
\left(
\omega t\right) \sum\limits_{k=0}^{\infty }\left( k+1\right) \left( \frac{%
\omega _{0}}{\omega }\right) ^{k}-\frac{\omega t}{2}\sum\limits_{k=0}^{%
\infty }\left( \frac{\omega _{0}}{\omega }\right) ^{k}\right] -  \nonumber \\
&  & i\left[ \sin \left( \omega t\right) q+\omega _{0}\left( \cos
\left( \omega
t\right) -1\right) p\right] \frac{\beta }{\omega }\sum\limits_{k=0}^{%
\infty }\left( \frac{\omega _{0}}{\omega }\right) ^{k}.
\end{eqnarray}
The resulting series may be summed in closed form thus yielding the
Floquet operator
\begin{equation}
F=-i\frac{\omega _{0}}{2}\left[ \left( q-\frac{\beta }{\omega
_{0}(\rho ^{2}-1)}\right) ^{2}+p^{2}\right] -i\frac{\beta
^{2}}{4\omega _{0}\left( \rho ^{2}-1\right) },
\end{equation}
with $\rho \equiv \omega /\omega _{0}$. Its eigenvalues are the
so-called \textit{Floquet eigenenergies} \cite{lefebvre97thf}
\begin{equation}
E_{n}=\omega _{0}\left( n+\frac{1}{2}\right) +\frac{\beta
^{2}}{4\omega _{0}\left( \rho ^{2}-1\right) }.
\end{equation}
The corresponding $\Lambda $ transformation after summation of the
series in (\ref{eq:FLambda}) is
\begin{equation}
\Lambda (t)=i\frac{\beta /\omega _{0}}{1-\rho ^{2}}\left[ \left(
\rho \sin
\omega t\right) q+(\cos \omega t-1)p+\left( \frac{2\rho ^{2}}{1-\rho ^{2}}%
+\cos \omega t\right) \frac{\beta \sin \omega t}{4\omega }\right] .
\end{equation}
Notice that, as they should, both operators are skew-Hermitian and
reproduce the exact solution of the problem.



\section{Numerical integration methods based on the Magnus expansion}
\label{section5}

\subsection{Introduction}

 The Magnus expansion as formulated in section 2 has found extensive use in
mathematical physics, quantum chemistry, control theory, etc,
essentially as a perturbative tool in the treatment of the linear
equation
\begin{equation}    \label{NI.1}
    Y^\prime = A(t) Y, \qquad Y(t_0) = Y_0.
\end{equation}
When the recurrence (\ref{eses})-(\ref{omegn}) is applied, one is
able to get explicitly the successive terms $\Omega_k$ in the
series defining $\Omega$ as linear combinations of multivariate
integrals containing commutators acting iteratively on the
coefficient matrix $A$, as in (\ref{O1})-(\ref{O4}). As a result,
with this scheme analytical approximations to the exact solution
are constructed explicitly. These approximate solutions are fairly
accurate inside the convergence domain, especially when high order
terms in the Magnus series are taken into account, as illustrated
by the examples considered in section \ref{section4}.

There are several drawbacks, however, involved in the procedure
developed so far, especially when one tries to find accurate approximations to the
solution for very long times. The first one is implicitly contained in the
analysis done in section \ref{section2}: the size of the
convergence domain of the Magnus series may be relatively small. The logarithm of the 
exact solution $Y(t)$ may have complex singularities and this implies that no series
expansion can converge beyond the first singularity. This
disadvantage may be up to some point avoided  by using different
pictures, e.g. the transformations shown in section \ref{PLT}, in
order to increase the convergence domain of the Magnus expansion, a fact
also illustrated by several examples in section 4.
Unfortunately, these preliminary transformations sometimes 
either do not guarantee convergence or the rate of convergence of
the series is very slow. In that case, accurate results can only
be obtained provided a large number of terms in the series are taken
into account.

The second drawback is the increasingly complex structure of the
terms $\Omega_k$ in the Magnus series: each $\Omega_k$ is a
$k$-multivariate integral involving a linear combination of
$(k-1)$-nested commutators of $A$ evaluated at different times
$t_i$, $i=1, \ldots, k$.  Although in some cases these expressions
can be computed explicitly (for instance, when the elements of $A$
and its commutators are polynomial or trigonometric functions), in
general a special procedure has to be designed to approximate
multivariate integrals and reduce the number of commutators
involved.

When the entries of the coefficient matrix $A(t)$ are complicated
functions of time or they are only known for certain values of
$t$, numerical approximation schemes are unavoidable. In many
cases it is thus desirable to obtain just numerical approximations to
the exact solution at many different times. This section is
devoted precisely to the Magnus series expansion as a tool to build
numerical integrators for equation (\ref{NI.1}).

Before
embarking ourselves in exposing the technical details contained in
this construction, let us first introduce several concepts which are commonplace
in the context of the numerical integration of differential equations.

 Given the general
(nonlinear) ordinary differential equation (ODE)
\begin{equation}\label{nde.1}
  {\bf x}' = {\bf f}(t,{\bf x}), \qquad \qquad {\bf x}(t_0)={\bf
  x}_0 \in \mathbb{C}^d,
\end{equation}
standard numerical integrators, such as Runge--Kutta and multistep methods,
proceed as follows. First the whole time interval $[t_0,t_f]$,  is split into $N$ subintervals,
$[t_{n-1},t_n], \ n=1,\ldots,N$, with $t_N=t_f$, and subsequently the value of
${\bf x}(t_n)$ is approximated with a time-stepping advance procedure of the form
\begin{equation}   \label{nde.2}
   \mathbf{x}_{n+1} = \Phi(h_n,\mathbf{x}_n,\ldots,h_0,\mathbf{x}_0)
\end{equation}
starting from $\mathbf{x}_0$. Here the map $\Phi$ depends on the
specific numerical method and $h_n=t_{n+1}-t_n$ are the time
steps. For simplicity in the presentation we consider a constant
time step $h$, so that $t_n=t_0+n \, h$. In this way one gets ${\bf
x}_{n+1}$ as an approximation to ${\bf x}(t_{n+1})$. In other
words, the exact evolution of the system (\ref{nde.1}) is replaced
by the discrete or numerical flow (\ref{nde.2}). The simplest of
all numerical methods for (\ref{nde.1}) is the \emph{explicit
Euler scheme}
\begin{equation}   \label{nde.3}
     \mathbf{x}_{n+1} = \mathbf{x}_n + h \, {\bf f}(t_n,{\bf x}_n).
\end{equation}
It computes approximations $\mathbf{x}_n$ to the values
$\mathbf{x}(t_n)$ of the solution using one explicit evaluation of
$\mathbf{f}$ at the already computed value $\mathbf{x}_{n-1}$.
 In general, the
numerical method (\ref{nde.2}) is said to be of order $p$ if,
assuming ${\bf x}_n= {\bf x}(t_n)$, then  ${\bf x}_{n+1}={\bf
x}(t_{n+1})+\mathcal{O}(h^{p+1})$. Thus, in particular, Euler
method is of order one.

Of course, more elaborate and efficient general purpose
algorithms, using several $\mathbf{f}$ evaluations per step, have
been proposed along the years for the numerical treatment of
equation (\ref{nde.1}). In fact, any standard software package and program
library contains dozens of routines aimed to provide numerical
approximations with several degrees of accuracy, including
(explicit and implicit) Runge--Kutta methods, linear multistep
methods, extrapolation schemes, etc, with fixed or adaptive step
size. They are designed in such a way that the user has to provide
only the initial condition and the function $\mathbf{f}$ to obtain
approximations at any given time.

This being the case, one could ask the following question:
if general purpose integrators are widely available for the
integration of the linear equation (\ref{NI.1}) (which is a particular
case of (\ref{nde.1})), what is the point in designing new and somehow
sophisticated algorithms for this specific problem?

It turns out that in the same way as for classical
time-dependent perturbation
theory, the qualitative properties of the exact
solution are not preserved by the numerical approximations
obtained by standard integrators.
 This motivates the study of the Magnus  expansion
with the ultimate goal of constructing numerical integration methods. We will
show that highly accurate schemes can be indeed designed, which
in addition
preserve qualitative properties of the system. The procedure
can also incorporate the
tools developed for the analytical treatment, such as preliminary
linear transformations, to end up with improved numerical
algorithms.

We first summarize the main features of the well know class of
Runge--Kutta methods as representative of integrators of the form
(\ref{nde.2}). They are introduced for the general nonlinear ODE
(\ref{nde.1}) and subsequently adapted to the linear case
(\ref{NI.1}).

\subsection{Runge--Kutta methods}
  \label{secRK}

The  Runge--Kutta (RK) class of methods are possibly the most
frequently used algorithms for numerically solving ODEs.
Among them, perhaps the most successful
during more than half a century has been the 4th-order method, which applied
to equation (\ref{nde.1}) provides the following numerical
approximation for the integration step $t_n \mapsto t_{n+1} = t_n + h$:
\begin{eqnarray}  \label{RK4}
 {\bf Y}_{1} & = & {\bf x}_{n}   \nonumber  \\
 {\bf Y}_{2} & = &
    {\bf x}_{n} + \frac{h}{2} {\bf f}(t_{n},{\bf Y}_{1})
      \nonumber  \\
 {\bf Y}_{3} & = &
    {\bf x}_{n} + \frac{h}{2} {\bf f}(t_{n}+\frac{h}{2},{\bf Y}_{2}) \\
 {\bf Y}_{4} & = &
    {\bf x}_{n} + h {\bf f}(t_{n}+\frac{h}{2},{\bf Y}_{3})  \nonumber  \\
 {\bf x}_{n+1} & = & {\bf x}_{n} +
 \frac{h}{6} \left( {\bf f}(t_{n},{\bf Y}_{1}) +
         2{\bf f}(t_{n}+\frac{h}{2},{\bf Y}_{2}) +
         2{\bf f}(t_{n}+\frac{h}{2},{\bf Y}_{3}) +
         {\bf f}(t_{n}+h,{\bf Y}_{4}) \right).      \nonumber
\end{eqnarray}

Notice that the function ${\bf f}$ can always be computed
explicitly because each ${\bf Y}_{i}$ depends only on the ${\bf
Y}_{j}, \ j<i$, previously evaluated.
 To measure the computational cost
of the method, it is usual to consider that the evaluation of the
function ${\bf f}(t,{\bf x})$ is the most consuming part. In this
sense, scheme (\ref{RK4})  requires four evaluations, which is
precisely the number of \emph{stages} (or inner steps) in the algorithm.

The general class of $s$-stage Runge--Kutta methods
are characterized by the real numbers $a_{ij}$,
$b_i$  ($i,j=1, \ldots, s$) and $c_i = \sum_{j=1}^s a_{ij}$, as
\begin{eqnarray}
 {\bf Y}_{i} & = & {\bf x}_{n} +
 h \sum_{j=1}^{s}a_{ij} \, {\bf f}(t_{n}+c_{j}h,{\bf Y}_{j}) ,
  \qquad i=1,\ldots,s  \nonumber  \\
 {\bf x}_{n+1} & = & {\bf x}_{n} +
 h \sum_{i=1}^{s}b_{i} \, {\bf f}(t_{n}+c_{i}h,{\bf Y}_{i}) ,
  \label{RK-general}
\end{eqnarray}
where ${\bf Y}_i$, $ i=1,\ldots,s$ are the intermediate stages.
For simplicity, the associated coefficients are usually displayed
with the so-called \emph{Butcher tableau}
\cite{butcher87tna,hairer93sod} as follows:
\begin{equation}    \label{RK-tablero}
\begin{array}{c|ccccc}
 c_1 &   a_{11} &  & \ldots &  & a_{1s} \\
 \vdots & \vdots &  &           &   &  \vdots \\
 c_s &   a_{s1} &  & \ldots &  & a_{ss} \\
 \hline
 &           b_1    &    &  \ldots &  &  b_s  \end{array}
\end{equation}
If $a_{ij}=0$, $j\geq i$, then the intermediate stages ${\bf Y}_i$ can be evaluated
recursively and the method is explicit. In that case the zero $a_{ij}$
coefficients (in the upper triangular part of the tableau) are omitted
for clarity. With
this notation, `the' 4th-order Runge--Kutta method (\ref{RK4}) can be expressed as
\begin{equation}  \label{RK4-tablero}
\begin{array}{c|cccc}
 0 &   & & & \\
  \frac{1}{2} &  \frac{1}{2} & & & \\
  \frac{1}{2} &   0 & \frac{1}{2} & &  \\
 1 &   0 & 0 & 1 &    \\
\hline  & \ \frac{1}{6} \ & \ \frac{2}{6} \ & \ \frac{2}{6} \ &
 \ \frac{1}{6}
\end{array}
\end{equation}
 Otherwise, the scheme is implicit and requires to
numerically solve a system of $s \, d$ nonlinear equations
of the form
\begin{equation}\label{implicit}
  {\bf y} = {\bf X}_{n} + h \, {\bf G}(h,{\bf y}),
\end{equation}
 where ${\bf y}=({\bf Y}_1,\ldots,{\bf Y}_s)^T,{\bf X}_n=({\bf x}_n,\ldots,{\bf x}_n)^T
 \in \mathbb{R}^{sd}$, and ${\bf G}$ is a function which depends on the method. A standard
 procedure to get $\mathbf{x}_{n+1}$ from
(\ref{implicit}) is applying simple iteration:
\begin{equation}\label{implicit_recurs}
  {\bf y}^{[j]} = {\bf X}_{n} +
  h \, {\bf G}(h,{\bf y}^{[j-1]} ), \qquad
  j=1,2,\ldots
\end{equation}
When $h$ is sufficiently small, the iteration starts with ${\bf
y}^{[0]}={\bf X}_{n}$ and stops once $ \| {\bf y}^{[j]}- {\bf
y}^{[j-1]}\| $ is smaller than a prefixed tolerance. Of course,
more sophisticated techniques can be used \cite{hairer93sod}.

After these general considerations, let us turn our attention now
to the linear equation (\ref{NI.1}).
When dealing with numerical methods applied to this
 equation, it is important to keep in
mind that the relevant small parameter here is no longer the norm of
the matrix $A(t)$ as in the analytical treatment, but the time step
$h$. For this reason, the concept of order of accuracy of a method is
different in this context. Now ${\bf f}(t,{\bf x})=
A(t){\bf x}$ and the general class of $s$-stage Runge--Kutta
methods adopt the more compact form
\begin{eqnarray}
 {\bf Y}_{i} & = & {\bf x}_{n} +
 h \sum_{j=1}^{s}a_{ij}A_j{\bf Y}_{j} ,
  \qquad i=1,\ldots,s  \nonumber  \\
 {\bf x}_{n+1} & = & {\bf x}_{n} +
 h \sum_{i=1}^{s}b_{i}A_i{\bf Y}_{i},
  \label{RK-lineal}
\end{eqnarray}
with $A_i=A(t_n+c_ih)$. In terms of matrices, this is equivalent to
\begin{eqnarray}\label{maRKs}
   \left\{
  \begin{array}{c}
   {\bf Y}_1 \\
   \vdots \\
   {\bf Y}_s
  \end{array} \right\} & = &
   \left\{
  \begin{array}{c}
   {\bf x}_n \\
   \vdots \\
   {\bf x}_n
  \end{array} \right\} +
  h \widetilde{A}
   \left\{
  \begin{array}{c}
   {\bf Y}_1 \\
   \vdots \\
   {\bf Y}_s
  \end{array} \right\}, \qquad \mbox{with} \qquad
    \widetilde{A} = \left(
  \begin{array}{ccc}
   a_{11}A_1 & \cdots & a_{1s}A_s \\
   \vdots &  &  \vdots  \\
   a_{s1}A_1 & \cdots & a_{ss}A_s
  \end{array} \right)               \nonumber  \\
  {\bf x}_{n+1} & = & {\bf x}_n +
  h \left(
  \begin{array}{ccc}
   b_{1}A_1 \  & \cdots &\ b_{s}A_s
  \end{array} \right)
   \left\{
  \begin{array}{c}
   {\bf Y}_1 \\
   \vdots \\
   {\bf Y}_s
  \end{array} \right\},
\end{eqnarray}
so that the application of the method for the integration step
$t_n \mapsto t_{n+1} = t_n + h$ can also be written as
\begin{equation}\label{}
  {\bf x}_{n+1} =  \left( I_d +
  h \left(
  \begin{array}{ccc}
   b_{1}A_1 \  & \cdots \  & b_{s}A_s
  \end{array} \right)
   \left( I_{sd} - h \widetilde{A}   \right)^{-1} \mathbb{I}_{sd\times d}
    \right){\bf x}_n,
\end{equation}
where $\mathbb{I}_{sd\times d}=(I_d,I_d,\ldots,I_d)^T$ and $I_d$
is the $d\times d$ identity matrix. For instance, taking $s=2$ and
using Matlab this can be easily implemented as follows
\[
\begin{array} {l}
 A = [ a11*A1 \ \ \  a12*A2 \ ; \ a21*A1 \ \ \  a22*A2 ];  \\
 x = (Id + h*[b1*A1 \ \ \ b2*A2]*((I2d - h*A)\backslash [Id \ \ \ Id]'))*x;
\end{array}
\]
where $b1,b2,a11,\ldots,a22$ are the coefficients of the method,
$A1,A2$ correspond to $A_1,A_2$ and $Id$, $I2d$ are the identity
matrices of dimension $d$ and $2d$, respectively.

  There exists an extensive literature about RK methods built for
many different purposes \cite{butcher87tna,hairer93sod}. It is
therefore reasonable to look for the most appropriate scheme to be
used for each problem. In practice, explicit RK methods are
preferred, because its implementation is usually simpler. They
typically require more stages than implicit methods and, in
general, a higher number of evaluations of the matrix $A(t)$,
although this is not always the case. For instance, the 4-stage
fourth order method (\ref{RK4}) only requires two new evaluations
of $A(t)$ per step ($A(t_n + h/2)$ and $A(t_n + h)$), which is
precisely the minimum number of evaluations needed by any fourth
order method.  In our
numerical examples we will also use the 7-stage sixth order method
with coefficients \cite[p. 203-205]{butcher87tna}
\begin{equation}\label{RK6-tablero}
 \begin{array}{c|ccccccc}
0 &  &  &  &  &  &  &  \\
\frac{5-\sqrt{5}}{10} & \frac{5-\sqrt{5}}{10} &  &  &  &  &  &  \\
\frac{5+\sqrt{5}}{10} & \frac{-\sqrt{5}}{10} &
\frac{5+2\sqrt{5}}{10} &  &
&  &  &  \\
\frac{5-\sqrt{5}}{10} & \frac{-15+7\sqrt{5}}{20} &
\frac{-1+\sqrt{5}}{4}
& \frac{15-7\sqrt{5}}{10} &  &  &  &  \\
\frac{5+\sqrt{5}}{10} & \frac{5-\sqrt{5}}{60} & 0 & \frac{1}{6} & \frac{%
15+7\sqrt{5}}{60} &  &  &  \\
\frac{5-\sqrt{5}}{10} & \frac{5+\sqrt{5}}{60} & 0 &
\frac{9-5\sqrt{5}}{12}
& \frac{1}{6} & \frac{-5+3\sqrt{5}}{10} &  &  \\
1 & \frac{1}{6} & 0 & \frac{-55+25\sqrt{5}}{12} &
\frac{-25-7\sqrt{5}}{12}
& 5-2\sqrt{5} & \frac{5+\sqrt{5}}{2} &   \\
\hline  & \frac{1}{12} & 0 & 0 & 0 & \frac{5}{12} & \frac{5}{12} &
\frac{1}{12}
\end{array}
\end{equation}
Observe that this method only requires three new evaluations of
the matrix $A(t)$ per step. This is, in fact, the minimum number
for a sixth-order method. Other implicit RK schemes widely used in the
literature involving the minimum number of stages at each order are
based on Gauss--Legendre collocation points \cite{hairer93sod}.
For instance, the corresponding
methods of order four and six (with two and three stages, respectively) have
the following coefficients:
\begin{equation}\label{RKGL46-tablero}
\begin{array}{c|cc}
 \frac{3-\sqrt{3}}{6} &   \frac14 & \frac{3-2\sqrt{3}}{12} \\
 \frac{3+\sqrt{3}}{6} & \frac{3+2\sqrt{3}}{12} & \frac14  \\
\hline  & \frac12  &  \frac12
\end{array}
 \qquad \qquad
\begin{array}{c|ccccc}
 \frac{5-\sqrt{15}}{10} & \frac{5}{36}  &   \;  &  \frac29 - \frac{\sqrt{15}}{15} &  \; &
  \frac{5}{36} - \frac{\sqrt{15}}{30}   \\
 \frac12 & \frac{5}{36} + \frac{\sqrt{15}}{24}  & \; &  \frac29 &  \; &
  \frac{5}{36} + \frac{\sqrt{15}}{24}  \\
 \frac{5+\sqrt{15}}{10} & \frac{5}{36} + \frac{\sqrt{15}}{30}  &  \; &
     \frac29 + \frac{\sqrt{15}}{15} &  \; &  \frac{5}{36} \\
\hline  & \frac{5}{18}   & &  \frac{4}{9}  &  &  \frac{5}{18}
\end{array}
\end{equation}

\subsection{Preservation of qualitative properties}
   \label{Qualitproperties}

Notice that the numerical solution provided by the class of
Runge--Kutta schemes, and in general by an integrator of the form
(\ref{nde.2}), is constructed as a sum of vectors in
$\mathbb{R}^d$. Let us point out some (undesirable) consequences of this
fact. Suppose, for instance, that $\mathbf{x}$ is a vector known
to evolve on a sphere. One does not expect that $\mathbf{x}_n$
built as a sum of two vectors as in (\ref{nde.2}) preserves this feature of the exact
solution, whereas approximations of the form $\mathbf{x}_{n+1} =
Q_n \mathbf{x}_n$, with $Q_n$ an orthogonal matrix, clearly lie
on the sphere.

In section \ref{NL-Magnus} we have introduced isospectral flows
(\ref{nlm5}), which include as a particular case the system
\begin{equation}   \label{nde.4}
   Y^\prime = [A(t),Y],   \qquad Y(t_0) = Y_0,
\end{equation}
with $A$ a skew-symmetric matrix.
As we have shown there, if $Y_0$ is a symmetric matrix, the exact solution can be factorized
as $Y(t) = Q(t) Y_0 Q^T(t)$, with $Q(t)$ an orthogonal matrix satisfying the equation
\begin{equation}   \label{nde.5}
    Q^\prime = A(t) Q, \qquad Q(0) = I
\end{equation}
and, in addition, $Y(t)$ and $Y(0)$ have the same eigenvalues.
This is also true when $Y$ is Hermitian, $A$ is skew-Hermitian and
$Q$ is unitary, in which case (\ref{nde.4}) and (\ref{nde.5})
can be interpreted as particular examples of
the Heisenberg and Schr\"odinger equations in
Quantum Mechanics, respectively.

When a numerical scheme of the form (\ref{nde.2}) is applied to
(\ref{nde.5}), in general, the approximations $Q_n$ will no longer
be unitary matrices and therefore $Y_n$ and $Y_0$ will not be unitarily
similar. As a result, the isospectral character of the system
(\ref{nde.4}) is lost in the numerical description. Observe that
explicit Runge--Kutta methods employ the ansatz that locally the
solution of the differential equation behaves like a polynomial in
$t$, so that one cannot expect that the approximate solution be a
unitary matrix. In this sense, explicit Runge--Kutta methods
present the same drawbacks in the numerical analysis of differential
equations as the standard time-dependent
perturbation theory when looking for analytical approximations.

Implicit Runge--Kutta methods, on the other hand, can be
considered as rational approximations,
and in some
cases the outcome they provide is a unitary matrix. For this class
of methods, however, matrices of
relatively large sizes have to be inverted, making the algorithms
computationally expensive. Furthermore, the evolution of many
systems (including highly oscillatory problems)
can not be efficiently approximated neither by
polynomial nor rational approximations.

With all these considerations in mind, a pair of questions arise in a quite natural way:
\begin{itemize}
 \item[(Q1)] Is it possible
to design numerical integration methods for equation (\ref{nde.1})
such that the corresponding
numerical approximations still preserve the main qualitative features
of the exact solution?
 \item[(Q2)] Since the Magnus expansion constitutes an extremely useful
 procedure for
 obtaining analytical (as opposed to numerical) approximate solutions to
 the linear evolution equation (\ref{NI.1}), is it feasible to construct
 efficient numerical integration schemes from the general formulation
 exposed in section \ref{section2}?
\end{itemize}

It turns out that both questions can be answered in the affirmative.
As a matter of fact, it has been trying to address (Q1) how the
field of \emph{geometric numerical integration} has been developed
during the last years. Here the aim is to reproduce the
qualitative features of the solution of the differential equation
which is being discretised, in particular its geometric
properties. The motivation for developing such
structure-preserving algorithms arises independently in areas of
research as diverse as celestial mechanics, molecular dynamics,
control theory, particle accelerators physics, and numerical
analysis
\cite{hairer06gni,iserles00lgm,leimkuhler04shd,mclachlan02sm,mclachlan06gif}.

Although diverse, the systems appearing in these areas have
one important common feature. They all preserve some
underlying geometric structure which influences
the qualitative nature of the phenomena they produce.
In the field of geometric numerical integration these properties are built into
the numerical method, which gives the method an improved
qualitative behaviour, but also allows for a significantly more
accurate long-time integration than with general-purpose methods.

In addition to the construction of new numerical algorithms,
an important aspect of geometric integration is the explanation
of the relationship between preservation of the geometric
properties of a numerical method and the observed favourable
error propagation in long-time integration.

Geometric numerical integration has been an active and
interdisciplinary research area during the last decade, and
nowadays is the subject of intensive development. Perhaps the most
familiar examples of \emph{geometric integrators} are symplectic
integration algorithms in classical Hamiltonian dynamics,
symmetric integrators for reversible systems and methods
preserving first integrals and numerical methods on manifolds
\cite{hairer06gni}.

A particular class of geometric numerical integrators are the
so-called \emph{Lie-group methods}. If the matrix $A(t)$ in
(\ref{NI.1}) belongs
to a Lie algebra $\mathfrak{g}$, the aim of Lie-group methods is
to construct numerical solutions staying in the corresponding Lie
group $\mathcal{G}$ \cite{iserles00lgm}.

With respect to question (Q2) above, it will be the subject of the
next subsection, where the main issues involved in the
construction of numerical integrators based on the Magnus
expansion are analyzed. The methods thus obtained preserve the
qualitative properties of the system and in addition are highly
competitive with other more conventional numerical schemes with
respect to accuracy and computational effort. They constitute, therefore, clear
examples of Lie-group methods.

\subsection{Magnus integrators for linear systems}
  \label{Mintegrators}

Since the Magnus series only converges locally, as we have pointed out before,
when the length of
the time-integration interval exceeds the bound provided by
Theorem \ref{conv-mag}, and in the spirit of any numerical
integration method, the usual procedure consists in dividing the
time interval $[t_0,t_f]$ into $N$ steps such that the Magnus series
converges in each subinterval $[t_{n-1},t_n]$, $n=1,\ldots,N$,
with $t_N = t_f$. In this way the solution of (\ref{NI.1}) at the
final time $t_f$ is represented by
\begin{equation}  \label{Mint.1}
  Y(t_N) =  \prod_{n=1}^N \exp(\Omega(t_n,t_{n-1})) \, Y_0,
\end{equation}
and the series $\Omega(t_n,t_{n-1})$ has to be appropriately truncated.

Early steps in this approach were taken in \cite{chan68imc} and
\cite{chang69eso}, where second and fourth order numerical schemes were
used for calculations of collisional inelasticity and potential
scattering, respectively. Those authors were well aware that the
resulting integration method ``would become practical only when
the advantage of being able to use bigger step sizes outweighs the
disadvantage in having to evaluate the integrals involved in the
Magnus series and then doing the exponentiation" \cite{chan68imc}.
Later on, by following a similar approach, Devries
\cite{devries85aho} designed a numerical procedure for determining
a fourth order approximation to the propagator employed in the
integration of the single channel Schr\"odinger equation, but it
was in the pioneering work \cite{iserles99ots} where Iserles and
N{\o}rsett carried out the first systematic study of the Magnus
expansion with the aim of constructing numerical integration
algorithms for linear problems. To design the new integrators, the
explicit time dependency of each term $\Omega_k$ had to be
analyzed, in particular its order of approximation in time to the
exact solution.

Generally speaking, the process of rendering the Magnus expansion
a practical numerical integration algorithm involves three steps.
First, the $\Omega$ series is truncated at an appropriate order.
Second, the multivariate integrals in the truncated series
$\Omega^{[p]} = \sum_{i=1}^p \Omega_i$
are replaced by conveniently chosen approximations.
Third, the exponential of the matrix $\Omega^{[p]}$ has to be
computed. We now briefly consider the first two issues, whereas
the general problem of evaluating the matrix exponential will be treated in section
\ref{exponen}.

We have shown in subsection \ref{time-symmetry} that the Magnus
expansion is time symmetric. As a consequence of eq. (\ref{t-s6}),
if $A(t)$ is analytic and one evaluates  its Taylor series centered
around the midpoint of a particular subinterval $[t_n, t_n+h]$,
then each term in $\Omega_k$ is an odd function of
$h$, and thus $\Omega_{2s+1} = \mathcal{O}(h^{2s+3})$ for $s \geq
1$. Equivalently, $\Omega^{[2s-2]} =
 \Omega + \mathcal{O}(h^{2s+1})$ and $\Omega^{[2s-1]} = \Omega +
\mathcal{O}(h^{2s+1})$. In other words, for achieving an
integration method of order  $2s$ ($s>1$)  only terms up to
$\Omega_{2s-2}$ in the $\Omega$ series are required
\cite{blanes00iho,iserles99ots}. For this reason, in general, only
even order methods are considered.

Once the series expansion is truncated up to an appropriate order,
the multidimensional integrals involved have to be computed or at
least conveniently approximated. Although at first glance this seems to
be a quite difficult enterprise, it turns out that their very structure allows
one to approximate all the multivariate integrals appearing in $\Omega$
just by evaluating $A(t)$ at the nodes of  a
univariate quadrature \cite{iserles99ots}.

To illustrate how this task can be accomplished, let us expand the
matrix $A(t)$ around the midpoint, $t_{1/2} \equiv t_n + h/2$, of the subinterval
$[t_n,t_{n+1}]$,
\begin{equation}   \label{3.2.1}
    A(t) = \sum_{j=0}^{\infty}  a_j \left(t - t_{1/2} \right)^j, \qquad \mbox{ where } \quad
       a_j = \frac{1}{j!} \, \frac{d^j A(t)}{d t^j} \big|_{t=t_{1/2}},
\end{equation}
and insert
the series (\ref{3.2.1}) into the recurrence defining the Magnus expansion
(\ref{eses})-(\ref{omegn}). In this
way one gets explicitly the expression of $\Omega_k$ up to order $h^6$ as
\begin{eqnarray}  \label{3.2.2}
  \Omega_1 & = & h a_0 + h^3 \frac{1}{12} a_2 + h^5 \frac{1}{80} a_4 +
  \mathcal{O}(h^{7})  \nonumber
  \\
  \Omega_2 & = & h^3 \frac{-1}{12} [a_0,a_1] + h^5 \Big( \frac{-1}{80}
     [a_0,a_3] + \frac{1}{240} [a_1,a_2] \Big) + \mathcal{O}(h^{7})
    \nonumber  \\
  \Omega_3 & = & h^5 \Big( \frac{1}{360} [a_0,a_0,a_2] - \frac{1}{240}
       [a_1,a_0,a_1] \Big) +  \mathcal{O}(h^{7})
      \\
  \Omega_4 & = & h^5 \frac{1}{720} [a_0,a_0,a_0,a_1] + \mathcal{O}(h^{7}),
  \nonumber
\end{eqnarray}
whereas $\Omega_5  =  \mathcal{O}(h^{7})$, $ \Omega_6  =  \mathcal{O}(h^{7})$
and $\Omega_7  =  \mathcal{O}(h^{9})$.
Here we write for clarity $[a_{i_1},a_{i_2}, \ldots, a_{i_{l-1}},a_{i_l}]
\equiv [a_{i_1},[a_{i_2},[\ldots,[a_{i_{l-1}},a_{i_l}]\ldots]]]$.
Notice that, as anticipated, only odd powers of $h$ appear in
$\Omega_k$ and, in particular, $\Omega_{2i+1} =
\mathcal{O}(h^{2i+3})$ for $i>1$.

Let us denote $\alpha_i \equiv h^i a_{i-1}$. Then
$[\alpha_{i_1},\alpha_{i_2}, \ldots,
\alpha_{i_{l-1}},\alpha_{i_l}]$ is an element of order
$h^{i_1 + \cdots + i_l}$.  In fact, the matrices
$\{\alpha_i\}_{i \ge 1}$ can be considered as the generators (with
grade $i$) of a graded free Lie algebra
$\mathcal{L}(\alpha_1,\alpha_2,\ldots)$ \cite{munthe-kaas99cia}.
It turns out that it is possible to build methods of order $p
\equiv 2s$ by considering only terms involving
$\alpha_1,\ldots,\alpha_{s}$ in $\Omega$. Then, these terms can
be approximated  by appropriate linear combinations of the
matrix $A(t)$ evaluated at different points. In particular,
up to order four  we have to approximate
\begin{equation}  \label{eq.2.3a}
  \Omega  =  \alpha_1 - \frac{1}{12} [\alpha_1,\alpha_2]
       +\mathcal{O}(h^5),
\end{equation}
whereas up to order six the relevant expression is
\begin{eqnarray}
  \Omega & = & \alpha_1 +  \frac{1}{12} \alpha_3 - \frac{1}{12} [\alpha_1,\alpha_2]
   + \frac{1}{240} [\alpha_2,\alpha_3] + \frac{1}{360} [\alpha_1,\alpha_1,\alpha_3]
    \label{eq.2.3b}    \\
   &  & - \frac{1}{240} [\alpha_2,\alpha_1,\alpha_2] + \frac{1}{720} [\alpha_1,\alpha_1,\alpha_1,\alpha_2]
   +\mathcal{O}(h^7). \nonumber
\end{eqnarray}
In order to present methods which
can be easily adapted for different quadrature rules we introduce
the averaged (or generalized momentum) matrices
\begin{equation}  \label{eq.2.4}
 A^{(i)}(h) \equiv  \frac{1}{h^{i}} \int_{t_n}^{t_n+h} \, \left(t -
   t_{1/2} \right)^i A(t) dt =
    \frac{1}{h^i} \int_{-h/2}^{h/2} t^i  A(t+t_{1/2}) dt
\end{equation}
for $i=0, \ldots, s-1$. If their exact evaluation is not possible
or is computationally expensive, a numerical quadrature may be
used instead. Suppose that $b_i$, $c_i$, $( i=1,\ldots,k)$, are
the weights and nodes of a particular quadrature rule, say $X$ (we
will use $X=G$ for Gauss--Legendre quadratures and $X=NC$ for
Newton--Cotes quadratures) of order $p$, respectively
\cite{abramowitz65hom},
\[
  A^{(0)} = \int_{t_n}^{t_n+h} A(t) dt = h \sum_{i=1}^k b_i A_i +
    \mathcal{O}(h^{p+1}),
\]
with $A_i \equiv A(t_n+c_ih)$. Then it is possible to
approximate all the integrals $A^{(i)}$ (up to the required order) just by
using only the evaluations $A_i$ at the nodes $c_i$ of the
quadrature rule required to compute $A^{(0)}$. Specifically,
\begin{equation}   \label{ru1}
  A^{(i)}= h \sum_{j=1}^k b_j\left(c_j-\frac12\right)^{i} A_j.
    \qquad\quad  i=0,\ldots,s-1,
\end{equation}
or equivalently, $A^{(i)}  =  h  \sum_{j=1}^k
\left(Q^{(s,k)}_{X}\right)_{ij} A_{j}$
with
$\left(Q^{(s,k)}_{X}\right)_{ij}=b_j\left(c_j-\frac12\right)^{i}$.

In particular, if fourth and sixth order Gauss--Legendre quadrature
rules are considered, then for
$s=k=2$ we have \cite{abramowitz65hom}
\[
  b_1=b_2=\frac{1}{2}, \quad c_1= \frac{1}{2} - \frac{\sqrt{3}}{6}, \quad c_2 =
\frac{1}{2} + \frac{\sqrt{3}}{6},
\]
to order four, whereas for $s=k=3$,
\[
 b_1=b_3= \frac{5}{18}, \quad b_2=\frac{4}{9}, \quad c_1 = \frac{1}{2} - \frac{\sqrt{15}}{10},
  \quad c_2 = \frac{1}{2},  \quad c_3 = \frac{1}{2} + \frac{\sqrt{15}}{10},
\]
to order six, so that
\begin{equation}\label{Gauss-Legendre}
  Q_G^{(2,2)}=   \left( \begin{array}{cc}
   \frac12 & \  \  \frac12 \\
    -\frac{\sqrt{3}}{12} & \frac{\sqrt{3}}{12}
  \end{array} \right), \qquad
  Q_G^{(3,3)}=   \left( \begin{array}{ccc}
   \frac{5}{18} & \  \  \frac49 &  \  \   \frac{5}{18}\\
    -\frac{\sqrt{15}}{36} & 0 & \frac{\sqrt{15}}{36} \\
    \frac{1}{24} & 0 & \frac{1}{24}
  \end{array} \right). \qquad
\end{equation}
Furthermore, substituting (\ref{3.2.1}) into (\ref{eq.2.4}) we
find (neglecting higher order terms)
\begin{equation}   \label{eq.2.4b}
   A^{(i)} = \sum_{j=1}^s \left(T^{(s)}\right)_{ij}  \alpha_j
   \; \equiv  \;
   \sum_{j=1}^s  \frac{1-(-1)^{i+j}}{(i+j)2^{i+j}}  \alpha_j,
   \qquad 0 \le i \le  s-1.
\end{equation}
If this relation is inverted (to order four, $s=2$, and six,
$s=3$) one has
\begin{equation}\label{CambioBase}
  R^{(2)}\equiv (T^{(2)})^{-1} = \left( \begin{array}{cc}
   1 & \ \ 0 \\
   0 & \ \ 12
  \end{array} \right), \qquad
   R^{(3)}=  \left( \begin{array}{ccc}
   \frac{9}{4} & \ 0 &  \ -15 \\
   0 & \ 12 & \ 0 \\
   -15 & \ 0 & \ 180
  \end{array} \right)
\end{equation}
respectively, so that the corresponding expression of $\alpha_i$
in terms of $A^{(i)}$ or
$A_j$ is given by
\begin{equation}  \label{CBaseMu}
   \alpha_{i} =  \sum_{j=1}^s  \left( R^{(s)} \right)_{ij} A^{(j-1)}
     =   h \sum_{j=1}^k \left(  R^{(s)}Q_{X}^{(s,k)}  \right)_{ij}
                   A_{j}.
\end{equation}
Thus, by virtue of (\ref{CBaseMu}) we can write $\Omega(h)$ in
terms of the univariate integrals (\ref{eq.2.4}) or in terms of
any desired quadrature rule. In this way, one gets the final
numerical approximations to $\Omega$. Fourth and sixth-order
methods can be obtained by substituting in (\ref{eq.2.3a}) and
(\ref{eq.2.3b}), respectively. The algorithm then provides an
approximation for $Y(t_{n+1})$ starting from $Y_n \approx Y(t_n)$,
with $t_{n+1} = t_n + h$.

The observant reader surely has noticed that, up to order $h^6$,
there are more terms involved in (\ref{3.2.2}) than those
considered in (\ref{eq.2.3a}) or (\ref{eq.2.3b}): specifically,
$\frac{1}{12} \alpha_3$ in (\ref{eq.2.3a}) and $\frac{1}{80}
\alpha_5$ and $-\frac{1}{80} [\alpha_1, \alpha_4]$ in
(\ref{eq.2.3b}). The reason is that $\Omega^{[6]}$ can be
approximated by $A^{(0)},A^{(1)},A^{(2)}$ up to order $h^6$ and
then, these omitted terms are automatically reproduced when either
$A^{(0)},A^{(1)},A^{(2)}$ are evaluated analytically or are
approximated by any symmetric quadrature rule of order six or higher.

  Another important issue involved in any approximation based on the
Magnus expansion is the number
of commutators appearing in $\Omega$. As it is already evident from
(\ref{3.2.2}), this number rapidly increases with the order, and so it might
constitute a major factor in the overall computational cost of the resulting
numerical methods. It is possible,
however, to design an optimization procedure aimed to reduce this
number to a minimum \cite{blanes02hoo}. For
instance, a straightforward counting of the number of commutators
 in (\ref{eq.2.3b})  suggests
that it seems necessary to compute seven commutators up to order
six in $h$, whereas the general analysis carried out in
\cite{blanes02hoo} shows that this can be done with only three
commutators. More specifically, the scheme
\begin{eqnarray}  \label{aprox6}
   C_1 & = & [\alpha_1,\alpha_2] , \nonumber\\
   C_2 & = & -\frac{1}{60} [\alpha_1,2\alpha_3+C_1] \\
   \Omega^{[6]} & \equiv & \alpha_1 + \frac{1}{12}\alpha_3
   + \frac{1}{240}[-20\alpha_1 - \alpha_3 + C_1 , \alpha_2 + C_2], \nonumber
\end{eqnarray}
verifies that $\Omega^{[6]}=\Omega+\mathcal{O}(h^7)$. Three is in fact the minimum
number of commutators required to get a sixth-order approximation
to $\Omega$.

This technique to reduce the number of commutators is indeed valid
for any element in a graded free Lie algebra. It has been used, in
particular, to obtain approximations to the Baker--Campbell--Hausdorff
formula up to a given order with the
 minimum number of commutators \cite{blanes03ool}.

As an illustration, next we provide
the relevant expressions for integration schemes of order 4 and 6, which
readily follow from the previous analysis.

\noindent \textbf{Order 4}. Choosing the Gauss--Legendre quadrature rule,
one has to evaluate

\begin{equation}\label{AiG4}
     A_1 =  A(t_n + (\frac{1}{2} - \frac{\sqrt{3}}{6}) h), \qquad
     A_2 =  A(t_n + (\frac{1}{2} + \frac{\sqrt{3}}{6}) h)
\end{equation}
and thus, taking $Q_G^{(2,2)}$ in (\ref{Gauss-Legendre}),
$R^{(2)}$ in (\ref{CambioBase}) and substituting in
(\ref{CBaseMu}) we find
\begin{equation}\label{alfasG4}
\alpha_1=\frac{h}{2}(A_1+A_2), \qquad
\alpha_2=\frac{h\sqrt{3}}{12}(A_2-A_1).
\end{equation}
Then, by replacing in (\ref{eq.2.3a}) we obtain
\begin{equation}  \label{or4GL}
    \begin{array}{ccc}
       &     \Omega^{[4]}(h) =  \frac{h}{2} (A_1 + A_2) - h^2 
          \frac{\sqrt{3}}{12} [A_1, A_2] &   \\
       &   Y_{n+1} =  \exp( \Omega^{[4]}(h) ) Y_n.
    \end{array}   
\end{equation}
Alternatively, evaluating $A$ at equispaced points, with $k=3$ and
$c_1=0,c_2=1/2,c_3=1; \ b_1=b_3=1/6, b_2=2/3$ (i.e., using the
Simpson rule to approximate  $\int_{t_n}^{t_n+h}A(s)ds$),
\[
   A_1 = A(t_n), \qquad A_2 = A(t_n + \frac{h}{2}), \qquad
    A_3 = A(t_n + h)
\]
we have instead  $\alpha_1=\frac{h}{6}(A_1+4A_2+A_3)$, 
$ \; \alpha_2=h(A_3-A_1)$, and then
\begin{eqnarray}  \label{or4NC1}
      \Omega^{[4]}(h) & = & \frac{h}{6} (A_1 + 4A_2 + A_3) -
          \frac{h^2}{72} [A_1 + 4A_2 + A_3, A_3-A_1].
\end{eqnarray}
It should be noticed that other possibilities not directly obtainable from the previous
analysis are equally valid. For instance, one could consider
\begin{eqnarray}  \label{or4NC2}
      \Omega^{[4]}(h) & = & \frac{h}{6} (A_1 + 4A_2 + A_3) -
          \frac{h^2}{12} [A_1, A_3].
\end{eqnarray}
Although apparently more $A$ evaluations are necessary in (\ref{or4NC1}) and
(\ref{or4NC2}), this
is not the case actually, since $A_3$ can be reused at the next integration
step.

\noindent \textbf{Order 6}.
In terms of Gauss--Legendre collocation points one has
\[
 A_1 = A \big( t_n + (\frac{1}{2} - \frac{\sqrt{15}}{10}) h \big), \quad
 A_2 = A \big( t_n + \frac{1}{2} h \big), \quad
 A_3 = A \big( t_n + (\frac{1}{2} + \frac{\sqrt{15}}{10}) h \big)
\]
and similarly we obtain
\begin{equation}   \label{4.1.4}
  \alpha_1 = h A_2,  \quad \alpha_2 = \frac{\sqrt{15}h}{3} (A_3 - A_1), \quad
  \alpha_3 = \frac{10h}{3} (A_3 - 2 A_2 + A_1),
\end{equation}
which are then inserted in (\ref{aprox6}) to get the approximation
$Y_{n+1} = \exp( \Omega^{[6]} ) Y_n$.

 If the matrix
$A(t)$ is only known at equispaced points, we can use the
Newton--Cotes (NC) quadrature values with $s=3$ and $k=5$,
$b_1=b_5= 7/90, b_2=b_4=32/90,b_3=12/90$ and $c_j=(j-1)/4, \
j=1,\ldots,5$. Then, using the corresponding matrix
$Q_{NC}^{(3,5)}$ from (\ref{ru1}) we get
\begin{eqnarray}  \label{4.1.5}
   \alpha_1 & = & \frac{1}{60} \big( -7 (A_1 + A_5) + 28 (A_2 + A_4) +
        18 A_3 \big)  \nonumber  \\
   \alpha_2 & = & \frac{1}{15} \big( 7 (A_5 - A_1) + 16 (A_4 - A_2) \big) \\
   \alpha_3 & = & \frac{1}{3} \big( 7 (A_1 + A_5) - 4 (A_2 + A_4) -
        6 A_3 \big).  \nonumber
\end{eqnarray}
Both schemes involve the minimum number of commutators (three)
and require three or four evaluations of the matrix $A(t)$ per
integration step (observe that $A_5$ can be reused in the next step in
the Newton--Cotes implementation because $c_1=0$ and $c_5=1$).

Higher orders can be treated in a similar way. For instance, an
8th-order Magnus method can be obtained with only six commutators
\cite{blanes02hoo}. Also variable step size techniques can be
easily implemented \cite{blanes00iho,iserles99oti}.

\subsubsection{From Magnus to Fer and Cayley methods}

For
arbitrary matrix Lie groups it is feasible to design numerical
methods based also on the Fer and Wilcox expansions, whereas for
the $J$-orthogonal group (eq. (\ref{j-ortho})) the Cayley transform
also maps the Lie algebra onto the Lie group
\cite{postnikov94lga} and thus it allows us to build a new class of Lie-group
methods. Here we briefly show how these integration
methods can be easily constructed from the previous schemes based
on Magnus. In other words, if the solution of (\ref{NI.1}) in a
neighborhood of $t_0$ is written as
\begin{eqnarray}
Y(t_0 + h) & = & \e^{\Omega(h)} \,Y_0  \hspace*{5.3cm} \mbox{Magnus}  \label{eq.2} \\
& = & \e^{F_1(h)} \e^{F_2(h)}\cdots Y_0 \qquad \qquad \qquad
 \quad \qquad \mbox{Fer}
\label{eq.3} \\
& = & \e^{S_1(h)} \e^{S_2(h)} \cdots \e^{S_2(h)} \e^{S_1(h)} Y_0   \hspace*{1.75cm}
\mbox{Symmetric  Fer}  \label{eq.4} \\
 & = & \left( I - \frac{1}{2} C(h) \right)^{-1} \left( I + \frac{1}{2} C(h)
\right) Y_0  \hspace*{0.9cm}  \mbox{Cayley}  \label{eq.6}
\end{eqnarray}
one may express the functions $F_{i}(h)$, $S_{i}(h)$ and $C(h)$ in terms of
the successive approximations to $\Omega $ and, by using the same
techniques as in the previous section, obtain the new methods. As the
schemes based on the Wilcox factorization are quite similar as the
Fer methods, they will not be considered here.

\subsubsection{Fer based methods}

To obtain integration methods based on the Fer factorization
(\ref{eq.3}) one applies the Baker--Campbell--Hausdorff (BCH)
formula after equating to the Magnus expansion (\ref{eq.2}). More
specifically, in the domain of convergence of expansions
(\ref{eq.2}) and (\ref{eq.3}) we can write
\[
\e^{\Omega(h)} = \e^{F_1(h)} \, \e^{F_2(h)} \ \cdots,
\]
where $F_1 = \Omega_1$ is the first term in the Magnus series,
$F_2 = O(h^3)$ and $F_3 = O(h^7)$. Then, a $p$-th order
algorithm with $3\leq p\leq 6$, based on the Fer expansion
requires to compute $F_2^{[p]}$ such that
\begin{equation}  \label{eq.4.1.4}
  Y(t_n + h) = \e^{F_1(h)} \, \e^{F_2^{[p]}(h)} \, Y(t_n)
  + \mathcal{O}(h^{p+1}).
\end{equation}
Taking into account that
\[
  \e^{\Omega^{[p]}(h)} = \e^{F_1(h)} \, \e^{F_2^{[p]}(h)}
  + \mathcal{O}(h^{p+1}),
\]
we have ($F_1=\Omega_1$)
\[
  \e^{F_2^{[p]}(h)}= \e^{-\Omega_1(h)} \, \e^{\Omega^{[p]}(h)}
  + \mathcal{O}(h^{p+1}).
\]
Then, by using the BCH formula and simple algebra to remove higher
order terms we obtain to order four
\begin{equation}  \label{eq.4.1.2}
F_2^{[4]} =  - \frac{1}{12} \big( [\alpha_1,\alpha_2] -
\frac{1}{2} [\alpha_1,\alpha_1,\alpha_2] \big),
\end{equation}
 so that two
commutators are needed in this case. A sixth-order method can be
similarly obtained with four commutators \cite{blanes02hoo}. These
methods are slightly more expensive than their Magnus counterpart
and they do not preserve the time-symmetry of the exact solution.
This can be fixed by the self-adjoint version of the Fer
factorization in the form (\ref{eq.4}) proposed in
\cite{zanna01tfe} and presented in a more efficient way in
\cite{blanes02hoo}. The schemes based on (\ref{eq.4}) up to order
six are given by
\begin{equation}\label{sym-Fer-n}
  Y(t_n + h) = \e^{S_1(h)} \e^{S_2^{[p]}(h)} \e^{S_1(h)} Y_n
\end{equation}
with $S_1 = \Omega_1/2$. A fourth-order method is given by
\begin{equation}  \label{eq.4.2.3}
  S_2^{[4]}(h) = -\frac{1}{12} [\alpha_1,\alpha_2]
\end{equation}
\begin{equation}  \label{eq.4.2.4}
\end{equation}
and a sixth-order one by
\begin{eqnarray}  \label{eq.4.2.5}
  s_1 & = & [\alpha_1, \alpha_2] \nonumber  \\
  r_1 & = & \frac{1}{120} \, [\alpha_1, -4 \alpha_3 + 3 s_1]  \nonumber  \\
 S_2^{[6]}(h) & = & \frac{1}{240} \,
       \big[ -20 \alpha_1 - \alpha_3 + s_1 \ , \ \alpha_2 + r_1
       \big].
\end{eqnarray}
To complete the formulation of the scheme, the $\alpha_i$ have to
be expressed in terms of the matrices $A_i$ evaluated at the
quadrature points (e.g., equations (\ref{4.1.4}) or
(\ref{4.1.5})).

\subsubsection{Cayley-transform methods}

We have seen in subsection \ref{notations} that for the
$J$-orthogonal group $\mathrm{O}_J(n)$ the Cayley transform (\ref
{cayley1}) provides a useful alternative to the exponential
mapping relating the Lie algebra to the Lie group. This fact is
particularly important for numerical methods where the evaluation
of the exponential matrix is the most computation-intensive part
of the algorithm.

If the solution of eq. (\ref{NI.1}) is written as
\begin{equation}  \label{cayl2}
Y(t) = \left( I - \frac{1}{2} C(t) \right)^{-1} \left( I +
\frac{1}{2} C(t) \right) Y_0
\end{equation}
then $C(t) \in \mathrm{o}_J(n)$ satisfies the equation
\cite{iserles01oct}
\begin{equation}  \label{cayl3}
  C^\prime= A - \frac{1}{2} [C,A] - \frac{1}{4} CAC, \qquad t
\geq t_0, \qquad C(t_0) = 0.
\end{equation}
Time-symmetric methods of order 4 and 6 have been obtained based
on the Cayley transform (\ref{cayl2}) by expanding the solution of
(\ref{cayl3}) in a recursive manner and constructing quadrature
formulae for the multivariate integrals that appear in the
procedure \cite{diele98tct,iserles01oct,marthinsen01qmb}. It turns out that
efficient Cayley based methods can be built directly from Magnus
based integrators \cite{blanes02hoo}. In particular, we get:

 \noindent \textbf{Order 4}:
\begin{equation}  \label{eq.4.3.2}
C^{[4]} = \Omega^{[4]} \big( I - \frac{1}{12} (\Omega^{[4]})^2
\big) = \alpha_1 - \frac{1}{12} [\alpha_1,\alpha_2] - \frac{1}{12}
\alpha_1^3 + O(h^5),
\end{equation}
where $C^{[4]} =  C(h) + O(h^5)$.

 \vspace*{0.1cm}
 \noindent \textbf{Order 6}:
\begin{equation}  \label{eq.51b}
C^{[6]} = \Omega^{[6]} \left( I - \frac{1}{12}
  ( \Omega^{[6]})^2 \Big( I - \frac{1}{10}
     (\Omega^{[6]})^2 \Big) \right) = C(h) + O(h^7).
\end{equation}
Three matrix-matrix products are required in addition to the three
commutators involved in the computation of $\Omega^{[6]}$, for a
total of nine matrix-matrix products per step.

\subsection{Numerical illustration: the Rosen--Zener model revisited}
\label{niarz}

Next we apply the previous numerical schemes to the integration of
the differential equation governing the evolution of a particular quantum
two-level system. Our purpose here is to illustrate the main issues
involved and compare the different approximations obtained with both
the analytical treatment done in section \ref{section4} and the exact
result. Specifically, we consider
the Rosen--Zener model in the
Interaction Picture already analyzed in subsection  \ref{2L}. In this case
the equation to be solved is $U_I' = \tH_I(t) U_I$, or
equivalently, equation (\ref{NI.1}) with $Y(t) = U_I(t)$
and coefficient matrix ($\hbar = 1$)
\begin{equation}   \label{marz}
   A(t) = \tH_I(t) =  -i V(s) \big(\sigma_1 \cos(\xi s) - \sigma_2 \sin(\xi s)
  \big) \equiv  -i \, {\bf b}(s) \cdot \boldsymbol{\sigma}.
\end{equation}
Here $V(s)=V_0/\cosh(s)$,  $\xi =\omega T$ and $s = t/T$.
Given the initial condition $|+\rangle\equiv(1,0)^T$ at
$t=-\infty$, our purpose is to get an approximate value for the transition probability 
to the state $|-\rangle\equiv(0,1)^T$ at $t=+\infty$. Its exact expression is
collected in the first line of eq. (\ref{P_Sech}), which we reproduce here for
reader's convenience:
\begin{equation}\label{neptex}
  P_{ex} = |(U_I)_{12}(+\infty,-\infty)|^2=
  \frac{\sin^2\gamma}{\cosh^2(\pi\xi/2)},
\end{equation}
with $\gamma=\pi V_0T$.

To obtain in practice a numerical approximation to $P_{ex}$ we
have to integrate the equation in a sufficiently large time interval. We take the
initial condition at $s_0=-25$ and the numerical integration is
carried out until $s_f=25$. Then, we determine $(U_I)_{12}(s_f,s_0)$.

As a first numerical test we take a fixed (and relatively large)
time step $h$ such that the whole numerical integration in the
time interval $s\in[s_0,s_f]$ is carried out with 50 evaluations of
the vector ${\bf b}(s)$ for all methods. In this way their computational cost is
similar.

To illustrate the
qualitative behaviour of Magnus integrators in comparison with
standard Runge--Kutta schemes, the following methods are
considered:
\begin{itemize}
\item {\bf Explicit first-order Euler} (E1):
   $Y_{n+1}=Y_n+hA(t_n)Y_n$ with $t_{n+1}=t_n+h$ and $h=1$
   (solid lines with squares in the figures).
\item {\bf Explicit fourth-order Runge--Kutta} (RK4), i.e., scheme
  (\ref{RK4})  with $h=2$, since only two
  evaluations of ${\bf b}(s)$ per step are required in the linear case (solid lines with triangles).
\item {\bf Second-order Magnus} (M2): we consider the midpoint rule (one evaluation per
  step) to approximate $\Omega_1$ taking $h=1$ (dashed lines), i.e.,
 \begin{equation}\label{mpoint}
  Y_{n+1}=\exp\big( -ih \, {\bf b}_n\cdot \boldsymbol{\sigma} \big) \,
  Y_n   = \left(
  \cos ( hb_n ) \,I - i \frac{\sin  (hb_n) }{h b_n} \, \boldsymbol{b}_n\cdot\boldsymbol{\sigma}
   \right) Y_n.
 \end{equation}
 with ${\bf b}_n\equiv{\bf b}(t_n+ h/2)$ and $b_n=\|{\bf b}_n\|$.
 The trapezoidal rule is equally valid by considering
 ${\bf b}_n\equiv({\bf b}(t_n)+{\bf b}(t_n+ h))/2$.
\item {\bf Fourth-order Magnus} (M4). Using the fourth-order
  Gauss--Legendre rule to approximate the integrals and taking
  $h=2$  one has the scheme (\ref{or4GL}) which for this problem
  reads
\begin{eqnarray}\label{expma4}
  {\bf b}_1 & = & {\bf b}(t_n+c_1h) , \quad
  {\bf b}_2={\bf b}(t_n+c_2h) , \nonumber \\
  {\bf d} & = & \frac{h}{2}\big( {\bf b}_1 + {\bf b}_2 \big)
     -h^2\frac{\sqrt{3}}{6} i ({\bf b}_1 \times {\bf b}_2)   \\
  Y_{n+1} & = & \exp\big( -ih \, {\bf d}\cdot \boldsymbol{\sigma} \big) \,
  Y_n. \nonumber
\end{eqnarray}
with $c_1=\frac12-\frac{\sqrt{3}}{6}, \
c_2=\frac12+\frac{\sqrt{3}}{6}$ (dotted lines).
\end{itemize}

 We choose $\xi=0.3$ and $\xi=1$, and each numerical integration is carried
out for different values of $\gamma$ in the range
$\gamma\in[0,2\pi]$. We plot the corresponding approximations to
the transition probability in a similar way as in Figure
\ref{plot:Sech3} for the analytical treatment. Thus we
obtain the plots of Figure~\ref{plot:Num_xi}.
 As expected, the performance of the methods deteriorates
as $\gamma$ increases. Notice also that the qualitative behaviour of
the different numerical schemes is quite similar as that exhibited by
the analytical approximations. Euler and Runge--Kutta methods
do not preserve unitarity and may lead to
transition probabilities greater than 1 (just like the standard perturbation theory).
On the other hand,
for sufficiently small values of
$\gamma$ (inside the convergence domain of the Magnus
series) the fourth-order Magnus method improves the result achieved by
the second-order, whereas for large values of
$\gamma$ a higher order method does not necessarily lead to a better
approximation.

In Figure~\ref{plot:Num_gamma} we repeat the experiment now taking
$\gamma=1$ and $\gamma=2$, and for different values of $\xi$. In
the first case, only the Euler method differs considerably from
the exact solution and in the second case this happens for both RK
methods.

\begin{figure}[htb]
\begin{center}
\includegraphics[height=7cm,width=14cm]{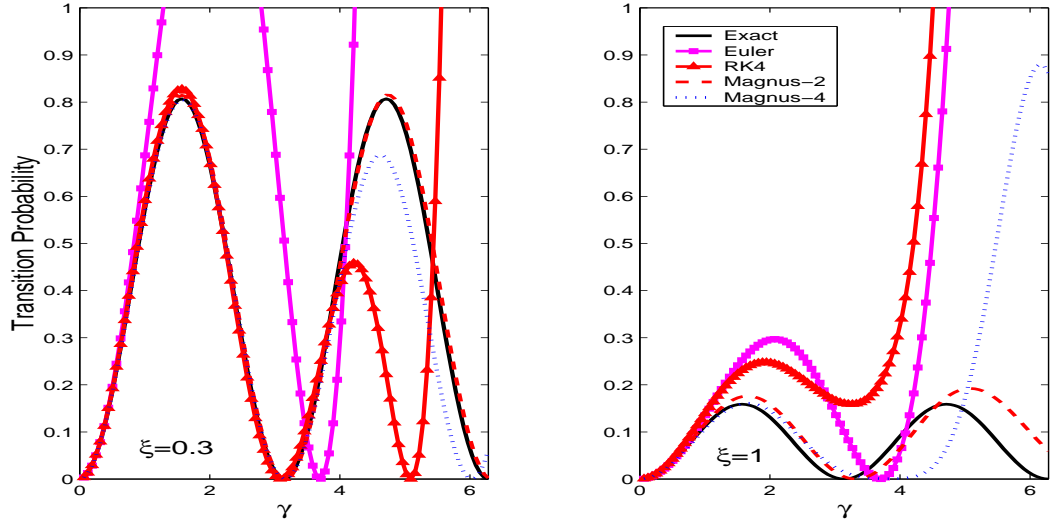}
\end{center}
\caption{{Rosen--Zener model:
 Transition probabilities as a function of $\gamma$, with $\xi=0.3$ and $\xi=1$. The curves
 are coded as follows. Solid line represents the exact result;
 E1: solid lines with squares; RK4: solid lines with triangles; M2:
 dashed lines; M4: dashed-broken lines (indistinguishable from exact result in left panel).}}
 \label{plot:Num_xi}
\end{figure}

\begin{figure}[htb]
\begin{center}
\includegraphics[height=7cm,width=14cm]{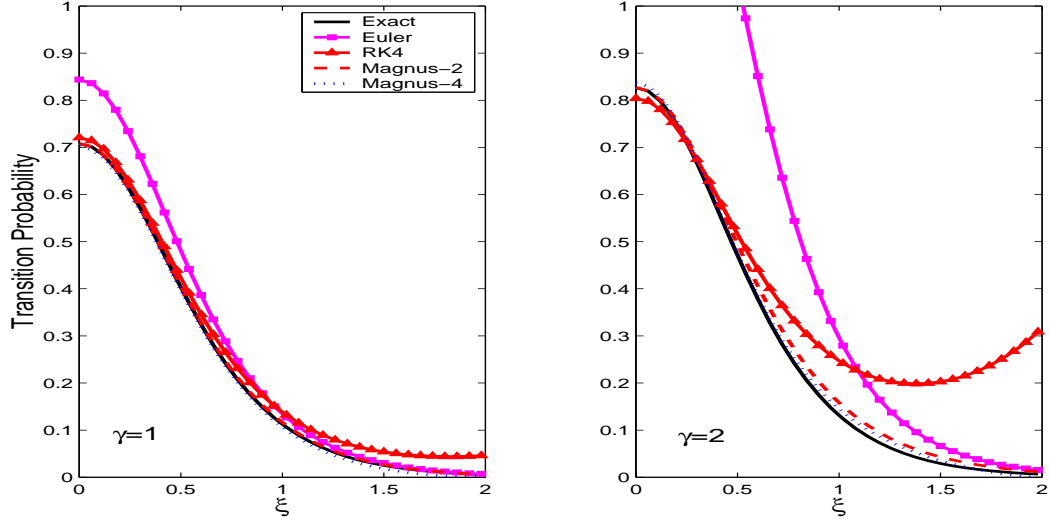}
\end{center}
\caption{{Rosen--Zener model:
 Transition probabilities as a function of $\xi$, with $\gamma=2$ and $\gamma=5$.
 Lines are coded as in Figure \ref{plot:Num_xi}.}}
 \label{plot:Num_gamma}
\end{figure}

To increase the accuracy, one can always take smaller time steps, but then
the number of evaluations of $A(t)$ increases, and this may be  computationally
very expensive for some problems due to the complexity of the time dependence and/or
the size of the matrix. In those cases, it is important to have reliable
numerical methods providing the desired accuracy as fast as
possible or, alternatively, leading to the best possible accuracy
at a given computational cost.

A good perspective of the overall performance of a given numerical integrator is
provided by the so-called efficiency diagram.
This efficiency plot is obtained by carrying out the numerical
integration with different time steps,
corresponding to different number of evaluations of $A(t)$.
 For each run one compares the corresponding approximation
with the exact solution and plot the error as a function of the
total number of matrix evaluations. The results are better
illustrated in a double logarithmic scale. In that case, the slope of the
curves should correspond, in the limit of very small time steps,
to (minus) the order of accuracy of the method.

To illustrate this issue, in Figure \ref{plot:Num_Efic} we collect
the efficiency plots of the previous schemes when  $\xi=0.3$ with
$\gamma=10$ (left) and $\gamma=100$ (right). We have also included
the results obtained by several higher order integrators, namely
the sixth-order RK method (RK6) whose coefficients are collected
in (\ref{RK-tablero}) and the sixth-order Magnus integrator (M6)
given by (\ref{aprox6}) and (\ref{4.1.4}).  We clearly observe
that the Euler method is, by far, the worst choice if accurate
results are desired. Notice the (double) logarithmic scale of the
plots: thus, for instance, when $\gamma=10$ the range goes approximately from
300 to 3000 evaluations of $A(t)$. Magnus integrators, in
addition to providing results in $\mathrm{SU}(2)$ by construction
(as the exact solution), show a better efficiency than
Runge--Kutta schemes for these examples, and this efficiency
increases with the value of $\gamma$ considered. The implicit
Runge-Kutta-Gauss-Legendre methods (\ref{RKGL46-tablero}) show
slightly better results than the explicit RK methods, but still
considerably worst than Magnus integrators.

\begin{figure}[tb]
\begin{center}
\includegraphics[height=7cm,width=14cm]{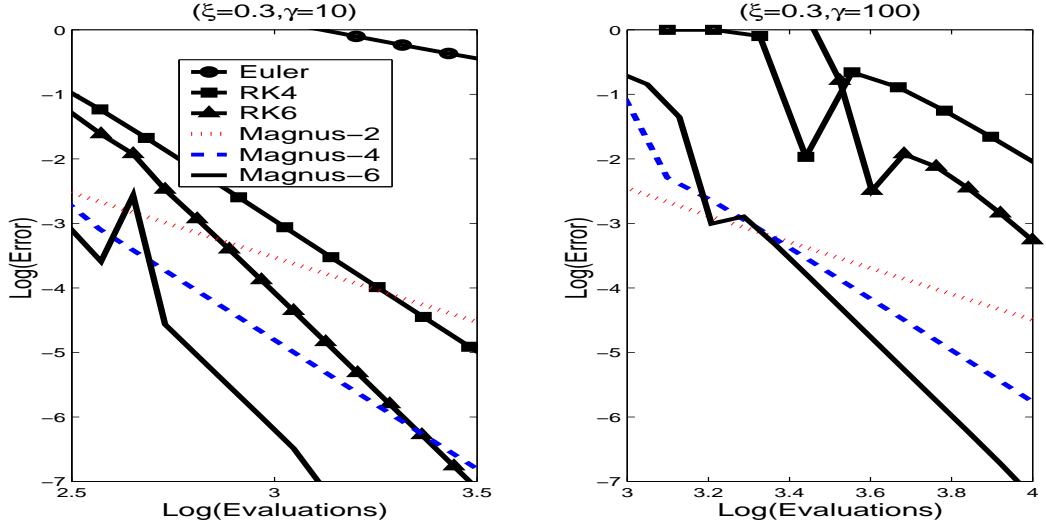}
\end{center}
\caption{{Rosen--Zener model:
 Error in the transition probabilities versus the number of evaluations
 of the Hamiltonian $H_I(s)$ for $\xi=0.3$ and $\gamma=10$ (left panel)
 and $\gamma=100$ (right panel).}}
 \label{plot:Num_Efic}
\end{figure}

If we compare Magnus integrators of different orders of accuracy,
we observe that the most efficient scheme is the second order
method M2 when relatively low accuracies are desired. For higher
accuracy, however, it is necessary to carry out a thorough
analysis of the computational cost of the methods for a given
problem before asserting the convenience of M4 or M6 with respect
to higher order schemes. For a fixed time step $h$, the
computational cost of a certain family of methods (such as those
based on Magnus) usually increases with the order. However, if one
fixes the number of $A(t)$ evaluations, this is not necessarily
the case  (sometimes higher order methods requires more commutators
but less exponentials).


Let us now compare the performance of the Magnus methods with
respect to other Lie-group solvers, namely Fer and Cayley methods.
We repeat the same experiments as in Fig.~\ref{plot:Num_Efic} but,
for clarity, only the results for the 6th-order methods are shown.
We consider the symmetric-Fer method given by (\ref{sym-Fer-n})
and (\ref{eq.4.2.5}) and the Cayley method (\ref{cayl2}) with
(\ref{eq.51b}) using in both cases the Gauss--Legendre quadrature.
The results obtained are collected in
 Figure \ref{plot:Num_Ef_MaCaFe}. We clearly
observe that the relative performance of the Cayley method
deteriorates by increasing the value of $\gamma$ similarly to the
RK6. In spite of preserving the qualitative properties, this
example shows that for some problems, polynomial or rational
approximations do not perform efficiently. Here, in particular, the
Magnus scheme is slightly more efficient that the symmetric Fer
method, although for other problems their performance is quite similar.

\begin{figure}[tb]
\begin{center}
\includegraphics[height=7cm,width=14cm]{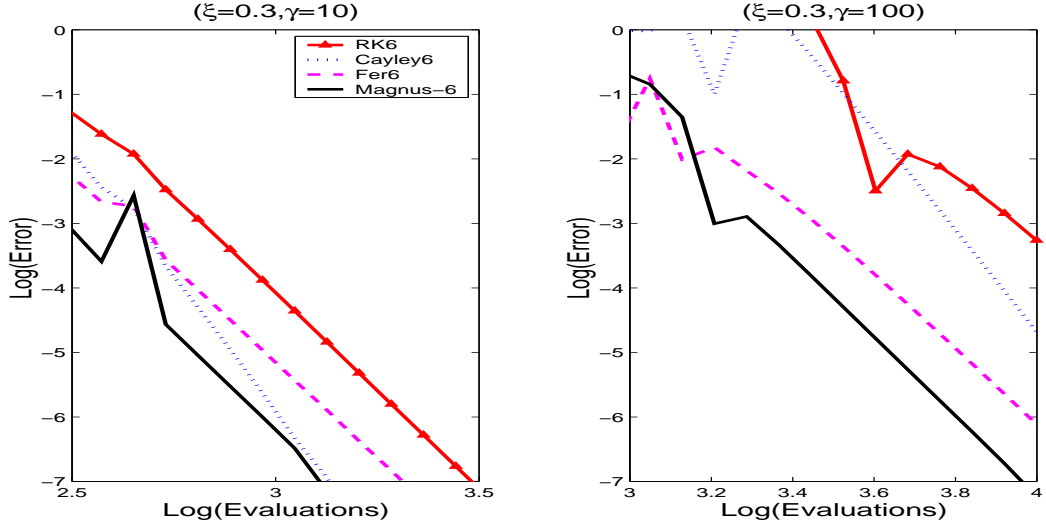}
\end{center}
\caption{{Rosen--Zener model:
 Same as Figure~\ref{plot:Num_Efic} where we compare the performance of the
 6th-order Magnus, symmetric-Fer, Cayley and RK6.}}
 \label{plot:Num_Ef_MaCaFe}
\end{figure}

\subsection{Computing the exponential of a matrix}
\label{exponen}

We have seen that the numerical schemes based on the Magnus expansion
provide excellent results when applied to equation (\ref{NI.1}) with
coefficient matrix (\ref{marz}). In fact, they are even more efficient than
several Runge--Kutta algorithms. Of course, for this particular example
the number of $A(t)$ evaluations is a good indication of the computational
cost required by the numerical schemes, since the evaluation of
$\exp(\Omega)$ can be done analytically by means of formula (\ref{exppaulis}).
In general, however, the matrix exponential has to be also approximated
numerically and thus the performance of the numerical integration algorithms based on
 the Magnus expansion strongly depends on this fact.
It may happen that the evaluation of $\exp(C)$, where $C$ is a (real or complex) $N \times N$ matrix,
 represents the major factor in the overall computational cost
required by this class of algorithms and is probably one of their most
problematic aspects.

As a matter of fact, the approximation of the
matrix exponential is among the oldest and most extensively researched
topics in numerical mathematics. Although many efficient algorithms
have been developed, the problem is still far from been solved in general.
It seems then reasonable to briefly summarize here some of the
most widely used procedures in this context.

Let us begin with two obvious but important remarks. (i) First, one
has to distinguish whether it is required to evaluate the full matrix
$\exp(C)$ or only the product $\exp(C)v$ for some given vector $v$. In the later
case, special algorithms can be designed requiring a much reduced computational
effort. This is specially true when $C$ is large and sparse (as often happens
with matrices arising from the spatial discretization of partial differential
equations).
(ii) Second, for the numerical integration methods based on ME
one has to compute
$\exp(C(h))$, where $C(h)=\mathcal{O}(h^p)$, $h$ is a (not too large) step size
and $p\geq 1$. In other words, the matrices to be exponentiated have
typically a small norm (usually restricted by the convergence bounds of the
expansion).

In any case,
prior to the computation of $\exp(C)$, it is significant to have
as much information about the exponent $C$ as possible.  Thus, for
instance, if the
matrix $C$ resides in a Lie algebra, then $\exp(C)$
belongs to the corresponding Lie group and one has to decide
whether this qualitative property has to be exactly preserved or
constructing a sufficiently accurate approximation (e.g., at a higher order
than the order of the integrator itself) is enough. Also when $C$ can be
split into different parts, one may consider a factorization of the form
$\exp(C) \approx \exp(C_1) \exp(C_2) \cdots \exp(C_m)$ if each individual
exponential is easy to evaluate exactly.

An important reference in this context is
\cite{moler78ndw} and its updated version \cite{moler03ndw}, where up to
nineteen (or twenty in \cite{moler03ndw})
different numerical algorithms for computing the exponential of a matrix are
carefully reviewed. An extensive software package for computing the matrix
exponential is Expokit, developed by R. Sidje, with 
Fortran and Matlab versions available \cite{sidje07exs,sidje98esp}. In addition to computing
the matrix-valued function $\exp(C)$ for small, dense matrices $C$, Expokit has
functions for computing the vector-valued function $\exp(C)v$ for both small, dense
matrices and large, sparse matrices.

\subsubsection{Scaling and squaring with Pad\'e approximation}

\label{sec.4.4}

Among one of the least dubious ways of computing $\exp(C)$ is by
scaling and squaring in combination with a diagonal Pad\'e
approximation \cite{moler03ndw}. The procedure is based on a fundamental
property of the exponential function, namely
\[
    \e^C = (\e^{C/j})^j
\]
for any integer $j$. The idea then is to choose $j$ to be a power of two for which
$\exp(C/j)$ can be reliably and efficiently computed,
and then to form the matrix $(\exp(C/j))^j$ by
repeated squaring. If the integer $j$ is chosen as the smallest power
of two for which $\|C\|/j < 1$, then $\exp(C/j)$ can be satisfactorily
computed by diagonal Pad\'e approximants of order, say, $m$. This is roughly the method
used by the built-in function \texttt{expm} in Matlab.

For the integrators based on the Magnus expansion, as
$C=\mathcal{O}(h^p)$ with $ p\geq 1$, one usually gets good approximations
with relatively small values of $j$ and $m$.

As is well known, diagonal Pad\'{e} approximants map the Lie algebra $%
\mathrm{o}_{J}(n)$ to the $J$-orthogonal Lie group $\mathrm{O}_{J}(n)$ and thus
constitute also a valid alternative to the evaluation of the
exponential matrix in
Magnus-based methods for this particular Lie group. More specifically, if $%
B\in \mathrm{o}_{J}(n)$, then $\psi _{2m}(tB)\in
\mathrm{O}_{J}(n)$ for sufficiently small $t\in \mathbb{R}$, with
\begin{equation}
\psi _{2m}(\lambda )=\frac{P_{m}(\lambda )}{P_{m}(-\lambda )},
\label{eq.24}
\end{equation}
provided the polynomials $P_{m}$ are generated according to the
recurrence
\begin{eqnarray*}
P_{0}(\lambda ) & = & 1, \ \qquad  P_{1}(\lambda )=2+\lambda \\
P_{m}(\lambda ) &=&2(2m-1)P_{m-1}(\lambda )+\lambda
^{2}P_{m-2}(\lambda ).
\end{eqnarray*}
Moreover,
$\psi _{2m}(\lambda )= \exp(\lambda )+\mathcal{O}(\lambda ^{2m+1})$ and $\psi
_{2}$ corresponds to the Cayley transform (\ref{eq.6}), whereas
for $m=2,3$ we have
\begin{eqnarray*}
\psi _{4}(\lambda ) &=&\left( 1+\frac{1}{2}\lambda
+\frac{1}{12}\lambda ^{2}\right) \Big/\left( 1-\frac{1}{2}\lambda
+\frac{1}{12}\lambda ^{2}\right)
\\
\psi _{6}(\lambda ) &=&\left( 1+\frac{1}{2}\lambda +\frac{1}{10}\lambda ^{2}+%
\frac{1}{120}\lambda ^{3}\right) \Big/\left( 1-\frac{1}{2}\lambda +\frac{1}{%
10}\lambda ^{2}-\frac{1}{120}\lambda ^{3}\right) .
\end{eqnarray*}
Thus, we can combine the optimized approximations to $\Omega $
obtained in subsection \ref{Mintegrators} for Magnus based methods
with diagonal Pad\'{e} approximants up to the corresponding order
to obtain time-symmetric integration schemes preserving the
algebraic structure of the problem without computing the matrix
exponential. For instance, the  ``Magnus--Pad\'e" methods thus
obtained involve, in addition to one matrix inversion, 3 and 8
matrix-matrix products for order 4 and 6, respectively.

Observe that since $\Omega^{[2n]} = \mathcal{O}(h)$ then
\[
   \psi_{2m}(\Omega^{[2n]}) = \exp(\Omega^{[2n]}) + \mathcal{O}(h^{2k+1}),
\]
where $k = \min \{m,n \}$. With $m=n$ we have a method of order
$2n$. However, for some problems this rational approximation to
the exponential may be not very accurate depending on the
eigenvalues of $\Omega^{[2n]}$. In this case one may take $m>n$,
thus giving a better approximation to the exponential and a more
accurate result by increasing slightly the computational cost of
the method. Of course, when the norm of the matrix $\Omega^{[2n]}$ is not so
small, this technique can be combined with scaling and squaring
\cite{dragt95com}.

Instead of using Pad\'e approximants for  the exponential
of the scaled matrix $B \equiv C/2^k$, Najfeld and Havel \cite{najfeld95dot} propose
a rational approximation for the matrix function
\begin{equation}  \label{nreex}
    H(B) = B \coth(B) = B \, \frac{\e^{2B}+I}{\e^{2B}-I},
\end{equation}
from which the exponential can be obtained as
\[
   \e^{2B} = \frac{H(B)+B}{H(B)-B}
\]
and then iteratively square the result $k$ times to recover the
exponential of the original matrix $C$. From the continued fraction
expansion of $H(B)$, it is possible to compute the first rational
approximations as
\[
    H_2(B) =  \frac{I + \frac{2}{5} B^2}{I + \frac{1}{15} B^2}, \qquad
    H_4(B) =  \frac{I + \frac{4}{9} B^2 + \frac{1}{63} B^4}
        {I + \frac{1}{9} B^2 + \frac{1}{945} B^4}
\]
and so on. Observe that the representation (\ref{nreex}) can be
regarded as a generalized Cayley transform of $B$ and thus it also provides
approximations in the group $\mathrm{O}_J(n)$. In  \cite{najfeld95dot} the
authors report a saving of about 30\% in the number of matrix multiplications
with respect to diagonal Pad\'e approximants when an optimal $k$ and a
rational approximation for $H(B)$ is used.

\subsubsection{Chebyshev approximation}

Another valid alternative
is to use polynomial approximations to the exponential of $C$ as a
whole. Suppose, in particular, that $C$ is a matrix
of the form $C = -i \tau H$, with $H$ Hermitian and $\tau>0$, as is the case
in Quantum Mechanics. In the Chebyshev approach, the evolution operator
$\exp(-i \tau H)$ is expanded in a truncated series of Chebyshev
polynomials, in analogy with the approximation of a scalar function
\cite{talezer84aaa}. As is
well known, given a function $F(x)$ in the interval $[-1,1]$, the Chebyshev
polynomial approximations are optimal, in the sense that the maximum
error in the approximation is minimal compared to almost all possible
polynomial approximations \cite{suli03ait}. To apply this procedure, one
has to previously bound the extreme eigenvalues $E_{\mathrm{min}}$
and $E_{\mathrm{max}}$ of $H$. Then a truncated Chebyshev
expansion of $\exp(-i x)$ on the interval $[\tau E_{\mathrm{min}}, \tau
E_{\mathrm{max}}]$ is considered:
\[
  \exp(-i x) \approx \sum_{n=0}^m c_n P_n(x),
\]
where
\[
    P_n(x) = T_n \left( \frac{2x - \tau E_{\mathrm{max}} - \tau
    E_{\mathrm{min}}}{\tau E_{\mathrm{max}} - \tau E_{\mathrm{min}}}
    \right)
\]
with appropriately chosen coefficients $c_n$. Here
$T_n(x)$ are the Chebyshev polynomials on the interval
$[-1,1]$ \cite{abramowitz65hom}, which can be determined via
the recurrence relation
\[
   T_{n+1}(x) = 2 x T_n(x) - T_{n-1}(x); \quad T_1(x) = x; \quad T_0(x) = 1.
\]
Finally, one uses the approximation
\begin{equation}   \label{cheb1}
  \exp(-i \tau H) \approx \sum_{n=0}^m c_n P_n(\tau H).
\end{equation}
This technique is frequently used in numerical quantum dynamics
to compute $\exp(-i \tau H) \psi_0$ over very long times. This can be done
with $m$ matrix-vector products if the approximation (\ref{cheb1}) is
considered with a sufficiently large truncation index $m$. In fact, the
degree $m$ necessary for achieving a specific accuracy depends linearly
on the step size $\tau$ and the spectral radius of $H$ \cite{nettesheim00mqc},
and thus an increase of the step size reduces the computational work
per unit step.
In a practical implementation, $m$ can be chosen such that the accuracy
is dominated by the round-off error \cite{leforestier91aco}. This approach
has two main drawbacks: (i) it is not unitary, and therefore the norm is not conserved
(although the deviation from unitarity is really small due to its extreme
accuracy),
and (ii) intermediate results are not obtained, since typically $\tau$ is very large.

\subsubsection{Krylov space methods}

As we have already pointed out, very often what is really required, rather than the
exponential of the matrix $C$ itself, is the computation of
$\exp(C)$ applied to a vector. In this situation,
evaluating $\e^C$ is somehow analogous to computing $C^{-1}$ to get the
solution of the linear system of equations
 $C x = b$ for many different $b$'s: other procedures are clearly far more
desirable. The computation of $\e^C v$ can be efficiently done with Krylov subspace methods,
in which approximations to the solution are obtained from the Krylov spaces spanned
by the vectors $\{v, Cv, C^2 v, \ldots, C^j v\}$ for some $j$ that is typically
small compared to the dimension of $C$ \cite{druskin95ksa,park86uqt}.
The Lanczos method for solving iteratively symmetric
eigenvalue problems is of this form \cite{watkins02fom}. If, as before, we let  $C = -i \tau H$,
the symmetric Lanczos process
generates recursively an orthonormal basis $V_{m} = [ v_{1} \cdots
v_{m}]$ of the $m$th Krylov subspace $K_{m}(H, u) = \mathrm{span}
(u, H u, \ldots, H^{m-1} u)$ such that
\[
   H V_{m} = V_{m} L_{m} + [0 \cdots 0 \, \beta_{m} \, v_{m+1}],
\]
where the symmetric tridiagonal $m \times m$ matrix $L_{m} = V^T_{m}
H V_{m}$ is the orthogonal projection of $H$ onto $K_{m}(H, u)$. Finally,
\[
   \exp(-i \tau H) u \approx V_{m}  \exp(-i \tau L_{m}) V_{m}^T u
\]
and the matrix exponential $\exp(-i \tau L_{m})$ can be computed by diagonalizing
$L_{m}$, $L_{m} = Q_{m} D_{m} Q_{m}^T$, as
\[
    \exp(-i \tau L_{m}) = Q_{m} \exp(-i \tau D_{m}) Q_{m}^T,
\]
with $D_{m}$ a diagonal matrix.
This iterative process is stopped when
\[
  \beta_{m} \left\| \left( \exp(-i \tau L_{m}) \right)_{m,m} \right\|
    < tol
\]
for a fixed tolerance. Very good approximations are often obtained with
relatively small values of $m$, and computable error bounds exist for the
approximation. This class of schemes require generally $\mathcal{O}(N^2)$
floating point operations in the computation of $\e^C v$.
More details are contained in the references
\cite{hochbruck99eif,hochbruck98eif,hochbruck97oks}.

\subsubsection{Splitting methods}
\label{splittingme}

  Frequently, one has to exponentiate a matrix
which can be split into several parts which are either solvable or easy to deal
with. Let us assume for simplicity that $C=A+B$, where the computation $\e^C$ is very
expensive, but
$\e^A$ and $\e^B$ are cheap and easy to evaluate. In such circumstances, it makes
sense to
approximate $\e^{\varepsilon C}$ with $\varepsilon$ a small
parameter, by the following scheme:
\begin{equation}   \label{split-standard}
   \psi^{[p]}_{h}  \equiv
    \e^{\varepsilon b_m B} \e^{\varepsilon a_m A} \ \cdots \
    \e^{\varepsilon b_1 B} \e^{\varepsilon a_1 A} \
   =  \e^{\varepsilon (A+B)} + \mathcal{O}(\varepsilon^{p+1})
\end{equation}
with appropriate parameters $a_i,b_i$. This can be seen as the
approximations to the solution at $t=\varepsilon$ of the equation
$Y^{\prime} = (A+B)Y$ by a composition of the exact solutions of
the equations $Y^{\prime}= AY$ and $Y^{\prime}= BY$ at times
$t=a_i \varepsilon$ and $t=b_i\varepsilon$, respectively. Two
instances of this kind of approximations are given by the well
known Lie--Trotter formula
\begin{equation}  \label{Lie-Trotter}
   \psi^{[1]}_{h}  =  \e^{ \varepsilon A} \e^{\varepsilon B}
\end{equation}
and the second order symmetric composition
\begin{equation}\label{leapfrog}
 \psi^{[2]}_{h} =  \e^{\varepsilon A/2}  \e^{ \varepsilon B}
  \e^{\varepsilon A/2},
\end{equation}
referred to as Strang splitting, St\"ormer, Verlet and leap-frog,
depending on the particular area where it is used.

Splitting methods have been considered in different contexts: in
designing symplectic integrators, for constructing volume-preserving
algorithms, in the numerical integration of partial differential equations, etc.
An extensive survey of the theory and practice of splitting methods can be found in
\cite{blanes02psp,hairer06gni,leimkuhler04shd,mclachlan02sm,mclachlan06gif}
and references therein.

Splitting methods are particularly useful in geometric
numerical integration. Suppose that
the matrix $C = A + B$ resides in a Lie algebra $\mathfrak{g}$. Then, obviously,
$\exp(C)$ belongs to the corresponding Lie group $\mathcal{G}$ and one is
naturally interested in getting approximations also in $\mathcal{G}$. In this respect,
notice that if $A,B \in
\mathfrak{g}$, then the scheme (\ref{split-standard}) also provides an approximation
in $\mathcal{G}$. It is worth noticing that other methods for the
approximation of the matrix exponential, e.g., Pad\'e approximants and Krylov
subspace techniques, are not guaranteed to map elements from $\mathfrak{g}$ to
$\mathcal{G}$. Although diagonal Pad\'e approximants map
the Lie algebra $\mathrm{o}_J$ to the underlying group, it is possible to show
that the only analytic function that maps $\mathfrak{sl}(n)$ into the special
linear group $\mathrm{SL}(n)$ approximating the exponential function up to a given order is the
exponential itself \cite{feng95vpa}. In consequence, diagonal Pad\'e
approximants only provide results in $\mathrm{SL}(n)$ if the computation
is accurate to machine precision.

In \cite{celledoni00ate}, Celledoni and Iserles devised a splitting
technique for obtaining an approximation to $\exp(C)$ in the Lie group $\mathcal{G}$ based
on a decomposition of $C \in \mathfrak{g}$ into low-rank matrices $C_i \in \mathfrak{g}$.
Basically, given a $n \times n$ matrix $C \in \mathfrak{g}$, they proposed to split
it in the form
\[
    C = \sum_{i=1}^k C_i,
\]
such that
\begin{enumerate}
   \item $C_i \in \mathfrak g$, for $i=1,\ldots,k$.
   \item Each $\exp(C_i)$ is easy to evaluate exactly.
   \item Products of such exponentials are computationally cheap.
\end{enumerate}
For instance, for the Lie algebra $\mathfrak{so}(n)$, the choice
\[
   C_i = \frac{1}{2} c_i e_i^T - \frac{1}{2} e_i b_i^T, \qquad i=1,\ldots,n,
\]
where $c_1,\ldots,c_n$ are the columns of $C$ and $e_i$ is the $i$-th vector of the
canonical basis of $\mathbb{R}^n$, satisfies the above requirements (with $k=n$), whereas in the
case of $\mathfrak{g} = \mathfrak{sl}(n)$ other (more involved) alternatives are possible
\cite{celledoni00ate}.

Proceeding as in (\ref{Lie-Trotter}), with
\[
   \psi^{[1]} = \exp(\varepsilon C_1) \exp(\varepsilon C_2) \cdots \exp(\varepsilon C_k)
\]
we get an order one approximation in $\varepsilon$ to $\exp(\varepsilon C)$,
whereas the symmetric composition
\begin{equation}  \label{strang-espe}
  \psi^{[2]} = \e^{\frac{1}{2} \varepsilon C_k} \e^{\frac{1}{2} \varepsilon C_{k-1}}
     \cdots \e^{\frac{1}{2} \varepsilon C_2}
       \e^{ \varepsilon C_1}   \e^{\frac{1}{2} \varepsilon C_2} \cdots
       \e^{\frac{1}{2} \varepsilon C_{k-1}} \e^{\frac{1}{2} \varepsilon C_k}
\end{equation}
provides an approximation of order two in $\varepsilon$, and this can be subsequently
combined with different techniques for increasing the order.

With respect to the computational cost, the results reported in
\cite{celledoni00ate} show that, up to order four in
$\varepsilon$, this class of splitting schemes are competitive
with the Matlab built-in function \texttt{expm} when machine
accuracy is not required in the final approximation. Running
\texttt{expm} on randomly generated matrices, it is possible to
verify that computing $\exp(C)$ to machine accuracy requires about
(20-30)$n^3$ floating point operations, depending on the
eigenvalues of $C$, whereas the 4th-order method constructed from
(\ref{strang-espe}) involves (12-15)$n^3$ operations when $C \in
\mathfrak{so}(100)$ \cite{celledoni00ate}. In the case of the
approximation of $\exp(C) v$ and $v \in \mathbb{R}^n$, the cost of
low-rank splitting methods drop down to $K n^2$, where $K$ is a
constant, and thus they are comparable to those achieved by
polynomial approximants \cite{hochbruck97oks}.

Splitting methods of the above type are by no means the only way to express
$\exp(C)$ as a product of exponentials of elements in $\mathfrak{g}$. For instance,
the Wei--Norman approach (\ref{wn2}) can also be implemented in this setting. Suppose
that $\dim \mathfrak{g} = s$ and let $\{X_1,X_2,\ldots, X_s\}$ be a basis of
$\mathfrak{g}$. In that case, as we have seen (subsection \ref{w-nfacto}),
it is possible to represent $\exp(t C)$
for $C \in \mathfrak{g}$ and sufficiently small $|t|$ in canonical coordinates of the
second kind,
\[
   \e^{t C} = \e^{g_1(t) X_1} \, \e^{g_2(t) X_2} \cdots \e^{g_s(t) X_s},
\]
where the scalar functions $g_k$ are analytic at $t=0$. Although the $g_k$s are
implicitly defined, they can be approximated by Taylor series.
The cost of the procedure
can be greatly reduced by choosing adequately the basis and exploiting the Lie-algebraic
structure \cite{celledoni01mft}.

Yet another procedure to get approximations of $\exp(C)$ in a
Lie-algebraic setting which has received considerable attention
during the last years is based on generalized polar decompositions
(GPD), an approach introduced in \cite{munthe-kaas01gpd} and
further elaborated in \cite{iserles05eco,zanna01gpd}. In
particular, in \cite{iserles05eco}, by bringing together GPD with
techniques from numerical linear algebra, new algorithms are
presented with complexity $\mathcal{O}(n^3)$, both when the
exponential is applied to a vector and to a matrix. This is
certainly not competitive with Krylov subspace methods in the
first case, but represents at least a 50\% improvement on the
execution time, depending on the Lie algebra considered, in the
latter. Another difference with respect to Krylov methods is that
the algorithms based on generalized polar decompositions
approximate the exponential to a given order of accuracy and thus
they are well suited to exponential approximations within
numerical integrators for ODEs, since the error is subsumed in
that of the integration method.
For a complete description of the procedure and its practical
implementation we refer the reader to \cite{iserles05eco}.

\subsection{Additional numerical examples}
\label{numexM}

The purpose of subsection \ref{niarz} was to illustrate the main
features of the numerical schemes based on the Magnus expansion in
comparison with other standard integrators (such as Runge--Kutta
schemes) and other Lie-group methods (e.g., Fer and Cayley) on
a solvable system. For larger systems the efficiency analysis is
more challenging since the (exact or approximate)
computation of exponential matrices
play an important role on the performance of the methods. It
makes sense then to analyze from this point of view more realistic problems where one necessarily
has to approximate the exponential in a consistent way.

As an illustration of this situation
we consider next two skew-symmetric
matrices $A(t)$ and $Y(0) = I$, so that the solution $Y(t)$ of
$Y^{\prime} = A(t) Y$ is
orthogonal for all $t$. In particular, the upper triangular
elements of the matrices $A(t)$ are as follows:
\begin{eqnarray}
\mbox{(a)} \;\;\;  A_{ij} & = & \sin \left( t(j^2 - i^2) \right) \qquad\qquad 1 \leq
i < j \leq
N  \label{ej3.1} \\
\mbox{(b)} \;\;\;  A_{ij} & = & \log \left( 1 + t \, \frac{j-i}{j+i} \right)
\label{ej3.2}
\end{eqnarray}
with $N=10$. In both cases $Y(t)$ oscillates with time, mainly due
to the time-dependence of $A(t)$ (in (\ref{ej3.1})) or the norm of
the eigenvalues (in (\ref{ej3.2})).

The integration is carried out in the interval $t \in [0,10]$ and
the approximate solutions are compared with the exact one at the final time $%
t_{\mathrm{f}}=10 $ (obtained with very high accuracy by using a
sufficiently small time step). The corresponding error at
$t_{\mathrm{f}}$ is computed for different values of the time step
$h$. The Lie-group solvers are implemented with Gauss--Legendre
quadratures and constant step size.

 First, we plot the accuracy of the different 4-th and 6-th order methods
as a function of the number of $A(t)$ evaluations. In contrast to the previous
examples, now there is not a closed formula for the matrix
exponentials appearing in the Magnus based integrators, so that
some alternative procedure must be applied. In particular, the computation
of $\e^C$ to machine accuracy is done by scaling--(diagonal
Pad\'e)--squaring, so that the result is correct up to round-off.
Figure \ref{nexo6.2} shows the results obtained for the problems
(\ref{ej3.1}) and (\ref{ej3.2}) with fourth- and
sixth-order numerical schemes based on Magnus and Cayley, and also
explicit Runge--Kutta methods. In
the first problem, Magnus and Cayley show a very similar performance, which happens
to be only slightly better than that of RK methods.

The situation changes drastically, however, for the second problem. Here the
behaviour of
Cayley and RK methods is essentially similar, whereas schemes
based on Magnus are clearly more efficient. The reason seems to be that
Cayley and RK
methods give poor approximations to the exponential, which, on the other hand,
has to be accurately approximated, since
the eigenvalues of
$A(t)$ may take large values.

With respect to symmetric Fer
methods, their efficiency is quite similar to that of Magnus if the matrix
exponentials are evaluated accurately up to machine precision.
 This is so for the matrix (\ref{ej3.1}) even if Pad\'e
approximants of relatively low order are used to replace the
exponentials.


On the other hand, the efficiency of ``Magnus-Pad\'e" methods (we
denote by MP$nm$ a Magnus method of order $n$ where the
exponential is approximated by a diagonal Pad\'e of order $m$, and
MP$n$ if $n=m$) is highly deteriorated for the problem
(\ref{ej3.2}), although it is always better than the corresponding
to Cayley schemes.

\begin{figure}[tb]
\begin{center}
\includegraphics[height=7cm,width=14cm]{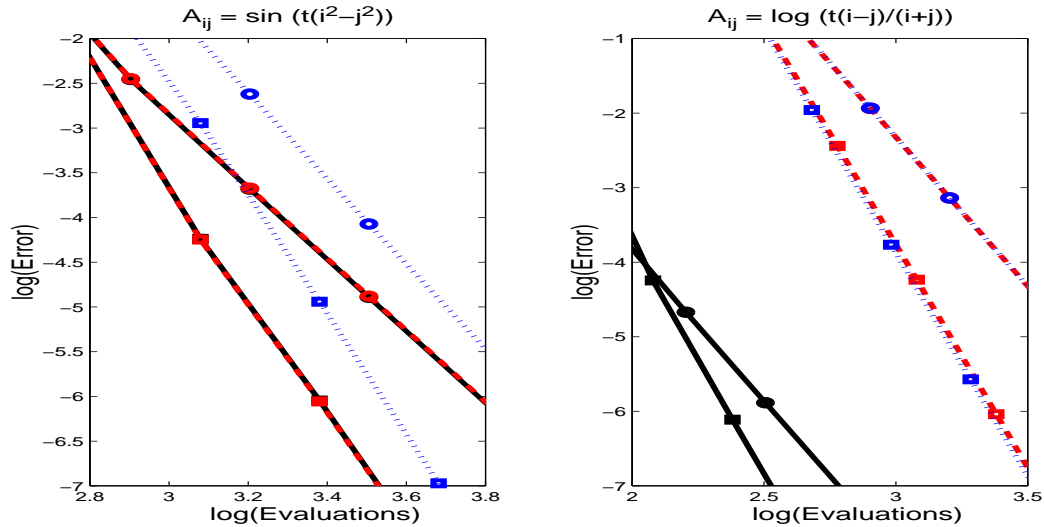}
\end{center}
\caption{{\protect\small Efficiency diagram corresponding to the
optimized 4-th (circles) and 6-th (squares) order Lie-group
solvers based on Magnus (solid lines) and Cayley (broken lines),
and the standard Runge--Kutta methods (dashed lines).}}
\label{nexo6.2}
\end{figure}

To better illustrate all these comments, in Figure \ref{nexo6.4} we
display the error in the solution corresponding to (\ref{ej3.1})
and (\ref{ej3.2}) as a function of time in the interval $t \in
[0,100]$ for $h=1/20$ as is obtained by the previous methods. We should stress that
all schemes require the same number of $A$
evaluations.

In the right
picture the exponentials appearing in the Magnus method are
computed using a Pad\'e approximant of order six (MP6), of
order eight (MP68) and to machine accuracy (M6). Observe the great
importance of evaluating the exponential as accurately as possible
for the matrix (\ref{ej3.2}): by increasing slightly the
computational cost per step in the computation of the matrix
exponential it is possible to improve dramatically the accuracy of
the methods. On the contrary, for problem (\ref{ej3.1}) the
meaningful fact seems to be that the integration scheme provides a
solution in the corresponding Lie group.

\begin{figure}[tb]
\begin{center}
\includegraphics[height=7cm,width=14cm]{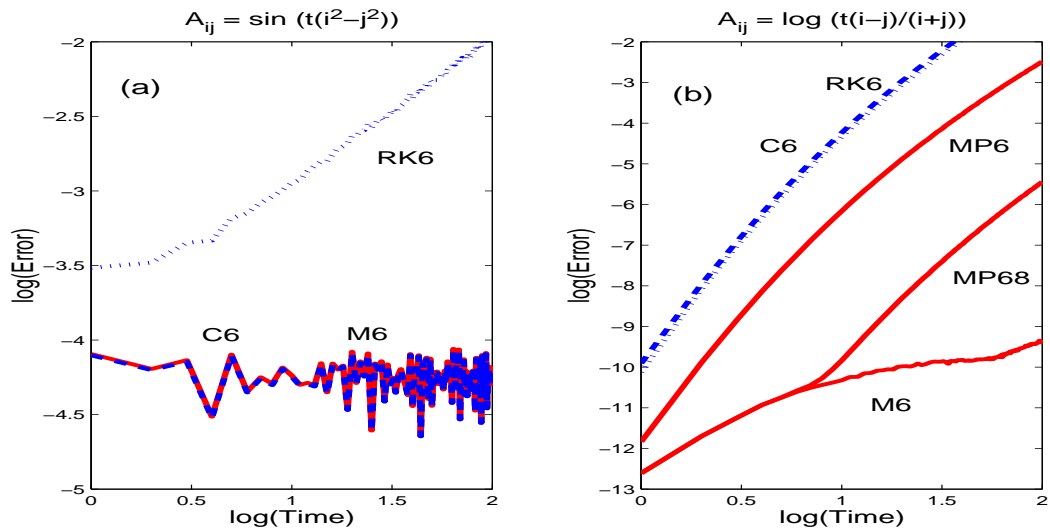}
\end{center}
\caption{{\protect\small Error as a function of time (in
logarithmic scale) obtained with different 6-th order integrators
for $h=1/20$: (a) problem (\ref{ej3.1}); (b) problem
(\ref{ej3.2}).}} \label{nexo6.4}
\end{figure}

\subsection{Modified Magnus integrators}

\subsubsection{Variations on the basic algorithm}

The examples collected in subsections \ref{niarz} and  \ref{numexM} show that the
numerical methods constructed from the Magnus expansion can be
computationally very efficient indeed. It is fair to say, however, that this
efficiency can be seriously affected when dealing with certain
types of problems.

Suppose, for instance, that one has to numerically integrate a
problem defined in $\mathrm{SU}(n)$. Although Magnus integrators are
unconditionally stable in this setting (since they preserve
unitarity up to roundoff independently of the step size $h$), in
practice only small values of $h$ are used for achieving accurate
results. Otherwise the convergence of the series is not assured.
Of course, the use of small time steps may render the algorithm
exceedingly costly.

In other applications the problem depends on several
parameters, so that the integration has to be carried out for
different values of the parameters. In that case the overall
integration procedure can be computationally very expensive.

In view of all these difficulties, it is hardly
surprising that several modifications of the standard algorithm of
Magnus integrators had been developed to try to minimize these undesirable
effects and get especially adapted integrators for particular
problems.

One basic tool used time and again in this context is performing a
preliminary linear transformation, similarly to those introduced
in section \ref{PLT}. These transformations can be carried out
either for the whole integration interval or at each step in the
process. Given an appropriately chosen transformation, $\tilde{Y}_0(t)$,
one factorizes $Y(t)=\tilde{Y}_0(t) \, \tilde{Y}_1(t)$, where the unknown $\tilde{Y}_1(t)$ satisfies the
equation
\begin{equation}\label{Pert}
  \tilde{Y}_1' = B(t) \tilde{Y}_1
\end{equation}
and $B(t)$ depends on $A(t)$ and $\tilde{Y}_0(t)$. This transformation
makes sense, of course, if $\|B(t)\|<\|A(t)\|$ and thus typically
$\tilde{Y}_0$ is chosen in such a way that the norm of $B$ verifies the
above inequality. As a consequence, Magnus integrators can be applied on
(\ref{Pert}) with larger time steps providing also more accurate
results.

  Alternatively, for problems  where
in addition to the time step $h$ there is another parameter ($E$, say),
one may analyze the Magnus expansion in terms of $h$ and $E$. This
allows us to identify which terms at each order in the series
expansion give the main contribution to the error, and design
methods which include these terms in their formulation. The
resulting schemes should provide then more accurate results at a
moderate computational cost without altering the convergence
domain.
As a general rule, it is always desirable to have in
advance as much information about the equation and the properties
of its solution as possible, and then to try to incorporate all this
information into the
algorithm.

Let us review some useful possibilities in this context. From (\ref{3.2.1}) one has
\begin{equation}\label{Taylor_mid2}
  A(t) = a_0 +  \tau a_1 +  \tau^2 a_2 + \cdots,
\end{equation}
where $\tau=t-t_{1/2}$.  The first term is exactly solvable
($a_0=A(t_{1/2})$) and, for many problems, it just provides the
main contribution to the evolution of the system. In that case it
makes sense to take
\[
 \tilde{Y}_0(t) = \e^{(t-t_n) a_0} = \e^{(t-t_n) A(t_{1/2})}
 \]
and subsequently integrate eq. (\ref{Pert}) with
\[
  B(t) = \e^{-(t-t_n) A(t_{1/2})}
  \left( A(t) -  A(t_{1/2}) \right)
    \e^{(t-t_n) A(t_{1/2})}.
\]
This approach has been considered in
\cite{iserles02otg,iserles02tga}, and shows an extraordinary
improvement when the system is highly oscillatory and the main
oscillatory part is integrated with $\tilde{Y}_0$. In those cases, the
norm of $B(t)$ is considerably smaller than $\|A(t)\|$, but $B(t)$ is still highly
oscillatory, so that especially adapted quadrature rules have to
be used in conjunction with the Magnus expansion
\cite{iserles04otn,iserles04oqm}.

  In some other problems the contributions from the derivatives
can also be significant, so that a more appropriate transformation
is defined by
\begin{equation}\label{magnus_2}
  \tilde{Y}_0(t) = \exp \left( \int_{t_n}^{t} A(\tau) d\tau \right).
\end{equation}
The resulting methods can be considered then as a combination of
the Fer or Wilcox expansions and the Magnus expansion. This
approach has been pursued in \cite{degani06rrc}.

On the other hand, it is known that several physically relevant
systems evolve adiabatically or almost-adiabatically. In that case
it seems appropriate to consider the adiabatic picture which
instantaneously diagonalizes $A(t)$ (subsection \ref{PLT}). This
analysis is carried out in
\cite{jahnke04lts,jahnke03nif,lorenz05aif}. In \cite{jahnke03nif}
the adiabatic picture is used   perturbatively, whereas in
\cite{jahnke04lts} it is shown that Magnus in the new picture
leads to  significant improvements.


  Alternatively, one can analyze the structure of the leading error
terms in order to identify the main contribution to the error at
each $\Omega_i$ in the Magnus series expansion. In most cases they
correspond to terms involving only $\alpha_1$ and its nested
commutators with $\alpha_2$
Thus, in particular, the standard fourth-order method given by (\ref{eq.2.3a}) can be
easily improved by including the dominant error term at higher
orders, i.e.,
\[
   \Omega^{[4]} = \alpha_1 - \frac{1}{12} [\alpha_1,\alpha_2]
    + \frac{1}{720}[\alpha_1,\alpha_1,\alpha_1,\alpha_2]
    - \frac{1}{30240}[\alpha_1,\alpha_1,\alpha_1,\alpha_1,\alpha_1,\alpha_2]+\ldots.
\]
We recall that using the fourth-order Gauss--Legendre quadrature
rule we can take $\alpha_1=\frac{h}{2}(A_1+A_2), \
\alpha_2=\frac{h\sqrt{3}}{12}(A_2-A_1)$ with $A_1,A_2$ given in
(\ref{AiG4}). The new method requires additional commutators but
the accuracy can be improved a good deal. This procedure is
analyzed in \cite{moan98eao}, where it is shown how
to sum up all terms of the form
$[\alpha_1,\alpha_1,\ldots,\alpha_1,\alpha_2]$. An error analysis
in the limit of very large values of $\|\alpha_1\|$ is done in
\cite{aparicio05neo,malham06ete}.

\subsubsection{Commutator-free Magnus integrators}
\label{cfmi}

All numerical methods based on the Magnus expansion appearing in
the preceding sections require the evaluation of a matrix
exponential which contain several nested commutators. As we have
repeatedly pointed out, computing the exponential is frequently
the most consuming part of the algorithm. There are problems where
the matrix $A(t)$ has a sufficiently simple structure which allows
one to approximate efficiently the exponential $\exp(A(t_i))$, or
the exponential of a linear combination of the matrix $A(t)$
evaluated at different points. In some sense, this is equivalent
to have efficient methods to compute or to approximate $\tilde{Y}_0$ in
(\ref{magnus_2}). It may happen, however, that the computation of
the matrix exponential is a much more involved task due to the
presence of commutators in the Magnus expansion. For this reason,
it makes sense to look for approximations to the Magnus expansion
which do not involve commutators whilst still preserving the same
qualitative properties. In other words, one may be interested in
compositions of the form
\begin{equation}\label{CF_expsInt}
 \Psi^{[n]}_m  \equiv  
 \exp\left(\int_{t_n}^{t_n+h} p_m(s)A(s) ds \right) \ \cdots \
 \exp\left(\int_{t_n}^{t_n+h} p_1(s)A(s) ds \right)
\end{equation}
where $p_i(s)$ are scalar functions chosen in such a way that
$\Psi^{[n]}_m=\e^{\Omega(t_n+h)}+ \mathcal{O}(h^{n+1})$. Alternatively,
instead of
the functions $p_i(s)$, it is possible to find
coefficients $\varrho_{i,j}, \ i=1,\ldots,m, j=1,\ldots,k$ such
that
\begin{equation}\label{CF_exps}
 \Psi^{[n]}_m  \equiv  
 \e^{\tilde{A}_m} \cdots
 \e^{\tilde{A}_1}, \qquad
 \mbox{with} \qquad
 \tilde{A}_i=h\sum_{j=1}^k \varrho_{i,j} A_j
\end{equation}
is an approximation of the same order. This procedure requires
first to compute  $A_j=A(t_n+c_jh), \ j=1,\ldots,k$ for some
quadrature nodes, $c_j$, of order $n$ or higher and, obviously,
the coefficients $\varrho_{i,j}$ will depend on this choice.
%
The process simplifies if one works in the associated graded free
Lie algebra generated by  $\{\alpha_i\}$, as in the sequel of eq.
(\ref{3.2.2}). Thus, achieving fourth-order integrators reduces
just to solve the equations for the new coefficients $a_{i,1}, a_{i,2}$
in
\begin{eqnarray}
 \Psi^{[4]}_m & \equiv  & 
 \exp\left( a_{m,1} \, \alpha_1+a_{m,2} \, \alpha_2 \right) \cdots
 \exp\left( a_{1,1} \, \alpha_1+a_{1,1} \, \alpha_2 \right) \label{CF4b}
\end{eqnarray}
with the requirement that $\Psi^{[4]}_m= \exp\left(  \Omega^{[4]}
\right)+O(h^5)$, where $ \Omega^{[4]}$ is given by
(\ref{eq.2.3a}). Here the dependence of  $a_{i,1}, a_{i,2}$ on the
coefficients $\varrho_{i,j}$ is determined through the existing
relation between the $\alpha_i$ and the $A_j$ given in
(\ref{CBaseMu}). The order conditions for the coefficients $a_{i,1}, a_{i,2}$ 
can be easily obtained from the Baker--Campbell--Hausdorff
formula. As we have already mentioned, time-symmetry is an
important property to be preserved by the integrators, whereas, at
the same time, also simplifies the analysis. Scheme (\ref{CF4b})
is time-symmetric if
\begin{equation}\label{symm-coef}
  a_{m+1-i,1}=a_{i,1}, \qquad
  a_{m+1-i,2}=-a_{i,2},, \qquad i=1,2,\ldots,m
\end{equation}
in which case the order conditions at even order terms are
automatically satisfied.
 As an illustration, the simple compositions
\begin{eqnarray}
  \Psi^{[4]}_2 & \equiv &
  \exp\left(\frac12 \alpha_1+\frac16 \alpha_2\right) \
 \exp\left(\frac12 \alpha_1-\frac16 \alpha_2\right)
 \label{2-CF} \\
 \Psi^{[4]}_3 & \equiv & \exp\left(\frac{1}{12} \alpha_2\right) \
\exp\left(\alpha_1\right) \ \exp\left(-\frac{1}{12}
\alpha_2\right)
   \label{3-CF}
\end{eqnarray}
are in fact fourth-order (commutator-free) methods requiring two
and three exponentials, respectively \cite{blanes00smf}. In
particular, scheme (\ref{2-CF}), when $\alpha_1, \alpha_2$ are
approximated using the fourth-order Gauss--Legendre quadrature as
shown in (\ref{AiG4}) and (\ref{alfasG4}) leads to the scheme
\begin{equation}
  \Psi^{[4]}_2  \equiv
  \exp\big(h(\varrho_{2,1} A_1+\varrho_{2,2} A_2)\big) \
  \exp\big(h(\varrho_{1,1} A_1+\varrho_{1,2} A_2)\big)   \label{2-CF-GL}
\end{equation}
with $\varrho_{1,1}=\varrho_{2,2}=\frac12+\frac{\sqrt{3}}{72},
\varrho_{1,2}=\varrho_{2,1}=\frac12-\frac{\sqrt{3}}{72}$. Methods
closely related to the scheme (\ref{3-CF}) are presented in
\cite{baye03fof,blanes00smf,lu00foc}, where they are applied to
the Schr\"odinger equation with a time-dependent potential. A
method quite similar to (\ref{2-CF}) is analyzed in
\cite{thalhammer06afo} through its application to parabolic
initial boundary value problems. A detailed study of fourth and
sixth order commutator-free methods is presented in
\cite{blanes06fas}.

On the other hand, very often the differential equation
(\ref{NI.1}) can be split into two parts, so that one has instead
\begin{equation}\label{sep-non-autonomo}
  Y' = \big( A(t) + B(t) \big) Y,
\end{equation}
where each part can be trivially or very efficiently solved. For
instance, the Schr\"odinger equation with a time-dependent
potential and, possibly, a time-dependent kinetic energy belongs
to this class. In principle, the following families of geometric
integrators are specially tailored for this problem:
\begin{description}
  \item[I-] The commutator-free Magnus integrators
  (\ref{CF_exps}), which in this case read
\begin{equation}\label{CF_expsAB}
 \Psi^{[n]}_m  \equiv  
 \e^{\tilde{A}_m+\tilde{B}_m} \cdots
 \e^{\tilde{A}_1+\tilde{B}_1}, \qquad
 \mbox{with} \qquad
 \tilde{A}_i=h\sum_{j=1}^k \varrho_{i,j} A_j, \ \
 \tilde{B}_i=h\sum_{j=1}^k \varrho_{i,j} B_j.
\end{equation}
 Assuming that $\e^{\tilde{A}_i}$ and  $\e^{\tilde{B}_i}$ are
easily computed, then each exponential can be approximated by a
conveniently chosen splitting method (\ref{split-standard})
\cite{mclachlan02sm}
\begin{equation}\label{}
\e^{\tilde{A}_i+\tilde{B}_i}\simeq
 \e^{b_s\tilde{B}_i}\e^{a_s\tilde{A}_i} \cdots
 \e^{b_1\tilde{B}_i}\e^{a_1\tilde{A}_i}.
\end{equation}
  \item[II-] If one takes the time variable in $A(t),B(t)$ as two new
  coordinates, one may use any splitting method as follows  \cite{sanzserna96cni}:
\begin{equation}\label{Split_AB(t)}
 \Psi^{[n]}_{l,h}  \equiv  
 \e^{b_lhB(w_l)}\e^{a_l hA(v_l)} \cdots
 \e^{b_1hB(w_1)}\e^{a_1 hA(v_1)},
\end{equation}
with
\[
 v_i=\sum_{j=1}^{i-1} b_j, \quad
 w_i=\sum_{j=1}^{i} a_j.
\]
and $b_0=0, \ A(v_i)\equiv A(t_n+v_ih), \ B(w_i)\equiv
B(t_n+w_ih)$.
\end{description}

Both approaches have pros and cons. By applying procedure I we may
get methods of order $2n$ with only $n$ evaluations of $A(t)$,
$B(t)$ using e.g. Gauss--Legendre quadratures, but if $m$ in
(\ref{CF_expsAB}) is large, the number of matrix exponentials to
be computed leads to exceedingly costly methods.  The approach II,
on the other hand, has the advantage of a smaller number of
stages, but also presents two drawbacks: (i) many evaluations of
$A(t),B(t)$ are required in general; (ii) for matrices $A$ and $B$
with a particular structure there are specially designed splitting
methods which are far more efficient, but these schemes are not
easily adapted to this situation.


Next we show how to combine splitting methods with
techniques leading to commutator free Magnus schemes to design efficient
numerical algorithms possessing the advantages of approaches I and II, and
at the same time generalizing the splitting idea
(\ref{split-standard}) to this setting
\cite{blanes06smf,blanes07smf}.

The starting point is similar as in previous schemes, i.e. we consider a
composition of the form
\begin{equation} \label{split-Magnus-n}
  \psi_{l,h}^{[n]} = 
  \e^{\tilde{B}_l} \  \e^{\tilde{A}_l} \cdots
  \e^{\tilde{B}_1} \ \e^{\tilde{A}_1} ,
\end{equation}
where the matrices $\tilde{A}_i$ and $\tilde{B}_i$ are
\begin{equation}\label{split-coefs}
  \tilde{A}_i = h \sum_{j=1}^k \rho_{ij}  A_j,
  \qquad
  \tilde{B}_i = h \sum_{j=1}^k \sigma_{ij} B_j,
\end{equation}
with appropriately chosen real parameters $\rho_{ij},\sigma_{ij}$
depending on the coefficients of the chosen quadrature rule.
Notice that $\e^{\tilde{A}_i}$ can be seen as the solution of the
initial value problem $Y^{\prime} = \hat{A}_i Y$, $Y(t_n)=I$ at
$t_{n+1}$, where $\tilde{A}_i = h \hat{A}_i$. Of course, the same
considerations apply to $\e^{\tilde{B}_i}$.

In many cases it is convenient to write the coefficients
$\rho_{ij},\sigma_{ij}$ explicitly in terms of the coefficients
$c_i$. Following \cite{blanes06smf} they can be written as
\begin{equation}\label{eq.2.11b}
    \rho_{ij}  =  \sum_{l=1}^s a_{i,l} \left( R^{(s)} Q_X^{(s,k)}\right)_{lj}, \qquad
    \sigma_{ij} = \sum_{l=1}^s b_{i,l} \left( R^{(s)} Q_X^{(s,k)}\right)_{ij}.
\end{equation}
where the coefficients for the matrices $R^{(s)}, \ s=2,3$ are
given in (\ref{CambioBase}) and for $Q_X^{(s,k)}$ (whose elements
depend on the coefficients $b_i,c_i$ for the quadrature rule) as
shown in (\ref{ru1}).


In this way, the coefficients $a_{ij}$ and $b_{ij}$ are
independent of the quadrature choice and can be obtained by
solving some order conditions (see  \cite{blanes06smf} for more
details).

This procedure allows us to analyse separately particular cases
for the matrices $A,B$ in order to build efficient methods. For
instance, in \cite{blanes06smf} the following particular cases are
considered: (i) when the matrices $A(t),B(t)$ have a general
structure; (ii) when they satisfy the additional constraint
$[B(t_i),[B(t_j),[B(t_k),A(t_l)]]]=0$ as it happens, for instance,
if $A$ corresponds to the kinetic energy and $B$ to the potential
energy (both in classical or quantum mechanics).

As an illustration, we consider the following 4th-order 6-stage
BAB composition
\begin{equation} \label{split-Magnus-4}
  \psi_{6,h}^{[4]} = 
  \e^{\tilde{B}_7} \  \e^{\tilde{A}_6} \  \e^{\tilde{B}_6} \cdots
  \e^{\tilde{A}_1} \  \e^{\tilde{B}_1}.
\end{equation}
In Table \ref{table2} we collect the coefficients $a_{ij},b_{ij}$
to be used in (\ref{eq.2.11b}) to obtain the coefficients
$\rho_{ij},\sigma_{ij}$ to be used in the scheme
(\ref{split-Magnus-4}) for two methods, denoted by
$\mathrm{GS}_6$-4 in the general case (whose coefficients
$a_{i1},b_{i1}$ correspond to $S_6$ in \cite{blanes02psp}) and
$\mathrm{MN}_6$-4 when $[B(t_i),[B(t_j),[B(t_k),A(t_l)]]]=0$ (the
coefficients $a_{i1},b_{i1}$ correspond to $\mathrm{SRKN}_6^b$ in
\cite{blanes02psp}).

Finally, one has to write the scheme in terms of the matrices
$A_i,B_i$. For instance, the composition (\ref{split-Magnus-4})
with the 4th-order Gauss--Legendre quadrature (i.e. taking
$Q^{(2,2)}$ in (\ref{Gauss-Legendre}) and $R^{(2)}$ in
(\ref{CambioBase}) to obtain the coefficients
$\rho_{ij},\sigma_{ij}$ in (\ref{eq.2.11b})) gives
\begin{eqnarray}
  \tilde{A}_i & = & \left( \frac{1}{2} a_{i1} - \sqrt{3} a_{i2}
     \right) h A_1 + \left( \frac{1}{2} a_{i1} + \sqrt{3} a_{i2}
     \right) h A_2 \nonumber \\
  \tilde{B}_i & = & \left( \frac{1}{2} b_{i1} - \sqrt{3} b_{i2}
     \right) h B_1 + \left( \frac{1}{2} b_{i1} + \sqrt{3} b_{i2}
     \right) h B_2.  \label{eq.2.16}
\end{eqnarray}

\begin{table}[tb]
\caption{Splitting methods of order 4 for separable non-autonomous
systems. $\mathrm{GS}_6$-4 is intended for general separable
problems, whereas $\mathrm{MN}_6$-4 can be applied when
$[B(t_i),[B(t_j),[B(t_k),A(t_l)]]]=0$. All the coefficients are given in terms of
$b_{11},a_{11},b_{21},a_{21},b_{31}$ for each method. }
\label{table2} { \footnotesize
\begin{tabular}{l}
\begin{tabular}{l|l}
 $\mathrm{GS}_6$-4  &  $\mathrm{MN}_6$-4   \\
\hline
   $b_{11}= 0.0792036964311957 $ &  $ b_{11}= 0.0829844064174052 $ \\
   $a_{11}= 0.209515106613362$   &  $ a_{11}=  0.245298957184271$  \\
   $b_{21}= 0.353172906049774  $ &  $ b_{21}= 0.396309801498368  $  \\
   $a_{21}=-0.143851773179818$   &  $ a_{21}=  0.604872665711080$   \\
   $b_{31}=-0.0420650803577195 $ &  $ b_{31}=-0.0390563049223486 $
\end{tabular}
\\
\begin{tabular}{lll}
   \hline
   $a_{31}= 1/2-(a_{11}+a_{21})$ &  $ b_{41}=1-2(b_{11}+b_{21}+b_{31})$ \\
   $a_{41}=a_{31}$   &  $b_{51}=b_{31}$ \\
   $a_{51}=a_{21}$   &  $b_{61}=b_{21}$ \\
   $a_{61}=a_{11}$   &  $b_{71}=b_{11}$
\end{tabular}
\\
\begin{tabular}{lll}
\hline
  $a_{12} = (2 a_{11}+2a_{21}+a_{31}-2b_{11}-2b_{21})/c$  &
     $b_{12} = (2 a_{11} + 2 a_{21} - 2 b_{11} - b_{21})/d$  \\
 $a_{22} = 0$  & $b_{22}=(-2a_{11}+b_{11})/d$ \\
  $a_{32} = -a_{11}/c$  & $b_{32} = b_{42} = 0$   \\
   $a_{42}=-a_{32}$   &  $b_{52}=-b_{32}$ \\
   $a_{52}=-a_{22}$   &  $b_{62}=-b_{22}$ \\
   $a_{62}=-a_{12}$   &  $b_{72}=-b_{12}$ \\
\end{tabular}
\\
\begin{tabular}{ll}
\hline
 $c = 12(a_{11}+2a_{21}+a_{31}-2b_{11}+2 a_{11} b_{11} - 2 b_{21} +
2 a_{11} b_{21})$   & \\
 $d = 12(2a_{21}-b_{11}+2 a_{11} b_{11} - 2 a_{21} b_{11} - b_{21}
+ 2 a_{11} b_{21})$  & \\
 \hline
\end{tabular}
\end{tabular}
}
\end{table}


\subsection{Magnus integrators for nonlinear differential
 equations}
\label{mifnde}

 The success of Magnus methods applied to the numerical integration of
linear systems has motivated several attempts to adapt the schemes
for solving time dependent nonlinear differential equations. For
completeness we present some recently proposed generalizations of
Magnus integrators. We consider two different problems: (i) a
nonlinear matrix equation defined in a Lie group, and (ii) a general
nonlinear equation to which the techniques of section \ref{GNLM}
can be applied.

\subsubsection{Nonlinear matrix equations in Lie groups}

As we have already mentioned, the strategy adopted by most
Lie-group methods for solving the nonlinear matrix differential equation
(\ref{nlm1}),
\[
   Y^{\prime} = A(t, Y) Y,  \quad\qquad Y(0) = Y_0 \in \mathcal{G}
\]
defined in a Lie group $\mathcal{G}$, whilst preserving its Lie
group structure, is to lift $Y(t)$ from $\mathcal{G}$ to the
underlying Lie algebra $\mathfrak{g}$ (usually with the
exponential map), then formulate and numerically solve there an
associated differential equation and finally map the solution back
to $\mathcal{G}$. In this way the discretization procedure works
in a linear space rather than in the Lie group. In particular, the
idea of the so-called Runge--Kutta--Munthe-Kaas class of schemes
is to approximate the solution of the
associated differential equation in the Lie algebra $\mathfrak{g}$
by means of a classical Runge--Kutta method
\cite{iserles00lgm,munthe-kaas98rkm,munthe-kaas99hor}.

To generalize Magnus integrators when $A=A(t,Y)$, an
important difference with respect to the linear case is that now
multivariate integrals depend also on the value of the (unknown)
variable $Y$ at quadrature points. This leads to implicit methods
and nonlinear algebraic equations in every step of the integration
\cite{zanna99car}, which in general cannot compete in efficiency
with other classes of geometric integrators such as splitting and
composition methods.

An obvious alternative is just to replace the integrals appearing
in the nonlinear Magnus expansion developed in section
\ref{NL-Magnus} by affordable quadratures, depending on the
particular problem. If, for instance, we use Euler's method to
approximate the first term in (\ref{nlm4}), $\Omega^{[1]}(h) = h
A(0,Y_0) + \cO(h^2)$ and $\Omega^{[2]}$ is discretized with the
midpoint rule, we get the second order scheme
\begin{eqnarray}  \label{nlmnum1}
  v_2 & \equiv & h
     A \left( \frac{h}{2}, \e^{\frac{h}{2} A(0,Y_0)} Y_0
           \right) = \Omega^{[2]}(h) +
     \cO(h^3) \nonumber \\
     Y_1 & = & \e^{v_2} Y_0.
\end{eqnarray}
The same procedure can be carried out at higher orders,
discretizing consistently the integrals appearing in
$\Omega^{[m]}(h)$ for $m>2$ \cite{casas06eme}.


\subsubsection{The general nonlinear problem}

In principle,  it is possible to adapt all methods built for
linear problems to the general nonlinear non-autonomous equation
(\ref{non-lin})
\[
 {\bf x}' = {\bf f}(t,{\bf x}),
\]
or equivalently, the operator differential equation
(\ref{nlp8}),
\[
 \frac{d}{dt}\Phi^t = \Phi^t L_{{\bf f}(t,{\bf y})},
 \qquad \quad  {\bf y}={\bf x}_0.
\]

 As we pointed out in section \ref{GNLM}, there are two problematic aspects
when
designing practical numerical schemes based on Magnus expansion in the nonlinear
case. The first one is how to
compute or approximate the truncated Magnus expansion (or its action on
the initial conditions). The second one is how to evaluate the required Lie
transforms. For example, to compute the Lie transform
$\exp(tL_{\bf f(y)})$ acting on ${\bf y}$ is equivalent to solve
the autonomous differential equation ${\bf x}' = {\bf f}({\bf x})
$ at $t=h$ with  ${\bf x}(0)={\bf y}$, or ${\bf x}(t)=\exp(tL_{\bf
f(x_0)}){\bf x}_0$ where ${\bf x}_0={\bf y}$ can be considered as
a set of coordinates.

Very often, the presence of Lie brackets in the exponent leads to fundamental difficulties, since
the resulting vector fields usually have very
complicate structures. Sometimes, however, this problem can be circumvented by
using the same techniques leading to
commutator-free Magnus integrators in the linear case. In any case, one should
bear in mind that the action of the exponentials in the methods designed for the linear
case has to be replaced by their corresponding maps. Alternatively, if the method is
formulated in terms of Lie transforms, the order of the exponentials has to be reversed,
according to equation (\ref{nlp7}).

Next we illustrate how to numerically
solve the problem
\begin{equation}\label{t-separable}
  {\bf x}' = {\bf f}_1(t,{\bf x}) + {\bf f}_2(t,{\bf x})
\end{equation}
using the scheme (\ref{split-Magnus-n}) with (\ref{eq.2.16}) and
the coefficients $a_{ij},b_{ij}$ taken from MN$_{6}4$ in
Table~\ref{table2}.

Let us consider the Duffing equation
\begin{equation} \label{Duffing}
  q'' + \epsilon q' + q^3 - q = \delta \cos(\omega t)
\end{equation}
which can be obtained from the  time-dependent
Hamiltonian
\begin{equation}\label{HamDuff1}
 H(q,p,t) = T(p,t) + V(q,t) = \e^{-\epsilon t} \, \frac12 p^2 +
   e^{\epsilon t}\left( \frac14 q^4 - \frac12 q^2
     - \delta \cos(\omega t) q \right)
\end{equation}
or equivalently from 
\begin{equation}\label{HamDuff2}
  \frac{d}{dt} \left\{ \begin{array}{l}
  q \\ p
   \end{array} \right\} = \left\{ \begin{array}{c}
    T'(t,p) \\ - V'(t,q)
   \end{array} \right\}  = \left\{ \begin{array}{c}
    e^{-\epsilon t} p \\ 0
   \end{array} \right\}  +  \left\{ \begin{array}{c}
         0 \\
        e^{\epsilon t}\left( q - q^3 + \delta \cos(\omega
        t)\right)
   \end{array} \right\}.
\end{equation}
Notice that this system has already the form (\ref{t-separable}), each part being
exactly solvable. In consequence, 
the splitting method shown in
(\ref{Split_AB(t)}) can be used here. The procedure is described as 
Algorithm 1 in Table~\ref{algorithms}.

\begin{table}[h!]
\caption{Algorithms for the numerical integration of
(\ref{HamDuff1}) or (\ref{HamDuff2}): (Algorithm 1) with scheme
(\ref{Split_AB(t)}), and (Algorithm 2) with scheme
(\ref{split-Magnus-4}).} \label{algorithms} {
 \small
\begin{tabular}{c}
\begin{tabular}{l|l}
\hline
  {\bf Algorithm 1:\ Standard split} &
  {\bf Algorithm 1:\ Magnus split}
 \\
\hline
 $ %
\begin{array}{l}
  q_{0}  = q(t_n); \quad 
  p_{0}= p(t_n);   \\
  t_a=t_n; \quad 
  t_b=t_n \\
     {\bf do} \ \ i=1,m \\
     \quad p_{i} = p_{i-1} - h a_i V'(t_a,q_{i-1})  \\
     \quad t_{a} = t_{a} + h a_i \\
     \quad q_{i} = q_{i-1} + h b_i T'(t_b,p_{i}) \\
     \quad t_{b} = t_{b} + h b_i \\
    {\bf enddo}  \\
\end{array}  $
              &
 $ %
\begin{array}{l}
  q_{0} = q(t_n); \quad 
  p_{0} = p(t_n);   \\
   {\bf do} \ \ i=1,k \\
     \quad T'_{i}(p) = T'(t_n+c_ih,p);  
     \quad V'_{i}(q) = V'(t_n+c_ih,q) \\
   {\bf enddo}  \\
   {\bf do} \ \ i=1,m \\
     \quad \tilde V_i(q) = \sigma_{i1}V'_1(q)+\cdots+ \sigma_{ik}V'_k(q) \\
     \quad \tilde T_i(p) = \rho_{i1}T'_1(p)+\cdots+ \rho_{ik}T'_k(p) \\
     \quad p_{i} = p_{i-1} - h \tilde V_i(q_{i-1})  \\
     \quad q_{i} = q_{i-1} + h \tilde T_i(p_{i}) \\
   {\bf enddo}  \\
\end{array}  $
 \\
 \hline
\end{tabular}
\end{tabular}
}
\end{table}

Observe that the leap-frog composition (\ref{leapfrog})
corresponds to $m=2$ and
\begin{equation}\label{leapfrog-coefs}
 a_1=a_2=\frac12, \quad b_1=1,b_2=0.
\end{equation}
Since $b_2=0$ one stage can be saved (with a trivial modification
of the algorithm) and the scheme is considered as a one stage
method. An efficient symmetric 5-stage fourth order integrator is given by
the coefficients ($m=6$)
\begin{equation}\label{Suzuki-coefs}
 a_i=\frac{\gamma_i+\gamma_{i-1}}{2}, \quad b_i=\gamma_i.
\end{equation}
$i=0,1,\ldots,6$ with $\gamma_0=\gamma_6=0$ and
$\gamma_1=\gamma_2=\gamma_4=\gamma_5=1/(4-4^{1/3}), \
\gamma_3=1-4\gamma_1$.

Alternatively, we can use the Magnus integrator
(\ref{split-Magnus-4}). Since the kinetic energy is quadratic in
momenta, we can apply the fourth-order method MN$_{6}$-4. If we take
the fourth-order Gauss-Legendre quadrature rule for the evaluation
of the time-dependent function then we can consider
(\ref{eq.2.16}), where the coefficients $a_{ij},b_{ij}$ are given
in Table~\ref{table2}. Here, $A(t)$ plays the role of $T'(t,p)$
and $B(t)$ the role of $V'(t,q)$ (they are not interchangeable,
otherwise the performance seriously deteriorates). The computation
of one time step is shown as Algorithm 2 in
Table~\ref{algorithms}.

We take $\epsilon=1/20$, $\delta=1/4$, $\omega=1$ and initial
conditions $q(0)=1.75$, $p(0)=0$. We integrate up to $t=10\, \pi$
and measure the average error in phase space in terms of the
number of force evaluations for different time steps (in
logarithmic scale). The results are shown in Figure \ref{fig0}.
The scheme MN$_{6}$4 has 6 stages per step, but only two
time-evaluations. For this reason in the figure we have considered
as the number of evaluations per step both two and six (left and
right curves connected by an arrow). It is evident the superiority
of the new splitting Magnus integrators for this problem. If the
time-dependent functions dominate the cost of the algorithm the
superiority is even higher. Surprisingly, the method shows
better stability than the leap-frog method, which attains the
highest stability possible among the splitting methods for
autonomous problems.

\begin{figure}[tb]
\begin{center}
\includegraphics[width=12cm]{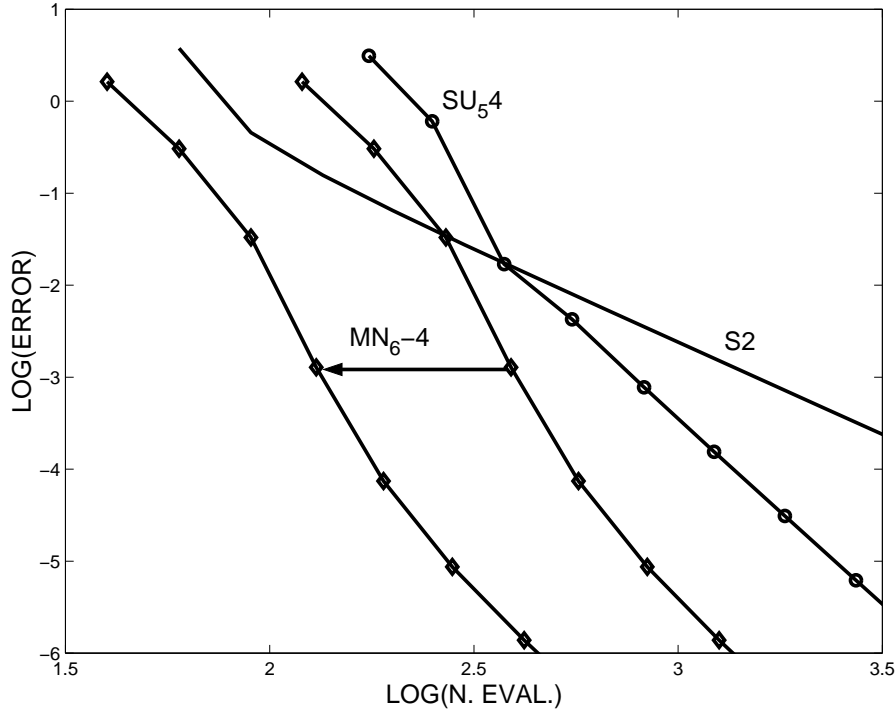}
\end{center}
\caption{{Average error versus number of force evaluations in the
numerical integration of (\ref{HamDuff2}) using second and fourth
order symplectic integrators for general separable systems (S2
corresponds to the second order leapfrog method with coefficients
(\ref{leapfrog-coefs}) and SU$_{5}4$ to the fourth order method
with coefficients (\ref{Suzuki-coefs})) and the fourth order
symplectic Runge--Kutta--Nystr\"om method MN$_{6}4$ with
initial conditions $q(0)=1.75$, $p(0)=0$ and $\epsilon=1/20$,
$\delta=1/4$, $\omega=1$.
}}
 \label{fig0}
\end{figure}




\section{Some applications of the numerical integrators based on ME}
\label{animag}

In this section we collect several examples where the numerical integration methods
based on the Magnus expansion have been applied in the recent literature. Special
attention is dedicated to the numerical integration of the Schr\"odinger
equation, since the Magnus series expansion has been extensively used in this
setting almost since its very formulation.
The time-independent Schr\"odinger equation can be considered as a particular
example of a Sturm--Liouville problem, so we also review  the applicability
of Magnus based  techniques in this context. Then we consider a particular
nonlinear system (the differential Riccati equation) which can be, in some sense,
linearized, so that at the end one
may work with finite-dimensional matrices. Finally, we summarize a recent but
noteworthy application: the design of new classes of numerical schemes
for the integration of stochastic differential equations.

\subsection{Case study: numerical treatment of the Schr\"odinger
  equation}

Before embarking ourselves in the use of numerical methods based on the Magnus
expansion in the integration of the Schr\"odinger equation, let us establish
first the theoretical framework which allows one to use numerical integrators 
in this setting for
obtaining approximate solutions in time and space.

\subsubsection{Time-dependent Schr\"odinger equation}

  To keep the treatment as simple as possible, we commence by considering the one-dimensional
time-dependent Schr\"odinger equation ($\hbar =1$)
\begin{equation} \label{Schr1}
 i \frac{\partial}{\partial t} \psi (t,x) = H \psi(t,x) \equiv
 -\frac{1}{2}\frac{\partial^2}{\partial x^2} \psi (t,x) + V(x)
 \psi (t,x)  ,
\end{equation}
with initial condition $\psi(0,x)=\psi_0(x)$. If we look for a
solution of the form $\psi(t,x)=\phi(t) \, \varphi(x)$, it is
clear that, by substituting into (\ref{Schr1}),  one gets
$\phi(t)=\e^{-i t E}$, where $E$ is a constant and $\varphi(x)$ is
the solution of the second order differential equation
\begin{equation}\label{time-indep}
  -\frac{d^2 \varphi}{dx^2} + V(x)\varphi = E \varphi.
\end{equation}
If $E>V$ the solution is oscillatory, whereas if $E<V$ the solution is
a linear combination of exponentially increasing and decreasing
functions. For bounded problems this last condition always takes place
at the boundaries. Since
\begin{equation}\label{binf}
  \int |\psi(x,t)|^2 dx =  \int |\varphi(x)|^2 dx < \infty,
\end{equation}
it is clear that the exponentially increasing solutions have to be
cancelled, and this can only occur for certain values of the
constant $E$, which are precisely the eigenvalues of the problem.

  Let us assume that the system has only a discrete spectrum and
denote by $\{E_n,\varphi_n\}_{n=0}^{\infty}$, with $E_i<E_j, \
i<j$, the complete set of eigenvalues and associated eigenvectors.
It is well known that we can take $\{\varphi_n\}_{n=0}^{\infty}$
as an orthonormal basis and, since (\ref{Schr1}) is linear, any
solution can be written as
\begin{equation} \label{ap1}
  \psi (t,x) = \sum_{n=0}^{\infty} c_n \, \e^{-i t E_n} \varphi_n(x).
\end{equation}
Using the standard notation for the inner product, one has
\begin{equation}\label{Sip2}
  \langle \varphi_n(x) | \psi (t,x) \rangle
  = \int \varphi_n^*(x)\, \psi(t,x) dx
  =  c_n \, \e^{-i t E_n}
\end{equation}
and
\[
  |\langle \varphi_n(x) | \psi (t,x) \rangle|^2 =  |c_n|^2
\]
is the probability to find the system in the eigenstate $\varphi_n$, so that
$\sum_n |c_n|^2=1$. The energy is given by
\begin{equation}\label{Sip3}
 \mathcal{E} = \langle \psi | H | \psi  \rangle
  = \int \psi^*(t,x) \, H \, \psi(t,x) dx
  =  \sum_{n=1}^{\infty}  \sum_{m=1}^{\infty} c_n^* c_m H_{n,m},
\end{equation}
where
\[
  H_{n,m} \equiv  \langle m | H | n \rangle
  =  \langle \varphi_m | H | \varphi_n \rangle.
\]
  In general, the coefficients $c_n$ decrease very fast with $n$ and,
in some cases, the system allows only a finite number of states.
In that situation, one may consider the Schr\"odinger equation as a finite
dimensional linear system where the Hamiltonian is a matrix with
elements $H_{n,m}$. This is precisely the case for the examples
examined in section 4.

  When the Hamiltonian is explicitly time-dependent, this procedure
is not longer valid. Instead one may use some alternative techniques which
we now briefly review.

 (i)  {\it Spectral Decomposition}.
Let us assume that the system is perturbed with a
time-dependent potential, i.e., equation (\ref{Schr1}) takes the form
\begin{equation} \label{Schr-t}
 i \frac{\partial}{\partial t} \psi(t,x)  = \hat{H}(t) \psi(t,x) \equiv
  (\hat{T} + \hat{V}(t)) \, \psi(t,x),
\end{equation}
where
\[
     \hat{T} \psi \equiv -\frac{1}{2}\frac{\partial^2 \psi}{\partial x^2},
      \qquad \hat{V}(t) \psi \equiv (V(x) + \tilde V(t,x) ) \, \psi.
\]
In this case we cannot use separation of variables. However, since $\{\varphi_n\}$
is a complete bases we can still write the solution as
\begin{equation} \label{ap11}
  \psi (t,x) \simeq \sum_{n=0}^{d-1} c_n(t) \, \e^{-i t E_n} \, \varphi_n(x),
\end{equation}
where $E_n$ and $\varphi_n$ are the exact
eigenvalues and eigenfunctions when
$\tilde V=0$, and the complex coefficients $c_n$ give the
probability amplitude to find the system in the state $\varphi_n$
($\sum_n |c_n(t)|^2=1$ for all $t$). Then, substituting
(\ref{ap11}) into (\ref{Schr-t}) we obtain the matrix equation
\begin{equation} \label{lin1}
 i \, \frac{d}{dt}{\bf c} (t) = {\bf H} (t) {\bf c} (t),  \qquad \qquad
 {\bf c} (0) = {\bf c}_0 ,
\end{equation}
where $ \ {\bf c}= (c_0,\ldots ,c_{d-1})^T \in \mathbb{C}^d \ $
and $ \,  {\bf H} \in \mathbb{C}^{d\times d}   \, $ is an
Hermitian matrix associated to the Hamiltonian
\[
   \, ({\bf H}(t))_{ij}= \langle \varphi_i |
\hat{H}(t)-\hat{H}_0 | \varphi_j \rangle \, \e^{i(E_i-E_j)t},   \qquad
i,j=1,\ldots,d
\]
and $\hat{H}_0 = \hat{H}(t=0)$.
 Given the initial wave function $\psi(0,x)$, the
components of ${\bf c}_0$ are determined by  $c_{0,i}=\langle \varphi_i
| \psi(0,x) \rangle$.

  Obviously, any complete basis can be used in this case, although
the norm of the matrix ${\bf H} (t)$ may depend on the
choice. In addition, the number of basis elements, i.e. the
minimum dimension $d$ necessary to obtain a sufficiently accurate
result, also depends on the chosen basis.

 (ii) {\it Space discretization}.
  This procedure intends to take profit of the structure of the Hamiltonian
$\hat{H}$ in (\ref{Schr-t}): $\hat{V} $ is diagonal in the coordinate space and $ \, \hat{T}
\, $ is diagonal in the momentum space.
  Let us assume that the system is defined in
the interval $ \, x \in [x_0, x_f] $ with periodic boundary
conditions. We can then split this interval in $ d $ parts of
length $ \Delta  x= (x_f-x_0)/d \, $ and consider $ \, c_n = \psi
(t,x_n) \, $ where $ x_n=x_0+n\Delta x, $ $ n=0,1,\ldots,d-1$.
Then a finite dimensional linear equation similar (but with a
different coefficient matrix $\mathbf{H}$) to equation (\ref{lin1}) results.
Since $ \, \hat{V} $ is diagonal in the coordinate space and $ \,
\hat{T}  \, $ is diagonal in momentum space, it is possible to use complex
Fast Fourier Transforms (FFTs) for evaluating the products ${\bf
H}{\bf c}$, where $ \, \hat{T} \psi (t,x_n) = \mathcal{F}^{-1} D_T
\mathcal{F} \psi (t,x_n) $, and $ \, D_T \, $ is a diagonal
operator.

\vspace*{0.3cm}

We thus see that whatever the procedure used (spectral decomposition or space
discretization), one ends up with a linear equation of the form
\begin{equation}    \label{nse1}
   i \frac{d\psi}{dt}(t) = H(t) \psi(t),  \qquad \psi(0) = \psi_{0}
\end{equation}
where now $\psi(t)$ represents a complex vector with $d$ components
which approximates the (continuous) wave function. The
computational Hamiltonian $H(t)$ appearing in (\ref{nse1})
is thus a space discretization
(or other finite-dimensional model) of $\hat{H}(t) = \hat{T} +
\hat{V}(t)$. Numerical difficulties come mainly from the unbounded
nature of the Hamiltonian and the highly oscillatory behaviour of
the wave function.

It is at this point when numerical algorithms based on the
Magnus expansion, such as they have been formulated in previous
sections, come into play for integrating in time the linear
system (\ref{nse1}). To put them in perspective, let us
introduce first some other numerical methods also used in this
context. Our exposition is largely based on the reference
\cite{lubich02ifq}.

\paragraph{The implicit midpoint rule}

The approximation to the solution of (\ref{nse1}) provided by
this scheme is implicitly defined by
\begin{equation}  \label{nse2}
  i \frac{\psi_{n+1} - \psi_n}{\Delta t} = H(t_{n + 1/2}) \, \frac{1}{2}
  (\psi_{n+1} + \psi_n),
\end{equation}
where $t_{n+1/2} = \frac{1}{2}(t_{n+1} + t_n)$. Here and in the sequel,
for clarity, we have denoted by $\Delta t$ the
time step size and $t_n = n \Delta t$. Alternatively,
\begin{equation}   \label{nse3}
   \psi_{n+1} = r(-i \Delta t H(t_{n + 1/2})) \, \psi_n, \quad
   \mbox{ with } \quad r(z) = \frac{1 + \frac{1}{2} z}{1 -
   \frac{1}{2} z}.
\end{equation}
Observe that, as $r$ is nothing but the Cayley transform, the
numerical propagator is unitary and consequently the Euclidean
norm of the discrete wave function is preserved along the
evolution: $\|\psi_{n+1} \| = \|\psi_n\|$. This is a crucial
qualitative feature the method shares with the exact solution,
contrarily to other standard numerical integrators, such as
explicit Runge--Kutta methods. From a purely numerical point of
view, the algorithm is stable for any step size $\Delta t$.

Another useful property of this numerical scheme is
time-symmetry: exchanging in (\ref{nse3}) $n$ by $n+1$ and $\Delta
t$ by $-\Delta t$ we get the same numerical method again.
Equivalently, $r(-z) = r(z)^{-1}$, exactly as the exponential
$\e^z$.

With respect to accuracy, it is not difficult to show that, if
$H(t)$ is bounded and sufficiently smooth, the error verifies
\begin{equation}   \label{nse4}
   \|\psi_n - \psi(t_n)\| = \mathcal{O}(\Delta t^2)
\end{equation}
uniformly for $n \Delta t$ in a time interval $[0, t_f]$. In other
words, the implicit midpoint rule is a second-order method. It
happens, however, that the constant in the term
$\mathcal{O}(\Delta t^2)$ depends on bounds of $H'$ and
$H^{\prime\prime}$ and on the norm of the third derivative of the solution
$\psi$. Since, in general, the wave function is highly oscillatory
in time, this time derivative can become large, and so the use of
very small time steps is mandatory.

\paragraph{The exponential midpoint rule}

Another possibility to get approximate solutions of (\ref{nse1})
consists in replacing $r(z)$ by $\exp(z)$ in (\ref{nse3}):
\begin{equation}   \label{nse5}
  \psi_{n+1} = \exp(-i \Delta t \, H(t_{n+1/2})) \, \psi_n.
\end{equation}
Now, instead of solving systems of linear equations as previously,
one has to compute the exponential of a large matrix times a
vector at each integration step. In this respect, the techniques
reviewed in subsection \ref{exponen} can be efficiently implemented. The
exponential midpoint rule (\ref{nse5}) also provides a unitary
propagator and it is time-symmetric. In addition, the error
satisfies the same condition (\ref{nse4}), but now the constant in
the $\mathcal{O}(\Delta t^2)$ term is independent of the time
derivatives of $\psi$ under certain assumptions on the commutator
$[H(t),H(s)]$ \cite{hochbruck03omi}. As a consequence, much larger
time steps can be taken to achieve the same accuracy as with the
implicit midpoint rule.

\paragraph{Integrators based on the Magnus expansion}

The method (\ref{nse5}) is a particular instance of a second order
Magnus method when the integral $\int_0^{\Delta t} H(s) ds$ is
replaced by the midpoint quadrature rule. In fact, we have already used it in
(\ref{mpoint}). Obviously, if higher order
approximations are considered, the accuracy can be enhanced a
great deal. This claim has to be conveniently justified, however,
since the order of the numerical methods based on Magnus has
been deduced when
$\|\Delta t H(t)\| \rightarrow 0$ and is obtained by studying the remainder
of the truncated Magnus series. In the Schr\"odinger equation, on
the other hand, one has to cope with discretizations of unbounded
operators, so in principle it is not evident how the previous results
on the order of accuracy apply in this context.
In \cite{hochbruck03omi}, Hochbruck and Lubich analyse
in detail the application of the fourth-order Magnus integrator
(\ref{or4GL}) to equation (\ref{nse1}), showing that it
works extremely well even with step sizes for which the corresponding
$\|\Delta t H(t)\|$ is large. In particular, the scheme
retains fourth order of accuracy in $\Delta t$
\emph{independently of the norm} of $H(t)$
when $H(t) = T + V(t)$, $T$ is a discretization of $-\frac{1}{2}
\frac{\partial^2}{\partial x^2}$ (with maximum eigenvalue
$E_{\mathrm{max}} \sim (\Delta
x)^{-2}$) and $V(t)$ is sufficiently smooth under the time step
restriction $\Delta t \sqrt{E_{\mathrm{max}}} \le Const.$ This is
so even when there is no guarantee that the Magnus series
converges at all.

\paragraph{Symplectic perspective}

The evolution operator corresponding to (\ref{Schr-t}) is not only unitary, but also symplectic
with canonical coordinates and momenta $\mbox{Re}(\psi)$ and $\mbox{Im}(\psi)$, respectively.
If we carry out a discretization in space, this symplectic structure is inherited by the 
corresponding equation (\ref{lin1}). It makes sense, then, to write
${\bf c}= {\bf q} + i {\bf p}$ and consider the equations satisfied by ${\bf q, p}\in\mathbb{R}^d$, namely
\begin{equation}\label{SchrClas1}
  \mathbf{q}' = \mathbf{H}(t) \mathbf{p}, \qquad \mathbf{p}' = -\mathbf{H}(t) \mathbf{q},
\end{equation}
which can be interpreted as the canonical equations corresponding to the Hamiltonian \cite{gray94chs}
\begin{equation}\label{SchrClas2}
 \mathcal{H}(t,{\bf q, p}) = {\bf p}^T {\bf H}(t) {\bf p} + {\bf q}^T {\bf H}(t)
  {\bf q}.
\end{equation}

Denoting ${\bf z=(q,p)}^T$, it is clear that
\[
  {\bf z}' = \left( {\bf A}(t) + {\bf B}(t)  \right)\, {\bf z}
\]
where
\begin{equation}\label{matMN}
  {\bf A}(t) = \left( \begin{array}{cc}
    0 & \;\; {\bf H}(t) \\
    0 & 0
  \end{array}
  \right), \qquad \quad
  {\bf B}(t) = \left( \begin{array}{cr}
    0 &  \;\; 0 \\
    -{\bf H}(t) & 0
  \end{array}
  \right).
\end{equation}
For this system it is possible, therefore, to apply the commutator-free Magnus integrators 
constructed in subsection \ref{cfmi}.  In addition, one has
\[
{\bf [B,[B,[B,A]]]=[A,[A,[A,B]]]=0}, 
\]
and this property allows us to use
especially designed and highly efficient integration methods
\cite{blanes07smf}.

\subsubsection{Time-independent Schr\"odinger equation}

Restricting ourselves to the time-independent Schr\"odinger
equation (\ref{time-indep}),  we next illustrate how Magnus
integrators can in fact be used to compute the discrete
eigenvalues defined by the problem. Although only the
Schr\"odinger equation in a finite domain is considered,
\begin{equation}   \label{tise1}
   -\frac{d^2 \varphi}{dx^2} + V(x) \varphi = \lambda \varphi, \qquad x \in (a,b)
\end{equation}
the procedure can be easily adapted to other types of eigenvalue
 problems, in which one has to find both $\lambda \equiv E$ and
$\varphi$. Here it is
 assumed that the potential is smooth, $V \in C^m(a,b)$ and, for simplicity,
 $\varphi(a)=\varphi(b)=0$.

Under these assumptions, it is well known that the eigenvalues are
real, distinct and bounded from below.
 The problem (\ref{tise1}) can be formulated in the
 special linear group $\mathrm{SL}(2)$,
\begin{equation}   \label{tise2}
  \frac{d \mathbf{y}}{dx} = \left(  \begin{array}{ccr}
                          0   &  \;\; & \  \   1  \\
                     V(x)-\lambda &  &  0
                       \end{array}  \right)  \mathbf{y},  \qquad
             x \in (a,b), \qquad  \mbox{ where } \quad
\mathbf{y} = (\varphi, d\varphi/dx)^T,
\end{equation}
so that the Magnus expansion can be applied in a natural way. As usual, rather than
approximating the fundamental solution of (\ref{tise2}) in the entire interval
$(a,b)$ by $\exp(\Omega)$, the idea is to partition the interval into $N$ small
subintervals, and then apply a conveniently discretized version of the Magnus
expansion. In this way, the convergence problem no longer restricts the size
$(b-a)$ \cite{moan98eao}.

For the sake of simplicity, let us consider the fourth-order method (\ref{or4GL}).
Writing
\[
     V_{n,1} = V(x_{n} + (\frac{1}{2}-\frac{\sqrt{3}}{6})h), \qquad
     V_{n,2} = V(x_{n} + (\frac{1}{2}+\frac{\sqrt{3}}{6})h),
\]
where $h = (b-a)/N$ and $x_n = a + h \, n$, we form
\[
    \sigma_{n}(\lambda) = \left(  \begin{array}{ccc}
        -\frac{\sqrt{3}}{12} h^2 (V_{n,1}-V_{n,2}) &   \;\;\; &   \   \  \  h  \\
        \frac{1}{2} h (V_{n,1}+V_{n,2}) - h \lambda &  &
           \frac{\sqrt{3}}{12} h^2 (V_{n,1}-V_{n,2})
            \end{array}  \right),
\]
for $n=0,1,\ldots,N-1$. Then, the fourth-order approximation to
the solution of (\ref{tise2}) at $x=b$ is
\begin{equation}   \label{tise3}
      \mathbf{y}(b) = \e^{\sigma_{N-1}(\lambda)} \cdots
             \e^{\sigma_{1}(\lambda)} \, \e^{\sigma_{0}(\lambda)}
             \mathbf{y}(a)
\end{equation}
and the values of $\lambda$ are obtained from (\ref{tise3}) by using repeatedly 
the expression of the exponential of a traceless matrix, eq. (\ref{expMat2}), and requiring
that $\varphi(a)=\varphi(b)=0$. The resulting nonlinear equation in $\lambda$ can be solved,
for instance, by Newton--Raphson iteration, which provides quadratic convergence
for starting values sufficiently near the solution \cite{moan98eao}.

Although by construction this procedure leads to a global order of
approximation $\mathcal{O}(h^p)$ if a $p$th-order Magnus method is
applied, it turns out that the error also depends on the magnitude
of the eigenvalue. Specifically, the error in a $p$th-order method
grows as $\mathcal{O}(h^{p+1} \lambda^{p/2-1})$ \cite{moan98eao},
and thus one expects poor approximations for large eigenvalues.
This difficulty can be overcome up to a point by analyzing the
dependence on $\lambda$ of each term in the Magnus series and
considering partial sums of the terms carrying the most
significant dependence on $\lambda$. For instance, it is possible
to design a sixth-order Magnus integrator for this problem with
error $\mathcal{O}(h^7 \lambda)$, which therefore behaves like a
fourth-order method when $h^2 \lambda \approx 1$, whereas the
standard sixth-order Magnus scheme, carrying an error of
$\mathcal{O}(h^7 \lambda^2)$, reduces to an order-two method
\cite{moan98eao}. In any case, getting accurate approximations
when $|\lambda| \rightarrow \infty$ is more problematic
\cite{jodar00soa}.

\subsection{Sturm--Liouville problems}

The system defined by (\ref{tise1}) with boundary conditions
$\varphi(a)=\varphi(b)=0$ is just one particular example of a
second order Sturm--Liouville problem \cite{pryce93nso,zettl05slt}. It is thus
quite natural to try to apply Magnus integrators to more general
problems within this class.

A second order Sturm--Liouville eigenvalue problem has the form
\begin{equation}  \label{st-1}
  \frac{d}{dx} \left( p(x) \frac{dy}{dx}(x) \right) + q(x) y(x) =
   \lambda \, r(x) y(x) \qquad \mbox{ on } (a,b)
\end{equation}
with separated boundary conditions which commonly have the form
\begin{equation}  \label{st-2}
   A_1 y(a) + A_2 p(a) y^\prime(a)  =  0  \qquad
   B_1 y(b) + B_2 p(b) y^\prime(b)  =  0
\end{equation}
for given constants $A_i, B_i$ and functions $p(x)$, $q(x)$ and $r(x)$.
Solving this problem means, of course,
determining the values $\lambda_n$ of $\lambda$ for which eq.
(\ref{st-1}) has a nontrivial (continuously differentiable
 square integrable) solution $y_n(x)$ satisfying
equations (\ref{st-2}) \cite{zettl05slt,bailey78aso}.

These and other higher order Sturm--Liouville problems can be
recasted as a linear matrix system of the form
\begin{equation} \label{st-3}
  Y' = (\lambda B + C(x)) \, Y
\end{equation}
by transforming to the so-called \emph{compound matrix} or modified
Riccati variables \cite{greenberg98ota,jodar00soa}. Here $B$ is a constant matrix.
When generalizing
the above treatment based on the Magnus expansion to this problem,
there is one
elementary but important remark worth to be stated explicitly:
\emph{unless the differential equation (\ref{st-3}) has the same
large $\lambda$-asymptotics as some differential equation with
$x$-independent coefficients, then it will be impossible to develop
a Magnus method which accurately approximate its solutions for large
$\lambda$} \cite{jodar00soa}.  The reason is that a Magnus method
approximates the solution by a discrete solution calculated using a
formula of the form $Y(x_{n+1}) = \exp(\sigma_n(\lambda)) Y(x_n)$; in
particular, on the first step $(x_0,x_1)$, the differential equation
is approximated by one in which the coefficient matrix is replaced
by the $x$-independent matrix $\sigma_0(\lambda)/(x_1-x_0)$.

In consequence, the attention should be restricted to systems for
which it is known that a suitable constant-coefficient system provides
the correct asymptotics. This is the case, in particular, for
equation (\ref{tise1}), and more generally for linear equations of order
$2n$ in which the $(2n-1)$st derivative is zero, such as
\[
  (-1)^n y^{(2n)} + \sum_{j=0}^{2n-2} q_j(x) y^{(j)} = \lambda y.
\]
Here the asymptotics are determined by the equation
$(-1)^n y^{(2n)} = \lambda y$ \cite{naimark68ldo}. Even then, the
methods developed in \cite{moan98eao} for equation
(\ref{tise1}) and implemented for systems with matrices of general
size in \cite{jodar00soa} require a $\lambda$-dependent step size
restriction of the form $h \le \mathcal{O}(|\lambda|^{-1/4})$
in order to be defined. Nevertheless, the analysis carried out in
\cite{jodar00soa} shows that the fourth order Magnus integrator
based on a two-point Gaussian quadrature appears to offer significant
advantages over conventional methods based on power series and library
routines.

Magnus integrators have also been successfully applied in the
somewhat related problem of computing the Evans function for
spectral problems arising in the analysis of the linear
stability of travelling wave solutions to reaction-diffusion
PDEs \cite{aparicio05neo}. In this setting, Magnus integrators
possess some appealing features in comparison, for instance, with
Runge--Kutta schemes: (1) they are unconditionally stable; (2) their
performance is superior in highly oscillatory regimes and (3) their
step size can be controlled in advance. Items (2) and (3) are due
to the fact that error bounds for Magnus methods depend only
on low order derivatives of the coefficient matrix, not (as for
Runge--Kutta schemes) on derivatives of the solution. Therefore,
performance and, correspondingly, the choice of optimal step size
remain uniform over any bounded region of parameter space
\cite{aparicio05neo}.

\subsection{The differential Riccati equation}

Let us consider now the two-point boundary value
problem in the $t$ variable defined by the linear differential equation
\begin{equation}\label{BVP_Riccati}
  {\bf y}' \equiv \left( \begin{array}{c}
    {\bf y}_1' \\ {\bf y}_2'
\end{array} \right) = \left( \begin{array}{cc}
A(t)  & \   \  B(t) \\
C(t) & D(t)
\end{array} \right) \
\left( \begin{array}{c}
    {\bf y}_1 \\ {\bf y}_2
\end{array} \right), \qquad 0<t<T
\end{equation}
with separated boundary conditions
\begin{equation}\label{sbc}
(K_{11} \; \; K_{12}) \left( \begin{array}{c}
    {\bf y}_1 \\ {\bf y}_2
\end{array} \right)_{t=0} = \; \gamma_1, \qquad
(K_{21} \;\; K_{22}) \left( \begin{array}{c}
    {\bf y}_1 \\ {\bf y}_2
\end{array} \right)_{t=T} = \, \gamma_2.
\end{equation}
Here $ A \in \mathbb{C}^{q \times q} \, $, $
B \in \mathbb{C}^{q \times p} $, $C \in \mathbb{C}^{p \times q} \, $, $ D \in
\mathbb{C}^{p \times p} \, $, whereas
${\bf y}_1, \gamma_2\in \mathbb{C}^{p}$, ${\bf y}_2,\gamma_1 \in
\mathbb{C}^{q}$ and the matrices $K_{ij}$ have appropriate
dimensions. We next introduce the time-dependent change of variables
(or picture)  ${\bf y}=Y_0(t) \, {\bf w}$,  with
\begin{equation}\label{cov1}
Y_0(t) = \left( \begin{array}{cc}
I_p  & \   \  0 \\
X(t) & I_q
\end{array} \right)
\end{equation}
and choose the matrix $ X\in \mathbb{C}^{p \times q}$ so as to ensure that in the
new variables $\mathbf{w} = Y_0^{-1}(t) \mathbf{y}$ the system assume the partly
decoupled structure \cite{dieci92nio}
\begin{equation}  \label{dececri}
  {\bf w}' \equiv \left( \begin{array}{c}
    {\bf w}_1' \\ {\bf w}_2'
\end{array} \right) = \left( \begin{array}{cc}
A+BX  & \  \  B \\
O & D-XB
\end{array} \right) \
\left( \begin{array}{c}
    {\bf w}_1 \\ {\bf w}_2
\end{array} \right),
\end{equation}
together with the corresponding boundary conditions for ${\bf w}$. It turns out that
this is possible if and only if $X(t)$ satisfies the so-called differential Riccati
equation \cite{dieci88art}
\begin{equation} \label{Riccati}
 X' =  C(t) + D(t) X - X A(t) - X B(t) X,  \qquad X(0) = X_0
\end{equation}
for some $X_0$. By requiring
\begin{equation}   \label{icrica1}
     X_0=-K_{12}^{-1}K_{11},
\end{equation}
then the boundary conditions (\ref{sbc}) also decouple as
\begin{equation}   \label{sbc12}
(O \;\; K_{12}) {\bf w}(0) = \gamma_1, \qquad
\big( K_{21}+K_{22}X(T) \;\;\; K_{22} \big) {\bf w}(T) = \gamma_2.
\end{equation}
Here we assume without loss of generality that $K_{12}$ is
invertible. In this way, the original boundary value problem can
be solved as follows \cite{dieci92nio,dieci88art}: (i) solve
equation (\ref{Riccati}) with initial condition (\ref{icrica1})
from $t=0$ to $t=T$; (ii) solve the $\mathbf{w}_2$-equation in
(\ref{dececri}) and (\ref{sbc12}), also from zero to $T$; (iii)
solve the $\mathbf{w}_1$-equation in (\ref{dececri}) from $t=T$ to
$t=0$ and recover ${\bf y}=Y_0(t) \, {\bf w}$. In other words, the
solution of the original two-point boundary value problem can be
obtained by solving a sequence of three different initial value
problems, one of which involves the nonlinear equation
(\ref{Riccati}). Obviously, steps (i) and (iii) can be solved
using numerical integrators based on the Magnus expansion. It
could be perhaps more surprising that these algorithms can indeed
be used to integrate the Riccati equation in step (ii).

Although the boundary value problem (\ref{BVP_Riccati}) is a convenient way to introduce
the differential Riccati equation (\ref{Riccati}), this equation arises in many fields
of science and engineering, such as linear quadratic
optimal control, stability theory, stochastic control,
differential games, etc. Accordingly, it has received
considerable attention in the literature, both focused to its theoretical aspects
 \cite{bittanti91tre,reid72rde}
and its numerical treatment \cite{dieci92nio,kenney85nio,schiff99ana}.

In order to apply Magnus methods to solve numerically
the Riccati equation, we first apply the transformation
\begin{equation}  \label{trans-ri1}
 X(t) \ = \ V(t) \, W^{-1}(t),
\end{equation}
with $V \in \mathbb{C}^{p \times q}$, $W \in \mathbb{C}^{q \times q}$
and $V(0)=X_0, W(0)=I_q$, in the region where $ W(t) $ is
invertible. Then eq. (\ref{Riccati}) is equivalent to the linear
system
\begin{equation}\label{RiccatiLin1}
\begin{array}{c}
\displaystyle
               Y' \, = \, S(t)  \, Y(t) \, , \qquad  Y(0) \, = \,
              \left( \begin{array}{c} I_{q} \\ X_{0}
               \end{array} \right) \,
\end{array}
\end{equation}
with
\begin{equation}\label{RiccatiLin2}
\begin{array}{c}
             \displaystyle
  Y(t) = \left( \begin{array}{c}
    W(t) \\ V(t)
\end{array} \right),  \qquad \qquad
  S(t) = \left( \begin{array}{cc}
A(t) & \  \  B(t) \\
C(t) & D(t)
\end{array} \right)
\end{array}
\end{equation}
so that the previous Magnus integrators for linear problems can be
applied here. Apparently, this system is similar to
(\ref{BVP_Riccati}), but now we are dealing with an initial value
problem and $Y$ is a matrix instead of a vector.

When dealing in general with the differential Riccati equation (\ref{Riccati}), it is
meaningful to distinguish the following three cases:
\begin{itemize}
 \item[(i)] The so-called symmetric Riccati equation, which corresponds to $q=p$,
$D(t) = -A(t)^{T}$ real, and $B(t)$, $C(t)$ real and symmetric
matrices. In this case, the solution satisfies $X^T=X$. It is straightforward to show
that this problem is equivalent to the treatment of the 
generalized time-dependent harmonic oscillator, described by the
Hamiltonian function
\[
 H = \frac12 {\bf p}^T B(t){\bf p} +
             {\bf p}^T A(t){\bf q} -
     \frac12 {\bf q}^T C(t) {\bf q}.
\]
The approximate solution attained by Magnus integrators when applied to
(\ref{RiccatiLin1})-(\ref{RiccatiLin2}) can be seen as the exact
solution corresponding to a perturbed symplectic matrix $\tilde{S}(t)\simeq
S(t)$. In other words, we are solving exactly a perturbed Hamiltonian system
so that the approximate solution, $\tilde{X}$, will shares several properties of the
exact solution, in particular $\tilde{X}^T=\tilde{X}$.
 \item[(ii)] The linear non-homogeneous problem
\begin{equation}\label{linear-non-homog}
  X' = D(t) X + C(t)
\end{equation}
corresponds to the particular case $A=0$ and $B=0$ in
(\ref{Riccati}).
 \item[(iii)] The problem
\begin{equation}\label{isospectral}
  X' = D(t) X + X A(t)
\end{equation}
is recovered from (\ref{Riccati}) by taking $C=0$ and $B=0$. It has been treated in
\cite{iserles01ame} by developing an \emph{ad hoc} Magnus-type expansion. Notice that
the case $p=q$, $D=-A$ corresponds to the linear isospectral system (\ref{nde.4}).
\end{itemize}

\subsection{Stochastic differential equations}


In recent years the use of stochastic differential equations (SDEs) has
become widespread in the simulation of random phenomena appearing
in physics, engineering, economics, etc, such as turbulent diffusion,
polymer dynamics and investment finance \cite{burrage04nmf}. Although
models based on SDEs can offer a more realistic representation of the
system than ordinary differential equations, the design of effective numerical
schemes for solving SDEs is, in comparison with ODEs, a less  developed field
of research. This fact notwithstanding, it is true that recently new classes of integration
methods have been constructed which automatically incorporate
conservation properties the SDE  possesses. Since some of the methods are based precisely
on the Magnus expansion, we briefly review here their main features, and refer
the reader to the more advanced literature on the subject
\cite{burrage04nmf,kloeden92nso,platen99ait}.

A SDE in its general form is usually written as
\begin{equation}   \label{sde1}
   dy(t) = g_0(t,y(t)) \, dt + \sum_{j=1}^d g_j(t,y(t)) \, dW_j(t), \qquad y(0) = y_0, \qquad
    y \in \mathbb{R}^m,
\end{equation}
where $g_j$, ($j \ge 0$),  are $m$-vector-valued functions. The function $g_0$ is the deterministic
continuous component  (called the \emph{drift coefficient}), the $g_j$, ($j \ge 1$), represent
the stochastic continuous components (the \emph{diffusion coefficients}) and $W_j$ are
$d$ independent Wiener processes. A Wiener process $W$ (also called Brownian motion)
is a stochastic process \cite{burrage04nmf} satisfying
\[
    W(0) = 0, \qquad E[W(t)] = 0, \qquad \mathrm{Var}[W(t)-W(s)] = t-s, \quad t > s
\]
which has independent increments on non-overlapping intervals. In other words, a Wiener process
is normally distributed with mean or expectation value $E$ equal to zero and variance $t$.

Equation (\ref{sde1}) can be written in integral form as
\begin{equation}   \label{sde2}
   y(t) = y_0 + \int_0^t g_0(s,y(s)) \, ds + \sum_{j=1}^d \int_0^t g_j(s,y(s)) \, dW_j(s).
\end{equation}
The $d$ integrals in (\ref{sde2}) cannot be considered as Riemann--Stieltjes integrals, since
the sample paths of a Wiener process are not of bounded variation. In fact, if different choices
are made for the point $\tau_i$ (in the subintervals $[t_{i-1},t_i]$ of a given partition) where
the function is evaluated, then the approximating sums for each $g_j$,
\begin{equation}   \label{sde3}
   \sum_{i=1}^N g_j(\tau_i, y(\tau_i)) (W_j(t_i) - W_j(t_{i-1})), \qquad
      \tau_i = \theta t_i + (1- \theta) t_{i-1},
\end{equation}
converge (in the mean-square sense) to different values of the integral, depending
on the value of $\theta$ \cite{burrage99hso}. Thus, for instance,
\[
   \int_a^b W(t) dW(t) = \frac{1}{2} (W^2(b) - W^2(a)) + (\theta - \frac{1}{2})(b-a).
\]
If $\theta=0$, then $\tau_i = t_{i-1}$ (the left-hand point of each subinterval) and the resulting
integral is called an It\^o integral; if $\theta=1/2$ (so that the midpoint is used instead), one has
a Stratonovich integral. These are the two main choices and, although they are related,
 the particular election depends ultimately on the nature of the process to be modeled
 \cite{burrage99hso}. It can be shown that the Stratonovich calculus
 satisfies the Riemann--Stieltjes
 rules of calculus, and thus it is the natural choice here.

 When dealing with numerical methods for solving (\ref{sde1}), there are two ways of
 measuring accuracy \cite{burrage04nmf}. The first is \emph{strong convergence}, essential when the
 aim is  to get numerical approximations to the trajectories which are close to the exact solution.
 The second is \emph{weak convergence}, when only certain moments of the solution are of
 interest. Thus, if $\hat{y}_n$ denotes the numerical approximation to $y(t_n)$ after $n$
 steps with constant step size $h = (t_n - t_0)/n$, then  the numerical solution
 $\hat{y}$ converges strongly to the exact solution $y$
 with strong global order $p$ if there exist $C>0$ (independent of $h$) and $\delta > 0$ such
 that
 \[
     E[\| \hat{y}_n - y(t_n) \|] \le C h^p, \qquad h \in (0,\delta).
 \]
 It is worth noticing that $p$ can be fractional, since the root mean-square order of the
 Wiener process is $h^{1/2}$. One of the simplest procedures for solving (\ref{sde1}) numerically
 is the so-called Euler--Maruyama method \cite{maruyama55cmp},
 \begin{equation}   \label{sde4}
    y_{n+1} = y_n + \sum_{j=0}^d J_j g_j(t_n,y_n),
\end{equation}
where
\[
    h = t_{n+1} - t_n, \quad J_0 = h, \quad J_j = W_j(t_{n+1}) - W_j(t_n), \quad j=1,\ldots,d.
\]
This scheme turns out to be of strong order $1/2$. Here the $J_j$ can be computed
as $\sqrt{h} N_j$, where the $N_j$ are $N(0,1)$ normally distributed independent
random variables \cite{burrage99hso}.

For the general non-autonomous linear Stratonovich problem defined by
\begin{equation}    \label{sde5}
   dy = G_0(t) y \, dt + \sum_{j=1}^d G_j(t) \, y \, dW_j, \qquad y(0) = y_0 \in \mathbb{R}^m
\end{equation}
the Magnus expansion for the deterministic case can be extended in a
quite straightforward way. It is well worth noticing
that in equation (\ref{sde5}), even when the functions $G_j$ are constant, there is no explicit solution
\cite{kloeden92nso}, unless all the $G_j$, $j \ge 0$, commute with one another, in which case it holds
that
\begin{equation}    \label{sde6}
    y(t) = \exp \left( G_0 \, t + \sum_{j=1}^d G_j W_j(t) \right) y_0.
\end{equation}
In many modeling situations, however, there is no reason to expect that the functions $G_j$ associated
with the Wiener processes commute. If for simplicity we only consider the autonomous case and
write
\[
    G(t) \equiv  G_0 \, dt + \sum_{j=1}^d G_j \, dW_j(t),
\]
then (\ref{sde5}) can be expressed as
\[
    dy = G(t) \, y \, dt, \qquad y(0) = y_0
\]
and thus one can formally apply the Magnus expansion to this
equation to get $y(t) = \exp(\Omega(t)) y_0$. The first term in
the series reads in this case
\[
   \int_0^t G(s) \, ds \equiv \int_0^t G_0 \, ds + \sum_{j=1}^d \int_0^t G_j \, dW_j(s) =
       G_0 \, t + \sum_{j=1}^d G_j J_j,
\]
where now
\[
    J_j = \int_0^t dW(s) = W_j(t) - W_j(0).
\]
By inserting these expressions into the recurrence associated with the Magnus series, Burrage
and Burrage \cite{burrage99hso}  show that
\begin{eqnarray}   \label{sde7}
   \Omega(t) & = & \sum_{j=0}^d G_j J_j + \frac{1}{2} \sum_{i=0}^d \sum_{j=i+1}^d [G_i, G_j]
       (J_{ji} - J_{ij}) \\
      & & + \sum_{i=0}^d \sum_{k=0}^d \sum_{j=k+1}^d [G_i,[G_j,G_k]]
       \left( \frac{1}{3} (J_{kji} - J_{jki}) + \frac{1}{12} J_i (J_{jk} - J_{kj}) \right) + \cdots, \nonumber
\end{eqnarray}
where the multiple Stratonovich integrals are defined by
\begin{equation}         \label{sde8}
   J_{j_1 j_2 \cdots j_l}(t) = \int_0^t \int_0^{s_l} \cdots \int_0^{s_2} dW_{j_1}(s_1) \cdots
       dW_{j_l}(s_l), \qquad j_i \in \{0,1,\ldots,d\}.
\end{equation}
Since not all the Stratonovich integrals are independent,  one has
to compute only $d(d+1)(d+5)/6$ stochastic integral evaluations to achieve strong order 1.5 with
the expression (\ref{sde7}) \cite{burrage99hso}. If, on the other hand, $\Omega(t)$ is truncated
after the first set of terms, then the resulting numerical approximation
\[
   y(t) = \exp \left( \sum_{j=0}^d G_j J_j \right) y_0
\]
has strong order $1/2$, but leads to smaller error coefficients than the Euler--Maruyama method
(\ref{sde4})
\cite{burrage99hso}. Furthermore, the error becomes smaller as the $G_i G_j$ terms get closer
to commuting and the scheme preserves the underlying structure of the problem.

One should notice at this point that  equation (\ref{sde1}) (or, in the linear case, equation (\ref{sde5})),
has formally the  same  structure as the nonlinear ODE (\ref{cf1}) appearing in control theory. Therefore,
the formalism developed there to get the Chen--Fliess series can be applied here with the
alphabet $I = \{ 0, 1, \ldots, d\}$ and the integrals
\[
   \left( \int_0 \mu \right)(t) = \int_0^t \mu(s) \, ds, \qquad
   \left( \int_i \mu \right)(t)  = \int_0^t \mu(s) \, dW_i(s), \quad i \ge 1,
\]
since the Stratonovich integrals satisfy the integration by parts rule.
In other words,  one can obtain the corresponding Magnus expansion for arbitrary (linear or nonlinear)
stochastic differential equations simply by following the same procedure as for deterministic ODEs.

With respect to nonlinear Stratonovich stochastic differential equations, it should be remarked that
the use of Lie algebraic techniques as well as the design of Lie group methods for obtaining
strong approximations when the solution evolves on a smooth manifold
 has received considerable attention in the recent literature
 \cite{lord06esi,malham07slg,misawa01ala}.



\section{Physical applications}

From previous sections it should be clear that ME has a strong
bearing on both Classical and Quantum Mechanics. As far as
Classical Mechanics is concerned this has been most explicitly
shown in section \ref{GNLM}. On its turn Quantum Mechanics has
been repeatedly  invoked as a source of applications, among
others, in sections \ref{PLT} and \ref{section4}. In this section
we present in a very schematic way and with no aim at completeness
some applications of ME in different areas of the physical
sciences. This will show that over the years ME has been one of
the preferred options to deal with equation (\ref{laecuacion})
which, under different appearances, pervades the entire field of
Physics. In the works mentioned here, almost exclusively
analytical methods are used and, in general, one must recognize
that in most, if not all, cases listed only the first two orders
of the expansion have been considered. In very specially simple
applications, due to particular algebraic properties of the
operators involved, this happens to be exact. Only with the more
recent advent of the numerical applications, as has been
emphasized in sections \ref{section5} and \ref{animag}, has the
expansion been carried in a more systematic way to higher orders.

\subsection{Nuclear, atomic and molecular physics}
\label{sec:nucatom}

As far as we know the first physical application of ME dates back to
1963. Robinson \cite{robinson63mce} published a brand new formalism
to investigate multiple Cou\-lomb excitations of deformed nuclei. As a
matter of fact, he states explicitly  that only after completion of
his work he discovered the ME. His derivation of ME formulas is
certainly worth reading.

The Coulomb excitation process yields information about the low
lying nuclear states. Prior to Robinson work, the theory was
essentially based on perturbation expansions which requires that the
bombarding energy is kept so low that no nuclear reaction takes
place. Even worse, if heavier ions are used as projectiles the
electric field exerted on the target nucleus is so strong that
perturbation methods fail.

The work by Robinson improved the so-called at that time
\emph{sudden approximation}, which is equivalent to the assumption
that all nuclear energy levels are degenerate. Results are reported
in that reference for rotational and vibrational nuclei.

As representatives of the applications of ME in the field of Atomic
Physics we mention several types of atomic collisions. The ME is
used in \cite{eichler77maf} to derive the transition amplitude and
the cross section for K-shell ionization of atoms by heavy-ion
impact. This is an important process in heavy-ion physics. The
theoretical investigations of these reactions always assumed that
the projectile is a relatively light ion such as a proton or an
$\alpha$ particle. The use of ME allowed to extend the studies to
the ionization of light target atoms by much heavier projectile
ions.

In \cite{wille81mef,wille86moi} the ME is applied to study the
time-evolution of rotationally induced inner-shell excitation in
atomic collisions. In this context the internuclear motion can be
treated classically and the remaining quantum-mechanical problem
for the electronic motion is then time-dependent. In particular,
in \cite{wille81mef} explicit results for Ne$^{+}$Ne collisions
are given as well as a study of the convergence properties of ME
with respect to the impact parameter.

The ME is applied in \cite{hyman85dma} to the theoretical study of
electron-atom collisions, involving many channels coupled by
strong, long-range forces. Then, as a test case, the theory is
applied to electron-impact excitation of the resonance transitions
of Li, Na and K. Computations up to second order are carried out
and the cross sections found are in good agreement with
experimental data for the intermediate-energy range.

The following examples illustrate the use of ME in Molecular
Physics. In \cite{cady74rsl} it is applied for the first time to the
theory of the pressure broadening of rotational spectra. Unlike the
previous approaches to the problem, the S-matrix obtained is
unitary. As a consequence of it the relative contributions on the
linewidth of the attractive and repulsive anisotropy terms in the
interaction potential may be calculated.

Floquet theory is applied in \cite{Milfeld83sea} to systems periodic
in time in the semiclassical approximation of the
radiation--quantum-molecule interaction in an intense field. The
paper contains an interesting discussion about the appropriateness
of the Schr\"odinger and Interaction pictures. One and two-photon
probability transitions are obtained up to second order in ME.
Noteworthy, formulas through fifth order in ME are given, in a less
symmetrical form.

In \cite{schek81aot} it is explored the applicability of ME to the
multiphoton excitation of a sparse level system for which the
rotating wave function approximation is not applicable. This
reference provides a method of treating the time-evolution of a
pumped molecular system in the low energy region, which is
characterized by a sparse distribution of bound vibrational states.

\subsection{Nuclear magnetic resonance: Average Hamiltonian
theory}
\label{sec:NMR}

By far this is the field where ME has been
most systematically used and so we consider it apart. From
elementary quantum mechanics it is known that a constant magnetic
field breaks the degeneracy of the energy levels of an atomic
nucleus with spin. If the nuclear spin is $s$ then $2s+1$ sublevels
appear. In a sample these states are occupied according to Boltzman
distribution with a exponentially distributed population. When a
time-dependent radio-frequency electromagnetic field of appropriate
frequency is applied then energy can be absorbed by certain nuclei
which are consequently promoted to higher levels. This is the
physical phenomenon of Nuclear Magnetic Resonance (NMR).

It was Evans \cite{evans68osa} and Haeberlen and Waugh
\cite{haeberlen68cae} who first applied the ME to NMR. Since that
time, the ME has been instrumental in the development of improved
techniques in NMR spectroscopy \cite{burum81meg}.

The major advantage of NMR is the possibility of modifying the
nuclear spin Hamiltonian almost at will and to adapt it to the
needs of the problem to be solved \cite{ernst86pon}. This
manipulation requires an external perturbation of the system that
can be either time-independent (changes of temperature, pressure,
solvents, etc.) or time-dependent (sample spinning, pulsed
radio-frequency fields). In the later context, the concept of
average Hamiltonian provides an elegant description of the effects
of a time-dependent perturbation applied to the system. It was
originally introduced in NMR by Waugh \cite{ernst86pon,waugh82tob}
to explain the effects of multiple-pulse sequences.

The basic idea of average Hamiltonian theory, for a system governed
by $H(t)$, consists in describing the effective evolution within a
fixed time interval by an average Hamiltonian $\overline{H}$. The
theory states that this is always possible provided $H(t)$ is
periodic. The average Hamiltonian depends, however, on the beginning
and the end of the time interval observed. It is right the average
Hamiltonian $\overline{H}$ which is obtained by means of the ME.

When the total Hamiltonian splits in a time-independent and a
time-dependent piece, $H(t)=H_0+H_1(t)$, with $H_1(t)$ periodic, an
interesting new picture is used, labeled \emph{toggling frame}. It
certainly reminds the Interaction Picture defined in equation (\ref{GInt}) but
is rather different. In (\ref{Ufac}) the operator $G(t)$ associated
to the toggling frame is given by the time-ordered expression
\begin{equation}\label{toggling}
G(t)=  \mathcal{T} \left( \exp{\int_0^t \tilde H_1(s)\, ds}  \right)
\end{equation}
and the key point here is whether the formal time-ordering is
solvable.

As already mentioned the interplay between NMR and ME has been
fruitful along the years and acted in both directions. To prove that
it is still alive we quote two recent papers directly dealing with
that mutual interaction. In \cite{vandersypen04ntf} the relevance of ME
through NMR for the new field of quantum information processing and
computing is envisaged. The authors of \cite{veshtort06sea} have recently
explored the fourth and sixth order of ME to design  a software
package for the simulation of NMR experiments. Although their
results are not yet conclusive their work shows the vitality of the
ME.

\subsection{Quantum Field Theory and High Energy Physics}

The starting point of any quantum field theory (QFT) calculation
is again equation (\ref{laecuacion}) which is conventionally
treated by time-dependent perturbation theory. So the first
question which arises is the connection between ME and Dyson-type
series. This has already been dealt with in subsection
\ref{sec2.4}. The main advantage of the first one is, as has
already repeatedly pointed out, that the unitary character of the
evolution operator is preserved at all orders of approximation. In
the historical development of QFT it was, however, Dyson approach
what was followed. The lost of unitarity was not thought to be of
great relevance in front of the problems presented by the
infinities appearing all over the place. Once renormalization idea
was introduced, this awful aspect of the theory was put also under
control. The results were, from the point of view of the
calculation of observable magnitudes, an unprecedented success:
the agreement between experimental results and their theoretical
counterparts was impressive.

So no wonder if alternatives to Dyson series, such as ME, did not
see popular acceptance. However, during the years there has been
interesting developments involving ME in the context of field
theory. In particular its use has shown to imply a re-ordering of
terms in the calculations in such a way that some infinities do
not appear and so make not necessary the introduction
of counterterms in the Hamiltonian. This is what happens for example in \cite%
{stefanovich01qft} where models are built in which ultraviolet
divergence appear neither in the Hamiltonian nor in the S-matrix.
In principle the results are valid for relativistic field theories
with any particle content and with minimal assumptions about the
form of the interaction.

ME as an alternative to conventional perturbation theory for
quantum fields has also been studied in \cite{dahmen86grf} where
normal products, Wick theorem and the like are used to deduce
graphical rules \textit{\`{a} la} Feynman for the terms $\Omega
_{i}$ for any value of $i$. This has been proved \cite
{dahmen88pao} helpful in the treatment of infrared divergences for
some QED processes such as the scattering of an electron on an
external potential or the bremsstrahlung of one hard photon, both
cases accompanied by the emission of an arbitrary number of soft
photons. An interesting feature of the ME based approach is that
the theory is free form infrared and mass divergences as a
consequence of the unitary character of the approximate
time-evolution operator \cite{dahmen88pao,dahmen82ido}. The method
is simpler than previous techniques based on re-summation of the
perturbation series to get rid of those divergences. Furthermore,
in contrast with the usual treatment, the resolution of the
detector is not an infrared regularization parameter.  An
application to Bhabha scattering (elastic electron-positron
scattering) is developed in \cite{dahmen91uaf}. The difficulties
of extending the results to Quantum Chromodynamics are commented
in \cite{dahmen86grf}.

Recently, an extension of the Magnus expansion 
has also been used in the context of Connes--Kreimer's Hopf algebra
approach to perturbative renormalization of quantum field theory 
\cite{connes00riqI,connes00riqII}. In particular,
in \cite{ebrahimi-fard08anb}, it is shown that this generalized ME allows one to
solve the Bogoliubov--Atkinson
recursion in this setting.

In the field of high energy physics ME has also found applications. Next
we just quote two instances: one referring to heavy ion collisions
and the other to elementary particle physics.

In collision problems the unitarity of the time evolution operator
imposes some bound on the experimentally observable cross
sections. When these magnitudes are theoretically calculated one
usually keeps only the lowest orders in conventional perturbation
theory. This may be harmless at relatively low energies but it may
lead to unitarity bounds violation as the energy increases. The
use of a manifestly unitary approximation scheme is then
necessary. ME provides such an scheme. In heavy ion collision at
sufficiently high energy and given kinematic configuration (small
impact parameter) that violation is produced for
\textit{e}$^{+}e^{-}$when analyzed
in the lowest-order time-dependent perturbation theory. In \cite%
{ionescu94uae} a remedy for this situation was advanced by the use
of first order ME. It is discussed how most theoretical approaches
are based on either lowest-order time-dependent perturbation
theory or the Fermi--Weizs\"aker--Williams method of virtual
photons. These approaches violate unitarity bounds for
sufficiently high collision energies and thus the probability for
single-pair creation exceeds unity. With some additional
assumptions a restricted class of diagrams associated with
electron-positron loops can be summed to infinite order in the
external charge. The electron-positron transitions amplitudes and
production probabilities obtained are manifestly unitary and gauge
invariant.

 In recent years there has been a great interest in neutrino oscillations and
its closely related solar neutrino problem. The known three
families of neutrinos with different flavors (electron, muon and
tau) were experimentally shown to be able of converting into each
other. The experiments were carried out with neutrinos of
different origins: solar, atmospheric, produced in nuclear
reactors and in particle accelerators. Here oscillation means that
neutrinos of a given flavour can, after propagation, change their
flavour. The accepted explanation for this phenomenon is that
neutrinos with a definite flavour have not a definite mass, and
the other way around. Let us denote by $\left| \nu _{\alpha
}\right\rangle $ the neutrinos of definite flavour with $\alpha $
the flavor index (i.e., electron, muon, tau) and by $\left| \nu
_{i}\right\rangle $ the neutrinos with well defined
distinct masses $m_{i},$ $i=1,2,3$ . Then the previous assertion means that $%
\left| \nu _{\alpha }\right\rangle $ will be a linear combination
of the different $\left| \nu _{i}\right\rangle. $

As neutrinos with different masses propagate with different
velocities this mixing allows for flavour conversion, i.e. for
neutrinos oscillations. ME enters the game in the solution of the
evolution operator in one basis. If one neutrino ``decouples''
from the other two then the problem reduces to one with only two
effective generations. Mathematically it is similar to the
two level system studied in Section 3. The reader is referred to \cite%
{dolivo92maf,dolivo90mea,dolivo96ntn,supanitsky08pee} for details of the calculations for two
and three generations.

\subsection{Electromagnetism}

The Maxwell equations govern the evolution of electromagnetic
waves. The equations are linear in its usual form and, although
they have been extensively studied, the complexity to obtain approximate
solutions is rather significant. When
reformulating the equations for a given problem, where the
geometry, the boundary conditions, etc. are considered and
appropriate discretisations are taken into account, it is frequent to end with
linear non-autonomous equations, so that 
the Magnus expansion can be of interest here.

To illustrate some possible applications, let us consider the
Maxwell equations
\begin{equation}\label{Maxwell}
  \left\{ \begin{array}{l} \displaystyle
    \frac{\partial \mathbf{H}}{\partial t} = - \frac{1}{\mu} \, \nabla
    \times \mathbf{E} \\
     \displaystyle
     \frac{\partial \mathbf{E}}{\partial t} = \frac{1}{\varepsilon} \, \nabla
    \times \mathbf{H} - \frac{1}{\varepsilon} \mathbf{J}(t)
  \end{array}  \right.
\end{equation}
where $\mathbf{H},\mathbf{E},\mathbf{J}$ are the magnetic and
electric field intensities, and the current density, respectively,
$\mu$ is the permeability and $\varepsilon$ is the permittivity.
After space discritisation, these equations turn into a large
linear system of non-homogeneous equations and Magnus integrators can in principle be applied. 
 Time
dependent contributions can also appear from boundary conditions
or external interactions. In some cases 
$\mathbf{J}=\sigma \, \mathbf{E}$ \cite{botchev08nid,jiang06sfd}, so that, 
if the conductivity $\sigma$ is not constant, Magnus
integrators can be useful.

Let us now consider the frequency domain Maxwell equations (with
$\mathbf{J}=\mathbf{0}$)
\begin{equation}\label{Maxwell2}
  \left\{ \begin{array}{l} \displaystyle
    \nabla \times \mathbf{E}  = i w \mu \mathbf{H}\\
     \displaystyle
    \nabla \times \mathbf{H}  = -iw\varepsilon \mathbf{E}
  \end{array}  \right.
\end{equation}
where $w$ is the angular frequency. These equations are of
interest for time-harmonic lightwaves propagating in a
wave-guiding structure composed of linear isotropic materials. If
one is interested in the $x$ and $y$ components of $\mathbf{H}$
and $\mathbf{E}$, and how they propagate in the $z$ direction, the
equations to solve, after appropriate discretisation, take the
form  \cite{lu06stf}
\begin{equation}\label{Maxwell2b}
    -iw\varepsilon \frac{d \mathbf{u}}{dz} = A(z) \mathbf{v}, \qquad
    -iw\mu \frac{d \mathbf{v}}{dz} = B(z) \mathbf{u}.
\end{equation}
Here $\mathbf{u}$, $\mathbf{v}$ are vectors and $A,B$ matrices
depending on $z$, which in this case play the role of evolution
parameter. A fourth-order Magnus integrator has been used in
\cite{lu06stf}. From section \ref{section5} we observe that higher order Magnus
integrators, combined with splitting methods could also lead to
efficient algorithms to obtain accurate numerical results with
preserved qualitative properties of the exact solution.


\subsection{Optics}

In the review paper \cite{dattoli88aav} one can find references to
some early applications of ME to Optics. For example to Hamiltonians
involving the generators of  $\mathrm{SU(2)}$, $\mathrm{SU(1,1)}$ and
Heisenberg--Weyl groups with applications to laser-plasma scattering
and pulse propagation in free-electron lasers. Here as
representatives of the more modern interest of ME in Optics we quote
two applications referring to Helmholtz equation and to the study of
Stokes parameters.

Helmholtz equation in one spatial dimension with a variable
refractive index $n(x)$ reads
\begin{equation}   \label{Helmholtz}
\psi^{\prime\prime}(x)+k^2n^{2}(x)\psi(x)=0,
\end{equation}
where $k$ is the wavenumber in vacuum.

Recently this time-honored wave equation has been treated in two
different ways, both using ME. From a more formal point of view, in
\cite{khan05wdm} Helmholtz equation is analyzed following the well
known procedure followed by Feshbach and Villars to convert the
second order relativistic quantum Klein--Gordon differential
equation for spin-0 particles in a first order differential equation
involving two components wave functions (the original wave function
and its time derivative). The evolution operator for Helmholtz
equation is then a $2\times2$ matrix which evolves according to the
fundamental equation (\ref{laecuacion}) with the only difference that
now the evolution parameter is $x$ instead of $t$. In
\cite{khan05wdm} the whole procedure is explained and the main
physical consequence, which amounts to the addition of correcting
terms to the Hamiltonian, is discussed in the case of an axially
symmetric graded-index medium, i.e. one in which the refractive index
is a polynomial.

Helmholtz equation has also been investigated with the help of ME
in \cite{lu06afo,lu06stf,lu07afo}. Here the propagation in a
slowly varying waveguide is considered and the boundary value
problem is converted into an initial value problem by the
introduction of appropriate operators which are shown to satisfy
equation (\ref{laecuacion}). Numerical methods to fourth order
borrowed from \cite{iserles99ots} are then used.

 Since mid 19th century the polarization
state of light, and in general electromagnetic radiation or any
other transverse waves, is described by the so called Stokes
parameters which constitute a four-dimensional vector
$\mathbf{S}(\omega)$ depending on the frequency $\omega$. When the
light traverses an optical element which acts on its polarization
state the in and out Stokes vectors are related by
\begin{equation} \label{mueller}
\mathbf{S}_{\mathrm{out}}(\omega)=M(\omega)\mathbf{S}_{\mathrm{in}}(\omega),
\end{equation}
where the $4\times4$ matrix $M(\omega)$ is called the Mueller
matrix. It can be proved \cite{reimer06cao,reimer06mmd} that it
satisfies the equation
\begin{equation}  \label{muellerdif}
M^{\prime}(\omega)=H(\omega) M(\omega),\qquad\qquad
M(\omega_0)=M_{0},
\end{equation}
where now the prime denotes derivative with respect to the real
independent variable $\omega$. For systems with zero
polarization-dependent loss (PDL) and no polarization mode
dispersion (PMD) $H(\omega)$ is constant whereas with PDL and PDM
the previous equation is just our equation (\ref{laecuacion}) and
the appropriateness of ME is apparent. The matrix $H(\omega)$ in
this application has an Hermitian and a non-Hermitian component.
ME has allowed a recursive calculation of successive orders of the
frequency variation of the Mueller matrix. This yields PMD and PDL
compensators that counteract the effects of PMD and PDL with
increased accuracy.

Also related to the use of Stokes vector one can mention the
so-called radiative transfer equation for polarized light. It is
relevant in Astrophysics to measure the magnetic fields in the Sun
and stars. That equation gives the variation of $\mathbf{S}(z)$
with the light path $z$
\[
\frac{d}{dz}\mathbf{S}(z)=-K(z)\mathbf{S}(z) +\mathbf{J},
\]
where $\mathbf{S}$ is the Stokes vector, $K$ is a $4\times 4$
matrix which describes absorption in the presence of Zeeman effect
and $\mathbf{J}$ stands for the emission term. In
\cite{lopez02ota,lopez99dia,semel99iot} ME is used to obtain an
exponential solution.

\subsection{General Relativity}

To illustrate once more the pervasive presence of the linear
differential equation (\ref{laecuacion}) let us mention reference
\cite{sanmiguel07ndo} in which the aim is to determine the time
elapsed between two events when the space-time is treated as in
General Relativity. Then it turns out to be necessary to solve a
two-point boundary value problem for null geodesics. In so doing
one needs to know a Jacobian whose expression involves
a~$8\times8$ matrix function obeying the basic equation
(\ref{laecuacion}). In \cite{sanmiguel07ndo} an eighth order
numerical method from \cite{blanes00iho} is used, which is proved
to be an efficient scheme.

\subsection{Search of Periodic Orbits}

The search of periodic orbits for some non-linear differential
autonomous equations, ${\bf x}'={\bf f(x)}$, $\mathbf{x}\in \mathbb{R}^d$
is of interest in Celestial Mechanics (periodic orbits of the
N-body problem) as well as in the general theory of dynamical systems.
Due to the complexity of this process, it is important to have
efficient numerical algorithms.

  The Lindstedt--Poincar\'e technique is frequently used to calculate periodic
orbits. An iterative process proposed in \cite{viswanath03sda} consists in 
starting with a guessed periodic orbit, and this guess is subsequently improved
by solving a correlation non-autonomous linear differential equation. The
numerical integration of this equation is carried out by
means of Magnus integrators.


\subsection{Geometric control of mechanical systems}
 \label{sec-control}

 Many mechanical systems studied in control theory
 can be modeled by an ordinary differential equation of the form  \cite{bullo05gco,kawski02tco}
 \begin{equation}   \label{control.1}
      \mathbf{x}^\prime(t) = \mathbf{f}_0(\mathbf{x}(t)) + \sum_{i=1}^m u_i(t) \;
\mathbf{f}_i(\mathbf{x}(t)),
\end{equation}
initialized at $\mathbf{x}(0) = \mathbf{p}$. Here $\mathbf{x} \in
\mathbb{R}^d$ represents all the possible states of the system,
$\mathbf{f}_i$ are (real) analytic vector fields and the function
$\mathbf{u} = (u_1, \ldots, u_m)$ (the \emph{controls}) are assumed
to be integrable with respect to time and taking values in a compact
subset $U \subset \mathbb{R}^m$. The vector field $\mathbf{f}_0$ is
called the \emph{drift} vector field, whereas $\mathbf{f}_i$, $i \ge
1$, are referred to as the \emph{control} vector fields. When
$\mathbf{f}_0 \equiv 0$, the system (\ref{control.1}) is called
`without drift', and its analysis is typically easier.

For a given set of controls $\{u_i\}$, equation (\ref{control.1})
with initial value $\mathbf{x}(0)$ is nothing but a dynamical
system, which can be analyzed and (approximately) solved by
standard techniques. In control theory, however, one is interested
typically in the \emph{inverse problem}: given a target
$\mathbf{x}(T)$, find controls $\{u_i\}$ that steer from
$\mathbf{x}(0)$ to $\mathbf{x}(T)$ \cite{kawski02tco}, perhaps by following a prescribed
path. Just to
illustrate these abstract considerations, a typical problem could be
the determination of a set of controls that drive the actions of a
robot during a task.

The first step is to guarantee that there exists a solution. This is
the problem of controllability. To characterize the controllability
of linear systems of the form
\begin{equation} \label{LQR}
 {\bf x}'(t) =  A(t) {\bf x}(t) + B(t) {\bf u}(t)
\end{equation}
is a relatively simple task thanks to an algebraic criterion known
as the Kalman rank condition \cite{bullo05gco,kalman63col}. This issue, however,
is much more involved for the nonlinear system (\ref{control.1})  \cite{kawski02tco}.

The interest of ME in control theory, as it has been already
discussed in section \ref{MECF}, stems from the approximate ansatz
it provides connecting the states $\mathbf{x}(0)$ and
$\mathbf{x}(T)$. Thus, the Magnus expansion can be used either to predict a state
$\mathbf{x}(T)$ for a given control $\mathbf{u}$ or to find
reasonable controls $\mathbf{u}$ which made reachable the target
$\mathbf{x}(T)$ from $\mathbf{x}(0)$. Of course, many sets of
controls may exist and it raises questions concerning the
\emph{cost} of every scheme and consequently the search for the
optimal choice. For instance, the ME has been used in non-holonomic
motion planning of systems without drift
\cite{duleba97lom,duleba98oac}. Among non-holonomic systems there are free-floating
robots, mobile robots and underwater vehicles \cite{duleba98oac,murray94ami}.

In the particular case of linear quadratic  optimal
control problems (appearing in in engineering problems as well as
in differential games) a given cost functional has to achieve a
minimum. When this happens, eq. (\ref{LQR}) can be written as
\cite{engwerda05lqd,reid72rde}
\begin{equation}\label{fi}
{\bf x}' \, = \, M(t,K(t)) {\bf x},
\end{equation}
where $ \, K(t) \in \mathbb{R}^{d \times d} \, $ has to solve a
Riccati differential equation similar to (\ref{Riccati}) with
final condition $K(T)=K_f$. In other words,  the Riccati equation has to be
integrated backward in time and then to use it as an input in
(\ref{fi}). As mentioned, the Riccati differential matrix equation
has received much attention
\cite{blanes00asw,dieci92nio,engwerda05lqd,jodar88arm,kenney85nio,reid72rde,schiff99ana},
but an efficient implementation to this problem requires further
investigation, and methods from ME can play an important role.


\section{Conclusions}

In this report we have thoroughly reviewed the abiding work on
Magnus expansion carried out during more than fifty years from very
different perspectives.

As a result of a real interdisciplinary activity some aspects of
the original formulation have been refined. This applies for
example to the convergence properties of the expansion which have
been much sharpened.

In other features much practical progress has been made. This is
the case of the calculation of the terms of the series both
explicitly or recurrently. New techniques, like the ones borrowed
from graph theory, have also profitably entered the play.

Although originally formulated for linear systems of ordinary
differential equations, the domain of usage of ME has enlarged to
include other types of problems with differential equations:
stochastic equations, nonlinear equations or Sturm-Liouville
problems.

In parallel with this developments in the mathematical structure
of the ME the realm of its applications has also widen with the
years. It is worth stressing in this respect the versatility of
the expansion to cope with new applications in old fields, like
NMR for instance, and at the same time its capability to generate
new contributions, like the generation of efficient numerical
algorithms for geometric integrators.

All these facts, historical and present, presented and discussed
in this report strongly support the idea that ME can be a very
useful tool for physicists.



\subsection*{Acknowledgments}

One of us (J.A.O) was introduced to the field of Magnus expansion
by Prof. Silvio Klarsfeld (Orsay) to whom we also acknowledge his
continuous interest in our work. Encouraging comments by Prof.
A.J. Dragt in early stages of our involvement with ME are deeply
acknowledged.  During the last decade  we have benefitted from
enlightening discussions with and useful comments by Prof. Arieh
Iserles. We would like to thank him for that as much as for the
kind hospitality extended to us during our visits to his group in
DAMTP in Cambridge. We also thank Prof. Iserles for providing us
with  his  Latex macros for the graphs in section \ref{graph} and Dr. Ander Murua
for reading various parts of the manuscript and providing useful feedback.
Friendly collaboration with Dr. Per Christian Moan is also
acknowledged.

The work of SB and FC has been partially supported by Ministerio
de Ciencia e Innovaci\'on (Spain) under project MTM2007-61572
(co-financed by the ERDF of the European Union), that of JAO and
JR by contracts MCyT/FEDER, Spain (Grant No. FIS2004-0912) and
Generalitat Valenciana, Spain (Grant No. ACOMP07/03). 
SB also aknowledges the support of the UPV trhough the project 20070307


\bibliographystyle{plain}
\bibliography{ourbib,geom_int,numerbib}




\end{document}